\documentclass[12pt]{article}
\usepackage{fullpage}
\usepackage{amssymb, amsmath, amsthm}
\usepackage{natbib}
\usepackage{setspace}
\usepackage{multirow}
\usepackage{titling}
\usepackage{graphicx}
\usepackage{caption}
\usepackage{subfigure}
\usepackage{algorithm}
\usepackage{algpseudocode}
\usepackage{enumitem}
\usepackage{url}

\usepackage{natbib}

\usepackage{color}

\usepackage[colorlinks,citecolor=blue,urlcolor=blue]{hyperref}
\RequirePackage{hypernat}

\usepackage{imakeidx}
\makeindex

\usepackage{booktabs}

\newcommand{\argmax}{\operatorname*{arg \ max}}
\newcommand{\argmin}{\operatorname*{arg \ min}}

\usepackage[usenames,dvipsnames]{xcolor}

\newcommand{\EE}{{\mathbb E}}
\newcommand{\PP}{{\mathbb P}}

\newcommand{\LL}{{\mathbb L}}
\newcommand{\GG}{{\mathbb G}}

\newcommand{\RR}{\mathbb{R}}

\newcommand{\FF}{{\mathbb F}}

\newcommand{\NN}{\mathbb{N}}

\newcommand{\indep}{\perp \!\!\! \perp}

\newcommand{\mc}[1]{\mathcal {#1}}

\theoremstyle{plain}

\newtheorem{theorem}{Theorem}
\newtheorem{proposition}{Proposition}

\newtheorem{lemma}{Lemma}

\newtheorem{assumption}{Assumption}

\newtheorem{assumptionalt}{Assumption}[assumption]

\def\spacingset#1{\renewcommand{\baselinestretch}%
{#1}\small\normalsize} \spacingset{1}

\newenvironment{assumptionp}[1]{
  \renewcommand\theassumptionalt{#1}
  \assumptionalt
}{\endassumptionalt}
\makeatletter
\newcommand{\myitem}[1]{%
\item[#1]\protected@edef\@currentlabel{#1}%
}
\makeatother

\title{Doubly robust estimation and inference for a log-concave counterfactual density}
\author{Daeyoung Ham \& Ted Westling \& Charles R. Doss\\
School of Statistics, University of Minnesota\\
Department of Mathematics and Statistics, University of Massachusetts Amherst\\
School of Statistics, University of Minnesota}
\date{}

\begin{document}
\maketitle

\begin{abstract}
We consider the problem of causal inference based on observational data (or the related missing data problem) with a binary or discrete treatment variable. 
In that context, we study inference for the counterfactual density functions and contrasts thereof, which can provide more nuanced information than counterfactual means and the average treatment effect.  
We impose the shape-constraint of log-concavity, a type of unimodality constraint, on the counterfactual densities, and then develop doubly robust estimators of the log-concave counterfactual density based on augmented inverse-probability weighted pseudo-outcomes. We provide conditions under which the estimator is consistent in various global metrics. 
We also develop asymptotically valid pointwise confidence intervals for the counterfactual density functions and differences and ratios thereof, which serve as a building block for more comprehensive analyses of distributional differences.
We also present a method for using our estimator to implement density confidence bands.
\end{abstract}

\doublespacing

\section{Introduction}

A common approach to comparing two distributions is to compare their means. 
In the context of causal inference, this leads to comparing the mean outcome under assignment of all units in a population to specific treatment levels. 
For instance, the average treatment effect (ATE) represents the difference between the mean outcome had all units been assigned to treatment and the mean outcome had all units been assigned to control. 

Since the mean is a coarse summary of a distribution, comparing distributions by comparing their means can fail to capture important information. 
For example, a small difference between means can be due to a small shift of the entire distribution or a larger shift on a small subset of the domain. 
Furthermore, two distributions can have the same mean but be qualitatively distinct. 
Hence, there is value in going beyond comparing means, and instead comparing entire distributions.

In this paper, we consider estimation and inference for the {\it counterfactual density function}, which is defined as the density of the potential outcome in the Neyman-Rubin causal model. 
Comparing the counterfactual densities under treatment and control can provide more nuanced conclusions than can be achieved by focusing on the counterfactual means. 

Nonparametric estimation and inference for density functions is more challenging than for means. 
Nonparametric density estimation usually involves careful selection of one or more tuning parameters, which can be difficult.  
In addition, due to the bias of nonparametric density estimators, forming confidence intervals (CI's) with good coverage can be challenging. 
Direct bias correction requires further tuning parameter choice(s) to estimate higher order derivatives of the density \citep{gasser1984estimating, Seijo.Sen.2011}, and properties of the resulting CI's depend on the local smoothness of the density \citep{Calonico-Cattaneo-Farrell_2018-JASA}.

In many problems, methods of estimation and inference become available that can avoid these challenges if the true function is known to satisfy a shape constraint such as monotonicity or convexity.  
Here, we will focus on the {\it log-concave} shape constraint.
Log-concave densities are unimodal and light-tailed, and so when these two characteristics are satisfied it is reasonable to consider using log-concavity. Furthermore, since many common parametric families, such as the normal family and exponential family, are log-concave, log-concavity can serve as a nonparametric generalization of these parametric families.  
In the non-causal setting, log-concavity allows for fully automatic estimation without the need to rely on or select tuning parameters (see e.g., \citealp{MR2459192}, \citealp{RufibachThesis}, \citealp{DR2009LC}), unlike many other nonparametric estimators.
Shape constraints  also  yield more efficient estimation when the shape assumption holds \citep{Birge:1989jn}.
In addition, using the log-concavity assumption, methods have been developed to form asymptotically valid CI's for densities without explicit bias correction. 
In the non-causal setting, \cite{Doss:2016ux,DosWel2019} and \cite{deng2022inference} developed methods for CI's based on log-concave estimators that do not rely on the selection of tuning parameters or on estimation of unknown limit distribution parameters, so that both the estimation and inference procedure are fully automatic. 
We refer the reader to \cite{recentprosamworth} for an in-depth review of nonparametric inference for log-concave densities and for additional references.

Motivated by these successes, in this paper we consider inference for the counterfactual density function under a log-concavity constraint. 
Estimation and inference for the counterfactual density function is more challenging than that for an ordinary density because in the causal setting, we need to adjust for potential confounding variables. 
As we will discuss in greater detail in Section \ref{section:Est}, our approach is built on a locally efficient, doubly robust estimator of the cumulative distribution function (CDF), which requires estimation of two nuisance functions: the conditional distribution functions and propensity score functions.
To use the log-concave projection operator (defined below) on a CDF, that CDF must be bona-fide, meaning it must be non-decreasing, which is not necessarily true for our initial covariate-adjusted CDF estimator.  
Thus, we isotonize the initial CDF estimator, and then proceed to define the log-concave density estimator.

To show  our asymptotic results for the log-concave density estimator, we need to show that our covariate-adjusted and isotonized CDF estimator is consistent in Wasserstein distance, which requires showing that nuisance function estimation error and the effect of isotonization are negligible.
The latter is done in two key lemmas, Lemmas \ref{lem1} and \ref{lem2} in Section \ref{section:consistency}, which may be of independent interest. 
To our knowledge, this is the first uniform convergence result in the literature on nonparametric counterfactual density estimation.
We further prove that our density estimator has the same pointwise convergence rate as in the non-causal setting \citep{BRW09}, by again showing the asymptotic negligibility of the isotonization and nuisance function estimation for the local behavior. 
However, as opposed to the non-causal limit distributions, we find that  the limit distribution in the causal setting has a scaling
factor that depends on the nuisance functions. 
We provide a doubly robust procedure for estimating the scaling factor. 
In addition, we provide asymptotically valid pointwise CI's based on the ideas in \cite{deng2022inference}. Leveraging the result that the limiting distributions of the counterfactual density estimates for two treatment levels are independent, we further propose asymptotically valid CI's for contrasts such as differences or ratios of counterfactual densities.

There have been several papers discussing uses of log-concavity, concavity, and convexity in a semiparametric setting. \cite{SY.ICA.LC} studied the use of log-concavity in independent component analysis, \cite{chen2015semiparametric} studied the use of log-concavity in time series models, and \cite{kuchibhotla2023semiparametric} studied the use of convexity in the single index regression model.  However, to the best of our knowledge, ours are the first results that either give the limit distribution or form confidence  intervals for a log-concave density in a semiparametric setting, as well as the first doubly-robust asymptotically valid CI's for the counterfactual density function in the causal setting.

Recently, \citet{kim2018causal} and \citet{KBWcounterfactualdensity} considered the problem of counterfactual density estimation.  
\citet{kim2018causal} developed an estimator based on kernel smoothing, and \citet{KBWcounterfactualdensity} developed estimators that are projections of the empirical distribution onto a parametric space based on $L_p$ distances and divergences.  Both approachesrely on careful selection of a tuning parameter: a bandwidth in \citet{kim2018causal}, \citet{KBWcounterfactualdensity}. In both papers, the tuning parameter was selected by minimizing an estimated pseudo-risk, which can be highly sensitive to the estimated risk function. More importantly, the CI's and confidence bands are valid for {\it smoothed} or  {\it projected} versions of the densities  and density effects. As a result, the proposed $(1-\alpha)$-level CI's do not necessarily have finite-sample or asymptotic $1-\alpha$ coverage for the actual target estimand.

As in the non-causal setting, our proposed log-concave estimator offers several advantages over smoothing-based methods. First, the CI's are asymptotically valid for the counterfactual densities and density contrasts rather than smoothed versions thereof. Second, our methods do not require careful tuning parameter selection, and are \textit{tuning-parameter-light} compared to existing methods. Our estimator is automatic up to the black-box nuisance estimators, which are also required for the methods of  \cite{kim2018causal} and \cite{KBWcounterfactualdensity}, and a grid for the monotone correction of the covariate-adjusted CDF estimator. This grid is needed to approximate a continuous isotonic projection \citep{groeneboom2010generalized}, and can be made arbitrarily fine without changing the asymptotic properties of the estimator (see Theorem \ref{prop:CI}), unlike bandwidths and basis functions, which require careful balancing of bias and variance. 
Our CI's require estimation of an additional nuisance parameter, for which we rely on a bandwidth. While the precise behavior of this nuisance estimator depends on the choice of bandwidth, asymptotic validity of the CI only relies on consistency of the estimator. We demonstrate in Lemma \ref{lemma:tuning.consistency} that a broad range of bandwidths yield a consistent estimator, and we show in simulations that our CI's are extremely robust to the choice of the bandwidth. The CI's for density contrasts use the same nuisance estimator as the CI's for the individual densities, so that no additional tuning is required.

The paper proceeds as follows.
In Section~\ref{setupnotation}, we define notation we will use and introduce our causal parameter of interest and its identification in the distribution of the observed data.
In Section~\ref{section:onestepCDF}, we introduce a doubly robust one-step estimator of the counterfactual CDF.  Since this estimator may not be monotonically increasing, in Section~\ref{section:correctedestimator}, we define a monotone correction procedure to ensure that our counterfactual CDF estimator is a proper CDF. In Section~\ref{subsec:LCCDE}, we apply the log-concave projection operator introduced by \citet{dss2011} to the monotonically corrected CDF estimator to obtain a log-concave counterfactual density estimator. 
In Section~\ref{section:consistency}, we provide conditions under which our estimator is uniformly consistent. In Section \ref{subsec:limitdist}, we provide conditions under which our estimator converges pointwise in distribution.
In Section~\ref{subsec:CIconst}, we use these results to develop pointwise and difference/log-ratio CI's and provide conditions under which the CI's are asymptotically valid.
In Section~\ref{subsec:band.const}, we propose a method for constructing a confidence band for a counterfactual density by adapting the approach of \citet{walther2022confidence} to our setting.
Furthermore, to the best of our knowledge, these results provide the first asymptotically valid counterfactual density difference/log-ratio CI's detecting distributional differences at a point.
We also introduce an estimator, pointwise/difference/log-ratio CI's, and confidence bands based on sample splitting, and provide analogous theoretical results for these methods.
Its introduction is postponed to Section~\ref{section:sample.splitting} of the supplementary material.
In Section \ref{Section:Simul} we present simulation studies assessing the finite-sample performance of our methods, and in Section~\ref{section:real.data.analysis} of the supplementary material we illustrate our method using data on the impact of a job training program on future earnings. 
Proofs of all theorems and lemmas, as well as additional simulation results, can be found in supplementary material.

\section{Causal setup}\label{setupnotation}
\subsection{Notation}\label{notation}

We assume that we observe $n$ independent and identically distributed samples $(\mathbf{Z}_1,\dots, \mathbf{Z}_n)$ of the generic tuple $\mathbf{Z}=(\mathbf{\mathbf{X}},A,Y)$ with support
$\mathcal{Z}=\mathcal{X}\times\mathcal{A}\times\mathcal{Y}$. 
Here, $\mathbf{X}\in \mathcal{X} \subseteq \mathbb{R}^d$ denotes a $d$-dimensional vector of potential confounders, $A\in\mathcal{A}=\{0,1\}$ a binary treatment, and $Y\in \mathcal{Y} \subseteq \mathbb{R}$ the outcome of interest. 
We let $P_*$ denote the true distribution of $\mathbf{Z}$. We let $\PP(C)$ and $\EE(\mathbf{X})$ denote the probability of an event $C$ and expectation of a random variable $\mathbf{X}$, which are with respect to $P_*$ unless otherwise noted. 
We then define $\eta_*(y|\mathbf{x},a):=\frac{\partial}{\partial y}\PP(Y\le y|\mathbf{\mathbf{X}}=\mathbf{x},A=a)$ as the true conditional density of $y$ given $X$ and $A$, and $\pi_*(a|\mathbf{x})=\PP(A=a|\mathbf{\mathbf{X}}=\mathbf{x})$ as the true propensity score function. 
We also use the notation $\eta_{a,*}(y|\mathbf{x}):=\eta_*(y|\mathbf{x},a), \pi_{a,*}(\mathbf{x}):=\pi_*(a|\mathbf{x})$ for convenience.

For a real-valued function $f$ and any measure $Q$ on $\mc{Z}$, we let $Q\{f(\mathbf{Z})\} \equiv Q f(\mathbf{Z}) :=\int_{\mathcal{Z}}f(\mathbf{z})dQ(\mathbf{z})$.
We use $\PP_n$ to denote the empirical measure of the data, so that $\PP_n\{f(\mathbf{Z})\}:=n^{-1}\sum_{i=1}^n f(\mathbf{Z}_i)$.
We let $\|f\|_p:=\{\int f(\mathbf{Z})^p dP_*(\mathbf{z})\}^{1/p}$ denote the $L_p(P_*)$ norm and $\|g\|_{\infty}=\sup_{x\in U}|g(x)|$ denote the supremum norm for a generic function $g$ and its domain $U$. 
We set $\|f\| := \| f\|_2$. We define the class multiplication of two classes of functions $\mc F$ and $\mc G$ as $\mc F \cdot \mc G:=\{fg:f\in \mc F, g\in \mc G\}$.
We let $\rightarrow_p$ and $\rightarrow_d$ denote convergence in probability and distribution, respectively, with respect to $P_*$, and we use ``a.s." for almost sure convergence and ``a.e." for almost everywhere with respect to $P_*$. 
We will say that $a \lesssim b$ if there exists a $c\in(0,\infty)$ such that $a\leq cb$, and that $a_n \lesssim b_n$ if there exists a $c\in(0,\infty)$ and $N$ such that $a_n\leq cb_n$ for all $n\geq N$.
For random vectors $U, V, W$, we let $U \indep V | W$ mean that $U$ and $V$ are conditionally independent given $W$. 
We let $I(A)$ denote a generic indicator function with an arbitrary statement $A$ which returns 1 if the statement $A$ is true, and returns 0 otherwise, and we let $I_{A}(x) :=I(x\in A)$.

\subsection{Causal parameter and its identification}\label{identification}

We let $Y^a$ be the Neyman-Rubin counterfactual outcome under assignment to treatment level $a$ \citep{Rubin1974EstimatingCE}. Throughout, we assume the Stable Unit Treatment Value Assumption (SUTVA); i.e., there is a unique version of treatment and control, and each unit's potential outcomes do not depend on any other units' treatment assignments.  
Our causal estimands of interest are the density functions $p_a$ of $Y^a$ for $a\in\{0,1\}$, which are known as the counterfactual density functions. 
Since we do not observe both the potential outcomes for each unit in the population, in order to estimate $p_a$ using the observed data, we first need to identify $p_a$ with a functional of the distribution $P_*$ of the observed data. To do so, we make the following assumptions.

\begin{assumptionp}{I} \label{assm:I} \phantom{blah}
\begin{enumerate} 
\myitem{(I1)} \label{assm:consistency} Consistency: $A=a$ implies that $Y=Y^a$ for $a \in \{0,1\}$.
\myitem{(I2)} \label{assm:noumconf}No unmeasured confounding: $(A\indep Y^a ) |\mathbf{\mathbf{X}}$.
\myitem{(I3)} \label{assm:positivity} Positivity: $\mathbb{P}(A=a|\mathbf{\mathbf{X}})>\epsilon_0$ almost surely for some $\epsilon_0>0$.
\myitem{(I4)} \label{assm:logconcv}Log-concavity: $p_a$ is a log-concave density for $a \in \{0,1\}$; i.e., there exist $\phi_a :\mathbb{R} \to \mathbb{R}$ for $a \in \{0,1\}$ such that $p_a = e^{\phi_a}$ and $\phi_a$ is concave. 
\end{enumerate}          
\end{assumptionp}

Assumptions \ref{assm:consistency}--\ref{assm:positivity} are causal assumptions which are commonly employed in the causal inference literature \citep{robins1986new}. 
Assumption \ref{assm:consistency} links the observed and potential outcomes. 
Assumption \ref{assm:noumconf} says that the potential outcome and the assigned treatment level are conditionally independent within strata of confounders; for this to hold, sufficiently many confounders must be collected. 
The positivity assumption ensures that each subject has a chance of having each treatment level $a$ regardless of confounder values. 
Under the Assumptions \ref{assm:consistency}--\ref{assm:positivity}, we have 
$p_a(y):=\int_{\mathcal{X}}\eta_*(y|\mathbf{x},a)dP_{\mathbf{\mathbf{X}}}(\bf x)$
for $a\in\{0,1\}$ (see Section 2 in \citealp{KBWcounterfactualdensity}).  
Assumption \ref{assm:logconcv} imposes the log-concave constraint on the counterfactual densities.   
We note that assumption \ref{assm:logconcv} is not needed for identification of the counterfactual density. 
However, any continuous log-concave density has a uniformly continuous CDF, since log-concavity guarantees that $p_a$ exists and is unimodal.  
Many commonly employed univariate densities are log-concave, including uniform, normal, logistic, Weibull with shape parameter $\ge 1$, Beta$(a,b)$ with $a,b\ge 1$, Gamma with shape parameter $\ge 1$, $\chi^2$ distribution with degrees of freedom $\ge 2$, and Laplace densities. 
Hence, log-concavity of $p_a$ is a nonparametric generalization of many commonly employed parametric models. 
As noted in the introduction, log-concavity is a reasonable assumption when the density of the outcome is expected to be unimodal and have sub-exponential tails.

We note that we could weaken Assumption~\ref{assm:logconcv} to allow the possibility of model misspecification by instead defining our causal parameter of interest as the Kullback-Leibler projection of the true counterfactual density onto the space of log-concave  densities \citep{dss2011}.  
We expect that our theoretical results would continue to hold for the log-concave projection in this case, but for clarity and simplicity we assume the log-concave assumption holds throughout the paper.

\section{Estimation}\label{section:Est}

\subsection{One-step CDF estimator}\label{section:onestepCDF}

In this section, we define our estimator of the log-concave counterfactual density function $p_a$. 
First, we introduce a doubly-robust one-step estimator of the counterfactual CDF using its efficient influence function. 
This one-step estimator is not guaranteed to be monotonic or contained in $[0,1]$, so we next define a correction procedure to enforce these constraints. 
Finally, we project this CDF estimator onto the space of log-concave distributions.

We define the counterfactual CDF under treatment level $a$ as $F_a(s) :=\int_{-\infty}^s p_a(y) dy =\mathbb{E}[I(Y^a\le s)]$. 
The efficient influence function $B_{a,\theta_{a,*}}$ of $F_a(s)$ relative to a nonparametric model is given by $D_{a,\theta_{a,*}}(s) - F_a(s)$ \citep{KBWcounterfactualdensity} for
\begin{align}
D_{a,\theta_{a,*}}(s)(\mathbf{Z}) ={I(A=a) \over \pi_{a,*}(\mathbf{\mathbf{X}})}\Big[I(Y\le s) -\phi_{a,*}(s|\mathbf{\mathbf{X}}) \Big]+\phi_{a,*}(s|\mathbf{\mathbf{\mathbf{X}}}),\label{term:Dthetastar}
\end{align}
where $\theta_{a,*}=(\phi_{a,*},\pi_{a,*})$ for $\phi_{a,*}(s|\mathbf{x}):=\int_{-\infty}^s \eta_{a,*}(y| \mathbf{x})dy$ the conditional distribution function of $Y^a$ given $\mathbf{X} = \mathbf{x}$ and $\pi_{a,*}(\mathbf{x}) := \PP(A = a| \mathbf{X} = \mathbf{x})$ the propensity score function. 

We use the efficient influence function to construct a \textit{one-step estimator} of $F_a(s)$ for each $s \in \mathbb{R}$ \citep{bickel1982adaptive, pfanzagl1982contributions}.
To do so, we let $\widehat{\phi}_{a}$ and $\widehat{\pi}_{a}$ be estimators of $\phi_{a,*}$ and $\pi_{a,*}$, respectively. 
We do not specify a particular form or method for estimating these nuisance parameters, but rather allow the user to choose their preferred method in a given problem setting.  
We provide high-level conditions on the complexity and rates of convergence of these nuisance estimators needed for our theoretical results in later sections. 
We then define
\begin{align}
   \widehat{F}_{a,n}(s)&=\frac{1}{n}\sum_{i=1}^n \left\{\frac{I(A_i=a)}{\widehat\pi_a(\mathbf{X}_i)}\Big[I(Y_i\le s)-\widehat\phi_a(s|\mathbf{X}_i)\Big]+\widehat\phi_a(s|\mathbf{X}_i) \right\},\label{form:onestepest}
\end{align}
as our one-step estimator of $F_a$ for each $a \in \{0,1\}$ and $s \in \mathbb{R}$.

\subsection{Monotone correction of the one-step estimator}\label{section:correctedestimator}

The one-step estimator of the counterfactual CDF $\widehat{F}_{a,n}(s)$ given in \eqref{form:onestepest} is not guaranteed to be a proper distribution function: it may not be contained in $[0,1]$, it may not be monotone, and it may not converge to 0 as $s \to -\infty$ and 1 as $s \to \infty$. 
This poses a problem because we ultimately aim to project the estimator onto the class of log-concave distributions, but to the best of our knowledge, this projection operation is currently only defined for proper distribution functions. 
Furthermore, even if it could be extended to an appropriate domain, properties of the log-concave projection are currently only known for proper distribution functions. 
In this section, we remedy this problem by defining a corrected version of $\widehat{F}_{a,n}$ that is a proper distribution function so that the log-concave projection can be applied.

Our correction procedure has three steps: projection onto $[0,1]$, projection onto monotone functions over a finite grid, and piecewise constant interpolation. 
For the first step, we define
\begin{equation}
    \widehat{F}^{c_0}_{a,n}(s)=  \widehat{F}_{a,n}(s)I\left(0<\widehat{F}_{a,n}(s)<1\right)+I\left(\widehat{F}_{a,n}(s)\ge 1\right),\label{def:projection01}
\end{equation}
for all $x\in\RR$ as the projection of $\widehat{F}_{a,n}$ onto $[0,1]$. Next, we define a finite and possibly random grid $\mc S_n(a)=\{s_1(a),\dots,s_{m_n}(a)\} \subset \mathbb{R}$ for $a\in\{0,1\}$.
We let $L_n(a)=s_1(a),U_n(a)=s_{m_n}(a)$, and $\delta_n(a)=s_{j+1}(a)-s_j(a)$ for every $j\in\{1,\cdots,m_n-1\}$. 
We require the following conditions for the grid $\mc S_n(a)$.
We suppress $(a)$ for our $L_n$, $U_n$, $\mc S_n$, $\delta_n$ and its elements for notational simplicity.
We denote the support of $p_a$ as $(\ell_{a},u_a)$ where $-\infty\le \ell_a < u_a \le \infty$.
\begin{assumptionp}{G}\label{assm:G} \phantom{blah}
For $a\in\{0,1\}$, the following hold:
\begin{enumerate}
    \myitem{(G1)}: \label{assm:condition.grid} $[s_2,s_{m_n}]\subset (\ell_{a},u_{a})$ for all $n$, and $L_n \to_p \ell_{a}$, $U_n \to_p u_{a}$, 
    \myitem{(G2)}: \label{assm:max.grid} $\delta_n=o_p(1)$.
\end{enumerate}   
\end{assumptionp}
In practice, we suggest setting
\begin{align}\label{eq:grid.setting}
\delta_n=\frac{Y_{a,n,\mathrm{max}}-Y_{a,n,\mathrm{min}}}{n}, \,L_n=Y_{a,\mathrm{min}}-\delta_n,\, \text{and }U_n= Y_{a,n,\mathrm{max}},
\end{align}
where $Y_{a,n,\mathrm{min}}:=\min_{1 \leq i \le n,\,A_i=a} Y_i$ and $Y_{a,n,\mathrm{max}}:=\max_{1 \leq i \leq n\,A_i=a} Y_i$. By \ref{assm:consistency}--\ref{assm:positivity}, this choice satisfies \ref{assm:condition.grid}.
As discussed in Section \ref{section:CI} below, any grid satisfying $\delta_n=O_p(n^{-1/2})$ is sufficient for our distributional results and validity of the corresponding CI's.
Therefore, practitioners can select arbitrarily fine grid lengths--albeit with a computational cost of $O(m_n)$ \citep{grotzinger1984projections}--so that our correction procedure can be viewed as an approximation of a continuous isotonic projection \citep{groeneboom2010generalized}. Given this, the grid is a computational rather than statistical tuning parameter.

We now define  
\begin{align}
    \left(\widehat{F}^{c}_{a,n}(s_2),\dotsc,\widehat{F}^{c}_{a,n}(s_{m_n-1}) \right) :={\rm argmin}_{v\in \mc C^{m_n-2}} \sum_{k=2}^{m_n-1}\left[ v_k -\widehat{F}^{c_0}_{a,n}(s_k) \right]^2,\label{def:isotonicongrid}
\end{align}
where $\mc C^k=\{(c_1,\cdots,c_k) \in \RR^k :c_1\le \cdots \le c_k\}$. 
That is, $\widehat{F}^{c}_{a,n}$  is the isotonic regression of $\widehat{F}^{c_0}_{a,n}$ on $\{s_2, \dotsc, s_{m_n - 1}\}$, which can be obtained using the Pool Adjacent Violators Algorithm \citep{PAVA.Ayer} using the \texttt{isoreg} function in \texttt{R} \citep{R.citation}.

We have now defined an estimator on the grid $\mc S_n$ that is monotone and contained in $[0,1]$. 
For the final step in our correction procedure, we extend this estimator to the entirety of $\mathbb{R}$ by defining the piecewise constant interpolation
\begin{equation}
\label{def:onestepcorrected}
        \widehat{F}^c_{a,n}(s)= \widehat F^c_{a,n}(s_k) I\left(s_2 \leq s < s_{m_n}\right) + I\left( U_n \leq s\right).
\end{equation}
Hence, our corrected estimator $\widehat{F}^c_{a,n}$ is a right-continuous step function with a finite number of jumps contained in the grid $\mc S_n$.

\subsection{Log-concave counterfactual density estimator}
\label{subsec:LCCDE}

We now use the log-concave projection operator defined in \citet{dss2011} to project $\widehat{F}^c_{a,n}$ onto the space of log-concave distributions and thereby obtain our log-concave counterfactual density estimator. 
We let $\mathcal{P}_1$ be the class of probability measures $P$ on $\mathbb{R}$ that are not point masses and that satisfy $\int_{\mathbb{R} } |x| dP(x)<\infty$. 
We also define $\mathcal{F}$ as the class of log-concave probability density functions on $\mathbb{R}$.
For any $Q \in \mc P_1$, the log-concave projection operator $\psi^{*}(Q)$ is then defined as
\begin{equation*}
    \psi^{*}(Q):={\rm argmax}_{f\in \mathcal{F}} \int_{\mathbb{R}} \log f \, dQ.
\end{equation*}
Existence and uniqueness $\psi^*(Q)$ follows from Theorem~2.2 of \citet{dss2011}. 
We slightly abuse notation by writing $\psi^*(F_Q) := \psi^*(Q)$, where $F_Q$ is the CDF according to $Q$.

We now define our log-concave counterfactual density estimator as
$\widehat{p}_{a,n} :=\psi^{*}(\widehat{F}^c_{a,n})$ for each $a=0,1$. 
In words, our estimator is the log-concave projection of the corrected one-step
counterfactual CDF estimator. We can compute $\psi^{*}(\widehat{F}^c_{a,n})$
by applying the active set algorithm of \citet{Rufibach:2007iw} and
\cite{Duembgen:2007vu}, which is implemented in the \texttt{activeSetLogCon} function in the R package \texttt{logcondens} \citep{logcondens.package}. The active set  algorithm takes as input weighted data
  points, so we pass in the points $\mc S_n$ with weights
  $\widehat F^c_{a,n}(s_{k})-\widehat F^c_{a,n}(s_{k-1})$ for each $s_k \in \mc{S}_n$.
 
 We summarize the steps to obtain $\widehat{p}_{a,n}$ as follows.

\begin{enumerate}	\myitem{(S1)}\label{Step1} Define a grid $\mc S_n=\{s_1,\ldots,s_{m_n}\}$  satisfying \ref{assm:condition.grid} and \ref{assm:max.grid}.
\myitem{(S2)}\label{Step2} Using the estimated nuisance functions $(\widehat\pi_a, \widehat{\phi}_{a})$, compute the doubly-robust one step CDF estimate $\widehat{F}_{a,n}(s)$ on $\mc S_n$ using \eqref{form:onestepest}.
\myitem{(S3)}\label{Step3} Compute $\widehat F^{c_0}_{a,n}$ as the projection of $\widehat F_{a,n}$ onto $[0,1]$  as in \eqref{def:projection01}.
\myitem{(S4)}\label{Step4} Apply the Pool Adjacent Violators Algorithm to $(\widehat{F}^{c_0}_{a,n}(s_2),\ldots,\widehat{F}^{c_0}_{a,n}(s_{m_n-1}))$ to obtain  \\
$(\widehat{F}^{c}_{a,n}(s_2),\ldots,\widehat{F}^{c}_{a,n}(s_{m_n-1}))$ as in \eqref{def:isotonicongrid}. Set $\widehat{F}^{c}_{a,n}(L_n)=0$ and $\widehat{F}^{c}_{a,n}(U_n)=1$.
\myitem{(S5)}\label{Step5} Apply the Active Set Algorithm to the points $\{s_k : 2\le k \le m_n\}$ with corresponding weights $\{\widehat F^c_{a,n}(s_{k})-\widehat F^c_{a,n}(s_{k-1}) : 2\le k \le m_n \}$ to obtain $\widehat{p}_{a,n} :=\psi^{*}(\widehat{F}^c_{a,n})$.
\end{enumerate}

\section{Consistency}\label{section:consistency}

In this section, we study double robust consistency of the proposed estimator. 
In Theorem \ref{prop:Consistency}, we prove that the log-concave MLE $\widehat p_{a,n}$ is uniformly consistent for the true counterfactual density $p_a$ with respect to exponentially weighted uniform and $L_1$ global metrics on the real line.  
To the best of our knowledge, Theorem \ref{prop:Consistency} is the first uniform convergence result in the literature on nonparametric counterfactual density estimation.
We begin by stating conditions we will use regarding the nuisance estimators. We discuss these conditions in detail following Theorem \ref{prop:Consistency}.

\begin{assumptionp}{E}\label{assm:E} \phantom{blah}
   There exist functions $\pi_{a,\infty},\phi_{a,\infty}$ such that:
   \begin{enumerate}
    \myitem{(E1)}: \label{assm:nuisanceconvergence}
    For $a\in\{0,1\}$,  the estimated nuisance functions $\widehat{\pi}_a,~\widehat\phi_a$ satisfy
    \begin{align*}
        P_*\left[\widehat{\pi}_a(\mathbf{X})-\pi_{a,\infty}(\mathbf{X}) \right]^2\rightarrow_p 0 \text{ and } P_*\left[\int_{-\infty}^{\infty}\left|\widehat\phi_a(s|\mathbf{X})-{\phi}_{a,\infty}(s|\mathbf{X})\right| \,ds\right]^2\rightarrow_p 0.
    \end{align*}
        \myitem{(E2)}: \label{assm:boundedpropscore} There exists  $K>0$ such that $\|1/{\pi_{a,\infty}}\|_{\infty},\| 1/{\widehat\pi_a}\|_{\infty}\le K$ a.s..
      \myitem{(E3)}: \label{assm:properCDF} $s \mapsto \widehat\phi_{a}(s|\mathbf{X})$ and $s \mapsto \phi_{a,\infty}(s|\mathbf{X})$ are a.s.\ proper conditional CDFs; i.e., for a.e.\ $\mathbf{X}$, they are monotonic in $s$, take values in $[0,1]$, and converge to $0$ and $1$ as $s$ converges to $-\infty$, and $\infty$, respectively. And, there exists $h \in L_2(P_*)$ such that $\int_{\mathbb{R}}|s| \, d\widehat\phi_a(s|\mathbf{x})\leq h(\mathbf{x})$ $P_*$-a.e.\ $\mathbf{x}$ for all $n>N_0$ for sufficiently large $N_0>0$.
     \myitem{(E4)}: \label{assm:DR}   There exist subsets $\mc S_1,\mc S_2$ of $\mc X \times \mc A$ such that $P_*(\mathcal{S}_1\bigcup \mathcal{S}_2)=1$, and 
    \begin{itemize}
    \item ${\pi}_{a,\infty}(\mathbf{x})={\pi}_{a,*}(\mathbf{x})$, for all $(\mathbf{x},a) \in \mathcal{S}_1$,
    \item ${\phi}_{a,\infty}(s|\mathbf{x})={\phi}_{a,*}(s|\mathbf{x})$, for all $(\mathbf{x},a) \in \mathcal{S}_2$ and $s\in\RR$.
    \end{itemize}
    \myitem{(E5)}: \label{assm:phiregularity} There exists $R \in L_2(P_*)$ such that $|\phi_{a,\infty}(s|\mathbf{X})-\phi_{a,\infty}(t|\mathbf{X})|\le |s-t|R(\mathbf{X})$ for all $s,t\in \RR$.

\end{enumerate}   
\end{assumptionp}

We next make estimator complexity assumptions to control the empirical process terms.
A class of functions $\mc R$ is called $P_*$-Glivenko-Cantelli if 
$\sup_{f\in \mc R} |(\PP_n-P_*)f| \rightarrow 0$ a.s..
More detailed description and examples about Glivenko-Cantelli classes can be found in Section 2.4 of \citet{vdvandW}. 
\begin{assumptionp}{EC-I}\label{assm:EC1} The estimators $\widehat{\pi}_a$  and $\widehat{\phi}_a$ belong to classes of measurable functions $\mc F_\pi$ and $\mc F_\phi$, respectively, where:
\begin{enumerate}
    \myitem{(EC1)}: \label{assm:basicGC} $\mc F_\pi$ and $\{\mathbf{x} \mapsto \phi(s|\mathbf{x}) : s\in\RR,\phi\in\mc F_\phi\}$ are $P_*$-Glivenko-Cantelli.
    \myitem{(EC2)}: \label{assm:integratedGC} The class of functions 
    \begin{align*}
        \left\{\int_{-\infty}^{t_1} \phi(s|\cdot) \, ds, \int_{t_2}^{\infty}[1-\phi(s|\cdot)] \, ds: t_1\in(-\infty,0],\,t_2\in[0,\infty),\phi\in\mc F_\phi\right\}
    \end{align*}
    is $P_*$-Glivenko-Cantelli.
\end{enumerate}
\end{assumptionp}

\noindent  We now  state the consistency of our log-concave counterfactual density estimator $\widehat{p}_{a,n}$ in  weighted $L_1$ and uniform metrics.

\begin{theorem}\label{prop:Consistency}
If conditions \ref{assm:consistency}--\ref{assm:logconcv}, \ref{assm:condition.grid}--\ref{assm:max.grid}, 
\ref{assm:nuisanceconvergence}--\ref{assm:phiregularity},
and \ref{assm:basicGC}--\ref{assm:integratedGC} hold, then 
\begin{align}
 \int_{\mathbb{R}} e^{\varepsilon|s|}|\widehat{p}_{a,n}(s)-p_a(s)| \, ds \rightarrow_p 0,\label{prop:consist-1}
\end{align}
as $n\rightarrow \infty$ for $a\in\{0,1\}$ and for all $\varepsilon \in [0,\alpha)$, where $\alpha > 0$ is such that $p_a(s) \le e^{-\alpha|s|+\beta}$  for all $s \in \mathbb{R}$ and some $\beta \in \mathbb{R}$. If in addition $p_a$ is continuous on $\mathbb{R}$, then
\begin{align}
 &\sup_{s\in\mathbb{R}} e^{\varepsilon|s|}|\widehat{p}_{a,n}(s)-p_a(s)| \rightarrow_p 0.\label{prop:consist-2}
\end{align}

\end{theorem}

\noindent We note that by Lemma 1 of \citet{Cule:2010gv}, \ref{assm:logconcv} implies that there always exist $\alpha >0$ and $\beta \in \mathbb{R}$ such that $p_a(s) \le e^{-\alpha|s|+\beta}$ for all $s \in \mathbb{R}$. 

We now discuss the conditions and result of Theorem~\ref{prop:Consistency}. We define 
\begin{align}\label{metric:Wasser}
    d_1(F,G)=\int_{\mathbb{R}} |F(s)-G(s)| \, ds,
\end{align}
as the Wasserstein distance between two univariate distribution functions $F,G$. Condition \ref{assm:nuisanceconvergence} requires that the $L_2$ norms of $\widehat\pi_a - \pi_{a,\infty}$ and $d_1 \left( \widehat\phi_a(\cdot|\mathbf{X}), {\phi}_{a,\infty}(\cdot|\mathbf{X})\right)$ converge in probability to zero.   Condition~\ref{assm:boundedpropscore} requires that $\widehat\pi_a$ and its limit $\pi_{a,\infty}$ are uniformly bounded below, and Condition~\ref{assm:properCDF} requires that $\widehat\phi_a$ and $\phi_{a,\infty}$ are proper CDFs, and $\widehat\phi_a$ has finite conditional first moment $P_*$-almost everywhere $\mathbf{X}$.  
The proper CDFs condition ensures that our one-step CDF estimator $\widehat F_{a,n}$ has bounded variation and converges uniformly in probability to $F_a$ over the real line, which is used to establish that the isotonized version, $\widehat F_{a,n}^c$, converges to $F_a$ in Wasserstein distance in probability.
Condition \ref{assm:DR} is satisfied if at least one, but not necessarily both, of the two nuisance estimators is consistent, namely, ${\pi}_{a,\infty}={\pi}_{a,*}$ or ${\phi}_{a,\infty}={\phi}_{a,*}$. 
Therefore,  Theorem~\ref{prop:Consistency} implies that $\widehat p_{a,n}$ is \textit{doubly-robust consistent}. 
Condition~\ref{assm:phiregularity} requires that $\phi_{a,\infty}$ is Lipschitz in its first argument, where the Lipschitz constant may depend on $\mathbf{x}$ but must be a square-integrable function of $\mathbf{x}$. 
This assumption is common in the nonparametric conditional distribution estimation literature, including \citet{hall2003order,meinshausen2006quantile,li2013optimal}, among others.
It is satisfied in particular when the conditional density corresponding to $\phi_{a,\infty}(\cdot|\mathbf{x})$ is bounded by $R(\mathbf{x})$.
Conditions \ref{assm:basicGC}--\ref{assm:integratedGC} restrict the complexity of the estimators to control empirical process terms.  
If the support of $Y^a$ is contained in $[-M,M]$ for some $M>0$, the function classes in \ref{assm:basicGC}--\ref{assm:integratedGC} can be constrained to $|s|\le M$. 

Consistency of conditional CDF estimators in expected Wasserstein distance is not presently available for many estimators, which makes verifying~\ref{assm:nuisanceconvergence} challenging. The following lemma gives sufficient conditions for the consistency assumption in~\ref{assm:nuisanceconvergence}.

\begin{lemma}\label{lemma:cond.E1}
Suppose that: (1) $\sup_{s\in[-L,L]}P_*\left|\widehat\phi_a(s|\mathbf{X})-\phi_{a,\infty}(s|\mathbf{X})\right|^2=o_p(1)$ for all $L>0$; (2) there exists a measurable function $\mathcal{U}_n : (0, \infty) \times \mathcal{Z} \to \mathbb{R}$ such that
\begin{align*}
P_*\left\{\int_{-\infty}^{-M} \widehat\phi_a(s|\mathbf{X}) \, ds\right\}^2 + P_*\left\{\int_{M}^{\infty} \left[1-\widehat\phi_a(s|\mathbf{X})\right] \, ds\right\}^2\leq \mathcal{U}_n(M,\mathbf{Z}),
\end{align*}
for every $M \in (0,\infty)$ and a.e.\ $\mathbf{Z}$, and for every $\epsilon, \delta>0$, there exist $N_{\epsilon,\delta}, M_{\epsilon,\delta} \in (0, \infty)$ such that $P_*\left(\mathcal{U}_n(M_{\epsilon,\delta},\mathbf{Z})>\delta\right) \leq \epsilon$ for all $n\geq N_{\epsilon,\delta}$ and $M\geq M_{\epsilon,\delta}$; and (3) there exists $T:\mathcal{X}\to \RR$ such that $T\in L_2(P_*)$, and $\int_\RR |s| \, d\phi_{a,\infty}(s|\mathbf{X})<T(\mathbf{X})$ for $P_*$-a.e.\ $\mathbf{X}$. Then
\begin{align*}
P_*\left[\int_{-\infty}^{\infty}\left|\widehat\phi_a(s|\mathbf{X})-{\phi}_{a,\infty}(s|\mathbf{X})\right| \,ds\right]^2\rightarrow_p 0.
\end{align*}
\end{lemma}
\medskip
The proof of Lemma \ref{lemma:cond.E1} is given in Section \ref{proof:cond.E1} of the supplementary material.
This lemma implies that a uniform light tail condition for $\widehat\phi_a$ and supremum convergence in probability of $L_2(P_*)$ distance between $\widehat\phi_a$ and $\phi_{a,\infty}$ on every compact interval ensures the second condition in Assumption \ref{assm:nuisanceconvergence}.
Many estimators $\widehat\phi_a$ and $\widehat{\pi}_a$, including nonparametric and machine learning estimators, can satisfy our conditions. For example, when the density of $\mathbf{X}$ is positive and bounded from above and below by positive constants on a compact support, quantile regression random forests \citep{meinshausen2006quantile} 
and monotone local linear estimators \citep{das2019nonparametric}
satisfy conditions \ref{assm:nuisanceconvergence} and \ref{assm:properCDF}
under regularity conditions \citep{elie2022random,xie2023uniform}.
We give an additional example of a semiparametric estimator that satisfies conditions \ref{assm:nuisanceconvergence} and \ref{assm:properCDF}.
We define $\mu_a : \mathbf{X} \mapsto \mathbb{E}(Y \mid A = a, \mathbf{X})$ and $\sigma_a^2 : \mathbf{X} \mapsto \mathrm{Var}(Y \mid A = a, \mathbf{X})$, and the class of functions $\mc F_H$ as follows:
\begin{align}\label{F.H.class}
	\mc F_H:=\left\{ \widehat \phi(s|\mathbf{x})=H((s-\hat\mu_a(\mathbf{x}))/\hat\mu_a(\mathbf{x})): \hat\mu_a(\mathbf{x})\in \mc F_1,\, \hat\sigma_a(\mathbf{x})\in \mc F_2\right\}.	
\end{align}
\noindent We provide the conditions under which the function class $\mc F_H$ satisfies assumptions \ref{assm:nuisanceconvergence}, \ref{assm:properCDF}, \ref{assm:basicGC}, and \ref{assm:integratedGC} in Section \ref{FH.Wasserstein.proof} of the supplementary material.

A key element in the proof of Theorem~\ref{prop:Consistency} is showing that certain properties of the one-step estimator $\widehat{F}_{a,n}$ carry over to the corrected one-step estimator $\widehat{F}_{a,n}^c$. While \citet{WvdlC2020} provided general results about monotone corrections using isotonic regression, some of their results assumed compact support, and their results do not address convergence in Wasserstein distance, which we need. We provide  two lemmas  extending the results of \citet{WvdlC2020} to unbounded domains, and to convergence in Wasserstein distance. Since these results may be of independent interest, we state them below. Proofs are given in Sections \ref{appd:lem1} and \ref{appd:lem2}, respectively, of the supplementary material.

\begin{lemma}\label{lem1}
If $F_a$ is uniformly continuous on $\RR$, $F_a(L_n) \to_p 0$, $F_a(U_n) \to_p 1$, $\delta_n \to_p 0$, and ${\rm sup}_{x\in \RR} |\widehat{F}_{a,n}(x)-F_a(x)|\rightarrow_p 0$, then $\widehat{F}^c_{a,n}(x) \rightarrow_p F_a(x)$ for all $x\in \mathbb{R}$.
\end{lemma}

\begin{lemma}\label{lem2}
  If $F_a$ is uniformly continuous on $\RR$, $\int_{\RR} |s| \, dF_a(s) < \infty$, $F_a(L_n) \to_p 0$, $F_a(U_n) \to_p 1$, $\delta_n \to_p 0$, ${\rm sup}_{x\in \RR} |\widehat{F}_{a,n}(x)-F_a(x)|\rightarrow_p 0$, and if
  \begin{align*}
     \delta_n\max_{2\le k\le m_n-1}\left|\sum_{j=2}^k \widehat{F}^{c_0}_{a,n}(s_{j})-\sum_{j=2}^k F_a(s_{j})\right|\rightarrow_p 0,  
  \end{align*}
then 
    $\int_{\mathbb{R}}|s| \, d \widehat{F}^c_{a,n}(s) \rightarrow_p  \int_{\mathbb{R}}|s| \, dF_a(s).$
\end{lemma}

\section{Limit distribution and confidence intervals}\label{section:CI}
\subsection{Limit distribution}\label{subsec:limitdist}

We now derive the joint limit distribution of $(\widehat p_{1,n}(s_0) - p_1(s_0),\widehat p_{0,n}(s_0) - p_0(s_0))$, properly rescaled, at a fixed point $s_0\in \mathbb{R}$. 
We define $\widehat{\varphi}_{a,n} = \log (\widehat{p}_{a,n})$ and $\varphi_a = \log( p_a)$. We first state the following regularity assumptions for the true density.
\begin{assumptionp}{R}\label{assm:R}\phantom{blah}
\begin{enumerate}
    \myitem{(R1)}: \label{assm:pregularity} $p_a(s_0)>0$, and there exists $\omega > 0$ such that $p_a$ is twice continuously differentiable in the neighborhood $I_{s_0,\omega}:=[s_0-\omega,s_0+\omega]$ of $s_0$.
    \myitem{(R2)}: \label{assm:pdiffwithk} If $p_a''(s_0)\neq 0$, then set $k=2$. Otherwise, assume that $k < \infty$ is the smallest positive even integer such that $\varphi_a^{(j)}(s_0)=0$ for $j=2,\dots,k-1$, and $\varphi_a^{(k)}(s_0)< 0$. In addition, $\varphi_a^{(k)}$ is continuous in a neighborhood of $s_0$.
\end{enumerate}
\end{assumptionp}
Assumption \ref{assm:R} is analogous to conditions (A3)-(A4) of \citet{BRW09}. We note that concavity of $\varphi_a$ implies that $k$ is an even integer (see page 7 in \citealp{BRW09}). Next, we state assumptions on the nuisance estimators that we will need. We recall $\mc S_1,~\mc S_2$, and $\mc F_\pi,~\mc F_\phi$ defined in Assumption \ref{assm:DR}, and for a function $f$ and a set $S \subset \mathbf{\mc Z}$, we define $\|f\|_{S}:=\|fI_S\|$, where $\|\cdot\|$ is the $L_2(P_*)$ norm. For any  $S \subset \mathbf{\mc Z}$ and $s, t \in \mathbb{R}$, we also define
\begin{align*}
&\mc G(s,t;\mathbf Z;S) := \|(\widehat\phi_a-\phi_{a,\infty})(t|\cdot)-(\widehat\phi_a-\phi_{a,\infty})(s|\cdot)\|_S.
\end{align*}

\begin{assumptionp}{E (cont.)}\label{assm:E2}\phantom{blah}
\begin{enumerate}
   \myitem{(E6)} \label{assm:L2conditionsforCI} For all $s_1,s_2\in  I_{s_0,\omega}$ defined in condition \ref{assm:pregularity} the following statements hold: 
    \begin{gather*}
    \mc G(s_1,s_2;\mathbf{Z};\mathcal{S}_1\cap\mathcal{S}_2) \|\widehat\pi_a-\pi_{a,\infty}\|_{\mathcal{S}_1\cap\mathcal{S}_2} = |s_2-s_1|M_1,\\
      \|\widehat\pi_a-\pi_{a,\infty}\|_{\mathcal{S}_1\cap\mathcal{S}_2^c} = o_p(n^{-k/(2k+1)}), \,\mc G(s_1,s_2;\mathbf{Z};\mathcal{S}_1\cap\mathcal{S}_2^c) = |s_2-s_1|M_2,\\
      \|\widehat\pi_a-\pi_{a,\infty}\|_{\mathcal{S}_1^c\cap\mathcal{S}_2} = o_p(1),\,\mc G(s_1,s_2;\mathbf{Z};\mathcal{S}_1^c\cap\mathcal{S}_2) = |s_2-s_1|M_3,
    \end{gather*}
    where $M_1$, $M_2$, and $M_3$ are random variables that do not depend on $s_1,s_2$ and such that $M_1$ and $M_3$ are $o_p(n^{-k/(2k+1)})$ and $M_2 = o_p(1)$.

    \myitem{(E7)}: \label{assm:phiestholder} There exists $R_1 \in L_2(P_*)$ and $\alpha \in (1/2,1]$ such that for every $t,s$ in a neighborhood of $s_0$ and $f \in\mc F_{\phi}$, $|f(t|\mathbf{X})-f(s|\mathbf{X})|\le R_1(\mathbf{X})|t-s|^{\alpha}$.
	
    \myitem{(E8)}: \label{assm:phiinfholder} There exists $R_2$ such that $P_*(R_2(\mathbf{X})>0)>0$ and 
    $
        R_2(\mathbf{X})|t-s|\le \left|\phi_{a,\infty}(t|\mathbf{X})-\phi_{a,\infty}(s|\mathbf{X})\right|
    $
    for all $t,s,\in(\ell_a,u_a)$,
\end{enumerate}
\end{assumptionp}

\begin{assumptionp}{EC (cont.)}\label{assm:EC2}\phantom{blah}
\begin{enumerate}
    \myitem{(EC3)}: \label{assm:pi-bracketing} There exists $V\in[0,2)$ such that for all $\varepsilon  > 0$,
        $\sup_Q {\rm log}N(\varepsilon,\mathcal{F}_{\pi},L_2(Q)) \lesssim \varepsilon^{-V}$, $\sup_Q {\rm log}N(\varepsilon,\mathcal{F}_{\phi},L_2(Q)) \lesssim \varepsilon^{-V}$.
\end{enumerate}
\end{assumptionp}

The limit distribution of $\widehat p_{a,n}$ involves the \textit{invelope process} $H_k$ of the integrated Brownian motion with drift $Y_k$ introduced by \citet{Groeneboom:2001fp,Groeneboom:2001jo}, which we define now. We refer the readers to \eqref{integrgaussprocess}--\eqref{H-k:definition} (see the preceding paragraph of Lemma \ref{lem:finallemmaforCI} in the supplementary material) for the definition of the processes $H_k$ and $Y_k$.

The following theorem provides the pointwise joint limit distribution of the estimator.

\begin{theorem}\label{prop:CI}
If \ref{assm:consistency}--\ref{assm:logconcv}, \ref{assm:condition.grid}, \ref{assm:nuisanceconvergence}--\ref{assm:phiinfholder},
\ref{assm:basicGC}--\ref{assm:pi-bracketing}, and
\ref{assm:pregularity}--\ref{assm:pdiffwithk} hold for $a\in\{0,1\}$, and $\delta_n=O_p(n^{-(k+1)/(2k+1)})$, then
\begin{align}\label{jointfororiginal}
  \begin{pmatrix}
    n^{k/(2k+1)}(\widehat p_{1,n}(s_0)-p_1(s_0))\\
n^{(k-1)/(2k+1)}(\widehat p'_{1,n}(s_0)-p'_1(s_0)) \\  
    n^{k/(2k+1)}(\widehat p_{0,n}(s_0)-p_0(s_0))\\
n^{(k-1)/(2k+1)}(\widehat p'_{0,n}(s_0)-p'_0(s_0))
  \end{pmatrix}
  \rightarrow_d 
  \begin{pmatrix}
   c_k(s_0,\varphi_1) H_{1k}^{(2)}(0)\\
d_k(s_0,\varphi_1) H_{1k}^{(3)}(0) \\ 
   c_k(s_0,\varphi_0) H_{0k}^{(2)}(0)\\
d_k(s_0,\varphi_0) H_{0k}^{(3)}(0) \\ 
  \end{pmatrix},
  \end{align}
and
 \begin{align}\label{jointforlogscale}
   \begin{pmatrix}
     n^{k/(2k+1)}(\widehat\varphi_{1,n}(s_0)-\varphi_1(s_0))\\
n^{(k-1)/(2k+1)}(\widehat\varphi'_{1,n}(s_0)-\varphi'_1(s_0))\\   
n^{k/(2k+1)}(\widehat\varphi_{0,n}(s_0)-\varphi_0(s_0))\\
n^{(k-1)/(2k+1)}(\widehat\varphi'_{0,n}(s_0)-\varphi'_0(s_0))  
  \end{pmatrix}
  \rightarrow_d 
  \begin{pmatrix}
   C_k(s_0,\varphi_1) H_{1k}^{(2)}(0)\\
D_k(s_0,\varphi_1) H_{1k}^{(3)}(0)  \\ 
   C_k(s_0,\varphi_0) H_{0k}^{(2)}(0)\\
D_k(s_0,\varphi_0) H_{0k}^{(3)}(0)  \\ 
  \end{pmatrix},
 \end{align}
 where, $\widehat p'_{a,n},\widehat\varphi'_{a,n}$ are the left derivatives of $\widehat p_{a,n},\widehat\varphi_{a,n}$, $H_{1k}$ and $H_{0k}$ are independent copies of the lower invelope $H_k$ of the process $Y_k$, and $c_k,d_k,C_k,D_k$ are given by
 \begin{align}
c_k^{2k+1}(s_0,\varphi_a)&=\frac{|\varphi^{(k)}_a(s_0)|p_a(s_0)}{(k+2)!{\chi_{\theta_a}}^{-k}},\,C_k^{2k+1}(s_0,\varphi_a)=\frac{|\varphi^{(k)}_a(s_0)|p_a(s_0)^{-2k}}{(k+2)!{\chi_{\theta_a}}^{-k}},
\label{ck}\\
d_k^{2k+1}(s_0,\varphi_a)&=\frac{|\varphi^{(k)}_a(s_0)|^3p_a(s_0)^{3}}{[(k+2)!]^3{\chi_{\theta_a}}^{-(k-1)}},\, D_k^{2k+1}(s_0,\varphi_a)=\frac{|\varphi^{(k)}_a(s_0)|^3p_a(s_0)^{-2(k-1)}}{[(k+2)!]^3{\chi_{\theta_a}}^{-(k-1)}},\label{dk}
 \end{align}
for $a\in\{0,1\}$, where $\chi_{\theta_a}=\EE \Big[\frac{\pi_{a,*}(\mathbf{X})}{\pi^2_{a,\infty}(\mathbf{X})}\eta_{a,*}(s_0|\mathbf{X})\Big]$.
\end{theorem}
The proof of Theorem~\ref{prop:CI} is provided in Section \ref{subsec:CI.proof.section} of the supplementary material. To the best of our knowledge, Theorem \ref{prop:CI} is the first convergence in distribution result for a nonparametric estimator of the counterfactual density function. In addition, we expect that distributional results for other nonparametric estimators would be asymptotically biased unless undersmoothing or bias correction were utilized. Furthermore, Theorem \ref{prop:CI} is the first distributional result we are aware of for a log-concave density in the presence of nuisance function estimation, as well as the first doubly robust limit distribution for a counterfactual density estimator. Each limit distribution involves information about the curvature of the true counterfactual density function, $\varphi_a^{(k)}$ for $a\in\{0,1\}$.
A variety of methods can be used to estimate the unknown curvature  to create plug-in CI's \citep{gasser1984estimating,SBW10.inconsist.boot,Seijo.Sen.2011}, but as we discussed in the introduction, such methods require careful selection of tuning parameters. In Section \ref{subsec:CIconst}, we present methods for constructing CI's that do not depend on careful tuning parameter selection.

Theorem \ref{prop:CI} also yields asymptotic results for contrasts of the counterfactual density functions using the delta method. For example, under the stated conditions, we have
\begin{align*}
&n^{k/(2k+1)}[\{\widehat p_{1,n}(s_0)-p_1(s_0)\}-\{\widehat p_{0,n}(s_0)-p_0(s_0)\}] \rightarrow_d c_k(s_0,\varphi_1) H_{1k}^{(2)}(0)-c_k(s_0,\varphi_0) H_{0k}^{(2)}(0),\\
&n^{k/(2k+1)}[\log(\widehat p_{1,n}(s_0)/\widehat p_{1,n}(s_0))-\log(p_1(s_0)/p_0(s_0))]\rightarrow_d  C_k(s_0,\varphi_1) H_{1k}^{(2)}(0)-C_k(s_0,\varphi_0) H_{0k}^{(2)}(0).
\end{align*}
To the best of our knowledge, the previous displays are the first results establishing the distributional convergence for contrasts of counterfactual density functions.

We now discuss the additional conditions required by Theorem~\ref{prop:CI}. Assumption R requires that $p_a$ is $k$ times continuously differentiable in a neighborhood of $s_0$, where $k$ is the smallest even integer such that  $\varphi_a^{(k)}(s_0) < 0$ in a neighborhood of $s_0$. It is assumed that $k < \infty$, so that $p_a$ is not affine at $s_0$. The rate of convergence of $\widehat{p}_{a,n}$ is $n^{-k / (2k+1)}$, so that the closer $p_a$ is to affine at $s_0$, the closer the rate of convergence is to the parametric rate $n^{-1/2}$. The isotonization grid in \eqref{eq:grid.setting} satisfies $\delta_n=O_p(n^{-(k+1)/(2k+1)})$.

Assumption \ref{assm:L2conditionsforCI} requires that the product of the rates of convergence of $\widehat \phi_{a}$ and $\widehat\pi_a$ to their true counterparts is faster than the rate of convergence of $\widehat p_{a,n}(s_0)$, $n^{-k/(2k+1)}$. 
When $k=2$, on the set $\mc S_1 \cap \mc S_2$ where both nuisance functions are consistent, we only require the product of the rates of convergences is faster than $n^{-2/5}$.
In particular, \ref{assm:L2conditionsforCI} permits that one of the nuisance estimators is misspecified, in which case the other nuisance estimator must converge faster than $n^{-k/(2k+1)}$ to the truth. For instance, on the set $\mc S_1\cap S_2^c$, where only $\widehat\pi_a$ is consistent, the condition requires that $\widehat\pi_a$ converges faster than $n^{-k/(2k+1)}$. Hence, Theorem~\ref{prop:CI} is a doubly-robust convergence in distribution result.
It also suggests the possibility of constructing doubly-robust confidence intervals for $p_a(s)$ for each $s\in\RR$, which 
we will do 
in Section \ref{subsec:CIconst} (with simulations in Section~\ref{Section:Simul}).
Assumption \ref{assm:L2conditionsforCI} also requires a Lipschitz type of assumption on $\widehat \phi_{a} - \phi_{a,\infty}$, which is easily satisfied.
For example, when both $\widehat{\phi}$ and $\phi_{\infty}$ are Lipschitz, then the Lipschitz condition directly follows.
In addition, the class $\mc F_H$ defined in \eqref{F.H.class} can be another example of \ref{assm:L2conditionsforCI}, as we will discuss below.
Assumption~\ref{assm:phiestholder} requires that the functions in $\mc F_\phi$ are all H{\"o}lder in their first argument with common exponent greater than $1/2$. 
For example, this condition is met when the 
conditional density of $\phi(\cdot|\mathbf{x}) \in \mc F_\phi$ is uniformly bounded by $R_1(\mathbf{x})$, where $R_1\in L_2(P_*)$.
Assumption \ref{assm:phiinfholder} is an analogue of condition (ii) in \citet{WvdlC2020} (see Section 4.1 therein), and controls the variation in $\phi_{a,\infty}$ from below. 
This assumption requires $R_2$ to be positive on a set of $\mathbf{x}$ with non-zero measure.
Condition \ref{assm:phiinfholder} holds when $\phi_{a,\infty}(s|\mathbf{x})$ has first-order derivatives (with respect to $s$) on $(\ell_a,u_a)$ that are bounded away from zero on a non-null set of $\mathbf{x}$.
It excludes scenarios where $\phi_{a,\infty}(\cdot|\mathbf{x})$ has zero derivatives over an interval of $s$ for $P_*$-almost everywhere $\mathbf{x}$.
This condition ensures that the isotonized one-step counterfactual CDF estimator is uniformly asymptotically equivalent to the one-step CDF estimator in the supremum norm at the rate $n^{-(k+1)/(2k+1)}$ within a shrinking neighborhood of $s_0$.
Consequently, it confirms that the isotonic correction of the one-step CDF estimator has a negligible impact on the limit distribution of the log-concave MLE.

Condition \ref{assm:pi-bracketing} requires that $\widehat\pi_a$ and $\widehat\phi_a$ are contained in function classes with finite uniform entropy integral, which is used to control certain empirical process terms. For example, parametric classes and $p$-dimensional H\"{o}lder classes with smoothness exponent $\gamma$ satisfying $p/\gamma<2$ satisfy this condition. Section 2.6 of \cite{vdvandW} contains these and further examples.
We also provide the conditions under which the function class $\mc F_H$ introduced in \eqref{F.H.class} satisfies assumptions \ref{assm:phiestholder} and \ref{assm:pi-bracketing} in Lemma \ref{lemma: meanvar2} and the following paragraph (see Section \ref{FH.Wasserstein.proof} in the supplementary material).

\subsection{Construction of confidence intervals}\label{subsec:CIconst}

We now propose a confidence interval for $p_{a}(s_0)$ at a fixed point $s_0$. While Theorem~\ref{prop:CI} could be used to construct confidence intervals, doing so would require estimating the asymptotic constants $c_k$ or $C_k$ in addition to $\chi_{\theta_a}$. Since $c_k$ and $C_k$ depend on the $k$th derivative of $\varphi_a$, these constants are difficult to estimate, and such a plug-in approach may result in substantial under-coverage in moderate sample sizes. Instead, we adapt the methods proposed in \cite{deng2022inference} to our setting, which removes the need to estimate $c_k$ or $C_k$, but not the need to estimate $\chi_{\theta_a}$. 

We recall that $\widehat{p}_{a,n}=\psi^*(\widehat{F}^c_{a,n})$.
Recalling that $\mc S_n$ is the grid used to isotonize the one-step counterfactual CDF estimator, we define the set of knots of $\widehat{p}_{a,n}$ as
\begin{align}	
\widehat{\mc L}_{a,n}:=\left\{t\in \mc S_n: \widehat\varphi'_{a,n}(t-)>\widehat\varphi'_{a,n}(t+)\right\} \cup \left\{s_{a_n},s_{b_n}\right\},\label{define.knots}
\end{align}
where $a_n:=\min\{k:1\le k\le m_n,\,\widehat{ F}_{a,n}>0\}$ and $b_n:=\max\{k:1\le k \le m_n,\, \widehat{ F}_{a,n}<1\}$.
The set $\widehat{\mc L}_{a,n}$ is well-defined and has finite cardinality because  $\widehat\varphi_{a,n}$ is piecewise linear and $\widehat\varphi_{a,n}=-\infty$ on $\RR\backslash [s_{a_n},s_{b_n}]$, and the knots only appear in the ordered observations, which is a subset of $\mc S_n$ in our case (see \citealp{DR2009LC} for a detailed justification).
We then define the two adjacent knots to $s_0$ as
\begin{align}\label{knots}
\tau_n^{+}(s_0;a):=\inf\{t\in\widehat{\mc L}_{a,n}:t>s_0\},~{\rm and}~\tau_n^{-}(s_0;a):=\sup\{t\in\widehat{\mc L}_{a,n}:t<s_0\}.
\end{align}
We suppress the dependence of $\tau_n^{+}$ and $\tau_n^{-}$ on $s_0$ and $a$  for notational simplicity.

As in Theorem 2.4 of \cite{deng2022inference}, we define 
\begin{align}
 &\LL_k^{(0)}:= \left(h_{k;-}^*+h_{k;+}^* \right)^{1/2} H_k^{(2)}(0) \text{ and }\LL_k^{(1)}:= \left(h_{k;-}^*+h_{k;+}^*\right)^{3/2}H_k^{(3)}(0),\label{pivots}
\end{align} 
where $h_{k;-}^*$ and $h_{k;+}^*$ are the absolute values of the location of the first touch points of the pair $(H_k,Y_k)$ defined prior to Theorem \ref{prop:CI} to 0 from the left and right, respectively (see the paragraph preceding Lemma \ref{lem:finallemmaforCI} of the supplement for more details). Quantiles of the distributions of $\LL_k^{(0)}$ and $\LL_k^{(1)}$ and their absolute values for $k= 2$ are displayed in Tables 1 and 2 of \cite{deng2022inference}.
Using Theorem \ref{prop:CI} and the method proposed by \citet{deng2022inference}, we define  symmetric $(1-\alpha)$-level CI's for $p_a(s_0)$ and $p_a'(s_0)$ as follows:
\begin{align}
&\mc I_{a,n}^{(0)}(\alpha;s_0) := \left[\widehat p_{a,n}(s_0)\pm   \left( \widehat\chi_{\theta_a} / \{n (\Delta \tau_{a,n}) \} \right)^{1/2} c_\alpha^{(0)}\right],\label{CI.for.density}\\
&\mc I_{a,n}^{(1)}(\alpha;s_0) := \left[\widehat p'_{a,n}(s_0)\pm \left ( \widehat\chi_{\theta_a} /  \{ n(\Delta\tau_{a,n})^3 \} \right)^{1/2} c_\alpha^{(1)}\right],\label{CI.for.density.deriv}
\end{align}
where $c_\alpha^{(j)}$ are the $1-\alpha$ quantiles of the distribution of $|\LL_k^{(j)}|$  for $j=0,1$, $\Delta \tau_{a,n} := \tau_n^{+}(s_0;a)-\tau_n^{-}(s_0;a)$ is the distance between the nearest knots to $s_0$, and $\widehat\chi_{\theta_a}$ is an estimator of $\chi_{\theta_a}$. 
We suppress the dependence of $\Delta \tau_{a,n} $ on $s_0$ for notational simplicity.
We will discuss estimation of  $\chi_{\theta_a}$ in Section~\ref{subsec:tuning.param} below.

The following theorem shows that the  CI's  proposed in~\eqref{CI.for.density} and~\eqref{CI.for.density.deriv} have asymptotically valid coverage as long as $\widehat\chi_{\theta_a}$ is consistent.
\begin{theorem}\label{prop:CI-construction}
Under the same assumptions as in Theorem \ref{prop:CI}, 
\begin{align}\label{joint:CI-construction}
\begin{pmatrix}
\sqrt{n(\Delta\tau_{1,n})}\{\widehat p_{1,n}(s_0)-p_1(s_0)\}\\
\sqrt{n(\Delta\tau_{1,n})^3}\{\widehat p'_{1,n}(s_0)-p'_1(s_0)\}\\
\sqrt{n(\Delta\tau_{0,n})}\{\widehat p_{0,n}(s_0)-p_0(s_0)\}\\
\sqrt{n(\Delta\tau_{0,n})^3}\{\widehat p'_{0,n}(s_0)-p'_0(s_0)\}\\
\end{pmatrix}
\rightarrow_d 
\begin{pmatrix}
-\sqrt{\chi_{\theta_1}}\LL_{1k}^{(0)}\\
-\sqrt{\chi_{\theta_1}}\LL_{1k}^{(1)}\\
-\sqrt{\chi_{\theta_0}}\LL_{0k}^{(0)}\\
-\sqrt{\chi_{\theta_0}}\LL_{0k}^{(1)}\\
\end{pmatrix},
\end{align}
where the processes $\LL_{1k}^{(i)}$ and $\LL_{0k}^{(i)}$ are independent copies of $\LL_{k}^{(i)}$ for $i=0,1$ defined in \eqref{pivots}.
Hence, if  $\widehat\chi_{\theta_a} \to_p \chi_{\theta_a}$, then for any $\alpha\in (0,1)$, and $a\in\{0,1\}$.
\begin{align*}
\lim_{n\rightarrow \infty}P_*(p_a(s_0)\in \mc I_{a,n}^{(0)}(\alpha;s_0))=\lim_{n\rightarrow \infty}P_*(p_a'(s_0)\in \mc I_{a,n}^{(1)}(\alpha;s_0))=1-\alpha.
\end{align*}
\end{theorem}

The proof of Theorem~\ref{prop:CI-construction} is provided in Section \ref{subsec:proof.CI.constr} of the supplementary material.
As noted above, our CI's are preferable to direct plug-in CI's based on Theorem \ref{prop:CI} because they do not require estimation of higher derivatives of $p_a$. The distance between the left and right knots adjacent to $s_0$ is used to standardize the distribution of $\widehat p_{a,n}(s_0)$ and $\widehat p_{a,n}'(s_0)$ instead. However, our CI's still require estimating $\chi_{\theta_a}$, which is the subject of the next section.

Theorem \ref{prop:CI-construction} can also be used to construct asymptotically valid CI's for contrasts of the counterfactual density functions; we get the following distributional approximation:
\begin{align*}
\left\{ \widehat p_{1,n}(s_0)-p_1(s_0)\right\} - \left\{\widehat p_{0,n}(s_0)-p_0(s_0)\right\}  \asymp_d -\sqrt{\frac{\hat \chi_{\theta_1}}{n(\Delta\tau_{1,n})}}\LL_{1k}^{(0)} +\sqrt{\frac{\hat \chi_{\theta_0}}{n(\Delta\tau_{0,n})}}\LL_{0k}^{(0)}.
\end{align*}
A symmetric $(1-\alpha)$-level CI for the density difference $p_1(s_0)-p_0(s_0)$ is then given by
\begin{align}\label{def:contrast.CI}
\mathcal{I}_{{\rm diff},\alpha,n}(\alpha;s_0):=\left\{\widehat p_{1,n}(s_0)-\widehat p_{0,n}(s_0)\right\}\pm c_{{\rm diff},\alpha,n}(s_0),
\end{align}
where $c_{{\rm diff},\alpha,n}(s_0)$ is the $1-\alpha$ quantile of $|\sqrt{\hat \chi_{\theta_0}/(n\Delta\tau_{0,n})}\LL_{0k}^{(0)}-\sqrt{\hat \chi_{\theta_1}/(n\Delta\tau_{1,n})}\LL_{1k}^{(0)}|$, and $v\pm w$ denotes the interval $[v-w,v+w]$ for $v,w\in\RR$.

Recalling that $\widehat\varphi_{a,n}=\log(\widehat p_{a,n})$, for $a\in\{0,1\}$,
Theorem \ref{prop:CI-construction}  also yields
\begin{align*}
\left\{\widehat\varphi_{1,n}(s_0)-\varphi_1(s_0)\right\}-\left\{\widehat\varphi_{0,n}(s_0)-\varphi_0(s_0))\right\} &\asymp_d \hat\alpha_0/\sqrt{n(\Delta\tau_{0,n})} \LL_{0k}^{(0)}- \hat\alpha_1/\sqrt{n(\Delta\tau_{1,n})} \LL_{1k}^{(0)},
\end{align*}
where $\hat\alpha_a:=\sqrt{\hat\chi_{\theta_a}}/\widehat p_{a,n}(s_0)$, for $a\in\{0,1\}$.
A symmetric $(1-\alpha)$-level CI for the log density ratio $\log(p_1(s_0)/p_0(s_0))$ is given by
\begin{align}\label{def:ratio.CI}
\mathcal{I}_{{\rm ratio},\alpha,n}(\alpha;s_0):=\log\left(\widehat p_{1,n}(s_0)/\widehat p_{0,n}(s_0)\right)\pm c_{{\rm ratio},\alpha,n}(s_0),
\end{align}
where $c_{{\rm ratio},\alpha,n}(s_0)$ is the $1-\alpha$ quantile of $|\hat\alpha_0/\sqrt{n(\Delta\tau_{0,n})}\LL_{0k}^{(0)}-\hat\alpha_1/\sqrt{n(\Delta\tau_{1,n})}\LL_{1k}^{(0)}|$. This CI can be exponentiated to provide a CI for the density ratio.
These two CI's, \eqref{def:contrast.CI}--\eqref{def:ratio.CI}, provide methods for identifying differences between the counterfactual densities.


\subsubsection{Doubly robust estimation of $\chi_{\theta_a}$}\label{subsec:tuning.param}

We now provide a doubly-robust estimator of $\chi_{\theta_a}$.  We define the limiting uncentered influence function of $F_a(s)$ as 
\begin{align}
   D_{a,\theta_{a,\infty}}(s)
    &={I(A=a) \over \pi_{a,\infty}(\mathbf{X})}\left[I(Y\le s) -\phi_{a,\infty}(s|\mathbf{X}) \right]+\phi_{a,\infty}(s|\mathbf{X}) ,\label{term:Dthetainfty}
\end{align}
with ${\theta}_{a,\infty}=({\phi}_{a,\infty},{\pi}_{a,\infty})$. 
We also define the estimated influence function as  $D_{a,\widehat\theta_a}$ for $\widehat{\theta}_a=(\widehat{\phi}_a,\widehat{\pi}_a)$, and we note that  $\widehat F_{a,n}(s) = \PP_n D_{a,\widehat\theta_a}(s)$ by \eqref{form:onestepest}.  
We then suggest
\begin{equation}\label{chi.theta.hat}
	\hat\chi_{\theta_{a},n} := \frac{\PP_n\left\{ D_{a,{\widehat\theta_{a}}}(s_0+h_n)-D_{a,{\widehat\theta_{a}}}(s_0) \right\}^2+\PP_n\left\{D_{a,{\widehat\theta_{a}}}(s_0-h_n)-D_{a,{\widehat\theta_{a}}}(s_0)\right\}^2}{2h_n},
\end{equation}
an estimator of $\chi_{\theta_a}$, where $h_n=n^{-b}$ for some $b > 0$ is a tuning parameter. We show in our simulations that the CI's are robust to the choice of this tuning parameter. We recommend setting $h_n = n^{-1/10}$ for simplicity. The following lemma, whose proof is given in Section \ref{subsec:tuning.proof} of the supplementary material,  demonstrates that $\hat\chi_{\theta_a,n}$ is a consistent estimator of $\chi_{\theta_a}$.
\begin{lemma}\label{lemma:tuning.consistency}
If the assumptions of Theorem \ref{prop:CI} hold, $P_* R_1^4 < \infty$, and $h_n^{-1}=O(n^{1/2})$, then $\hat\chi_{\theta_a,n}\rightarrow_p\chi_{\theta_a}$.
\end{lemma}

\subsection{Construction of confidence bands}\label{subsec:band.const}
In this section, we propose an approach for constructing confidence bands for the counterfactual density function $p_a$ over the real line, leveraging recent advancements in log-concave density estimation. Our method builds upon the framework developed by \citet{walther2022confidence}, which provides simultaneous confidence bounds for a log-concave density with finite-sample guarantees at the $1-\alpha$ confidence level. Their approach relies on the use of observed order statistics. To extend this method to the causal setting, we utilize order statistics generated from the estimated log-concave counterfactual distribution function, denoted as $\widehat F^{-1}_{a,n}(U_{(i)})$, where $U_{(i)}$ is the $i$th order statistic obtained from an independent sample of size $n$ drawn from a uniform distribution on $(0,1)$, for $i=1,\ldots,n$.
Subsequently, the upper and lower simultaneous confidence bounds $(\hat\ell_{a,n},\hat u_{a,n})$ for $\log(p_a)$ can be obtained by applying Algorithm 1 of \citet{walther2022confidence} for each $a\in\{0,1\}$.
We leave the theoretical validation and further exploration of these bands through empirical studies for future research.

\section{Simulation study}\label{Section:Simul}

In this section, we conduct numerical experiments to assess our proposed estimator's performance.  For a given sample size $n$, we simulate data for each $i=1,\ldots,n$ in the following steps.
First,  we generate $\mathbf{X}_i=(X_{i1},\ldots,X_{i4})$, where $X_{i1},\ldots,X_{i4}$ are i.i.d.\ from $U[0,1]$, i.e., the continuous uniform distribution on $[0,1]$. 
Given $\mathbf{X}_i=\mathbf{x}_i$, we sample $A_i$ from a Bernoulli distribution with probability $p_i=\exp{(v_i)}/[1+\exp{(v_i)}]$ for $v_i=-1.5+0.25x_{i1}+0.5x_{i2}+0.75x_{i3}+x_{i4}$.
Finally, given $\mathbf{X}_i=\mathbf{x}_i$ and $A_i=a_i$, we generate $Y_i$ from $U[(8-4a_i)+(2a_i-2)s_i, (8-4a_i)+2a_is_i]$, where $s_i=\sum_{j=1}^4 x_{ij}$.  
The closed-form  expression of the marginal density $p_a$ of $Y^a$ is given in Section \ref{section:pas}
of the supplementary material, and  $p_0$ and $p_1$ are displayed in Figure~\ref{p1.and.p0}. The means of $p_0$ and $p_1$ are both equal to 6 and the variances are both equal to $16/9$, but the shapes of $p_0$ and $p_1$ are very different. 

For each $n\in\{500,1000,2500,4000,6000,8000\}$, we simulate $1000$ datasets using the above method. For each dataset, we estimate the counterfactual densities using our proposed estimator $\widehat p_a$ and the basis expansion method of \citet{KBWcounterfactualdensity}. We attempted to compare our estimator to \citet{kim2018causal} as well, but the estimator may be negative and so requires truncation of negative values to zero and renormalization, and we experienced numerical instability in this computation. Hence, we omitted this method from our comparisons. In general, we expect that many of the strengths and weaknesses of kernel and shape-constrained density estimators are likely to carry over to the causal setting.  We also compare to the log-concave MLE without covariate adjustment. 
We call this the ``naive log-concave MLE.''

To assess the double-robustness of the estimators, we consider three settings as follows: both $\widehat\pi_{a}$ and $\widehat\phi_{a}$ are well-specified (Case 1); only $\widehat\pi_{a}$ is well-specified (Case 2); only $\widehat\phi_{a}$ is well-specified (Case 3).
The details for estimating the nuisance functions can be found in Section \ref{subsec:est.nuisances} of the supplementary material. 
We also include the  sample splitting version of our estimator studied in Section \ref{section:sample.splitting} of the supplementary material with $K_0 = 5$ folds.

For the basis method of \citet{KBWcounterfactualdensity}, we use their one-step projection estimator with a cosine basis series.   
We select this tuning parameter in an oracle fashion.
Table \ref{table:basis.number} in Section \ref{appd:basis.table} of the supplementary material contains the oracle number of basis functions selected for each case and sample size. 
The details for the projection estimator are provided in Section \ref{subsec:est.KBW} of the supplementary material.

We measure each estimator's performance using the average $L_1$ distance between the estimated density and the truth over the 1000 replications for each sample size. 
We experienced some numerical instability when computing the $L_1$ distance of the basis expansion method, especially when the number of basis functions was more than 25.
In our reported $L_1$ averages, we dropped the instances where we could not compute the  $L_1$ distance. We also compare the empirical coverage of 95\% CI's based on each estimator.
We construct CI's for our estimators using the procedures described in Sections~\ref{subsec:CIconst} and~\ref{subsec:CI-split} of the supplement. We use $h_n = n^{-1/10}$ for the tuning parameter in the estimator of $\chi_{\theta_a}$. We construct 95\% CI's for the naive log-concave estimator using the  procedure of \citet{deng2022inference}. For the basis expansion method, we construct CI's using the \texttt{npcausal} package \citep{npcausal}.

\begin{figure}[ht!]
    \centering
     \begin{subfigure}[$L_1$ error ($p_1$, Case 1)]{\label{Hell.for.p1.c1}\includegraphics[width=40mm]{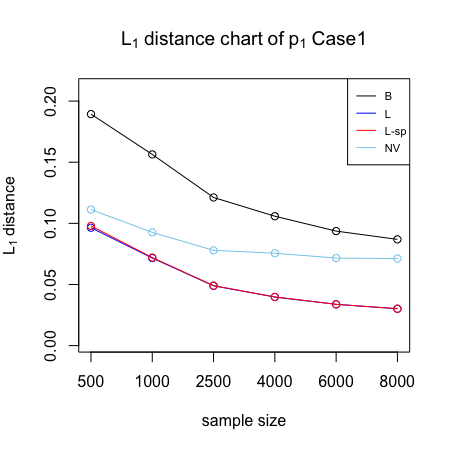}}
    \end{subfigure}
    \hspace{0.2cm}
    \begin{subfigure}[$L_1$ error ($p_1$, Case 2)]{\label{Hell.for.p1.c2}\includegraphics[width=40mm]{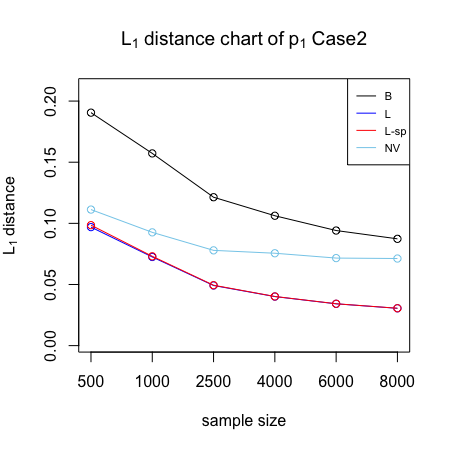}}
    \end{subfigure}
    \hspace{0.2cm}
    \begin{subfigure}[$L_1$ error ($p_1$, Case 3)]{\label{Hell.for.p1.c3}\includegraphics[width=40mm]{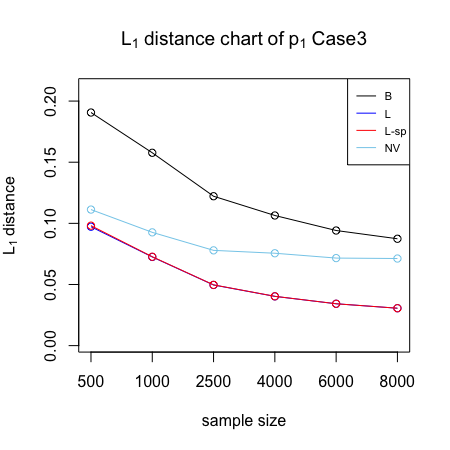}}
    \end{subfigure}\\
   \vspace{0.2cm}
    \begin{subfigure}[CI coverage probability ($p_1$, Case 1, $n=8000$)]
    	{\label{8000P1C1}\includegraphics[width=50mm]{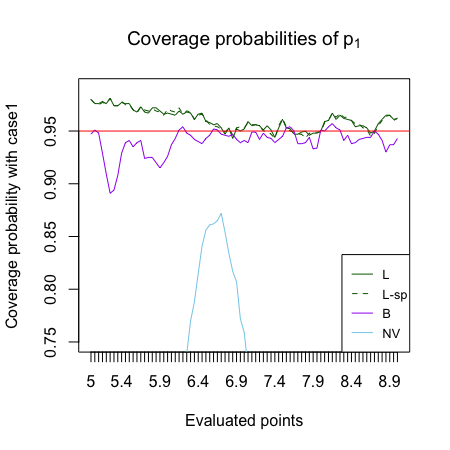}}
    \end{subfigure}
   \hspace{0.2cm}
   \begin{subfigure}[Average CI width ($p_1$, Case 1, $n=8000$)]
   	{\label{8000P1C1w}\includegraphics[width=50mm]{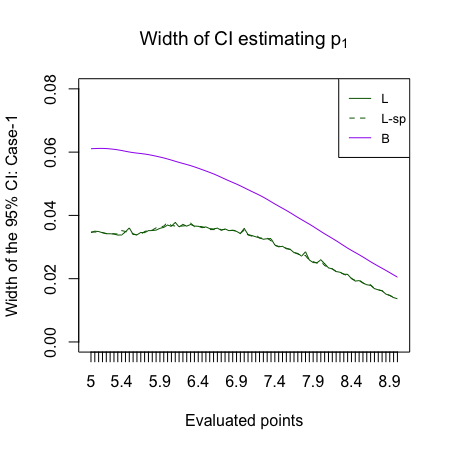}}
   \end{subfigure}\\
   \vspace{0.2cm}
       \begin{subfigure}[True $p_1$ and $p_0$]{\label{p1.and.p0}\includegraphics[width=50mm]{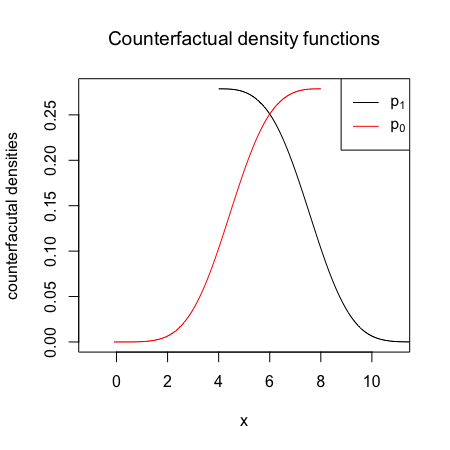}}
   \end{subfigure}
   \hspace{0.2cm}
    \begin{subfigure}[Density estimates ($p_1$, Case 1, $n=4000$)]{\label{pilot}\includegraphics[width=50mm]{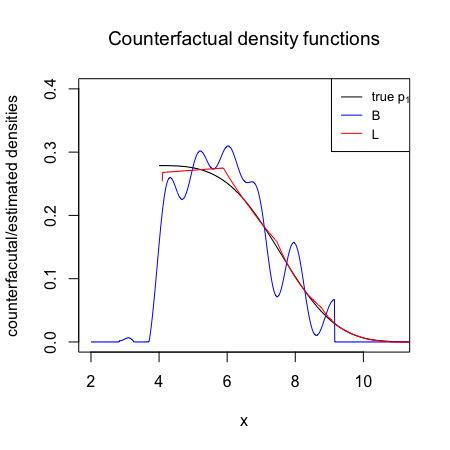}}
   \end{subfigure}
    \caption{\label{fig:Hellinger.linecharts} (a)--(c) Average $L_1$ distance between estimators and true $p_1$. 
    In Cases 1, 2, and 3, both nuisance functions, only the propensity score, and only the conditional CDF are well-specified, respectively.
    (d) Empirical coverage probabilities of 95\% CI's. 
    (e) The corresponding widths for each CI from (d). The lines for L and L-sp cannot be visually distinguished for (a)--(e).
    (f) The true counterfactual density functions used for the simulations.
    ``B", ``L", ``L-sp", and ``NV" stand for the basis expansion, log-concave, log-concave with sample splitting, and naive log-concave estimators, respectively. 
    (g) The true density $p_1$ and the density estimates for $p_1$ by L and B from a single simulation based on $n=4000$ in Case 1.
 }
\end{figure}

Figures \ref{Hell.for.p1.c1}--\ref{Hell.for.p1.c3} display the average $L_1$ distances of the estimators of $p_1$ as a function of $n$ in the three nuisance estimation scenarios. The results for $p_0$ are very similar, and can be found in the supplementary material. 
Our proposed log-concave estimator consistently had the smallest average $L_1$ distance of the three methods for all $n$ values.
The average $L_1$ distance decreased as a function of $n$ for our estimator and the basis expansion estimator, but not for the naive log-concave estimator, which was expected because the counterfactual and marginal densities are different due to confounding.
The average $L_1$ distances for our method with and without sample splitting were very similar. We expect the difference to be more substantial if more complicated nuisance estimators were used.
The average $L_1$ distances were similar across the three nuisance estimator specifications, validating the double-robustness of the log-concave and basis expansion methods.

Figures \ref{8000P1C1}--\ref{8000P1C1w} display the coverage and average width of 95\% CI's for $p_1$ for $n = 8000$ in the well-specified nuisance scenario (Case 1). Figures \ref{fig:kplot2s.case.1.p1}--\ref{fig:kplot2s.case.3.p0} and \ref{fig:wplots.case.1.p1}--\ref{fig:wplots.case.3.p0} in the supplementary material display the results for other cases, $p_0$, and other sample sizes.
Our CI's exhibited some undercoverage where the true density is close to zero when the sample size was smaller (i.e., $n=500,1000$), and exhibited overcoverage where the true density is large.  
\citet{deng2022inference} observed similar phenomena in the non causal setting. Sample splitting did not meaningfully change coverage in this simulation design.
The basis expansion method had substantial fluctuation in the coverage from point to point, including undercoverage for some points even with $n = 8000$.
This is because the basis expansion method is centered around an approximation to the density rather than around the true density, and this bias interferes with constructing CI's with valid coverage.
The average widths of 95\% CI's for our method were also smaller than those of the basis expansion method in all sample sizes and cases.
Furthermore, the average widths of the 95\% CI's for our methods were nearly identical across Cases 1 through 3.
Figure \ref{pilot} displays a single simulation result in Case 1 for estimating $p_1$ with $n=4000$.
The number of basis function was 18, which was selected because it minimized average $L_1$ distance as described above.
As opposed to our estimator, the basis expansion method suffers from boundary issues as well as instability across the domain, which might be exacerbated by the bounded support.
Finally, the naive log-concave did not achieve nominal coverage in any case because it is inconsistent, demonstrating again that covariate adjustment in the causal setting is essential for valid estimation and inference. 

Figures \ref{fig:contplots.ww}--\ref{fig:contplots.mw} illustrate the empirical coverage rates of symmetric 95\% CI's for the difference $p_1-p_0$, and Figures \ref{fig:contplot2s.ww}--\ref{fig:contplot2s.mw} display the empirical coverage rates of symmetric 95\% CI's for the log-ratio $\log(p_1/p_0)$.
The exact procedure for computing quantiles of the reference distributions is provided in Section \ref{subsec:quantile.generation} of the supplementary material.
For $n=500,\,1000$ the average coverage for both diference and log-ratio CI's was around 98\%, indicating a notable degree of conservatism.
As the sample size increases to $n=6000,\,8000$, the coverage decreases to approximately 97\% yet still remains conservative. 
Hence, these CI's were roughly twice as conservative as the individual density CI's.



\spacingset{2} 

\clearpage

\title{Supplementary material for ``Doubly robust estimation and inference for a log-concave counterfactual density"}
\maketitle

\appendix

\section{Sample splitting}\label{section:sample.splitting}

In Theorems \ref{prop:Consistency}, \ref{prop:CI}, and \ref{prop:CI-construction}, we required that the nuisance estimators $\widehat\pi_a$ and $\widehat \phi_a$ fall into classes of functions that satisfy complexity conditions in order to control empirical process terms. In particular, Assumption~\ref{assm:EC1} used for Theorem~\ref{prop:Consistency} required that $\widehat\pi_a$ and certain transformations of  $\widehat \phi_a$ fall in to $P_*$- Glivenko-Cantelli classes, and Assumption~\ref{assm:pi-bracketing} used for Theorems \ref{prop:CI} and \ref{prop:CI-construction} required that $\widehat\pi_a$ and $\widehat \phi_a$ fall into classes that satisfy uniform entropy bounds. The \textit{sample splitting} (also known as \textit{cross-fitting} or \textit{double machine learning}) approach has been shown to avoid such complexity constraints, which yields improved performance in high-complexity regimes \citep{van2011cross,Cherno.double.debiased,belloni2018uniformly,kennedy2019nonparametric}. In this section, we propose a counterfactual density estimator based on sample splitting, and we provide analogues of Theorems \ref{prop:Consistency}, \ref{prop:CI}, and \ref{prop:CI-construction} demonstrating that the asymptotic behavior of this estimator does not rely on complexity constraints on $\widehat\pi_a$ and $\widehat \phi_a$.

We assuming that $N := n/K_0$ is an integer for an integer $K_0\ge 2$ for convenience. We randomly partition the  indices $\{1,\ldots,n\}$ into $\mc V_{n,1},\ldots,\mc V_{n,K}$ where the cardinality of $\mc V_{n,k}$ satisfies $|\mc V_{n,k}|=N$ for each $k=1,\ldots,K_0$.  For our proofs of Theorems~\ref{prop:Cons-split}--\ref{prop:CI-construction-split} below, we require the number of folds satisfies $K_0=O_p(1)$.
For each $k\in\{1,\ldots,K_0\}$, we denote 
$\mc T_{n,k} := \{1,\ldots, n\} \setminus \mc V_{n,k} $ 
as the indices of the training set for the $k$-th fold, and we assume that the nuisance estimators $\widehat\theta_{a,-k}=\{\widehat \pi_{a,-k},\, \widehat\phi_{a,-k}\}$ for the $k$-th fold are functions of the the training set $\{ \mathbf{Z}_i : i \in \mc T_{n,k}\}$.
We then define the cross fitted one-step estimator $\widehat F_{a,n}^{K_0}$ as
\begin{align}\label{crossfitted.onestep}
\widehat F^{K_0}_{a,n}(s)=\frac{1}{K_0}\sum_{k=1}^{K_0}\frac{1}{N}\sum_{i\in\mc V_{n,k}} \left\{\frac{I(A_i=a)}{\widehat\pi_{a,-k}(\mathbf{X}_i)}\left[I(Y_i\le s)-\widehat\phi_{a,-k}(s|\mathbf{X}_i)\right] +\widehat\phi_{a,-k}(s|\mathbf{X}_i)\right\}, 
\end{align}
for each $s\in\RR$.
We then apply Steps \ref{Step3}--\ref{Step5} with $\widehat F^{K_0}_{a,n}$ in place of $\widehat F_{a,n}$ to arrive at our cross-fitted log-concave counterfactual density estimator $\widehat p^{K_0}_{a,n}$.

\subsection{Consistency}

We now provide an analogue of the consistency result Theorem \ref{prop:Consistency} for the sample splitting estimator. We begin by stating conditions we will require.
\begin{assumptionp}{E$'$}\label{assm:E.split} \phantom{blah}
There exist functions $\pi_{a,\infty},\phi_{a,\infty}$ such that:
\begin{enumerate}
    \myitem{(E0$'$)}: \label{assm:samplesplit} The estimators $\widehat\pi_{a,-k}$ and $\widehat\phi_{a,-k}$ are obtained from sample splitting for each $k\in\{1,\ldots,K_0\}$.
        \myitem{(E1$'$)}: \label{assm:nuisanceconvergence.split} For $a \in \{0,1\}$, the estimated nuisance functions $\widehat{\pi}_{a,-k},~\widehat\phi_{a,-k}$ satisfy
    \begin{gather*}
        \max_{1\le k\le K_0}P_*\left[ \widehat{\pi}_{a,-k}(\mathbf{X})-\pi_{a,\infty}(\mathbf{X}) \right]^2\rightarrow_p 0,\\
        \max_{1\le k\le K_0} P_*\left[\int_{-\infty}^{\infty}\left|\widehat\phi_{a,-k}(s|\mathbf{X})-{\phi}_{a,\infty}(s|\mathbf{X}))\right| \, ds\right]^2\rightarrow_p 0.
    \end{gather*}
        \myitem{(E2$'$)}: \label{assm:boundedpropscore.split} There exists $K \in (0, \infty)$ such that $\|1/{\pi_{a,\infty}}\|_{\infty},\| 1/{\widehat\pi_{a,-k}}\|_{\infty}\le K$ a.s.\ for all $k\in\{1,\ldots,K_0\}$.
        \myitem{(E3$'$)}: \label{assm:properCDF.split} $s \mapsto \widehat\phi_{a,-k}(s|\mathbf{X})$ and $s \mapsto \phi_{a,\infty}(s|\mathbf{X})$ are a.s.\ proper conditional CDFs for all $k\in\{1,\ldots,K_0\}$. And, there exists $h \in L_2(P_*)$ such that $\int_{\mathbb{R}}|s| \, d\widehat\phi_{a,-k}(s|\mathbf{x})\leq h(\mathbf{x})$ $P_*$-a.e.\ $\mathbf{x}$, for all $k\in\{1,\ldots,K_0\}$ and all $n>N_0$ for sufficiently large $N_0>0$.
\end{enumerate}   
\end{assumptionp}

We have the following consistency result for the sample splitting estimator, which is an analogue of Theorem \ref{prop:Consistency}.  Our proof is provided in Appendix \ref{subsec:proof.cons.split}.
 
\begin{theorem}\label{prop:Cons-split}
If conditions~\ref{assm:consistency}--\ref{assm:logconcv}, 
\ref{assm:condition.grid}--\ref{assm:max.grid}, \ref{assm:samplesplit}--\ref{assm:properCDF.split}, and \ref{assm:DR}--\ref{assm:phiregularity} hold, then
\begin{align}
 \int_{\mathbb{R}} e^{\varepsilon|s|}|\widehat{p}^{K_0}_{a,n}(s)-p_a(s)| \, ds \rightarrow_p 0,\label{prop:consist-1split}
\end{align}
as $n\rightarrow \infty$ for $a\in\{0,1\}$ and for all $\varepsilon \in (0,\alpha)$, where $\alpha>0$ is such that $p_a(s) \le e^{-\alpha|s|+\beta}$ for all $s \in \mathbb{R}$ and  some $\beta \in \mathbb{R}$. If in addition $p_a$ is continuous on $\mathbb{R}$, then
\begin{align}
 &\sup_{s\in\mathbb{R}} e^{\varepsilon|s|}|\widehat{p}^{K_0}_{a,n}(s)-p_a(s)| \rightarrow_p 0.\label{prop:consist-2split}
\end{align}
\end{theorem}

Conditions \ref{assm:nuisanceconvergence.split}--\ref{assm:properCDF.split} are analogous to \ref{assm:nuisanceconvergence}--\ref{assm:properCDF}, and discussion of the latter  can be found following Theorem \ref{prop:Consistency}. Notably,  the sample splitting scheme enables us to avoid the Glivenko-Cantelli conditions \ref{assm:basicGC}--\ref{assm:integratedGC}, so that the nuisance estimators may be arbitrarily complicated.

\subsection{Limit distribution}

We now demonstrate that the estimator based on sample splitting has the same asymptotic distribution as the original estimator provided in Theorem~\ref{prop:CI}. We first state a condition we will require. For $k \in\{1,\ldots,K_0\}$, we define
\[ \mc G_{-k}(s,t;\mathbf Z;S)=\|(\widehat\phi_{a,-k}-\phi_{a,\infty})(t|\cdot)-(\widehat\phi_{a,-k}-\phi_{a,\infty})(s|\cdot)\|_S.\]

\begin{assumptionp}{E$'$ (cont.)}\label{assm:E2}\phantom{blah}
\begin{enumerate}
   \myitem{(E6$'$)} \label{assm:L2conditionsforCI.split} For all $s_1,s_2\in  I_{s_0,\omega}$, the following statements hold: 
    \begin{gather*}
    \max_{1\le k \le K_0}\mc G_{-k}(s_1,s_2;\mathbf{Z};\mathcal{S}_1\cap\mathcal{S}_2) \|\widehat\pi_{a,-k}-\pi_{a,\infty}\|_{\mathcal{S}_1\cap\mathcal{S}_2} = |s_2-s_1|M_1',\\
     \max_{1\le k \le K_0} \|\widehat\pi_{a,-k}-\pi_{a,\infty}\|_{\mathcal{S}_1\cap\mathcal{S}_2^c} = o_p(n^{-k/(2k+1)}),\,\max_{1\le k \le K_0} \mc G_{-k}(s_1,s_2;\mathbf{Z};\mathcal{S}_1\cap\mathcal{S}_2^c) = |s_2-s_1|M_2',\\
      \max_{1\le k \le K_0}\|\widehat\pi_{a,-k}-\pi_{a,\infty}\|_{\mathcal{S}_1^c\cap\mathcal{S}_2} = o_p(1),\,\max_{1\le k \le K_0}\mc G_{-k}(s_1,s_2;\mathbf{Z};\mathcal{S}_1^c\cap\mathcal{S}_2) = |s_2-s_1|M_3',
    \end{gather*}
    where $M_1'$, $M_2'$, and $M_3'$ are random variables that do not depend on $s_1,s_2$ and such that $M_1'$ and $M_3'$ are $o_p(n^{-k/(2k+1)})$ and $M_2' = o_p(1)$.
\end{enumerate}
\end{assumptionp}

\noindent 
We now state the analogue of Theorem \ref{prop:CI} for the estimator based on sample splitting. The proof is given in Section \ref{subsec:proof.CI.split}.
We define $\widehat\varphi^{K_0}_{a,n}:=\log(\widehat p^{K_0}_{a,n})$ for $a\in\{0,1\}$.

\begin{theorem}\label{prop:CI.split}
If \ref{assm:consistency}-\ref{assm:logconcv}, \ref{assm:condition.grid},  \ref{assm:samplesplit}--\ref{assm:properCDF.split}, \ref{assm:DR}--\ref{assm:phiregularity},  \ref{assm:L2conditionsforCI.split}, \ref{assm:phiinfholder}, and \ref{assm:pregularity}-\ref{assm:pdiffwithk} hold for $a\in\{0,1\}$, $\delta_n=O_p(n^{-(k+1)/(2k+1)})$, then
\begin{align}\label{jointfororiginal:both.0and1.split}
  \begin{pmatrix}
    n^{k/(2k+1)}(\widehat p^{K_0}_{1,n}(s_0)-p_1(s_0))\\
n^{(k-1)/(2k+1)}(\widehat p'^{K_0}_{1,n}(s_0)-p'_1(s_0)) \\  
    n^{k/(2k+1)}(\widehat p^{K_0}_{0,n}(s_0)-p_0(s_0))\\
n^{(k-1)/(2k+1)}(\widehat p'^{K_0}_{0,n}(s_0)-p'_0(s_0))
  \end{pmatrix}
  \rightarrow_d 
  \begin{pmatrix}
   c_k(s_0,\varphi_1) H_{1k}^{(2)}(0)\\
d_k(s_0,\varphi_1) H_{1k}^{(3)}(0) \\ 
   c_k(s_0,\varphi_0) H_{0k}^{(2)}(0)\\
d_k(s_0,\varphi_0) H_{0k}^{(3)}(0) \\ 
  \end{pmatrix},
  \end{align}
and
 \begin{align}\label{jointforlogscale:both.0and1.split}
   \begin{pmatrix}
     n^{k/(2k+1)}(\widehat\varphi^{K_0}_{1,n}(s_0)-\varphi_1(s_0))\\
n^{(k-1)/(2k+1)}(\widehat\varphi'^{K_0}_{1,n}(s_0)-\varphi'_1(s_0))\\   
n^{k/(2k+1)}(\widehat\varphi^{K_0}_{0,n}(s_0)-\varphi_0(s_0))\\
n^{(k-1)/(2k+1)}(\widehat\varphi'^{K_0}_{0,n}(s_0)-\varphi'_0(s_0))  
  \end{pmatrix}
  \rightarrow_d 
  \begin{pmatrix}
   C_k(s_0,\varphi_1) H_{1k}^{(2)}(0)\\
D_k(s_0,\varphi_1) H_{1k}^{(3)}(0)  \\ 
   C_k(s_0,\varphi_0) H_{0k}^{(2)}(0)\\
D_k(s_0,\varphi_0) H_{0k}^{(3)}(0)  \\ 
  \end{pmatrix},
 \end{align}
\end{theorem}
 
Condition \ref{assm:L2conditionsforCI.split} is analogous to \ref{assm:L2conditionsforCI}, which was discussed following Theorem~\ref{prop:CI}. As above, sample splitting allows us to avoid the entropy condition \ref{assm:pi-bracketing} and the H{\"o}lder condition \ref{assm:phiestholder} used to control empirical process terms.

\subsection{Confidence intervals}\label{subsec:CI-split}

Finally, we demonstrate that confidence intervals based on the sample-splitting estimator constructed in the same manner as those defined in Section~\ref{subsec:CIconst} are asymptotically valid.  As with the proof of Theorem \ref{prop:CI-construction}, the proof of Theorem~\ref{prop:CI-construction-split} below is a direct consequence of Theorem \ref{prop:CI.split}, so we omit it. 

Analogously to \eqref{define.knots} and \eqref{knots}, we denote the knots of $\widehat\varphi_{a,n}^{K_0}=\log(\widehat p_{a,n}^{K_0})$ as $\widehat{\mc L}^{K_0}_{a,n}$ and the knots adjacent to $s_0$ as $\tau_n^{+;K_0}(s_0;a)$ and $\tau_n^{-;K_0}(s_0;a)$. 
We then define $\Delta \tau_{a,n}^{K_0} = \tau_n^{+;K_0}(s_0;a)-\tau_n^{-;K_0}(s_0;a)$.
We suppress the dependence of $\Delta \tau_{a,n}^{K_0}$ on $s_0$ for notational simplicity.
As in \eqref{CI.for.density} and \eqref{CI.for.density.deriv}, we define  symmetric $(1-\alpha)$-level CI's for $p_a(s_0)$ and $p_a'(s_0)$ based on the sample splitting estimator as
\begin{align}
	&\mc I_{a,n}^{(0);K_0}(\alpha;s_0) := \left[\widehat p^{K_0}_{a,n}(s_0)\pm   \left( \widehat\chi^{K_0}_{\theta_a} / \{n (\Delta \tau^{K_0}_{a,n}) \} \right)^{1/2} c_\alpha^{(0)}\right],\label{CI.for.density.split}\\
	&\mc I_{a,n}^{(1);K_0}(\alpha;s_0) := \left[(\widehat p^{K_0}_{a,n})'(s_0)\pm \left ( \widehat\chi^{K_0}_{\theta_a} /  \{ n(\Delta\tau^{K_0}_{a,n})^3 \} \right)^{1/2} c_\alpha^{(1)}\right],\label{CI.for.density.deriv.split}
\end{align}
where $\widehat\chi^{K_0}_{\theta_a}$ is an estimator of $\chi_{\theta_a}$. We have the following result regarding asymptotic validity of these CI's.
\begin{theorem}\label{prop:CI-construction-split}
Under the same assumptions as Theorem \ref{prop:CI.split},
\begin{align}\label{joint.split:CI-construction}
\begin{pmatrix}
\sqrt{n(\Delta\tau^{K_0}_{1,n})}(\widehat p^{K_0}_{1,n}(s_0)-p_1(s_0))\\
\sqrt{n(\Delta\tau^{K_0}_{1,n})^3}((\widehat p^{K_0}_{1,n})'(s_0)-p'_1(s_0))\\
\sqrt{n(\Delta\tau^{K_0}_{0,n})}(\widehat p^{K_0}_{0,n}(s_0)-p_0(s_0))\\
\sqrt{n(\Delta\tau^{K_0}_{0,n})^3}((\widehat p^{K_0}_{0,n})'(s_0)-p'_0(s_0))\\
\end{pmatrix}
\rightarrow_d 
\begin{pmatrix}
-\sqrt{\chi_{\theta_1}}\LL_{1k}^{(0)}\\
-\sqrt{\chi_{\theta_1}}\LL_{1k}^{(1)}\\
-\sqrt{\chi_{\theta_0}}\LL_{0k}^{(0)}\\
-\sqrt{\chi_{\theta_0}}\LL_{0k}^{(1)}\\
\end{pmatrix},
\end{align}
where the processes $\LL_{1k}^{(i)}$ and $\LL_{0k}^{(i)}$ are independent copies of $\LL_{k}^{(i)}$ for $i=0,1$ defined in \eqref{pivots}.
Hence, if $\widehat\chi^{K_0}_{\theta_a} \to_p \chi_{\theta_a}$, then for any $\alpha > 0$ and $a \in \{0,1\}$,
\begin{align*}
\lim_{n\rightarrow \infty}P_*(p_a(s_0)\in \mc I_{a,n}^{(0);K_0}(\alpha;s_0))=\lim_{n\rightarrow \infty}P_*(p_a'(s_0)\in \mc I_{a,n}^{(1);K_0}(\alpha;s_0))=1-\alpha.
\end{align*}
\end{theorem}
As in \eqref{def:contrast.CI}, a symmetric $(1-\alpha)$-level CI for $p_1(s_0)-p_0(s_0)$ based on the sample splitting estimator is 
\begin{align}\label{def:contrast.CI.split}
\mathcal{I}^{K_0}_{{\rm diff},n}(\alpha;s_0):=\left\{\widehat p^{K_0}_{1,n}(s_0)-\widehat p^{K_0}_{0,n}(s_0)\right\}\pm c^{K_0}_{{\rm diff},n}(s_0),
\end{align}
where $c^{K_0}_{{\rm diff},n}(s_0)$ is the $(1-\alpha)$ quantile of $\left|\sqrt{\hat \chi^{K_0}_{\theta_{0}}/(n\Delta\tau_{0,n})}\LL_{0k}^{(0)}-\sqrt{\hat \chi^{K_0}_{\theta_{1}}/(n\Delta\tau_{1,n})}\LL_{1k}^{(0)}\right|$.
As in \eqref{def:ratio.CI}, a symmetric $(1-\alpha)$-level CI for $\log(p_1(s_0)/p_0(s_0))$ based on the sample splitting estimator is
\begin{align}\label{def:ratio.CI.split}
\mathcal{I}^{K_0}_{{\rm ratio},\alpha,n}(\alpha;s_0):=\log\left(\widehat p^{K_0}_{1,n}(s_0)/\widehat p^{K_0}_{0,n}(s_0)\right)\pm c^{K_0}_{{\rm ratio},\alpha,n}(s_0),
\end{align}
where $c^{K_0}_{{\rm ratio},\alpha,n}(s_0)$ is the $1-\alpha$ quantile of $\left|\hat\alpha^{K_0}_0/\sqrt{n(\Delta\tau^{K_0}_{0,n})}\LL_{0k}^{(0)}-\hat\alpha^{K_0}_1/\sqrt{n(\Delta\tau^{K_0}_{1,n})}\LL_{1k}^{(0)}\right|$, and $\hat\alpha^{K_0}_a:=\sqrt{\hat \chi^{K_0}_{\theta_{a}}}/\widehat p^{K_0}_{a,n}(s_0)$, for $a\in\{0,1\}$.

To estimate $\chi_{\theta_a}$ using sample splitting, we propose the estimator $\widehat\chi^{K_0}_{\theta_{a},n}$ defined as 
\begin{align}
\begin{split}
 &\frac{1}{2h_nK_0} \sum_{k=1}^{K_0}
  \left[ \PP^k_n\left\{ D_{a,{\widehat\theta_{a,-k}}}(s_0+h_n) - D_{a,{\widehat\theta_{a,-k}}}(s_0)\right\}^2 \right. \\
  &\qquad\qquad\qquad\qquad  \left. +  \PP^k_n\left\{ D_{a,{\widehat\theta_{a,-k}}}(s_0-h_n) - D_{a,{\widehat\theta_{a,-k}}}(s_0)\right\}^2
  \right],
  \end{split}
    \label{chi.theta.hat.split} 
\end{align}
where $\PP_n^k$ is the empirical distribution of the data in the $k$-th fold $\mc V_{n,k}$, and $D_{a,{\widehat\theta_{a,-k}}}$ is the estimated influence function for $\widehat{\theta}_{a,-k}=(\widehat{\phi}_{a,-k},\widehat{\pi}_{a,-k})$ obtained from $\mathcal{T}_{n,k}$.
We again suggest $h_n=n^{-1/10}$ as a tuning parameter. As in Lemma \ref{lemma:tuning.consistency}, consistency of $\widehat\chi^{K_0}_{\theta_{a},n}$ holds under the same conditions as Theorem \ref{prop:CI.split} as long as $h_n^{-1}=O(\sqrt{n})$. Hence, this estimator does not require conditions controlling the complexity of the nuisance estimators.

The sample splitting estimator can be used to construct a uniform confidence band using the approach discussed in Section \ref{subsec:band.const}.

\section{Data analysis}\label{section:real.data.analysis}

We use our method to analyze the \texttt{lalonde} dataset from the R package \texttt{cobalt} \citep{cobalt.pack}.
The data was used by \citet{dehejia.Wahba} to study the effectiveness of a job training program on earnings \citep{lalonde1986evaluating}.
The treatment $A$ is an indicator of participation in the National Supported Work Demonstration job training program.
The data contains 185 people with $A = 1$ and a comparison sample of 429 people with $A = 0$ from the Population Survey of Income Dynamics.
The outcome $Y$ is the real earnings in US dollars measured in 1978, several years after completion of the program. We scaled the outcome by $10^4$.
The covariates include demographic variables (age, race, education, and marital status), and two previous earnings levels measured in 1974 and 1975, which we also scaled by $10^4$.

\begin{figure}
    \centering
    \begin{subfigure}[95\% CI's]  {\label{LC.real}\includegraphics[width=42mm]{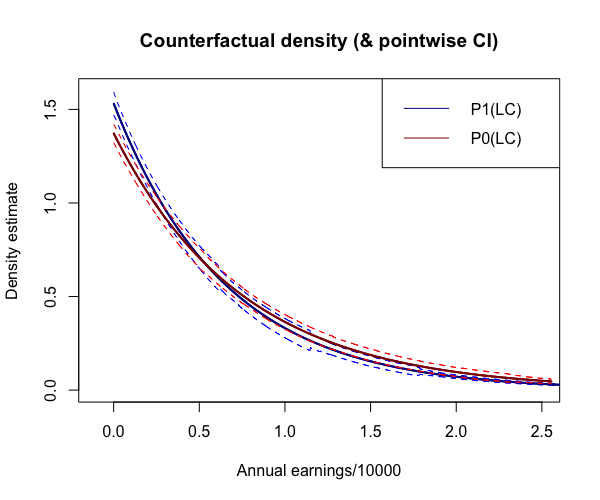}}
   \end{subfigure}
    \hspace{0.2cm}
    \begin{subfigure}[95\% CI's via sample splitting]{\label{LC.realsp}\includegraphics[width=42mm]{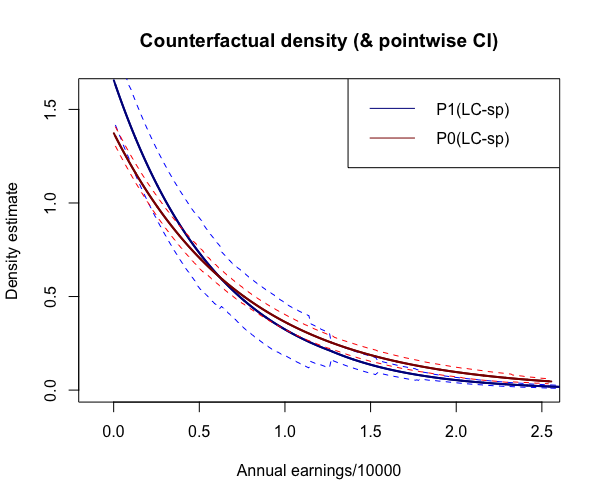}}
   \end{subfigure}
    \hspace{0.2cm}
    \begin{subfigure}[95\% CI's via projection method]    {\label{realbasis}\includegraphics[width=40mm]{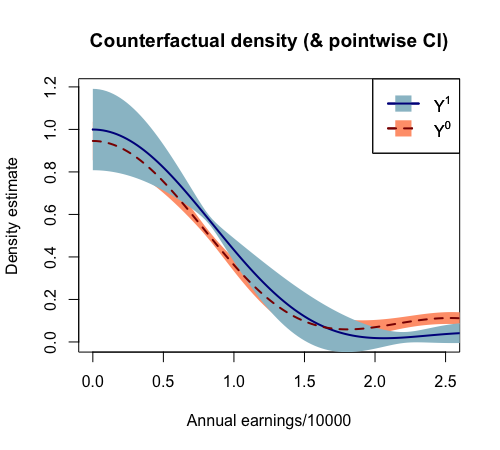}}
    \end{subfigure}\\
   \begin{subfigure}[95\% diference CI's]  {\label{LC.real.contr}\includegraphics[width=50mm]{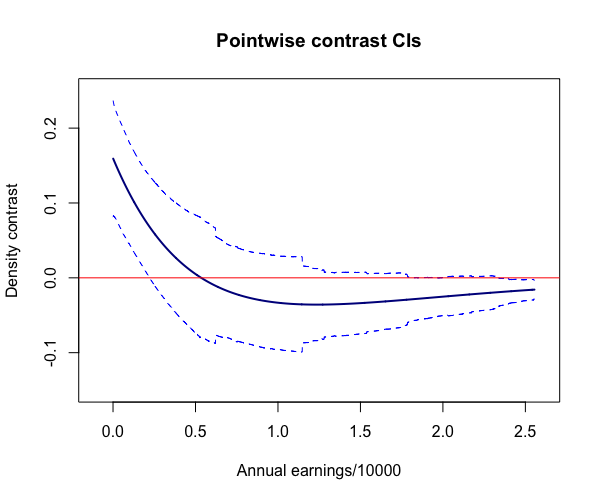}}
   \end{subfigure}
    \hspace{0.2cm}
    \begin{subfigure}[95\% difference CI's via sample splitting]{\label{LC.realsp.contr}\includegraphics[width=50mm]{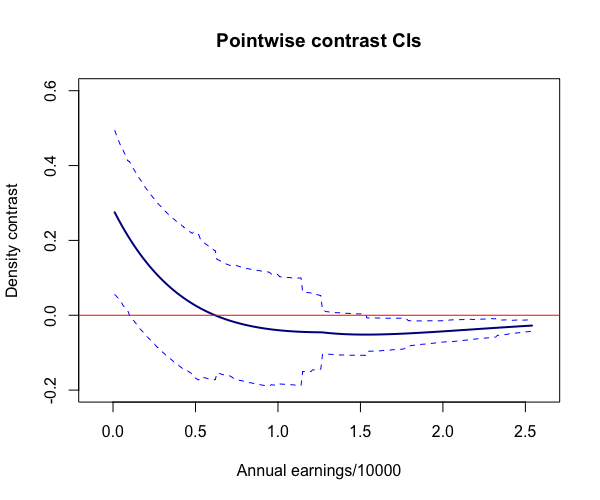}}
   \end{subfigure}\\
    \begin{subfigure}[95\% log-ratio CI's]  {\label{LC.real.ratio}\includegraphics[width=50mm]{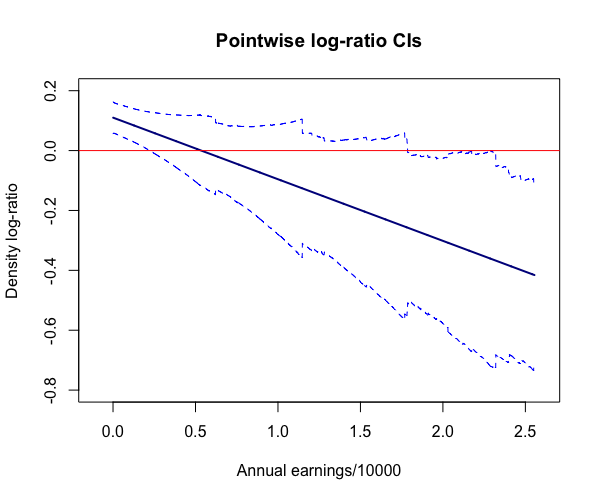}}
   \end{subfigure}
    \hspace{0.2cm}
    \begin{subfigure}[95\% log-ratio CI's via sample splitting]{\label{LC.realsp.ratio}\includegraphics[width=50mm]{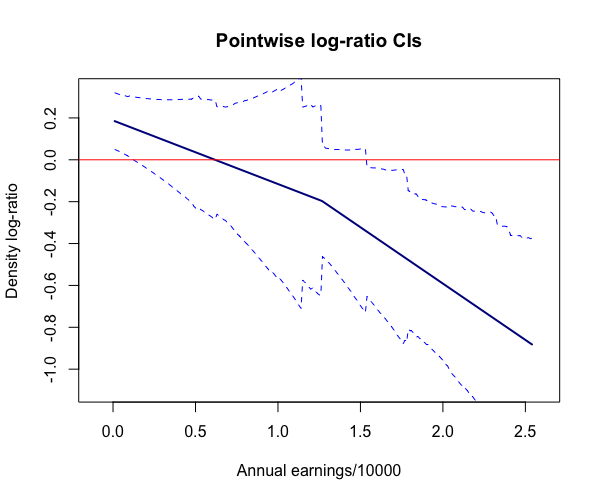}}
   \end{subfigure}\\
    \caption{\label{fig:Real.data} (a--c): Density estimates and corresponding 95\% pointwise CI's for the Lalonde data.  P1 and P0 stand for the counterfactual density estimates for each treatment and control group, respectively. (a) Our log-concave estimator. (b) Our estimator with five-fold sample splitting. (c) Basis expansion estimator \citep{KBWcounterfactualdensity}.
    (d--e): 95\% pointwise difference (P1$-$P0) CI's for the Lalonde data. (d) Our log-concave estimator. (e) Our estimator with five-fold sample splitting.
    (f--g): 95\% pointwise log-ratio ($\log({\rm P1}/{\rm P0})$) CI's for the Lalonde data. (f) Our log-concave estimator. (g) Our estimator with five-fold sample splitting.}
  \end{figure}

We estimated the density of real earnings in 1978 for treatment and control  using our proposed method adjusting for the covariates listed above.
We estimated the conditional distribution function using the location-scale estimator with a gamma distribution as in Lemma~\ref{lemma: meanvar}.
We estimated the propensity score and the conditional mean of the outcome using random forests via the R package \texttt{ranger} \citep{ranger.pack}. 
For the conditional mean, we included interactions between the treatment and covariates as additional predictors in the random forest.
To estimate the conditional variance, we used another random forest with these same predictors and outcome $(Y_i-\hat \mu_i)^2$, where $\hat\mu_i$ is the fitted value of the conditional mean for the $i$th observation. 
We considered other several other conditional distribution estimators, including location-scale estimators with exponential and uniform CDFs, as well as quantile regression random forest \citep{meinshausen2006quantile,elie2022random}. We did not find significant differences between the results, so we only present the results from the gamma conditional distribution function.
We also computed the log-concave estimator with sample splitting with $K_0 = 5$ folds and the same nuisance estimators as above (see Section \ref{section:sample.splitting}).
We set the tuning parameter for estimating $\chi_{\theta_a}$ to be $n^{-1/10}$. 
For comparison, we estimated the counterfactual density and corresponding 95\% CI's  using the basis expansion method \citep{KBWcounterfactualdensity} using the \texttt{cdensity} function in the R package \texttt{npcausal} \citep{npcausal}.
We again used 5 fold sample splitting and random forests for the nuisance estimators.

Figure \ref{fig:Real.data} displays the estimated densities.
The densities under treatment and control are similar in shape, but perhaps surprisingly there are regions where the CI's in the upper and lower tails do not overlap.
By contrast, a 95\% CI for the ATE (using the \texttt{ate} function in the
\texttt{npcausal} package with default settings) was $[-0.10, 0.56]$,
indicating there was not a statistically significant difference in the
counterfactual means.  Our proposed log-concave estimator does not have the
same boundary issue that the kernel density estimator and/or the projection
method encounter on this dataset.

Figures \ref{LC.real.contr}--\ref{LC.realsp.contr} display the pointwise difference CI's derived from our estimator, both with and without the application of sample splitting.
Figures \ref{LC.real.ratio}--\ref{LC.realsp.ratio} illustrate the pointwise density log-ratio CI's derived from both of our estimators.
Both proposed methods reveal significant differences, specifically at the upper and lower extremes.
Therefore, if we have accounted for all sources of confounding and the true densities are log-concave, then there is evidence of an effect of treatment on the upper and lower tails of the distribution of income.

\section{Empirical processes tools}\label{appd:A}

Given classes $\mathcal{F}_1,\dots,\mc F_k$ of functions $\mc F_i:\mc X\mapsto \mc R$ and a function $\varphi:\RR^k \mapsto \RR$, let $\varphi(\mathcal{F}_1,\dots,\mc F_k)$ be the class of functions $x\mapsto \varphi(f_1(x),\dots,f_k(x))$, where $f_i\in \mc F_i$, for $i=1,\dots,k$. The following proposition is Theorem 3 in \citet{preserveGC}.
\begin{proposition}\label{gcpreserve}
Suppose that $\mathcal{F}_1,\dots,\mc F_k$ are $P$-Glivenko-Cantelli classes of functions, and that $\varphi:\RR^k\mapsto \RR$ is continuous. Then $\mc H:=\varphi(\mc F_1,\dots, \mc F_k)$ is a $P$-Glivenko-Cantelli class given that it has an integrable envelope.
\end{proposition}

The following lemma that controls empirical process terms under sample splitting scheme and its proof were provided in Lemma C.3 of \citet{kim2018causal}.
\begin{lemma}\label{kkklem}
(Sample-splitting). Let $\PP_n$ denote the empirical measure over a set $D_{1,n} = (Z_1, \cdots , Z_n)$, which is i.i.d. from $\PP$. Let $\hat f$ be a sample operator (e.g., estimator) constructed in a separate, independent sample set $D_{2,m}$ with m observations. Then we have
\begin{align*}
    \PP \Big((\PP_n-\PP)\hat f\Big)^2 \le \frac{1}{n}\PP \hat{f}^2
\end{align*}
\end{lemma}

We state the following proposition that controls empirical processes indexed by classes of functions that change with $n$. It is Theorem 2.11.24 in \citet{vdvandW}. For each $n$, suppose that a class of measurable functions $\mc F_n=\{f_{n,t}:t\in T\}$ indexed by a totally bounded semimetric space $(T,\rho)$ admits an envelope function (sequence) $F_n$. And, further suppose that the classes $\mc F_{n,\delta}=\{f_{n,s}-f_{n,t}:\rho(s,t)<\delta\}$ and $\mc F^2_{n,\delta}$ are measurable under the probability measure $P$. 
\begin{proposition}\label{prop:vdv.2.11.24}
Suppose that the following holds,
\begin{equation}
\begin{split}\label{vdv.display}
&PF_n^2=O(1),\\
&PF_n^2I(F_n>\eta\sqrt{n})\rightarrow 0, \qquad{\rm for~every~}\eta>0,\\
\sup_{\rho(s,t)<\delta_n}&P(f_{n,s}-f_{n,t})^2\rightarrow 0, \qquad{\rm for~every~}\delta_n \downarrow 0,\\
\end{split}
\end{equation}
and,
\begin{align}\label{vdv.UEI}
 \sup_Q\int_0^{\delta_n} \sqrt{\log N\left(\epsilon\|F_n\|_{Q,2},\mc F_n, L_2(Q)\right)}d\epsilon \rightarrow 0,   \qquad{\rm for~every~}\delta_n \downarrow 0.
\end{align}
Then the sequence $\{\sqrt{n}(\PP_n-P)f_{n,t}:t\in T\}$ is asymptotically tight in $\ell^{\infty}(T)$. Moreover, given that the sequence of covariance functions $Pf_{n,s}f_{n,t}-Pf_{n,s}Pf_{n,t}$ converges pointwise on $T\times T$, the sequence $\{\sqrt{n}(\PP_n-P)f_{n,t}:t\in T\}$ converges in distribution to a Gaussian process.
\end{proposition}

Here we state the following propositions that give bounds for each $L_p$-norm (resp. $L_1$-norm) of $\GG_n=\sqrt{n}(\PP_n-P)$ for function classes $\mc F$ that admit a finite uniform entropy (resp. bracketing) integral.
\begin{proposition}\label{prop:vdv.2.14.1}
Let $\mc F$ be a $P$-measurable class of measurable functions with measurable envelope function $F$. Then, 
\begin{align*}
\left\|\|\GG_n\|_{\mc F}\right\|_{P,p}\lesssim J(1,\mc F)\|F\|_{P,2\vee p} ,\qquad 1\le p,   
\end{align*}
where the entropy integral $J(1,\mc F)$ is given by
\begin{align*}
J(1,\mc F)=\sup_Q \int_0^\delta \sqrt{1+\log N\left(\epsilon\|F\|_{Q,2},\mc F,L_2(Q)\right)}d\epsilon,
\end{align*}
where the supremum is taken over all discrete probability measures $Q$ with $\|F\|_{Q,2}>0$.
\end{proposition}
\begin{proposition}\label{prop:vdv.2.14.16}
Let $\mc F$ be a $P$-measurable class of measurable functions with measurable envelope function $F$. Then, 
\begin{align*}
\left\|\|\GG_n\|_{\mc F}\right\|_{P,1}\lesssim J_{[]}(1,\mc F,L_2(P))\|F\|_{P,2},
\end{align*}
where the entropy integral $J_{[]}(\delta,\mc F,L_2(P))$ (for $\delta>0$) is given by
\begin{align*}
J_{[]}(\delta,\mc F,L_2(P))=\int_0^\delta \sqrt{1+\log N_{[]}\left(\epsilon\|F\|_{P,2},\mc F,L_2(P)\right)}d\epsilon.
\end{align*}
\end{proposition}

\section{Proofs of theorems, lemmas}

\subsection{Proof of Lemma \ref{lemma:cond.E1}}\label{proof:cond.E1}
Fix $\varepsilon, \delta > 0$. For any fixed $M>0$, we have
\begin{align*}
d_1 \left( \widehat\phi_a(\cdot|\mathbf{X}), {\phi}_{a,\infty}(\cdot|\mathbf{X})\right)&=\int_{-\infty}^{\infty}\left|\widehat\phi_a(s|\mathbf{X})-{\phi}_{a,\infty}(s|\mathbf{X})\right| \, ds \\
&= \int_{-M}^{M}\left|\widehat\phi_a(s|\mathbf{X})-{\phi}_{a,\infty}(s|\mathbf{X})\right|\, ds \\
&\indent + \int_{(-\infty,-M]\cup [M,\infty)}\left|\widehat\phi_a(s|\mathbf{X})-{\phi}_{a,\infty}(s|\mathbf{X})\right|\, ds\\
&\leq \int_{-M}^{M}\left|\widehat\phi_a(s|\mathbf{X})-{\phi}_{a,\infty}(s|\mathbf{X})\right| \, ds \\
&\indent +\int_{-\infty}^{-M} \widehat\phi_a(s|\mathbf{X}) ds + \int_{-\infty}^{-M}\phi_{a,\infty}(s|\mathbf{X}) \, ds\\
&\indent +\int_{M}^{\infty} \left[1-\widehat\phi_a(s|\mathbf{X})\right] \, ds + \int_{M}^{\infty}\left[1-\phi_{a,\infty}(s|\mathbf{X})\right] \, ds.
\end{align*}
Thus,
\begin{align*}
P_*\left\{d_1 \left( \widehat\phi_a(\cdot|\mathbf{X}), {\phi}_{a,\infty}(\cdot|\mathbf{X})\right)\right\}^2&\leq P_*\left\{\int_{-M}^{M}\left|\widehat\phi_a(s|\mathbf{X})-{\phi}_{a,\infty}(s|\mathbf{X})\right|ds\right\}^2 \\
&\indent +P_*\left\{\int_{-\infty}^{-M} \widehat\phi_a(s|\mathbf{X}) ds\right\}^2 + P_*\left\{\int_{M}^{\infty} \left[1-\widehat\phi_a(s|\mathbf{X})\right] \, ds\right\}^2\\
&\indent + P_*\left\{\int_{-\infty}^{-M}\phi_{a,\infty}(s|\mathbf{X})\, ds\right\}^2 + P_*\left\{\int_{M}^{\infty}\left[1-\phi_{a,\infty}(s|\mathbf{X})\right] ds\right\}^2
\end{align*}
Define
\begin{align*}
    D_{n,M}^{(1)} &:= P_*\left\{\int_{-M}^{M}\left|\widehat\phi_a(s|\mathbf{X})-{\phi}_{a,\infty}(s|\mathbf{X})\right|ds\right\}^2 \\
    D_{n,M}^{(2)} &:=    P_*\left\{\int_{-\infty}^{-M} \widehat\phi_a(s|\mathbf{X}) ds\right\}^2 + P_*\left\{\int_{M}^{\infty} \left[1-\widehat\phi_a(s|\mathbf{X})\right] \, ds\right\}^2\\
    d_{M}^{(3)} &:=  P_*\left\{\int_{-\infty}^{-M}\phi_{a,\infty}(s|\mathbf{X})\, ds\right\}^2 + P_*\left\{\int_{M}^{\infty}\left[1-\phi_{a,\infty}(s|\mathbf{X})\right] ds\right\}^2.
\end{align*} 
Then, our goal is to show that for all $\varepsilon, \delta > 0$, there exists $N = N_{\varepsilon, \delta}$ and $M = M_{\varepsilon, \delta}$ such that for all $n \geq N$,
\begin{align}\label{Lemma1.WTS}
P_* \left( D_{n,M}^{(1)} > \delta\right) < \varepsilon, \quad  P_* \left( D_{n,M}^{(2)} > \delta\right) < \varepsilon, \quad d_{M}^{(3)} < \delta.
\end{align}

Since $ \int_\RR |s| \, d\phi_{a,\infty}(s|\mathbf{X})<T(\mathbf{X})$ and $T\in L_2(P_*)$,  $\int_{-\infty}^{-M} \phi_{a,\infty}(s|\mathbf{X}) \, ds$ and $\int_{M}^{\infty} \left[1-\phi_{a,\infty}(s|\mathbf{X})\right] \, ds$ both decrease to 0 as $M \to \infty$  for $P_*$-a.e.\ $\mathbf{X}$. Therefore, by the monotone convergence theorem, 
\begin{align*}
\lim_{M\rightarrow\infty}P_*\left\{\int_{-\infty}^{-M}  \phi_{a,\infty}(s|\mathbf{X})\, ds\right\}^2=0, {~\rm and ~} \lim_{M\rightarrow\infty}P_*\left\{\int_{M}^{\infty}  \left[1-\phi_{a,\infty}(s|\mathbf{X})\right] \, ds\right\}^2=0,
\end{align*}
Hence, we have that for all $\delta > 0$, there exists $M_3(\delta)$ such that for all $M \geq M_3(\delta)$, $d_{M}^{(3)} < \delta$. 

\begin{align*}
P_*\left\{\int_{-\infty}^{-M} \widehat\phi_a(s|\mathbf{X}) ds\right\}^2 + P_*\left\{\int_{M}^{\infty} \left[1-\widehat\phi_a(s|\mathbf{X})\right] ds\right\}^2\leq \mathcal{U}_n(M;\mathbf{Z}),
\end{align*}
$M_2(\varepsilon, \delta)$ and $N_2(\varepsilon, \delta)$ such that for all $M \geq M_2(\varepsilon, \delta)$ and $n \geq N_2(\varepsilon, \delta)$, $P_*\left(\mathcal{U}_n(M;\mathbf{Z})>\delta\right)  <\epsilon$.

Finally, by the Cauchy-Schwartz inequality,
\begin{align*}
P_*\left\{\int_{-M}^{M}\left|\widehat\phi_a(s|\mathbf{X})-{\phi}_{a,\infty}(s|\mathbf{X})\right|\, ds\right\}^2 &\leq  P_*\left\{2M\int_{-M}^M \left|\widehat\phi_a(s|\mathbf{X})-{\phi}_{a,\infty}(s|\mathbf{X})\right|^2\, ds\right\}\\
&\leq 2M\int_{-M}^M P_*\left|\widehat\phi_a(s|\mathbf{X})-{\phi}_{a,\infty}(s|\mathbf{X})\right|^2 \, ds\\
&\leq 4M^2 \sup_{s\in[-M,M]} P_*\left|\widehat\phi_a(s|\mathbf{X})-{\phi}_{a,\infty}(s|\mathbf{X})\right|^2.
\end{align*}
Since $\sup_{s\in[-M,M]} P_*\left|\widehat\phi_a(s|\mathbf{X})-{\phi}_{a,\infty}(s|\mathbf{X})\right|^2=o_p(1)$ for any $M>0$ by assumption, for all $\varepsilon, \delta, M > 0$, there exists $N_1(\varepsilon, \delta, M)$ such that for all $n \geq N_1(\varepsilon, \delta, M)$, $P_* \left( D_{n,M}^{(1)} > \delta\right) < \varepsilon$.

For any $\varepsilon, \delta > 0$, we let $M_{\epsilon,\delta} := \max\{ M_2(\varepsilon, \delta), M_3(\delta)\}$ and $N_{\epsilon,\delta} := \max\{ N_1(\varepsilon, \delta, M_{\epsilon,\delta}), N_2(\varepsilon, \delta)\}$. This implies that, for any $n \geq N_{\epsilon,\delta}$, \eqref{Lemma1.WTS} holds.
This completes the proof.

\subsection{A semi-parametric class of functions $\mc F_H$}\label{FH.Wasserstein.proof}

Recall the class of functions $\mc F_H$ defined in \eqref{F.H.class}.
We provide a lemma under which $\mc F_H$ satisfies conditions  \ref{assm:nuisanceconvergence}, \ref{assm:properCDF}, \ref{assm:basicGC}, and \ref{assm:integratedGC}.

\begin{lemma} \label{lemma: meanvar}
Suppose $\widehat\phi_a(s | \mathbf{x}) = H((s-\hat\mu_a(\mathbf{x}))/\hat\sigma_a(\mathbf{x}))$ for a fixed CDF $H$ with mean 0 and variance 1 and estimators $\hat\mu_a$ and $\hat\sigma_a$ such that  $\hat\mu_{a} \in \mc{F}_1$ and $\hat\sigma_{a} \in \mc{F}_2$, where $\mc{F}_1$ and $\mc F_2$ are classes of measurable functions uniformly bounded by $K>0$, and $\mc{F}_2$ is also uniformly bounded away from $0$ by $1/K$, the function $L:\RR \mapsto [0,\infty)$ defined as
\begin{align}\label{function.L}
L(s):=\sup_{K(s-K)<s_1\neq s_2\le K(s+K)} \frac{|H(s_2)-H(s_1)|}{|s_2-s_1|},
\end{align}
satisfies $\int_{\RR}(|s|+1)L(s)ds<\infty$, and $\| \hat\mu_{a}-\mu_{a,\infty} \| \rightarrow_p 0$ and $\|\hat\sigma_{a} -\sigma_{a,\infty} \| \rightarrow_p 0$.
Then conditions  \ref{assm:nuisanceconvergence} and \ref{assm:properCDF} are satisfied.

Assume that $\mc F_1$ and $\mc F_2$ satisfy the aforementioned conditions.
Suppose that $\widehat\phi_a \in \mc F_H$ and $\mc F_1 \times \mc F_2$ is endowed with $L_1(P_*)$ distance, where
\begin{align*}
	L_1((f_1,g_1),(f_2,g_2))(P_*)=P_*\left[\left|f_1-f_2\right|+\left|g_1-g_2\right|\right].	
\end{align*}
Without loss of generality, assume that the support of $Y^a$ $(a\in\{0,1\})$ is $\RR$.
If the $L_1(P_*)$ bracketing number satisfies $N_{[]}(\varepsilon,\mc F_1\times \mc F_2,L_1(P_*))<\infty$ for every $\varepsilon>0$ and $H^{-1}$ is Lipschitz continuous on any compact interval contained in $(0,1)$, then condition \ref{assm:basicGC} is satisfied. If in addition $J_1(t) := \int_{-\infty}^t H(s)ds$ and $J_2(t)=\int_{t}^{\infty} (1-H(s))ds$ satisfy that $J_1^{-1}$ and $J_2^{-1}$ are Lipschitz continuous on every compact interval $[t,J_1(0)]$ and $[t,J_2(0)]$, respectively, for every $0<t$, then condition \ref{assm:integratedGC} is satisfied. 
\end{lemma}

\begin{proof}
First, we show \ref{assm:nuisanceconvergence}, since \ref{assm:properCDF} is trivial.
Under the conditions in Lemma \ref{lemma: meanvar}, we have
\begin{align*}
&P_*\int_{\RR}\left|H\left(\frac{s-\hat\mu_a(\mathbf{X})}{\hat\sigma_{a}(\mathbf{X})}\right)-H\left(\frac{s-\mu_{a,\infty}(\mathbf{X})}{\sigma_{a,\infty}(\mathbf{X})}\right)\right|ds\\
&\indent \leq P_*\int_{\RR} L(s)\left|\frac{s-\hat\mu_a(\mathbf{X})}{\hat\sigma_{a}(\mathbf{X})}-\frac{s-\mu_{a,\infty}(\mathbf{X})}{\sigma_{a,\infty}(\mathbf{X})}\right|ds\\
&\indent \leq P_*\int_{\RR} L(s)\left|\frac{s(\hat\sigma_{a}(\mathbf{X})- \sigma_{a,\infty}(\mathbf{X}))}{\hat\sigma_{a}(\mathbf{X})\sigma_{a,\infty}(\mathbf{X})}+\frac{\hat\sigma_{a}(\mathbf{X})(\mu_{a,\infty}(\mathbf{X})-\hat\mu_{a}(\mathbf{X}))+\hat\mu_{a}(\mathbf{X})(\hat\sigma_{a}(\mathbf{X})-\sigma_{a,\infty}(\mathbf{X}))}{\hat\sigma_{a}(\mathbf{X})\sigma_{a,\infty}(\mathbf{X})}\right|ds\\
&\indent \leq P_*\int_{\RR} L(s)\left[(K^2|s|+K^3)|\hat\sigma_{a}(\mathbf{X})-\sigma_{a,\infty}(\mathbf{X})|+K|\hat\mu_{a}(\mathbf{X})-\mu_{a,\infty}(\mathbf{X})|\right]ds\rightarrow 0.
\end{align*}

Now, we show \ref{assm:basicGC} and \ref{assm:integratedGC}.
For each $\epsilon>0$, there exists $0<\epsilon_0<1$ which satisfies $H(s)<\epsilon$ when $s\le K(H^{-1}(\epsilon_0)+K)$, and $1-H(s)<\epsilon$ when $s\ge K(H^{-1}(1-\epsilon_0)-K)$.
Define $L:=\max_{s\in\RR} L(s)$ and $L_{\epsilon_0}$ as a Lipschitz constant for $H^{-1}$ on the interval $[\epsilon_0,1-\epsilon_0]$.
We further let $K_{H,\epsilon_0}$ as $\max\{|H^{-1}(\epsilon_0)|,|H^{-1}(1-\epsilon_0)|\}$.
By the bracketing number condition on the class $\mc F_1\times \mc F_2$, we have finite number of brackets $(f_{l,i},g_{l,i}),\,(f_{u,i},g_{u,i})$ where $L_1((f_{l,i},g_{l,i}),(f_{u,i},g_{u,i}))(P_*)\le \epsilon/(2L(K^3+K_{H,\epsilon_0}))$ for $1\le i\le N_{\epsilon}$ where $N_{\epsilon}=N_{\epsilon,L,K,K_{H,\epsilon_0}}=N_{[]}(\epsilon/(2L(K^3+K_{H,\epsilon_0})),\mc F_1\times \mc F_2,L_1(P_*))$.
Now we construct the brackets for the class $\mc F_H$ at $H((s-f)/g)$ by $H((H^{-1}(\epsilon_j)-f_{u,i})/g_{u,i})$, $H((H^{-1}(\epsilon_{j+i})-f_{l,i})/g_{l,i})$, where $(f,g)$ has lower and upper brackets $(f_{l,i},g_{l,i}),(f_{u,i},g_{u,i})$ and $H^{-1}(\epsilon_j) \le s< H^{-1}(\epsilon_{j+1})$, and $0\le(\epsilon_{j+1}-\epsilon_j)\le \epsilon /(2LL_{\epsilon_0}K)$ for $1\le i\le N'_\epsilon-1$, $\epsilon_1=\epsilon_0,\,\epsilon_{N'_\epsilon}=1-\epsilon_0$.
As in Section \ref{FH.Wasserstein.proof}, assuming that $K>1$ without loss of generality, for $\epsilon_0\le s\le 1-\epsilon_0$, we have
\begin{align*}
	&P_*\left|H\left(\frac{H^{-1}(\epsilon_{j})-f_{u,i}}{g_{u,i}}\right)-H\left(\frac{H^{-1}(\epsilon_{j+1})-f_{l,i}}{g_{i,i}}\right)\right|\\
	&\indent \le L\left[P_*\left|\frac{g_{u,i}(f_{u,i}-f_{l,i})+f_{u,i}(g_{u,i}-g_{l,i})}{g_{u,i}g_{i,i}}\right|\right.\\
	&\qquad\qquad\left.+P_*\left|\frac{g_{l,i}(H^{-1}(\epsilon_{j})-H^{-1}(\epsilon_{j+1}))+H^{-1}(\epsilon_{j+1})(g_{l,i}-g_{u,i})}{g_{u,i}g_{i,i}}\right|\right]\\
	&\indent \le L(K^3+K_{H,\epsilon_0}) L_1((f_{l,i},g_{l,i}),(f_{u,i},g_{u,i}))(P_*) + LL_{\epsilon_0}K(\epsilon_{j+i}-\epsilon_j)\\
	&\indent \le \epsilon.
\end{align*}
In addition, since $H(s)<\epsilon$ on $s\in(-\infty,K(H^{-1}(\epsilon_0)+K)]$, $0$ and $H(K(H^{-1}(\epsilon_0)+K)$ can serve as lower and upper bracket for any $H((s-f)/g)\in\mc F_H$ with $s\in(-\infty,K(H^{-1}(\epsilon_0)+K)]$.
Analogous brackets $1$ and $H(K(H^{-1}(1-\epsilon_0)-K)$ work for the case $s\ge K(H^{-1}(1-\epsilon_0)-K)$.
Thus, we have $N_{[]}(\epsilon, \mc F_H,L_1(P_*))\le N_\epsilon N'_{\epsilon}+2<\infty$, and this implies the condition \ref{assm:basicGC}. 

To show \ref{assm:integratedGC}, first we notice that, the primitive $\int_{-\infty}^{t}H((s-f)/g)ds$ is $gJ_1((t-f)/g)$ in $t\in(-\infty,0]$, similarly one can see the primitive $\int_0^{\infty} \left(1-H((s-f)/g)\right)ds$ is $gJ_2((t-f)/g)$ in $t\in[0,\infty)$.
$J_1$ and $J_2$ are uniformly bounded on $(-\infty,0],\,[0,\infty)$ by sufficiently large $K_J>0$, respectively, since $H$ has finite first moment, which implies $\int_{-\infty}^0 H(s)ds+\int_{0}^{\infty} \left(1-H(s)\right)ds <\infty$.
Recalling that $\mc F_2$ is uniformly bounded by $K$, as in the derivation steps to show \ref{assm:basicGC}, there exists $M_{J_1}<0$ such that $J(s)<\epsilon$ when $s<K(J_1^{-1}(M_{J_1})+K)$, for each $\epsilon>0$.
Now, for general bracket pairs, we have,
\begin{align*}
	&P_*\left|g_{l,i}J_1\left(\frac{J_1^{-1}(\epsilon_{j})-f_{u,i}}{g_{u,i}}\right)-g_{u,i}J_1\left(\frac{J_1^{-1}(\epsilon_{j+1})-f_{l,i}}{g_{i,i}}\right)\right|\\
	&\le P_*\left|\left(g_{l,i}-g_{u,i}\right)J_1\left(\frac{J_1^{-1}(\epsilon_{j})-f_{u,i}}{g_{u,i}}\right)\right|\\
	&\indent\qquad+P_*\left|g_{u,i}\left[J_1\left(\frac{J_1^{-1}(\epsilon_{j})-f_{u,i}}{g_{u,i}}\right)-J_1\left(\frac{J_1^{-1}(\epsilon_{j+1})-f_{l,i}}{g_{i,i}}\right)\right]\right|\\
	&\indent \le K_JP_*\left|g_{u,i}-g_{l,i}\right|+ L_{J_1}KP_*\left|\frac{g_{u,i}(f_{u,i}-f_{l,i})+f_{u,i}(g_{u,i}-g_{l,i})}{g_{u,i}g_{i,i}}\right|\\
	&\qquad\qquad+L_{J_1}KP_*\left|\frac{g_{l,i}(J_1^{-1}(\epsilon_{j})-J_1^{-1}(\epsilon_{j+1}))+J_1^{-1}(\epsilon_{j+1})(g_{l,i}-g_{u,i})}{g_{u,i}g_{i,i}}\right|,
\end{align*}
where $L_{J_1}$ is a Lipschitz constant of $J_1^{-1}$ on $[M_{J_1},J_1(0)]$. 
Analogous derivation to steps for proving \ref{assm:basicGC} can be directly applied to control the last term above, since $J_1$ has global Lipschitz constant $1$ because of condition \ref{assm:properCDF}.
In addition, similar reasoning can be applied to $J_2$ on $t\ge 0$. So we omit the proof.
\end{proof}

Many conditional distributions can satisfy the conditions of Lemma~\ref{lemma: meanvar}, such as log-concave CDFs.
Moreover, any union of classes of $\mc F_H$ over finitely many  $H$
satisfies the conditions as well.

We next provide a follow-up to  Lemma~\ref{lemma: meanvar} under which \ref{assm:phiestholder} and \ref{assm:pi-bracketing} hold for the special case $\widehat\phi_a(s | \mathbf{X}) = H( (s - \hat\mu_a(\mathbf{X})) / \hat\sigma_a(\mathbf{X})) \in \mc F_H$.

\begin{lemma}\label{lemma: meanvar2}
Suppose the setup and conditions of Lemma~\ref{lemma: meanvar} hold.  Define the Lipschitz constant of $H^{-1}$ on $[\epsilon, 1-\epsilon]$ as
\begin{align}\label{L.eps}
L_\varepsilon:=\sup_{\epsilon \le s_1\neq s_2\le 1-\epsilon} \frac{|H^{-1}(s_2)-H^{-1}(s_1)|}{|s_2-s_1|}.
\end{align}
If $\log L_\varepsilon\lesssim (1/\varepsilon)^V$, $\max\{|H^{-1}(\varepsilon)|,\, |H^{-1}(1-\varepsilon)|\} \lesssim \varepsilon^{1 - V /V_0}$, and $ \log N_{[]}(\varepsilon,\mc F_1\times \mc F_2,L_1(P_*))\lesssim (1/\varepsilon)^{V_0}$ for $0\le V_0<V<2$, then \ref{assm:phiestholder} and the second statement of~\ref{assm:pi-bracketing} hold. 
\end{lemma}
 
A wide range of conditional distributions  satisfy this condition, such as normals, exponentials, and gammas.
In addition, we note that the convergence rate of $\mc G(s,t;\mathbf Z;S)$ is controlled by 
\begin{align*}
\mc G(s,t;\mathbf Z;S) \lesssim |t-s|\left(\|\widehat\mu_{a}-\mu_{a,\infty}\|+\|\widehat\sigma_{a}-\sigma_{a,\infty}\|\right).
\end{align*}
Hence, when $\widehat\phi_a \in \mc F_H$, \ref{assm:L2conditionsforCI} is satisfied under sufficient rates of convergence of $\|\widehat\pi_a-\pi_{a,\infty}\|$ and $\|\widehat\mu_{a}-\mu_{a,\infty}\|+\|\widehat\sigma_{a}-\sigma_{a,\infty}\|$.

\begin{proof}
\ref{assm:phiestholder} is trivial.
We showed that $N_{[]}(\epsilon, \mc F_H,L_1(P_*))$ is controlled by $N_\epsilon N'_\epsilon$ in Section \ref{FH.Wasserstein.proof}. If $H(K(H^{-1}(\epsilon_0)+K))=\epsilon$, then one can approximate $\epsilon_0$ with $H\left(\frac{H^{-1}(\epsilon)-K}{K}\right)$. Similarly, if $1-H(s)=H(K(H^{-1}(1-\epsilon_0)-K))$, then one can approximate $1-\epsilon_0$ with $H\left(\frac{H^{-1}(1-\epsilon)+K}{K}\right)$.
Since we assumed $K>>1$, one can approxiamtely control $H^{-1}(\epsilon)\le (H^{-1}(\epsilon)-K)/K$ and $H^{-1}(1-\epsilon)\ge (H^{-1}(1-\epsilon)+K)/K$.
This implies that one can approximately control $[\epsilon_0,1-\epsilon_0]$ by $[\epsilon,1-\epsilon]$.
Thus, we have $N'_\epsilon \lesssim \frac{L_{\epsilon}}{\epsilon}$ approximately.
Hence, one could control $\log N'_\epsilon$ approximately by $\log (L_{\epsilon}/\epsilon)\lesssim (1/\epsilon)^V-\log(\epsilon)$.
Furthermore, $K_{H,\epsilon_0}=\max\{|H^{-1}(\epsilon_0)|,\, |H^{-1}(1-\epsilon_0)|\}$ can be further approximately bounded by $K_{H,\epsilon}:=\max\{|H^{-1}(\epsilon)|,|H^{-1}(1-\epsilon)|\}$.
This yields
\begin{align*}
\log N_{\epsilon}&\lesssim \log N_{[]}(\epsilon/K_{H,\epsilon},\mc F_1\times \mc F_2,L_1(P_*))\\
&\le (K_{H,\epsilon}/\epsilon)^{V_0} \lesssim (1/\epsilon)^V,
\end{align*}
since $K_{H,\epsilon}^{V_0}\lesssim (1/\epsilon)^{V-V_0}$. 
Thus, $\log N_\epsilon + \log N'_\epsilon\lesssim (1/\epsilon)^V-\log(\epsilon)$, and this implies condition \ref{assm:pi-bracketing}.
\end{proof}

\subsection{Proof of Theorem \ref{prop:Consistency}}\label{appd:consistency}
Recall that $d_1$ denotes the Wasserstein distance; see \eqref{metric:Wasser}. In the proof of Theorem \ref{prop:Consistency} (see Section \ref{sec:consistency.proof} below), we will check that $d_1(\widehat{F}^c_{a,n},F_a)\rightarrow_p 0$ implies the claimed convergence, and $d_1(\widehat{F}^c_{a,n},F_a)\rightarrow_p 0$ if and only if $\widehat{F}^c_{a,n}(s) \rightarrow_p F_a(s)$ for all $s \in \mathbb{R}$ and $\int_{\mathbb{R}}|s|d\widehat{F}^c_{a,n}(s)\rightarrow_p \int_{\mathbb{R}}|s| \, dF_a(s)$.
Lemmas \ref{lem1} and \ref{lem2} provide conditions on $\widehat F_{a,n}\,\widehat F_{a,n}^{c_0}$ that are used for establishing $\widehat{F}^c_{a,n}(s) \rightarrow_p F_a(s)$ for all $s \in \mathbb{R}$ and $\int_{\mathbb{R}}|s|d\widehat{F}^c_{a,n}(s)\rightarrow_p \int_{\mathbb{R}}|s| \, dF_a(s)$.
We first provide proofs of Lemmas \ref{lem1} and \ref{lem2}.

\subsubsection{Proof of Lemma \ref{lem1}}\label{appd:lem1}
\begin{proof}
For any  $x\in\mathbb{R}$, $|\widehat{F}^{c_0}_{a,n}(x) - F_a(x)|\le  |\widehat{F}_{a,n}(x) - F_a(x)|$. Moreover,
\begin{align*}
   |\widehat{F}^c_{a,n}(x) - F_a(x)| &= I(x<L_n)F_a(x) + I(x\ge U_n)[1-F_a(x)]\\
   &\qquad + I(L_n\le x <U_n)|\widehat{F}^c_{a,n}(x)-F_a(x)|\\
   &\le F_a(L_n)+ [1-F_a(U_n)]+I(L_n\le x <U_n)|\widehat{F}^c_{a,n}(x)-F_a(x)|\\
   &= o_p(1)+ I(L_n\le x <U_n)|\widehat{F}^c_{a,n}(x)-F_a(x)|.
\end{align*}
If $L_n\le x <U_n$, then there exists $k\in\{1,\dotsc,m_n-1\}$ such that $x\in [s_k,s_{k+1})$. Then,
\begin{align}
    |\widehat{F}^c_{a,n}(x)-F_a(x)| &= |\widehat{F}^c_{a,n}(s_{k})-F_a(x)|\nonumber\\
    &\le |\widehat{F}^c_{a,n}(s_{k})-F_a(s_{k})| + |F_a(s_{k})-F_a(x)|,
\end{align}
and $|F_a(s_{k})-F_a(x)| = o_p(1)$ because $F_a$ is uniformly continuous and $|s_k - x| \leq \delta_n \to_p 0$. We also have
\begin{align*}
    \max_k|\widehat{F}^c_{a,n}(s_k)-F_a(s_k)|&\le \max_k|\widehat{F}^{c_0}_{a,n}(s_k)-F_a(s_k)|\le \max_k |\widehat{F}_{a,n}(s_k)-F_a(s_k)|\rightarrow_p 0.
\end{align*}
where the first inequality is from Theorem~1(i) of~\citet{WvdlC2020}.
\end{proof}

\subsubsection{Proof of Lemma \ref{lem2}}\label{appd:lem2}
\begin{proof}
   We can write
   \begin{align*}
    \int_{\mathbb{R}}|s| \, d\left\{\widehat{F}^c_{a,n}(s)-F_a(s)\right\}&= \int_{-\infty}^{L_n}|s| \, d\left\{\widehat{F}^c_{a,n}(s)-F_a(s)\right\}+ \int_{L_n}^{U_n} |s| \, d\left\{\widehat{F}^c_{a,n}(s)-F_a(s)\right\}\\ 
    &\indent+ \int_{U_n}^{\infty}|s| \, d\left\{\widehat{F}^c_{a,n}(s)-F_a(s)\right\}.
   \end{align*}
  Since $\widehat{F}^c_{a,n}(s) = 0$ for $s \leq L_n$ and $\widehat{F}^c_{a,n}(s) = 1$ for $s \geq U_n$,
   \begin{align*}
     \int_{-\infty}^{L_n}|s| \, d\left\{\widehat{F}^c_{a,n}(s) - F_a(s)\right\} &=  -\int_{-\infty}^{L_n} |s| \, dF_a(s), \text{ and} \\
      \int_{U_n}^{\infty}|s| \, d\left\{\widehat{F}^c_{a,n}(s)-F_a(s)\right\}&= -\int_{U_n}^{\infty}|s| \, dF_a(s)
   \end{align*}
   Since $\int_{\mathbb{R}}|s| \, dF_a(s)<\infty$, $F_a(L_n)\rightarrow_p 0$, and $F_a(U_n)\rightarrow_p 1$ by assumption, $\int_{-\infty}^{L_n} |s| \, dF_a(s) = o_p(1)$ and $\int_{U_n}^{\infty}|s|dF_a(s) = o_p(1)$. 

Now, it suffices to study the limiting behavior of $\int_{L_n}^{U_n} |s| \, d\left\{\widehat{F}^c_{a,n}(s)-F_a(s)\right\}$.
First, we have
\begin{align*}
\int_{L_n}^{U_n}|s|d\left\{\widehat{F}^c_{a,n}(s)-F_a(s)\right\} &=  \int_{L_n \wedge 0}^{U_n \wedge 0}|s| \, d\left\{\widehat{F}^c_{a,n}(s)-F_a(s)\right\} +  \int_{L_n \vee 0}^{U_n \vee 
 0}|s| \, d\left\{\widehat{F}^c_{a,n}(s)-F_a(s)\right\} \\ 
 &= -\int_{L_n \wedge 0}^{U_n \wedge 0}sd\left\{\widehat{F}^c_{a,n}(s)-F_a(s)\right\} +  \int_{L_n \vee 0}^{U_n \vee 
 0}s\, d\left\{\widehat{F}^c_{a,n}(s)-F_a(s)\right\}.
   \end{align*}
Furthermore, we obtain
\begin{align*}
    \int_{L_n \wedge 0}^{U_n \wedge 0}s \, d\left\{\widehat{F}^c_{a,n}(s)-F_a(s)\right\}  &= (U_n \wedge  0)\left\{ \widehat{F}^c_{a,n}(U_n \wedge  0)-F_a(U_n \wedge 0) \right\} \\
    &\qquad - (L_n \wedge 0)\left\{ \widehat{F}^c_{a,n}(L_n \wedge  0)  - F_a(L_n \wedge 0) \right\}\\
    &\qquad- \int_{L_n \wedge 0}^{U_n \wedge 0} \left\{ \widehat{F}^c_{a,n}(s)-F_a(s) \right\} \, ds.
\end{align*}
Since $\widehat{F}^c_{a,n}(U_n) = 1$, 
\begin{align*}
    \left| (U_n \wedge  0)\left\{ \widehat{F}^c_{a,n}(U_n \wedge  0)-F_a(U_n \wedge 0) \right\} \right| \leq |U_n| [1 - F_a(U_n)],
\end{align*}
which is $o_p(1)$ by assumption. Similarly, since $\widehat{F}^c_{a,n}(L_n) = 0$, 
\begin{align*}
\left| (L_n \wedge 0)\left\{ \widehat{F}^c_{a,n}(L_n \wedge  0)  - F_a(L_n \wedge 0) \right\} \right| &\le |L_n| F_a(L_n),
\end{align*}
which is $o_p(1)$ by assumption. By a similar derivation for $\int_{L_n \vee 0}^{U_n \vee 0}s \, d\left\{\widehat{F}^c_{a,n}(s)-F_a(s)\right\}$, we then have
\begin{align*}
\int_{L_n}^{U_n}|s| \, d\left\{\widehat{F}^c_{a,n}(s)-F_a(s)\right\}  &=  - \int_{L_n \wedge 0}^{U_n \wedge 0} \left\{ \widehat{F}^c_{a,n}(s)-F_a(s) \right\} \,ds \\
&\qquad+ \int_{L_n \vee 0}^{U_n \vee 0} \left\{ \widehat{F}^c_{a,n}(s)-F_a(s) \right\}\,ds+o_p(1).
\end{align*}
Let $J_1:=\{k:1\le k\le m_n-1, L_n \vee 0 \le s_k < U_n \vee 0 \}$. Then,
\begin{align*}
\int_{L_n\vee 0}^{U_n\vee 0} \left\{\widehat{F}^c_{a,n}(s)-F_a(s)\right\}\, ds &= -\int_{L_n\vee 0}^{U_n\vee 0}\left\{ \left[1-\widehat{F}^c_{a,n}(s)\right]-\left[1-F_a(s)\right] \right\} \, ds\\ 
&= -\delta_n \sum_{k\in J_1} \left\{ \left[1-\widehat{F}^c_{a,n}(s_{k})\right]-\left[1- F_a(s_{k})\right]\right\}\\
&\qquad-\delta_n \sum_{k\in J_1}\left[1- F_a(s_{k})\right] +\int_{L_n\vee 0}^{U_n\vee 0} \left[1 -F_a(s)\right] \, ds\\
&= \delta_n \sum_{k\in J_1}\left\{ \widehat{F}^c_{a,n}(s_{k})-F_a(s_{k})\right\}\\
&\qquad -\delta_n \sum_{k\in J_1} \left[1- F_a(s_{k})\right] +\int_{L_n\vee 0}^{U_n\vee 0}\left[1 -F_a(s)\right] \, ds.
\end{align*}
We then have
\begin{align*}
\sum_{k\in J_1} \left\{ \widehat{F}^c_{a,n}(s_{k})- F_a(s_{k}) \right\} &=\sum_{1 \le k,\, s_k< U_n\vee 0} \left\{ \widehat{F}^c_{a,n}(s_{k})- F_a(s_{k})\right\}-\sum_{1 \le k,\, s_k< L_n\vee 0} \left\{\widehat{F}^c_{a,n}(s_{k})-F_a(s_{k})\right\}\\
&=\sum_{2 \le k,\, s_k< U_n\vee 0} \left\{ \widehat{F}^c_{a,n}(s_{k})- F_a(s_{k})\right\}-\sum_{2 \le k,\, s_k< L_n\vee 0} \left\{\widehat{F}^c_{a,n}(s_{k})- F_a(s_{k})\right\}.
\end{align*}
Thus, 
\begin{align*}
\left|\delta_n \sum_{k\in J_1}\left\{ \widehat{F}^c_{a,n}(s_{k})-F_a(s_{k})\right\}\right|&\le 2\delta_n\max_{2\le k\le m_n-1}\left|\sum_{j=2}^k \left\{ \widehat{F}^c_{a,n}(s_{j})- F_a(s_{j})\right\}\right|.
\end{align*}
Now, by Marshall's inequality (see, e.g., ExerCI'se 3.1-c in \citealp{groeneboom2014nonparametric}), we have
\begin{align*}
2\delta_n\max_{2\le k\le m_n-1}\left|\sum_{j=2}^k \left\{ \widehat{F}^c_{a,n}(s_{j})- F_a(s_{j})\right\}\right| &\le 2\delta_n\max_{2\le k\le m_n-1}\left|\sum_{j=2}^k \left\{ \widehat{F}^{c_0}_{a,n}(s_{j})- F_a(s_{j})\right\}\right|,
\end{align*}
which is $o_p(1)$ by assumption. Next, since $1-F_a$ is non-increasing, we have
\begin{align*}
\delta_n\sum_{k\in J_1} \left[ 1-F_a(s_{k+1})\right] \le \int_{L_n\vee 0}^{U_n\vee 0}\left[ 1- F_a(s)\right] \, ds \le \delta_n\sum_{k\in J_1} \left[1-F_a(s_{k})\right].
\end{align*}
Hence, 
\begin{align*}
&\left|\delta_n\sum_{k\in J_1} \left[1-F_a(s_{k})\right] - \int_{L_n\vee 0}^{U_n\vee 0}\left[ 1- F_a(s)\right] \, ds \right| \\
&\qquad \leq \delta_n\left| \sum_{k\in J_1} \left[1-F_a(s_{k})\right] - \sum_{k\in J_1} \left[ 1-F_a(s_{k+1})\right]\right| \\
&\qquad = \delta_n\left|\left[1-F_a(s_{k_2 + 1})\right] - \left[ 1-F_a(s_{k_1})\right]\right| \\
&\qquad \leq 2 \delta_n = o_p(1)
\end{align*}
where $k_1$ and $k_2$ are the minimal and maximal elements of $J_1$, respectively. Therefore, 
\[\int_{L_n\vee 0}^{U_n\vee 0} \left\{\widehat{F}^c_{a,n}(s)-F_a(s)\right\}\, ds = o_p(1).\]
An analogous argument shows that
\begin{align*}
\int_{L_n \wedge 0 }^{U_n \wedge 0} \left\{\widehat{F}^c_{a,n}(s)-F_a(s) \right\} \, ds = o_p(1), 
\end{align*}
which completes the proof.
\end{proof}

\subsubsection{Proof of the main consistency result}\label{sec:consistency.proof}
For two cumulative distribution functions $F,G:{\mathbb{R}} \mapsto [0,1]$, the Lévy distance between $F$ and $G$ is defined as
\begin{align}\label{deflevy}
    {\displaystyle L(F,G):=\inf\{\varepsilon >0:F(x-\varepsilon )-\varepsilon \leq G(x)\leq F(x+\varepsilon )+\varepsilon ,\;\forall x\in \mathbb {R} \}.}
\end{align}
We state and prove the following lemma which relates the Levy distance and the Wasserstein distance.
\begin{lemma}\label{lem:levy}
For any CDFs $F$ and $G$ on $\mathbb{R}$, $L^2(F,G)\le d_1(F,G)$.
\end{lemma}
\begin{proof}
    If $L(F, G) = 0$ then the result is trivial. If $L(F,G) > 0$, then for any $h \in \mathbb{R}$ such that $0 < h < L(F,G)$, by definition of the Levy distance, either $F(x - h) - h > G(x)$ or $G(x) > F(x + h) + h$ for some $x \in \mathbb{R}$. If $G(x) > F(x + h) + h$, then by monotonicity of $F$ and $G$, $G(y) > F(y) + h$ for all $y \in (x, x + h)$. Hence,
\begin{align*}
    d_1(F,G)\ge \int_x^{x+h} \left| F(y)-G(y) \right| dy \ge h^2.
\end{align*}
If $F(x - h) - h > G(x)$, then an identical calculation yields the same result. Hence, $d_1(F,G) \ge h^2$ for all  $0 < h < L(F,G)$, and taking the limit as $h\uparrow L(F,G)$ yields the result.
\end{proof}
\noindent\textbf{Proof of Theorem \ref{prop:Consistency}}\\
Now, we give the main proof of Theorem \ref{prop:Consistency}.
\begin{proof}
We first show that the result follows if $d_1(\widehat{F}^c_{a,n},F_a)\rightarrow_p 0$. 
If $d_1(\widehat{F}^c_{a,n},F_a)\rightarrow_p 0$,
then any subsequence $\{n_k\}$ of $\mathbb{N}$ has a further subsequence $\{n_{k_l}\}$ for which $d_1(\widehat{F}^c_{a,n_{k_l}},F_a)\rightarrow_{a.s.} 0$. Proposition 2-(c) of~\citet{Cule:2010gv}, this implies that for any $\varepsilon \in (0, \alpha)$,
\begin{align}\label{dssreferral}
\int_{\mathbb{R}} e^{\varepsilon|s|}| \widehat{p}_{a,n_{k_l}}(s)-p_a(s)| \, ds \rightarrow_{a.s.} 0,
\end{align}
since $p_a(s) \leq e^{-\alpha|s| + \beta}$ by assumption. Hence, for any subsequence $\{n_k\}$ of $\mathbb{N}$, there exists a further subsequence $\{n_{k_l}\}$ such that $\int_{\mathbb{R}} e^{\varepsilon|s|}| \widehat{p}_{a,n_{k_l}}(s)-p_a(s)| \, ds \rightarrow_{a.s.} 0$ holds, which implies that $\int_{\mathbb{R}} e^{\varepsilon|s|}| \widehat{p}_{a,n_{k_l}}(s)-p_a(s)| \, ds \rightarrow_{p} 0$. Since this implies the total variation distance converges in probability to zero, and convergence in total variation implies convergence in distribution, by a similar subsequence argument and the second statement of Proposition 2-(c) of~\citet{Cule:2010gv}, continuity of $p_a$ further implies that $\sup_{s\in\mathbb{R}} e^{\varepsilon|s|}|\widehat{p}_{a,n}(s)-p_a(s)| \rightarrow_p 0$.  Therefore, the result follows if $d_1(\widehat{F}^c_{a,n},F_a)\rightarrow_p 0$. Furthermore, $d_1(\widehat{F}^c_{a,n},F_a)\rightarrow_p 0$ if and only if  $\widehat{F}^c_{a,n}(s) \rightarrow_p F_a(s)$ for all $s \in \mathbb{R}$ and  $\int_{\mathbb{R}}|s|d\widehat{F}^c_{a,n}(s)\rightarrow_p \int_{\mathbb{R}}|s| \, dF_a(s)$ (see, e.g., page 407 of \citealp{wasserreview}), so it suffices to show these two statements.

We start by showing that $\widehat{F}^c_{a,n}(s) \rightarrow_p F_a(s)$ for all $s \in \mathbb{R}$. If $\sup_{s\in\RR} | \widehat{F}_{a,n}(s) - F_a(s)| \to_p 0$, then since the other conditions of Lemma~\ref{lem1} hold by assumption, it follows that $\widehat{F}^c_{a,n}(s) \rightarrow_p F_a(s)$ for all $s \in \mathbb{R}$. Hence, it is sufficient to show that $\sup_{s\in\RR} | \widehat{F}_{a,n}(s) - F_a(s)| \to_p 0$.

For each $s \in \mathbb{R}$, by Assumptions \ref{assm:consistency}, \ref{assm:noumconf}, and the tower property, we have
\begin{align}
P_*D_{a,{\theta_{a,\infty}}}(s) &=  P_*\left[{I(A=a) \over \pi_{a,\infty}(\mathbf{X})}\left\{I(Y\le s) -\phi_{a,\infty}(s|\mathbf{X}) \right\}+\phi_{a,\infty}(s|\mathbf{X})\right]\nonumber\\
&=P_*\left[\frac{\pi_{a,*}(\mathbf{X})}{\pi_{a,\infty}(\mathbf{X})}\left\{ \phi_{a,*}(s|\mathbf{X})-\phi_{a,\infty}(s|\mathbf{X}) \right\}+\phi_{a,\infty}(s|\mathbf{X})\right]\nonumber.
\end{align}
Furthermore, by~\ref{assm:DR},
\begin{align*}
   &P_*\left[\frac{\pi_{a,*}(\mathbf{X})}{\pi_{a,\infty}(\mathbf{X})}\left\{\phi_{a,*}(s|\mathbf{X})-\phi_{a,\infty}(s|\mathbf{X} )\right\}+\phi_{a,\infty}(s|\mathbf{X})\right] \\
   &\qquad= P_*\left[I(\mathbf{X} \in \mc S_1)\frac{\pi_{a,*}(\mathbf{X})}{\pi_{a,\infty}(\mathbf{X})} \left\{\phi_{a,*}(s|\mathbf{X})-\phi_{a,\infty}(s|\mathbf{X})\right\}+I(\mathbf{X} \in \mc S_1)\phi_{a,\infty}(s|\mathbf{X})\right]  \\
   &\qquad\qquad  + P_*\left[I(\mathbf{X} \in \mc S_2)\frac{\pi_{a,*}(\mathbf{X})}{\pi_{a,\infty}\mathbf{X})}\left\{\phi_{a,*}(s|\mathbf{X})-\phi_{a,\infty}(s|\mathbf{X})\right\}+I(\mathbf{X} \in \mc S_2)\phi_{a,\infty}(s|\mathbf{X})\right] \\
   &\qquad= P_*\left[I(\mathbf{X} \in \mc S_1)\frac{\pi_{a,*}(\mathbf{X})}{\pi_{a,*}(\mathbf{X})} \left\{\phi_{a,*}(s|\mathbf{X})-\phi_{a,\infty}(s|\mathbf{X})\right\}+I(\mathbf{X} \in \mc S_1)\phi_{a,\infty}(s|\mathbf{X})\right]  \\
   &\qquad\qquad  + P_*\left[I(\mathbf{X} \in \mc S_2)\frac{\pi_{a,*}(\mathbf{X})}{\pi_{a,\infty}\mathbf{X})}\left\{\phi_{a,*}(s|\mathbf{X})-\phi_{a,*}(s|\mathbf{X})\right\}+I(\mathbf{X} \in \mc S_2)\phi_{a,*}(s|\mathbf{X})\right] \\
   &\qquad= P_*\left[I(\mathbf{X} \in \mc S_1) \phi_{a,*}(s|\mathbf{X})\right]   + P_*\left[I(\mathbf{X} \in \mc S_2)\phi_{a,*}(s|\mathbf{X})\right] \\
   &= P_*\left[\phi_{a,*}(s|\mathbf{X})\right] = F_a(s).
\end{align*}
Hence, $P_*D_{a,\theta_{a,\infty}}(s)=F_a(s)$.

Now, since $\widehat{F}_{a,n}(s) = \mathbb{P}_nD_{a,\widehat{\theta}_a}(s)$ and $P_*D_{a,\theta_{a,\infty}}(s)=F_a(s)$, by adding and subtracting terms we can write
\begin{equation}
\label{eq:consistency-Fhat-basic-decomposition}
 \widehat{F}_{a,n}(s)- F_a(s) = (\mathbb{P}_n- P_*)D_{a,\widehat{\theta}_a}(s) + P_*\left[D_{a,\widehat{\theta}_a}(s)-D_{a,{\theta_{a,\infty}}}(s)\right].
\end{equation}

For the first term of~\eqref{eq:consistency-Fhat-basic-decomposition}, we note that by conditions~\ref{assm:boundedpropscore} and~\ref{assm:basicGC}, for all $s \in \mathbb{R}$,
\begin{equation}\label{eq:GCproof.by.function.class}
\begin{split}
 D_{a,\widehat{\theta}_a}(s) &: (Y, A, \mathbf{X}) \mapsto  \frac{I(A=a)}{\widehat\pi_{a}(\mathbf{X})}\left[I(Y\le s)-\widehat\phi_{a}(s|\mathbf{X}) \right]+\widehat\phi_{a}(s|\mathbf{X}) \\
&\in \mc F^a := \left\{  g ( f_1, \pi, f_2, \phi_{a}) : f_1 \in \mc F_1, \pi \in \mc F_\pi, f_2 \in \mc F_2, \phi_a \in \mc F_\phi \right\},
\end{split}
\end{equation}
where $g : \{0,1\} \times [1,K] \times [0,1] \times [0,1] \to \mathbb{R}$ is given by $g(a,b,c,d) = \frac{a}{b} (c - d) + d$, $\mc F_1$ is the singleton class containing the function $A \mapsto I(A = a)$, and $\mc F_2 := \{ Y \mapsto I(Y \leq s) : s \in \mathbb{R} \}$.  Here, $g$ is a continuous function, and $\mc F_1$,  $\mc F_\pi$, $\mc F_2$, and $\mc F_\phi$ are $P_*$-Glivenko-Cantelli by~\ref{assm:basicGC} and Example 2.6.1 of~\cite{vdvandW}. Thus, by Theorem 3 of~\cite{preserveGC}, $\mc F^a$ is $P_*$-Glivenko-Cantelli as well. Hence,
\begin{align*}
   \sup_{s\in\RR} \left|(\mathbb{P}_n- P_*)D_{a,\widehat{\theta}_a}(s)\right|&\leq \sup_{f \in \mc F^a} \left|(\mathbb{P}_n- P_*) f\right| \to_{a.s.} 0.
\end{align*}

Using total expectation, we can write the second summand in~\eqref{eq:consistency-Fhat-basic-decomposition} as
\begin{align}
    P_*&\left[D_{a,\widehat{\theta}_a}(s)-D_{a,{\theta_{a,\infty}}}(s) \right]\nonumber \\
    &=P_*\left[I(A=a) \left\{ I(Y\le s) -\widehat{\phi}_a(s|\mathbf{X})\right\} \left\{ \frac{1}{\widehat{\pi}_a(\mathbf{X})} - \frac{1}{{\pi}_{a,\infty}(\mathbf{X})} \right\}  \nonumber \right.\\
    &\left.\qquad\qquad + \left\{ 1 - \frac{I(A = a)}{{\pi}_{a,\infty}(\mathbf{X})} \right\}\left\{\widehat{\phi}_a(s|\mathbf{X}) - \phi_{a,\infty}(s|\mathbf{X})  \right\} \right]\nonumber\\
    &=P_*\left[\frac{\pi_{a,*}(\mathbf{X})}{\widehat{\pi}_a(\mathbf{X})\pi_{a,\infty}(\mathbf{X})} \left\{ \widehat{\phi}_a(s|\mathbf{X}) - \phi_{a,\infty}(s|\mathbf{X}) \right\} \left\{ \widehat{\pi}_a(\mathbf{X}) - \pi_{a,\infty}(\mathbf{X}) \right\}   \right.\label{indist.rem.pt1}\\
    &\qquad\qquad+ \frac{\pi_{a,*}(\mathbf{X})}{\widehat{\pi}_{a}(\mathbf{X})\pi_{a,\infty}(\mathbf{X})} \left\{ {\phi}_{a,\infty}(s|\mathbf{X}) - \phi_{a,*}(s|\mathbf{X}) \right\} \left\{ \widehat{\pi}_a(\mathbf{X}) - \pi_{a,\infty}(\mathbf{X}) \right\} \label{indist.rem.pt2}\\
    &\left.\qquad\qquad + \left\{ 1 - \frac{\pi_{a,*}(\mathbf{X})}{{\pi}_{a,\infty}(\mathbf{X})} \right\}\left\{\widehat{\phi}_a(s|\mathbf{X}) - \phi_{a,\infty}(s|\mathbf{X})  \right\} \right]\label{indist.rem.pt3},
\end{align} 
Hence, by \ref{assm:boundedpropscore} and the Cauchy-Schwartz inequality, we have
\begin{align*}
     \sup_{s\in\RR}\left| P_*\left[D_{a,\widehat{\theta}_a}(s)-D_{a,{\theta_{a,\infty}}}(s) \right]\right| &\lesssim K^2 \left\{  P_*\left[ \widehat{\pi}_a(\mathbf{X})-{\pi}_{a,\infty}(\mathbf{X}) \right]^2 \right\}^{1/2} \\
     &\qquad + K \left\{\sup_{s\in\RR}P_*\left[\widehat{\phi}_a(s|\mathbf{X})-{\phi}_{a,\infty}(s|\mathbf{X})\right]^2\right\}^{1/2},
\end{align*}
up to a constant, since $\widehat\theta_a,\theta_{a,\infty},\theta_{a,*}$ are less than a universal upper bound 1.
By~\ref{assm:nuisanceconvergence}, the first term is $o_P(1)$. To show that~\ref{assm:nuisanceconvergence} also implies the  second term is $o_P(1)$, we introduce the Lévy metric, which is associated with the Wasserstein distance (see Lemma~\ref{lem:levy}). 
We define
\[\hat{L}(\mathbf{X}):=L\left(\phi_{a,\infty}(\cdot|\mathbf{X}), \widehat\phi_a(\cdot|\mathbf{X})\right)\]
for all $\mathbf{X} \in \mc X$. Then, by the definition of the Lévy metric, for all $s \in \mathbb{R}$ and $\mathbf{X} \in \mc X$ we have
\begin{align*}
 {\phi}_{a,\infty}(s-\hat L(\mathbf{X})|\mathbf{X}) -\hat L(\mathbf{X})  \le \widehat{\phi}_a(s|\mathbf{X}) \le {\phi}_{a,\infty}(s+\hat L(\mathbf{X})|\mathbf{X})+\hat L(\mathbf{X}).
\end{align*}
Hence,
\begin{align*}
     &|\widehat{\phi}_a(s|\mathbf{X})-{\phi}_{a,\infty}(s|\mathbf{X})| \\
     &\quad =  \max\left\{ \widehat\phi_a(s|\mathbf{X})-{\phi}_{a,\infty}(s|\mathbf{X}), {\phi}_{a,\infty}(s|\mathbf{X}) - \widehat\phi_a(s|\mathbf{X}) \right\}\\
     &\quad = \max\left\{ \left[ \widehat\phi_a(s|\mathbf{X})- {\phi}_{a,\infty}(s + \hat{L}(\mathbf{X})|\mathbf{X}) \right] +\left[{\phi}_{a,\infty}(s + \hat{L}(\mathbf{X})|\mathbf{X}) - {\phi}_{a,\infty}(s |\mathbf{X})\right] , \right. \\
     &\quad\qquad\qquad \left. \left[ {\phi}_{a,\infty}(s|\mathbf{X})  -  {\phi}_{a,\infty}(s - \hat{L}(\mathbf{X})|\mathbf{X})\right] + \left[{\phi}_{a,\infty}(s - \hat{L}(\mathbf{X})|\mathbf{X}) - \widehat\phi_a(s|\mathbf{X}) \right]\right\} \\
     &\quad\le \max\left\{ \hat{L}(\mathbf{X}) +\left[{\phi}_{a,\infty}(s + \hat{L}(\mathbf{X})|\mathbf{X}) - {\phi}_{a,\infty}(s |\mathbf{X})\right] , \right. \\
     &\quad\qquad\qquad \left. \left[ {\phi}_{a,\infty}(s|\mathbf{X})  -  {\phi}_{a,\infty}(s - \hat{L}(\mathbf{X})|\mathbf{X})\right] + \hat{L}(\mathbf{X}) \right\} \\
     &\quad=   w\left(s, \hat L(\mathbf{X}) | \mathbf{X}\right)+\hat L(\mathbf{X}),
\end{align*}
where we define
\begin{align*}
    w(s,k | \mathbf{X}) :=\max \{{\phi}_{a,\infty}(s+k|\mathbf{X})-{\phi}_{a,\infty}(s|\mathbf{X}),{\phi}_{a,\infty}(s|\mathbf{X})-{\phi}_{a,\infty}(s-k|\mathbf{X})\}
\end{align*}
for all $k,s \in \mathbb{R}$ and $\mathbf{X} \in \mc X$. Hence, 
\begin{align}
     \sup_{s\in\RR}P_*\left\{ \widehat{\phi}_a(s|\mathbf{X})-{\phi}_{a,\infty}(s|\mathbf{X})\right\}^2 &\le 2\sup_{s\in\RR}P_* w(s, \hat L(\mathbf{X}) | \mathbf{X})^2 + 2P_* \hat L(\mathbf{X})^2. \label{aftersup}
\end{align}
By Lemma~\ref{lem:levy} (and Assumption~\ref{assm:properCDF}), we have
\[ P_*\hat L(\mathbf{X})^2 \leq P_*\int_{-\infty}^{\infty}|\widehat\phi_a(s|\mathbf{X})-{\phi}_{a,\infty}(s|\mathbf{X}))| \, ds, \]
so by~\ref{assm:nuisanceconvergence}, we have $P_*\hat L(\mathbf{X})^2=o_p(1)$. By Jensen's inequality, the fact that $|w(\cdot,\cdot|\cdot)|\le 1$, Assumption \ref{assm:phiregularity}, and the Cauchy-Schwartz inequality, we have 
\begin{align*}
    \sup_{s\in\RR}P_* w(s, \hat L(\mathbf{X}) | \mathbf{X})^2  
    &\le \sup_{s\in\RR}P_* \left|w(s, \hat L(\mathbf{X}) | \mathbf{X}) \right|\\
    &\le P_* \sup_{s\in\RR} \left|w(s, \hat L(\mathbf{X}) | \mathbf{X})\right| \\
    &\leq P_* \left\{ \hat L(\mathbf{X}) R(\mathbf{X})\right\}\\
    &\leq \left[ P_* \left\{ |\hat L(\mathbf{X})| ^{2} \right\}\right]^{1/2} \left[ P_* \left\{ |R(\mathbf{X})|^2 \right\}\right]^{1/2},
\end{align*}
which is $o_p(1)$. 

We have now shown that both summands of~\eqref{eq:consistency-Fhat-basic-decomposition} are $o_p(1)$ uniformly on $s\in\RR$, so we conclude that $\sup_{s\in\RR} |\widehat{F}_{a,n}(s)- F_a(s)| = o_p(1)$, and hence $\widehat{F}_{a,n}^c(s) \to_p F_a(s)$ for all $s \in \mathbb{R}$.

We now show that $\int_{\RR} |s| \, d\widehat{F}_{a,n}^c(s) \to_p \int_{\RR} |s| \, dF_a(s)$. If we can show that
\begin{align}
    \delta_n\max_{2\le k\le m_n-1}\left|\sum_{j=2}^k \widehat{F}^{c_0}_{a,n}(s_{j})-\sum_{j=2}^k F_a(s_{j})\right|\rightarrow_p 0,\label{want.to.show}
\end{align}
then the result follows by  Lemma~\ref{lem2}, since the other conditions of Lemma~\ref{lem2} hold by assumption and by the derivation above. We have
\begin{align}
&\delta_n\max_{2\le k\le m_n-1}\left|\sum_{j=2}^k \widehat{F}^{c_0}_{a,n}(s_{j})-\sum_{j=2}^k F_a(s_{j})\right|\nonumber\\
&\le \delta_n \max_{2\le k\le m_n-1} \left\{\left|\sum_{j=2}^{k} \left[ \widehat{F}_{a,n}^{c_0}(s_{j}) - \widehat{F}_{a,n}(s_{j})\right] \right|  + \left|\sum_{j=2}^{k} \widehat{F}_{a,n}(s_{j}) - \sum_{j=2}^{k} F_a(s_{j})\right| \right\}\nonumber\\
&\le \delta_n \max_{2\le k\le m_n-1} \left|\sum_{j=2}^{k} I\left\{ \widehat{F}_{a,n}(s_{j}) < 0\right\} \widehat{F}_{a,n}(s_{j})\right| \label{First.moment.term1}\\
&\qquad + \delta_n \max_{2\le k\le m_n-1} \left|\sum_{k=2}^{j} I\left\{ \widehat{F}_{a,n}(s_{j}) > 1\right\} \left[ 1-\widehat{F}_{a,n}(s_{j})\right]\right|\label{First.moment.term2}\\
&\qquad + \delta_n \max_{2\le k\le m_n-1} \left|\sum_{j=2}^k \widehat{F}_{a,n}(s_{j})-\sum_{j=2}^k 
 F_a(s_{j})\right|.\label{First.moment.term3}
\end{align}
We show that each of these three summands is $o_p(1)$ in turn. First,
we study \eqref{First.moment.term1}. We define $A_n$ as $\max\{s_k:\widehat F_{a,n}(s_k)<0, 1\le k \le m_n-1\}$.  We note that only the case where $A_n > L_n$ matters, since if $A_n = L_n$ then the expression is zero. We then have
\begin{align*}
&\delta_n \max_{2\le k\le m_n-1} \left|\sum_{j=2}^{k} I\left\{ \widehat{F}_{a,n}(s_{j}) < 0\right\} \widehat{F}_{a,n}(s_{j})\right| \\
&\qquad= -\delta_n \sum_{k=2}^{m_n-1} I\left\{ \widehat{F}_{a,n}(s_{k}) < 0\right\} \widehat{F}_{a,n}(s_{k})  \\
&\qquad\leq -\delta_n \sum_{k=1, s_k \leq A_n} I\left\{ \widehat{F}_{a,n}(s_{k}) < 0\right\} \widehat{F}_{a,n}(s_{k}) \\
&\qquad \leq \delta_n \sum_{k=1, s_k \leq A_n} \left|\widehat{F}_{a,n}(s_{k}) \right| \\
&\qquad = \left[ \delta_n \sum_{k=1, s_k \leq A_n} \left|\widehat{F}_{a,n}(s_{k}) \right|  - \int_{L_n}^{A_n} \left|\widehat{F}_{a,n}(s) \right| \, ds\right] +  \int_{L_n}^{A_n} \left|\widehat{F}_{a,n}(s) \right| \, ds.
\end{align*}
We can write $\widehat F_{a,n}(s) = \widehat F_{a,n,1}(s) - \widehat F_{a,n,2}(s)$ for 
\begin{align*}
\widehat F_{a,n,1}(s)&= \frac{1}{n} \sum_{i=1}^n\left[\frac{I(A_i=a)}{\widehat \pi_a(\mathbf{X}_i)}I(Y_i\le s)+\widehat\phi_a(s|\mathbf{X}_i)\right] \\
\widehat F_{a,n,2}(s) &=  \frac{1}{n} \sum_{i=1}^n\frac{I(A_i=a)}{\widehat \pi_a(\mathbf{X}_i)}\widehat\phi_a(s|\mathbf{X}_i).
\end{align*}
By Assumptions~\ref{assm:boundedpropscore} and~\ref{assm:properCDF},  $\widehat F_{a,n,1}(s)$ and $\widehat F_{a,n,2}(s)$ are uniformly bounded monotone functions. Therefore, there exists a constant $C < \infty$, independent of $n$, such that
\begin{align*}
\left|\int_{L_n}^{A_n} \left|\widehat F_{a,n}(s)\right| \, ds - \delta_n \sum_{k=1,s_k\le A_n} \left|\widehat F_{a,n}(s_k)\right|\right|\leq C\delta_n,
\end{align*}
which is $o_p(1)$. To show that $\int_{L_n}^{A_n} \left|\widehat{F}_{a,n}(s) \right| \, ds \to_p 0$, we  define $\varepsilon_n:=\sup_{s\in \RR}|\widehat F_{a,n}(s)-F_a(s)|$. 
Since $\widehat F_{a,n}(A_n) < 0$ if $A_n > L_n$, we have $F_a(A_n) \leq \widehat F_{a,n}(A_n) + \varepsilon_n < \varepsilon_n$, so $A_n \leq F_a^{-1}(\varepsilon_n)$. If $u_{a}> -\infty$, then since $s_2 > u_{a}$ by assumption, we have $A_n \geq L_n = s_2 - \delta_n > u_{a}- \delta_n$. Then, since $|\widehat{F}_{a,n}|$ is uniformly bounded, $\int_{L_n}^{A_n} \left|\widehat{F}_{a,n}(s) \right| \, ds$ is bounded up to a constant by $A_n - L_n \leq F_a^{-1}(\varepsilon_n) - u_{a}+ \delta_n$. Since $\varepsilon_n \to_p 0$, $F_a^{-1}(\varepsilon_n) \to_p u_{a}$, so this is $o_p(1)$. If $u_{a}=-\infty$, then $A_n\le F_a^{-1}(\varepsilon_n)\rightarrow_p -\infty$ by uniform continuity of $F_a$. Hence, $\int_{L_n}^{A_n}\left|\widehat F_{a,n}(s)\right|ds \leq \int_{-\infty}^{A_n} \left|\widehat F_{a,n}(s)\right|ds \rightarrow_p 0$,
as long as  $\int_{-\infty}^0 \left|\widehat F_{a,n}(s)\right|ds<\infty$, which we now show. By definition and Assumption \ref{assm:boundedpropscore},
\begin{align}
 \mathbb{E}\int_{-\infty}^0 \left|\widehat F_{a,n}(s)\right| \, ds &=\mathbb{E}\int_{-\infty}^0\left| \PP_n D_{a,\widehat{\theta}_a}(s) \right| \, ds\ \\
  &\le \mathbb{E} \left\{\PP_n\int_{-\infty}^0  \left|D_{a,\widehat{\theta}_a}(s) \right|ds\right\}\nonumber\\
  &\le  \mathbb{E}\left\{\frac{1}{n}\sum_{i=1}^n  \frac{I(A_i=a)}{\widehat\pi_a(\mathbf{X}_i)}I(Y_i\le 0)(-Y_i) \right\}\nonumber\\
  &\qquad+\mathbb{E}\left\{\frac{1}{n}\sum_{i=1}^n \left| \left[1-\frac{I(A_i=a)}{\widehat\pi_a(\mathbf{X}_i)}\right] \int_{-\infty}^0 \widehat\phi_a(s|\mathbf{X}_i) \,ds \right| \right\}\nonumber\\
  &\le K\mathbb{E}\left|Y^a\right|+(1+K)\mathbb{E}\left\{\int_{-\infty}^0 \widehat\phi_a(s|\mathbf{X}) \, ds \right\},\label{f}
\end{align}
which is finite by log-concavity of the distriution of $Y^a$ and Assumption \ref{assm:properCDF}. Hence, $\int_{-\infty}^0 \left|\widehat F_{a,n}(s)\right| \, ds < \infty$ almost surely. Therefore,~\eqref{First.moment.term1} is $o_p(1)$. Similar reasoning can be applied to show that~\eqref{First.moment.term2} is $o_p(1)$, so we omit the details.

Finally, we prove that
\begin{align}\label{WTS:integratedGC}
\delta_n \max_{2\le k\le m_n-1} \left|\sum_{j=2}^k \widehat{F}_{a,n}(s_{j})-\sum_{j=2}^k F_a(s_{j})\right|=o_p(1),   
\end{align}
from which it will follow that $\int_{\mathbb{R}} |s| \, d\widehat{F}_{a,n}(s) \to_p \int_{\mathbb{R}} |s| \, dF_{a}(s)$. 
Then, where $\Tilde{F}_{a,n}(x) := \widehat{F}_{a,n}(x)I(x\le 0)+[\widehat{F}_{a,n}(x)-1]I(x>0)$ and $\Tilde{F}_{a}(x) := F_{a}(x)I(x\le 0)+[F_{a}(x)-1]I(x>0)$, we have 
\begin{align*}
&\delta_n\left|\sum_{j=2}^k \widehat{F}_{a,n}(s_{j})-\sum_{j=2}^k F_a(s_{j})\right|\\
&\qquad=\delta_n\left|\sum_{j=2}^k \Tilde{F}_{a,n}(s_{j})-\sum_{j=2}^k \Tilde F_a(s_{j})\right|\\
&\qquad\leq\left|\int_{L_n}^{s_k} \left[\Tilde{F}_{a,n}(s) - \Tilde F_a(s)\right] \, ds \right|\\
&\qquad\qquad +\left|\int_{L_n}^{s_k} \Tilde{ F}_{a,n}(s) \, ds-\delta_n\sum_{j=2}^k \Tilde{F}_{a,n}(s_{j})\right| + \left|\int_{L_n}^{s_k} \Tilde F_a(s) \, ds-\delta_n\sum_{j=2}^k \Tilde F_a(s_{j}) \right|.
\end{align*}
As above, we can write $\widehat{F}_{a,n}$ as the difference of two monotone functions, and $F_a$ is monotone by definition, so that
\begin{align*}
\max_{2\le k\le m_n-1}\left|\int_{L_n}^{s_k} \Tilde{F}_{a,n}(s) \, ds-\delta_n\sum_{j=2}^k \Tilde{ F}_{a,n}(s_{j})\right| &= o_p(1) \text{ and} \\
\max_{2\le k\le m_n-1}\left|\int_{L_n}^{s_k} \Tilde F_a(s) \, ds-\delta_n\sum_{j=2}^k \Tilde F_a(s_{j})\right| &=o_p(1).
\end{align*}
Furthermore,
\begin{align*}
&\max_{2\le k\le m_n-1} \left|\int_{L_n}^{s_k} \left[\Tilde{F}_{a,n} - \Tilde F_a\right]  \right| \\
&\qquad\leq \sup_{L_n< t\le U_n} \left|\int_{L_n}^{s_k} \left[\Tilde{F}_{a,n} - \Tilde F_a\right]  \right|\\
&\qquad=\sup_{L_n< t\le U_n} \left|\int_{-\infty}^t \left[\Tilde{F}_{a,n} - \Tilde F_a\right]  -\int_{-\infty}^{L_n} \left[\Tilde{F}_{a,n} - \Tilde F_a\right]  \right|\\
&\qquad\le \left|\int_{-\infty}^{L_n}\left[\Tilde{F}_{a,n} - \Tilde F_a\right]  \right| +\sup_{L_n< t\le U_n} \left|\int_{-\infty}^t \left[\Tilde{F}_{a,n} - \Tilde F_a\right]  \right| \\
&\qquad \le 2 \sup_{t \in \mathbb{R}} \left|\int_{-\infty}^t \left[\Tilde{F}_{a,n} - \Tilde F_a\right]  \right| \\
&\qquad\le 2\sup_{t \leq 0} \left|\int_{-\infty}^t \left[ \Tilde{ F}_{a,n} - \Tilde F_a \right]\right|+ 2\sup_{ 0 < t} \left|\int_{-\infty}^t \left[ \Tilde{ F}_{a,n}- \Tilde F_a \right]\right|\\
&\qquad\le  2\sup_{t \leq 0} \left|\int_{-\infty}^t \left[\widehat{ F}_{a,n} -  F_a \right]\right|+2\left|\int_{-\infty}^0 \left[ \widehat F_{a,n}- F_a\right] \right| +  2\sup_{0< t } \left|\int_0^t \left[ \left\{1-\widehat F_{a,n}\right\} -  \left\{1- F_a\right\} \right] \right|\\
&\qquad\le  4\sup_{t \leq 0} \left|\int_{-\infty}^t \left[\widehat{ F}_{a,n} -  F_a \right]\right|+ 2\sup_{0< t } \left|\int_0^t \left[ \left\{1-\widehat F_{a,n}\right\} -  \left\{1- F_a\right\} \right] \right|,
\end{align*}
where the fifth inequality holds due to $\int_{-\infty}^t=\int_{-\infty}^0+\int_0^t$. We also have
\begin{align*}
&\sup_{0< t } \left|\int_0^t \left[ \left\{1-\widehat F_{a,n}\right\} -  \left\{1- F_a\right\} \right] \right|\\
&\qquad \le\sup_{0< t } \left|\int_t^\infty \left[ \left\{1-\widehat F_{a,n}\right\} -  \left\{1- F_a\right\} \right] \right|+\left|\int_{0}^{\infty}  \left[ \left\{1-\widehat F_{a,n}\right\} -  \left\{1- F_a\right\} \right] \right|\\
&\qquad \le 2\sup_{0\leq  t }\left|\int_t^\infty \left[ \left\{1-\widehat F_{a,n}\right\} -  \left\{1- F_a\right\} \right] \right|. 
\end{align*}
Thus, \eqref{WTS:integratedGC} follows if
\begin{align*}
\sup_{t\le 0}\left|\int_{-\infty}^t \left[\widehat{ F}_{a,n} -  F_a \right]\right|=o_p(1) \text{ and}\\
\sup_{0\leq  t }\left|\int_t^\infty \left[ \left\{1-\widehat F_{a,n}\right\} -  \left\{1- F_a\right\} \right] \right|=o_p(1).
\end{align*}
We only check the first statement in the preceding display, since similar reasoning can be applied to the second statement.
We use the decomposition established in \eqref{eq:consistency-Fhat-basic-decomposition}:
\begin{align}
 \widehat F_{a,n}(s)-F_a(s) &= \PP_n D_{a,\widehat{\theta}_a}(s)-P_*D_{a,{\theta_{a,\infty}}}(s) \nonumber \\
 &= (\mathbb{P}_n- P_*)D_{a,\widehat{\theta}_a}(s)+P_*\left[D_{a,\widehat{\theta}_a}(s)-D_{a,{\theta_{a,\infty}}}(s)\right].\label{firstmomentdecomp}
\end{align}
For the second term in \eqref{firstmomentdecomp}, as in \eqref{indist.rem.pt1}--\eqref{indist.rem.pt3}, we can write 
\[\int_{-\infty}^t  P_* \left[D_{a,\widehat{\theta}_a}(s)-D_{a,{\theta_{a,\infty}}}(s) \right] \, ds = P_* \left[ Q_1(t,\mathbf{Z})+Q_2(t,\mathbf{Z})+Q_3(t,\mathbf{Z})\right]\]
for
\begin{align}
    Q_1(t,\mathbf{Z})&= \frac{\widehat{\pi}_a(\mathbf{X})-{\pi}_{a,\infty}(\mathbf{X})}{\widehat{\pi}_a(\mathbf{X})}\int_{-\infty}^t\left\{\widehat{\phi}_a(s|\mathbf{X})-{\phi}_{a,\infty}(s|\mathbf{X})\right\} \, ds,\label{first-second-pt2}\\
    Q_2(t,\mathbf{Z})&= \frac{{\pi}_{a,*}(\mathbf{X})\left\{{\widehat{\pi}_a(\mathbf{X})}-{\pi}_{a,\infty}(\mathbf{X})\right\}}{{\pi}_{a,\infty}(\mathbf{X}){\widehat{\pi}_a(\mathbf{X})}} \int_{-\infty}^t \left\{{\phi}_{a,\infty}(s|\mathbf{X})-{\phi}_{a,*}(s|\mathbf{X})\right\}\, ds, \label{second-second-pt2}\\
     Q_3(t,\mathbf{Z})&= \frac{ {\pi}_{a,\infty}(\mathbf{X})-{\pi}_{a,*}(\mathbf{X})}{\pi_{a,\infty}(\mathbf{X})}\int_{-\infty}^t\left\{\widehat{\phi}_a(s|\mathbf{X})-{\phi}_{a,\infty}(s|\mathbf{X})\right\} \, ds.\label{third-second-pt2}
\end{align}
By the Cauchy-Schwarz inequality, we have
\begin{align*}
\sup_{t\leq 0}\left|P_*Q_1(t,\mathbf{Z})\right|^2 &\le K^2\sup_{t\le 0} P_*\left[\int_{-\infty}^t\left\{\widehat{\phi}_a(s|\mathbf{X})-{\phi}_{a,\infty}(s|\mathbf{X})\right\} \, ds\right]^2\\
&\qquad \qquad \times P_*\left[\widehat{\pi}_a(\mathbf{X})-{\pi}_{a,\infty}(\mathbf{X})\right]^2\\
\sup_{t\leq 0}\left|P_*Q_2(t,\mathbf{Z})\right|^2 &\le K^4C_1P_*\left[\widehat{\pi}_a(\mathbf{X})-{\pi}_{a,\infty}(\mathbf{X})\right]^2\\
\sup_{t\leq 0}\left|P_*Q_3(t,\mathbf{Z})\right|^2 &\le K^2\sup_{t\leq 0} P_*\left[\int_{-\infty}^t\left\{\widehat{\phi}_a(s|\mathbf{X})-{\phi}_{a,\infty}(s|\mathbf{X})\right\}\, ds\right]^2
\end{align*}
where $C_1=2P_*\left[\int_{-\infty}^0 {\phi}_{a,\infty}(s|\mathbf{X})ds\right]^2 + 2P_*\left[\int_{-\infty}^0 {\phi}_{a,*}(s|\mathbf{X})ds\right]^2$. 
And, $P_*\left[\widehat{\pi}_a(\mathbf{X})-{\pi}_{a,\infty}(\mathbf{X})\right]^2 = o_p(1)$ holds by assumption~\ref{assm:nuisanceconvergence}.
We also have
\begin{align*}
    \sup_{t\leq 0}P_*\left[\int_{-\infty}^t\left\{\widehat{\phi}_{a}(s|\mathbf{X})-{\phi}_{a,\infty}(s|\mathbf{X})\right\} \, ds \right]^2 &\le  P_*\left[\int_{-\infty}^{\infty} \left|\widehat{\phi}_{a}(s|\mathbf{X})-{\phi}_{a,\infty}(s|\mathbf{X})\right|\, ds\right]^2,
\end{align*}
which is also $o_p(1)$ by~\ref{assm:nuisanceconvergence}. 
This further implies that $\int_{-\infty}^0 {\phi}_{a,\infty}(s|\mathbf{X}) \, ds\in L_2(P_*)$ by Assumption \ref{assm:properCDF}. Additionally, by~\ref{assm:consistency} and~\ref{assm:noumconf}, Jensen's inequality, and the tower property,
\begin{align*}
    P_*\left[\int_{-\infty}^0 {\phi}_{a,*}(s|\mathbf{X})\, ds \, \right]^2 &\leq P_* \left[\int_{-\infty}^{\infty} |s| \, d{\phi}_{a,*}(s|\mathbf{X})\right]^2 \\
    &=\mathbb{E}\left[\mathbb{E}(|Y^a||\mathbf{X})\right]^2\\
    &\le\mathbb{E}\mathbb{E}\left[(Y^a)^2|\mathbf{X}\right]\\
    &=\mathbb{E}\left[(Y^a)^2 \right],
\end{align*}
which is finite by~\ref{assm:logconcv} since log-concave distributions have finite moments. Thus, $C_1 < \infty$. We conclude that $\sup_{t \leq 0} \left| P_* Q_j(t, \mathbf{X}) \right| = o_p(1)$ for $j \in \{1,2,3\}$, which implies that
\begin{align*}
    \sup_{t\in(-\infty,0]}\left|\int_{-\infty}^t P_*\left[ D_{a,\widehat{\theta}_a}(s)-D_{a,{\theta_{a,\infty}}}(s)\right] \, ds\right| \rightarrow_p 0.
\end{align*}


Finally, we show that the first term in \eqref{firstmomentdecomp} is $o_p(1)$. We have
\begin{align}\label{firstmoment.emp.proc.term}
&\int_{-\infty}^t (\mathbb{P}_n- P_*)D_{a,\widehat{\theta}_a}(s) \, ds \nonumber\\
&\quad = (\mathbb{P}_n- P_*)\left[\frac{I(A=a)}{\widehat\pi_a(\mathbf{X})}\left\{I(Y\le t)(t-Y)-\int_{-\infty}^t\widehat{\phi}_a(s|\mathbf{X}) \, ds\right\}+\int_{-\infty}^t\widehat{\phi}_a(s|\mathbf{X}) \, ds\right].
\end{align}
Hence, similar to~\eqref{eq:GCproof.by.function.class}, we can write
\begin{align*}
    \sup_{t \leq 0} \left| \int_{-\infty}^t (\mathbb{P}_n- P_*)D_{a,\widehat{\theta}_a}(s) \, ds  \right| &\leq \sup_{f \in \mc \mc F^a_{\rm int}} \left| (\mathbb{P}_n- P_*) f \right|,
\end{align*}
where 
\begin{align*}
    \mc F^a_{\rm int} := \left\{ g(f_1, \pi, f_3, f_4) : f_1 \in \mc F_1, \pi \in \mc F_\pi, f_3 \in \mc F_3, f_4 \in \mc F_4\right\}
\end{align*}
for $g$ and $\mc F_1$ defined following~\eqref{eq:GCproof.by.function.class}, $\mc F_3:= \left\{Y\mapsto I(Y\le t)(t-Y):t\in(-\infty,0]\right\}$, and
$$\mc F_4:=\left\{\mathbf{X}\mapsto\int_{-\infty}^t \phi_a(s|\mathbf{X})\, ds :  \phi_a: \in\mc F_\phi,\, t\in(-\infty,0]\right\}.$$
The classes $\mc F_1$ and $\mc F_\pi$ are $P_*$-Glivenko Cantelli as noted above. 
For $\mc F_3$, we note that $\sup_{t \leq 0} |(\PP_n-P_*)(Y-t)| =|(\PP_n-P_*)Y| \to_{a.s.}  0$ because $P_* |Y| < \infty$, so the class $\{Y \mapsto Y - t : t \leq 0\}$ is $P_*$-Glivenko Cantelli. 
Hence, since $\{Y \mapsto I(Y \leq t) : t \leq 0\}$ is also $P_*$-Glivenko Cantelli and $Y \mapsto |Y|$ is an envelope for $\mc F_3$, $\mc F_3$ is $P_*$-Glivenko Cantelli by Proposition \ref{gcpreserve}. 
Finally, $\mc F_4$ is $P_*$-Glivenko-Cantelli by Assumption \ref{assm:integratedGC}. 
Also by~\ref{assm:properCDF}, an envelope function for $\mc F^a_{\rm int}$ is given by
\[ K I(A = a)\left\{ |Y| + h(\mathbf{X}) \right\} + h(\mathbf{X}) \]
Since $\EE[I(A=a) |Y|] \le \EE |Y^a|<\infty$, and $h(\mathbf{X})\in L_2(P_*)$ by assumption \ref{assm:properCDF}, this envelope function is integrable. 
Thus, by Theorem 3 in \citet{preserveGC} (see Proposition \ref{gcpreserve}), $\mc F^a_{\rm int}$ is $P_*$-Glivenko-Cantelli.
This implies \eqref{firstmoment.emp.proc.term} is $o_p(1)$, so that 
\begin{align*}
\sup_{t\le 0} \left|\int_{-\infty}^{t} \left[ \widehat F_{a,n}(s)-F_a(s)\right] \, ds\right|=o_p(1).   
\end{align*}
This further yields that \eqref{First.moment.term3} is $o_p(1)$, which concludes the proof.
\end{proof}

\subsection{Proof of Theorem \ref{prop:CI}}
\subsubsection{Lemmas needed for proof Theorem \ref{prop:CI}}
\label{subsec:lemmas.for.CI}
We define three basic processes,
\begin{align*}
    &\widehat G_{a,n}(x) = \int_{-\infty}^x  \widehat{p}_{a,n}(t) dt = \int_{s_1}^x {\rm exp} (\widehat{\varphi}_{a,n})(t) dt,\\
    &\widehat H_{a,n}(x)=\int_{s_1}^x \widehat G_{a,n}(t) dt,\\
    &\widehat H^c_{a,n}(x)=\int_{s_1}^x \widehat{F}^c_{a,n}(t) dt.
\end{align*}
The following lemma is a slight modification of Lemma A.1 in \citet{BW07}, which is itself an extension of the pioneering work of \citet{Kim:1990ue}.
\begin{lemma}\label{emptermlemma}
Let $\mc F$ be a collection of functions defined on $[s_0-\delta,s_0+\delta]^2\times \mathbb{R}^m$ with small $\delta>0$ and arbitrary positive integer $m$. Suppose that for a fixed $s_1\in [s_0-\delta,s_0+\delta]$ and $R>0$, such that $s_0-\delta \le s_1 \le s_2 \le s_1+R\le s_0+\delta$, the collection 
\begin{align*}
    \mathcal{F}_{s_0,R}=\{f_{s_1,s_2}: f_{s_1,s_2}(\mathbf{x})=f(s_1,s_2,\mathbf{x})\in \mathcal{F},s_0-\delta \le s_1 \le s_2 \le s_1+R\le s_0+\delta\}
\end{align*}
admits an envelope $F_{s_0,,R}$, such that
\begin{align*}
    \mathbb{E} F^2_{s_0,R}(\mathbf{X})\le K_0R^{2t-1}, ~~R\le R_0
\end{align*}
for some $t\ge 1/2$ and $K_0>0$, depending only on $s_0$ and $\delta$. Moreover, suppose that 
\begin{align*}
    \sup_{Q}\int_0^1\sqrt{{\rm log}N(\eta\|F_{s_1,R}\|_{Q,2},\mathcal{F}_{s_0,R},L_2(Q))}d\eta <\infty.
\end{align*}
Then, for each $\epsilon>0$, there exist random variables $M_n$ of order $O_P(1)$ which does not depend on $s_1,s_2$ and $R_0>0$, such that
\begin{align*}
    |(\PP_n-\PP)f_{s_1,s_2}|\le \epsilon|s_2-s_1|^{l+t}+n^{-(l+t)/(2l+1)}M_n~~~~{\rm for}~|s_2-s_1| \le R_0
\end{align*}
for $f\in \mathcal{F}_{s_0,R}$ and $l>0$.
\end{lemma}
\begin{proof}
The proof is identical to the proof of Lemma A.1 of \citet{BW07} where one can prove the same result for $\mathcal{F}_{s_0,R}$ in a similar fashion to $\mathcal{F}_{x,R}$ in their proof.   
\end{proof}

The following lemma about analytical properties of $p_a,\varphi_a$ is identical to Lemma 4.2 in \citet{BRW09}.
\begin{lemma}\label{techtermlemma}
If \ref{assm:pregularity} and \ref{assm:pdiffwithk} hold, then, for $a\in\{0,1\}$, we have
\begin{align}
    p_a^{(j)}(s_0)=[\varphi'_a(s_0)]^jp_a(s_0)~~~~{\rm for}~j=1,\dots,k-1
\end{align}
and, for $j=k$,
\begin{align}
    p_a^{(k)}(s_0)=(\varphi_a^{(k)}(s_0)+[\varphi'_a(s_0)]^k)p_a(s_0).
\end{align}
\end{lemma}

The following lemma shows that the distance between two adjacent knots around the target point $s_0$ follows the same rate of convergence with the non-causal log-concave MLE \citep{BRW09}. Recall that $\tau_n^{+},\tau_n^{-}$ are defined in \eqref{knots} and $\delta_n$ denotes the regular grid length for the isotonic regression grid (see \eqref{def:isotonicongrid} and the preceding texts).
\begin{lemma}\label{successiveknots}
If \ref{assm:consistency}--\ref{assm:logconcv}, \ref{assm:condition.grid}, \ref{assm:nuisanceconvergence}--\ref{assm:phiinfholder},
\ref{assm:basicGC}--\ref{assm:pi-bracketing}, and
\ref{assm:pregularity}--\ref{assm:pdiffwithk} hold, and if $\delta_n=O_p(n^{-(k+1)/(2k+1)})$, then, for $a \in \{0,1\}$, 
    \begin{align*}
        \tau_n^{+}-\tau_n^{-}=O_p(n^{-1/(2k+1)}).
    \end{align*}
\end{lemma}
\begin{proof}
The proof is aligned with the proof of Theorem 4.3 and Lemma 4.4 from \cite{BRW09}. We similarly define
\begin{align*}
    R_{1n}&:=\int \Delta (y)(\widehat p_{a,n}-p_a)(y) dy,\\
    R_{2n}&:=\int \Delta (y)d(\widehat F^c_{a,n}-F_a)(y),
\end{align*}
where 
\begin{align*}
  \Delta (y)=(y-\tau_n^{-})I_{[\tau_n^{-},\Bar{\tau}]}(y)+(\tau_n^{+}-y)I_{[\Bar{\tau},\tau_n^{+}]}(y)-\frac{1}{4}(\tau_n^{+}-\tau_n^{-})I_{[\tau_n^{-},\tau_n^{+}]}(y),
\end{align*}
and $\Bar{\tau}=\frac{\tau_n^{-}+\tau_n^{+}}{2}$. Following the same steps in \cite{BRW09} for the log-concave MLE $\widehat p_{a,n}$, one can easily verify that
\begin{align}
    &-R_{1n}\le \frac{\tau_n^{+}-\tau_n^{-}}{2n}-R_{2n} \label{ineqforR1nandR2n}\\
    &R_{1n}=M_kp_a(s_0)\varphi_a^{(k)}(s_0)(\tau_n^{+}-\tau_n^{-})^{k+2}+o_p((\tau_n^{+}-\tau_n^{-})^{k+2}),\label{R1ndescription}
\end{align}
where $M_k>0$ depends only on $k$. Furthermore, since uniform convergence of $\widehat F_{a,n}$ to $F_a$ (see proof of Theorem \ref{prop:Consistency}) holds with $F_a(s_0)>0$, and consistency from Theorem \ref{prop:Consistency} implies $\tau_n^{+}-\tau_n^{-}=o_p(1)$, we can identify $\widehat F_{a,n}$ with $\widehat F^{c_0}_{a,n}$ locally for the term $R_{2n}$.
Thus, we identify $\widehat F_{a,n}$ with $\widehat F^{c_0}_{a,n}$ throughout the following proof.
We further decompose $R_{2n}$ as follows:
\begin{align}
   \int \Delta (y)d(\widehat F^c_{a,n}-F_a)(y)= \int \Delta (y)d(\widehat F^c_{a,n}-\widehat F_{a,n})(y) + \int \Delta (y)d(\widehat F_{a,n}-F_a)(y). \label{deltadecompose}
\end{align}
We start with analyzing the second term in \eqref{deltadecompose}. With integration by parts, one can check
\begin{align}
\int \Delta (y)d(\widehat F_{a,n}-F_a)(y)&=\int_{\Bar{\tau}}^{\tau_n^{+}} \Big[(\widehat F_{a,n}-F_a)(y)-(\widehat F_{a,n}-F_a)(2\Bar{\tau}-y)\Big]dy\label{onestepdecomp.pt1}\\
&\indent-\frac{\tau_n^{+}-\tau_n^{-}}{4}\Big[(\widehat F_{a,n}-F_a)(\tau_n^{+})-(\widehat F_{a,n}-F_a)(\tau_n^{-})\Big].\label{onestepdecomp.pt2}    
\end{align}
To examine terms \eqref{onestepdecomp.pt1} and \eqref{onestepdecomp.pt2}, we exploit the decomposition $\widehat F_{a,n}(s)-F_a(s)=\PP_n D_{a,\widehat\theta_a}(s)-P_* D_{a,{\theta_{a,\infty}}}(s)=(\PP_n-P_*)D_{a,\widehat\theta_a}(s) + P_*(D_{a,\widehat\theta_a}(s)-D_{a,{\theta_{a,\infty}}}(s))$. Then \eqref{onestepdecomp.pt2} can be expressed as
\begin{align*}
 \frac{\tau_n^{+}-\tau_n^{-}}{4} \Big[(\PP_n-P_*)(D_{a,\widehat\theta_a}(\tau_n^{+})&-D_{a,\widehat\theta_a}(\tau_n^{-})) \\
 &+ P_*\Big((D_{a,\widehat\theta_a}-D_{a,{\theta_{a,\infty}}})(\tau_n^{+})-(D_{a,\widehat\theta_a}-D_{a,{\theta_{a,\infty}}})(\tau_n^{-})\Big)\Big]. 
\end{align*}
For $[s_1,s_2]\subseteq [s_0-\delta,s_0+\delta]$, 
\begin{equation}\label{empterm1}
\begin{split}
  (\PP_n-P_*)(D_{a,\widehat\theta_a}(s_2)-D_{a,\widehat\theta_a}(s_1))&= (\PP_n-P_*)\Big[\frac{I(A=a)}{\widehat\pi_a(\mathbf{X})}I(s_1<Y\le s_2)\\
  &\indent+(\widehat\phi_a(s_2|\mathbf{X})-\widehat\phi_a(s_1|\mathbf{X}))\Big(1-\frac{I(A=a)}{\widehat\pi_a(\mathbf{X})}\Big)\Big].
\end{split}
\end{equation}
To apply Lemma \ref{emptermlemma}, we line up with the conditions for two terms in \eqref{empterm1}. Firstly, $\{(a,b]:a<b\}$ has VC dimension of 2 (see Example 2.6.1 in \citealp{vdvandW}). Thus, since the function class $\mathcal{F}_1=\{I_{(s_1,s_2](\cdot)}:[s_1,s_2]\subseteq I_{s_0,\omega}, s_1\le s_2\le s_1+R\}$ allows an envelope function of constant 1,
\begin{align*}
   N(\epsilon,\mathcal{F}_1,L_2(Q))\le C_3\Big(\frac{1}{\epsilon}\Big)^2
\end{align*}
for any probability measure $Q$ and a constant $C_3>0$ (see Theorem 2.6.7 in \citealp{vdvandW}). Similar reasoning can be applied to the indicator function on a single point. In addition, with Assumption \ref{assm:boundedpropscore}, \ref{assm:pi-bracketing}, and Lemma 5.1 in \citet{vdvvl06}, a function class $\mathcal{F}_2:=\Big\{\frac{I_{\{a\}}(x_1)}{\pi_a(\mathbf{x_2})}: \pi_a \in \mathcal{F}_{\pi}, x_1\in \{0,1\}, \mathbf{x_2}\in \mathbb{R}^d\Big\}$ satisfies, for arbitrary probability measure $Q$,
\begin{align*}
    \sup_Q N(\epsilon K, \mathcal{F}_2, L_2(Q)) \le  \sup_Q N(\epsilon , \mathcal{F}_{\pi}, L_2(Q)) 
\end{align*}
by Assumption \ref{assm:boundedpropscore} and knowing that the function class $\mathcal{F}_{\pi}$ allows the same envelope as $\mathcal{F}_1$.
Thus, the whole function class $\mathcal{F}_{t1}:=\mathcal{F}_1\cdot \mathcal{F}_2$ for the first term in \eqref{empterm1} has finite uniform entropy integral, by Assumption \ref{assm:pi-bracketing} and Lemma 5.1 from \citet{vdvvl06},
\begin{align*} 
 \sup_Q {\rm log} N(\epsilon , \mathcal{F}_{t1}, L_2(Q)) \lesssim \epsilon^{-V}-{\rm log}(\epsilon) 
\end{align*}
up to a constant, with $V\in[0,2)$. Furthermore, the class allows an envelope $KI(s_0-\omega<Y\le s_0+\omega)I(A=a)$ which satisfies
\begin{align*}
    \mathbb{E}\Big(KI(s_0-\omega<Y\le s_0+\omega)I(A=a)\Big)^2 &= K^2P_*(s_0-\omega<Y^a\le s_0+\omega)\\
    &\le C_4\omega K^2
\end{align*}
for some constant $C_4>0$. Hence, the first term in \eqref{empterm1} satisfies the assumptions of Lemma \ref{emptermlemma} with $t=1,l=k$, and therefore, there exist a random variable $B_1$ which has order of $O_p(n^{-(k+1)/(2k+1)})$ and is independent of $s_1,s_2$ such that
\begin{align}\label{empterm:identity.part}
    \sup_{[s_1,s_2]\subset I_{s_0,\omega},s_1\le s_2\le s_1+R} \Big|(\PP_n-P_*)\Big[\frac{I(A=a)}{\widehat\pi_a(\mathbf{X})}I(s_1<Y\le s_2)\Big|\le \epsilon|s_2-s_1|^{k+1}+B_1,
\end{align}
for arbitrary $\epsilon>0$.
Next, we set up a similar reasoning for the second term in \eqref{empterm1}. With Assumption \ref{assm:phiestholder} and Lemma 5.1 from \citet{vdvvl06}, a function class $\mathcal{F}_3:=\{\phi_a(s_2|\cdot)-\phi_a(s_1|\cdot): [s_1,s_2]\subseteq [s_0-\delta,s_0+\delta], s_1\le s_2\le s_1+R, \phi_a \in \mathcal{F}_{\phi}\}$ satisfies
\begin{align*}
 \sup_Q{\rm log} N(\epsilon, \mathcal{F}_3, L_2(Q)) \lesssim  \epsilon^{-V} 
\end{align*}
up to a constant with $V\in[0,2)$, and supremum taken over any probability measure $Q$. Thus, similar to the reasoning above combined with Assumption \ref{assm:boundedpropscore} and Lemma 5.1 in \citet{vdvvl06}, one can easily show that a function class $\mathcal{F}_{t2}:=\mathcal{F}_3\cdot(1-\mathcal{F}_2)$ for the second term satisfies
\begin{align*}
 \sup_Q {\rm log} N(\epsilon, \mathcal{F}_3, L_2(Q)) \lesssim  \epsilon^{-V}.  
\end{align*}
This yields the finite uniform entropy integral of this function class. From Assumption \ref{assm:phiestholder}, this class has an envelope $(1+K)R_1(\mathbf{X})\omega^{\alpha}$, and, the envelope satisfies
\begin{align*}
\mathbb{E}\Big((1+K)R_1(\mathbf{X})R^{\alpha}\Big)^2=\omega^{2\alpha}(1+K)^2\mathbb{E}R_1^2(\mathbf{X}).
\end{align*}
Hence, this further implies that the second term satisfies the conditions in Lemma \ref{emptermlemma} with some $t>1$ and $l=k$. In other words, for each $\epsilon>0$, there exist a random variable $B_2$ which has order of $O_p(n^{-(k+t)/(2k+1)})$ and is independent of $s_1,s_2$ such that
\begin{align}\label{empterm:phi.part}
    \sup_{[s_1,s_2]\subset I_{s_0,\omega},s_1\le s_2\le s_1+R} \Big|(\widehat\phi_a(s_2|\mathbf{X})-\widehat\phi_a(s_1|\mathbf{X}))\Big(1-\frac{I(A=a)}{\widehat\pi_a(\mathbf{X})}\Big)\Big|\le \epsilon|s_2-s_1|^{k+t}+B_2.   
\end{align}
Combining the two empirical process terms \eqref{empterm:identity.part} and \eqref{empterm:phi.part},  for the empirical process part of \eqref{onestepdecomp.pt2}, we obtain
\begin{equation}\label{pt1.decom.emp}
\begin{split}
\Bigg|(\PP_n-P_*)\Big\{\frac{\tau_n^{+}-\tau_n^{-}}{4}&\Big[(\widehat F_{a,n}-F_a)(\tau_n^{+})-(\widehat F_{a,n}-F_a)(\tau_n^{-})\Big]\Big\}\Bigg|\\
&\le 2\epsilon(\tau_n^{+}-\tau_n^{-})^{k+2}+(\tau_n^{+}-\tau_n^{-})O_p\Big(n^{-(k+1)/(2k+1)}\Big),    
\end{split}    
\end{equation}
for each $\epsilon>0$.
We now check the empirical process term of \eqref{onestepdecomp.pt1}. Since, in \eqref{empterm:identity.part} and \eqref{empterm:phi.part}, $B_1$ and $B_2$ do not depend on $s_1,s_2$, we have
\begin{equation}\label{pt2.decom.emp}
\begin{split}
 \Big|\int_{\Bar{\tau}}^{\tau_n^{+}} &(\PP_n-P_*) (D_{a,\widehat\theta_a}(y)-D_{a,\widehat\theta_a}(2\Bar{\tau}-y)) dy\Big|\\
 &\le \int_{\Bar{\tau}}^{\tau_n^{+}} \Big|(\PP_n-P_*) (D_{a,\widehat\theta_a}(y)-D_{a,\widehat\theta_a}(2\Bar{\tau}-y))\Big| dy\\
 &\le 2\epsilon(\tau_n^{+}-\tau_n^{-})^{k+2}+(\tau_n^{+}-\tau_n^{-})O_p\Big(n^{-(k+1)/(2k+1)}\Big),  
\end{split}    
\end{equation}
due to $2y-2{\Bar{\tau}}\le(\tau_n^{+}-\tau_n^{-})$ in $y\in[{\Bar{\tau}},\tau_n^{+}]$.

On the other hand, for the remainder term analyses for \eqref{onestepdecomp.pt1} and \eqref{onestepdecomp.pt2}, we exploit a decomposition as follows
\begin{align}
  P_*&\Big[(D_{a,\widehat\theta_a}-D_{a,{\theta_{a,\infty}}})(\tau_n^{+})-(D_{a,\widehat\theta_a}-D_{a,{\theta_{a,\infty}}})(\tau_n^{-})\Big]\nonumber\\
  &=P_*\Big[\frac{\widehat\pi_a(\mathbf{X})-\pi_{\infty}(\mathbf{X})}{\widehat\pi_a(\mathbf{X})}\Big((\widehat\phi_a-\phi_{a,\infty})(\tau_n^{+}|\mathbf{X})-(\widehat\phi_a-\phi_{a,\infty})(\tau_n^{-}|\mathbf{X})\Big)\label{remainder1}\\
  &\indent+\frac{\pi_{a,*}(\mathbf{X})(\widehat\pi_a(\mathbf{X})-\pi_{a,\infty}(\mathbf{X}))}{\widehat\pi_a(\mathbf{X})\pi_{a,\infty}(\mathbf{X})}\Big((\phi_{a,\infty}-\phi_{a,*})(\tau_n^{+}|\mathbf{X})-(\phi_{a,\infty}-\phi_{a,*})(\tau_n^{-}|\mathbf{X})\Big)\label{remainder2}\\
  &\indent+\frac{\pi_{a,\infty}(\mathbf{X})-\pi_{a,*}(\mathbf{X})}{\widehat\pi_a(\mathbf{X})}\Big((\widehat\phi_a-\phi_{a,\infty})(\tau_n^{+}|\mathbf{X})-(\widehat\phi_a-\phi_{a,\infty})(\tau_n^{-}|\mathbf{X})\Big)\Big].\label{remainder3}
\end{align}
Then, the absolute values of three terms in \eqref{remainder1}, \eqref{remainder2}, \eqref{remainder3} are order of $(\tau_n^{+}-\tau_n^{-})o_p(n^{-k/(2k+1)})$ from Assumption \ref{assm:boundedpropscore}, \ref{assm:phiregularity}, and \ref{assm:L2conditionsforCI}, and since true $p_a$ is unimodal and $P_*[\phi_{a,*}(\tau_n^{+}|\mathbf{X})-\phi_{a,*}(\tau_n^{-}|\mathbf{X})]=F_a(\tau_n^{+})-F_a(\tau_n^{-})$ which is $(\tau_n^{+}-\tau_n^{-})O_p(1)$. Thus, we have the following result for the remainder term of \eqref{onestepdecomp.pt2}.
\begin{equation}\label{reference.Remainder.1}
\begin{split}
(\tau_n^{+}&-\tau_n^{-}) P_*\Big((D_{a,\widehat\theta_a}-D_{a,{\theta_{a,\infty}}})(\tau_n^{+})-(D_{a,\widehat\theta_a}-D_{a,{\theta_{a,\infty}}})(\tau_n^{-})\Big)\\
    &\asymp (\tau_n^{+}-\tau_n^{-})^2o_p(n^{-k/(2k+1)}),
\end{split}
\end{equation}
up to a constant. Since, similarly to \eqref{pt2.decom.emp}, the random variables $M_i$s ($i=1,2,3)$ in Assumption \ref{assm:L2conditionsforCI} does not depend on $s_1,\,s_2$, the following holds for the remainder term of \eqref{onestepdecomp.pt2}.
\begin{equation}\label{reference.Remainder.2}
\begin{split}
P_*&\Big\{\int_{\Bar{\tau}}^{\tau_n^{+}}\Big[(D_{a,\widehat\theta_a}-D_{a,{\theta_{a,\infty}}})(y)-(D_{a,\widehat\theta_a}-D_{a,{\theta_{a,\infty}}})(2\Bar{\tau}-y)\Big] \Big\}\\
&\asymp (\tau_n^{+}-\tau_n^{-})^2o_p(n^{-k/(2k+1)}),
\end{split}
\end{equation}
up to a constant.
Hence, we obtain a result for one step estimator $\widehat F_{a,n}$ as follows.
\begin{align}\label{onestepresult}
  \Big|\int \Delta (y)d(\widehat F_{a,n}-F_a)(y)\Big| &\le4\epsilon(\tau_n^{+}-\tau_n^{-})^{k+2}+(\tau_n^{+}-\tau_n^{-})O_p(n^{-(k+1)/(2k+1)})\nonumber\\
  &\indent+(\tau_n^{+}-\tau_n^{-})^2o_p(n^{-k/(2k+1)}). 
\end{align}
Now it remains to study the remaining term from \eqref{deltadecompose},
\begin{align*}
    \Big|\int \Delta (y)d(\widehat F^c_{a,n}-\widehat F_{a,n})(y)\Big|.
\end{align*}
Analogously to the former expansion for $\widehat F_{a,n}-F_a$ (see \eqref{onestepdecomp.pt1} and \eqref{onestepdecomp.pt2}), we have
\begin{align*}
   \int \Delta (y)d(\widehat F^c_{a,n}-\widehat F_{a,n})(y)&=\int_{\Bar{\tau}}^{\tau_n^{+}}\Big((\widehat F^c_{a,n}-\widehat F_{a,n})(y)-(\widehat F^c_{a,n}-\widehat F_{a,n})(2\Bar{\tau}-y)\Big)dy\nonumber\\
   &\indent-\frac{\tau_n^{+}-\tau_n^{-}}{4}\Big[(\widehat F^c_{a,n}-\widehat F_{a,n})(\tau_n^{+})-(\widehat F^c_{a,n}-\widehat F_{a,n})(\tau_n^{-})\Big].
\end{align*}
By a straightforward application of Lemma \ref{lem:monotone.correction.ext}, we have
\begin{align}
 \frac{\tau_n^{+}-\tau_n^{-}}{4}\Big[(\widehat F^c_{a,n}-\widehat F_{a,n})(\tau_n^{+})-(\widehat F^c_{a,n}-\widehat F_{a,n})(\tau_n^{-})\Big]=(\tau_n^{+}-\tau_n^{-})O_p(n^{-(k+1)/(2k+1)}).\label{term:didremainder}
\end{align}
Similarly, one can easily check
\begin{align}
    \Big|\int_{\Bar{\tau}}^{\tau_n^{+}}\Big((\widehat F^c_{a,n}-\widehat F_{a,n})(y)-(\widehat F^c_{a,n}-\widehat F_{a,n})(2\Bar{\tau}-y)\Big)dy\Big|=(\tau_n^{+}-\tau_n^{-})O_p(n^{-(k+1)/(2k+1)})\label{term:integratedremainder}.
\end{align}

Hence, combining the above results \eqref{onestepresult}, \eqref{term:didremainder} and \eqref{term:integratedremainder} with \eqref{R1ndescription}, \eqref{deltadecompose}, and plugging them altogether into \eqref{ineqforR1nandR2n}, this yields
\begin{align*}
M_kp_a(s_0)&|\varphi_a^{(k)}(s_0)|(\tau_n^{+}-\tau_n^{-})^{k+2}+o_p((\tau_n^{+}-\tau_n^{-})^{k+2})\\
&\le \frac{3(\tau_n^{+}-\tau_n^{-})}{2n}+ 4\epsilon (\tau_n^{+}-\tau_n^{-})^{k+2}+(\tau_n^{+}-\tau_n^{-})O_p(n^{-(k+1)/(2k+1)})\\
  &\indent+(\tau_n^{+}-\tau_n^{-})^2o_p(n^{-k/(2k+1)}). 
\end{align*}
Similarly to the proof of Lemma 4.4 in \citet{BRW09}, since $\epsilon>0$ is arbitrary and uniform consistency from Theorem \ref{prop:Consistency} with $\varphi_a^{(k)}(s_0)<0$ implies $(\tau_n^{+}-\tau_n^{-})=o_p(1)$, this proves the claim.
\end{proof}

The following lemma is a straightforward extension of Theorem 2 in \citet{WvdlC2020}.
\begin{lemma}\label{lem:monotone.correction.ext}
Under the same assumptions for Lemma \ref{successiveknots}, if an interval $[l_n,u_n]$ that contains $s_0$ satisfies $|u_n-l_n|=o_p(1)$, we have
\begin{align*}
    \sup_{s\in[l_n,u_n]}|\widehat F^c_{a,n}(s)-\widehat F_{a,n}(s)| = O_p(n^{-(k+1)/(2k+1)})
\end{align*}
for $a \in \{0,1\}$.
\end{lemma}

\begin{proof}
The proof follows the steps in Lemma 1, 2, and 3 of
\citet{WvdlC2020}.
We give a sketch here. 
First, Assumption \ref{assm:positivity} and \ref{assm:phiinfholder} line up with the condition (i) and (ii) in Section 4.1 (see page 16) of \citet{WvdlC2020} which correspond to condition (B) and (C) therein (see page 8). Since $\delta_n=O_p(n^{-(k+1)/(2k+1)})$ and we verified that 
\begin{align*}
 |(\widehat F_{a,n}(t)-\widehat F_{a,n}(s))-(F_a(t)-F_a(s))|&\le\epsilon|t-s|^{k+1}+O_p(n^{-(k+1)/(2k+1)})\\
 &\indent+|t-s|o_p(n^{-k/(2k+1)}),  
\end{align*}
for arbitrary small $\epsilon>0$ in the proof of Lemma \ref{successiveknots}, this yields
\begin{align*}
    \sup_{|t-s|<cw_n} w_n^{-1}|(\widehat F_{a.n}(t)-\widehat F_{a.n}(s))-(F_a(t)-F_a(s))|=O_p(1),
\end{align*}
which meets the condition (A) in \citet{WvdlC2020} (see page 8), where $w_n=n^{-(k+1)/(2k+1)}$.
When we define $\kappa_n:=\sup\{|t-s|:t,s\in[L_n,U_n],s\le t, \widehat F_{a,n}(t)\le \widehat F_{a,n}(s)\}$, then by the same procedure in Lemma 2 of \citet{WvdlC2020}, we obtain $\kappa_n=O(n^{-(k+1)/(2k+1)})$. This yields the conclusion with Lemma 3 in \citet{WvdlC2020}.
\end{proof}

The following lemma is analogous to Lemma 4.5 in \citet{BRW09}.
\begin{lemma}
\label{lem:local-tightness-phi}
For any $T>0$, under the same assumptions as in Lemma \ref{successiveknots}, for $a \in \{0,1\}$, we have
\begin{align}
    &\sup_{|t|\le T} |\widehat\varphi'_{a,n}(s_0+v_n t)-\varphi'_a(s_0)|=O_p(n^{-(k-1)/(2k+1)})\label{taylorlemma1}\\
    &\sup_{|t|\le T} |\widehat\varphi_{a,n}(s_0+v_n t)-\varphi_a(s_0)-v_n t\varphi'_a(s_0)|=O_p(n^{-k/(2k+1)}),\label{taylorlemma2}
\end{align}
where $v_n=n^{-1/(2k+1)}$. This implies that for $ \widehat\varpi_{a,n,1}(u):=\sum_{j=k+1}^{\infty}\frac{1}{j!}(\widehat\varphi_{a,n}(u)-\varphi_a(s_0))^j$,
\begin{align}
 \widehat\varpi_{a,n,1}(u)=o_p(n^{-k/(2k+1)}),   \label{taylorlemma3} 
\end{align}
uniformly in $u\in[s_0-tn^{-1/(2k+1)},s_0+tn^{-1/(2k+1)}]$, where $|t|\le T$.
Furthermore, if we define, for any $u\in\mathbb{R}$,
\begin{align*}
    \widehat e_{a,n}(u)=\widehat p_{a,n}-\sum_{j=0}^{k-1}\frac{p_a^{(j)}(s_0)}{j!}(u-s_0)^j-p_a(s_0)\frac{[\varphi'_a(s_0)]^k}{k!}(u-s_0)^k,
\end{align*}
then
\begin{align}
 \sup_{|t|\le T}&|\widehat e_{a,n}(s_0+v_n t)-p_a(s_0)(\widehat\varphi_{a,n}(s_0+v_n t)-\varphi_a(s_0)-v_n t\varphi'_a(s_0))| \label{taylorlemma4} \\
 &=o_p(n^{-k/(2k+1)}).
\end{align}
\end{lemma}
\begin{proof}
The proof for \eqref{taylorlemma1}, \eqref{taylorlemma2} is straightforward extension from the proof for Lemma 4.3, and Lemma 4.4 in \citet{Groeneboom:2001jo} with (our) Lemma \ref{successiveknots}. Then, the proof for \eqref{taylorlemma3}, \eqref{taylorlemma4} is identical to the proof for Lemma 4.5 in \citet{BRW09}, relying on \eqref{taylorlemma1} and \eqref{taylorlemma2}.
\end{proof}
We construct local processes
\begin{align*}
    \widehat G^{\rm loc}_{a,n}(t)&=r_n\int_{s_0}^{s_n(t)}\Big(\widehat F^c_{a,n}(v)-\widehat F^c_{a,n}(s_0)-\int_{s_0}^v\sum_{j=0}^{k-1}\frac{p_a^{(j)}(s_0)}{j!}(u-s_0)^j du\Big)dv,
\end{align*}
and
\begin{align*}
    \widehat H^{\rm loc}_{a,n}(t)&=r_n\int_{s_0}^{s_n(t)} \int_{s_0}^v \Big(\widehat p_{a,n}(u)-\sum_{j=0}^{k-1}\frac{p_a^{(j)}(s_0)}{j!}(u-s_0)^j\Big) du dv+\widehat A_{a,n}t+\widehat B_{a,n},
\end{align*}
where 
\begin{align*}
    \widehat A_{a,n}&=r_nv_n(\widehat G_{a,n}(s_0)-\widehat F^c_{a,n}(s_0)),\\
    \widehat B_{a,n}&=r_n(\widehat H_{a,n}(s_0)-\widehat H^c_{a,n}(s_0)),
\end{align*}
where $r_n=n^{(k+2)/(2k+1)},v_n=n^{-1/(2k+1)},s_n(t)=s_0+v_nt$.
We further define modified processes of each $\widehat G^{\rm loc}_{a,n}$ and $\widehat H^{\rm loc}_{a,n}$ as follows.
\begin{align*}
    \widehat G^{\rm locmod}_{a,n}(t)&=\frac{r_n}{p_a(s_0)}\int_{s_0}^{s_n(t)}\Big(\widehat F^c_{a,n}(v)-\widehat F^c_{a,n}(s_0)-\int_{s_0}^v\sum_{j=0}^{k-1}\frac{p_a^{(j)}(s_0)}{j!}(u-s_0)^j du\Big)dv\\
    &\indent-r_n\int_{s_0}^{s_n(t)} \int_{s_0}^v \widehat\varpi_{a,n}(u)dudv\\
    &=\frac{1}{p_a(s_0)}\widehat G^{\rm loc}_{a,n}(t)-r_n\int_{s_0}^{s_n(t)} \int_{s_0}^v \widehat\varpi_{a,n}(u)dudv,
\end{align*}
and
\begin{align*}    
    \widehat H^{\rm locmod}_{a,n}(t)&=r_n\int_{s_0}^{s_n(t)} \int_{s_0}^v(\widehat\varphi_{a,n}(u)-\varphi_a(s_0)-(u-s_0)\varphi'_a(s_0))dudv+\frac{\widehat A_{a,n}t+\widehat B_{a,n}}{p_a(s_0)}\\
    &=\frac{1}{p_a(s_0)}\widehat H^{\rm loc}_{a,n}(t)-r_n\int_{s_0}^{s_n(t)} \int_{s_0}^v \widehat\varpi_{a,n}(u)dudv,
\end{align*}
and
\begin{align*}
    \widehat\varpi_{a,n}(u)&=\sum_{j=2}^{\infty}\frac{1}{j!}[\widehat\varphi_{a,n}(u)-\varphi_a(s_0)]^j-\sum_{j=2}^{k-1} \frac{[\varphi'_a(s_0)]^j}{j!}(u-s_0)^j\\
    &=\widehat\varpi_{a,n,1}+\sum_{j=2}^{k}\frac{1}{j!}[\widehat\varphi_{a,n}(u)-\varphi_a(s_0)]^j-\sum_{j=2}^{k-1} \frac{[\varphi'_a(s_0)]^j}{j!}(u-s_0)^j.
\end{align*}
We let $W$ denote a standard two-sided Brownian motion starting at $0$. For each $t\in \RR, k\in \NN$, we defined the integrated Gaussian process with drift $Y_k$ as
\begin{align}\label{integrgaussprocess}
   Y_k(t):=\begin{cases}
       \int_0^t W(s)ds-t^{k+2}, &{\rm if~} t\ge 0,\\
       \int_t^0 W(s)ds-t^{k+2}, &{\rm if~} t<0.
   \end{cases} 
\end{align}
The invelope process $H_k$ of $Y_k$ is then the unique process satisfying:
\begin{equation}\label{H-k:definition}
 \begin{gathered}
 	H_k(t)\le Y_k(t) \text{ for all } t\in\RR, \\
         \int \left[ H_{k}(t)-Y_{k}(t)\right] \, dH^{(3)}_{k}(t) = 0, \text{ and}  \\
         H_k^{(2)} \text{ is concave}.
 \end{gathered}
 \end{equation}
We study the asymptotic behavior of localized processes (at the `log density level') in Lemma~\ref{lem:finallemmaforCI}.

\begin{lemma}\label{lem:finallemmaforCI}
Let $T>0$. Under the same assumptions as in Lemma \ref{successiveknots}, the following holds for $a \in \{0,1\}$. 
\begin{enumerate}
    \myitem{(a)} \label{CIlem:a}$\widehat G^{\rm locmod}_{a,n}$ converges weakly in $C[-T,T]$ to the driving process $G_{\alpha_a,k,\sigma_a}$, where
    \begin{align}
     G_{\alpha_a,k,\sigma_a}(t):=\alpha_a\int_0^t W(s)ds-\sigma_a t^{k+2},   
    \end{align}
    for $\alpha_a=\frac{1}{p_a(s_0)}\sqrt{\EE \Big[\frac{\pi_{a,*}(\mathbf{X})}{\pi^2_{a,\infty}(\mathbf{X})}\eta_{a,*}(s_0|\mathbf{X})\Big]}$, $\sigma_a=|\varphi_a^{(k)}(s_0)|/(k+2)!$.
    \myitem{(b)} \label{CIlem:b}The following inequality holds.
    \begin{align*}
     \widehat G^{\rm locmod}_{a,n}(t)-\widehat H^{\rm locmod}_{a,n}(t)\ge 0,  
    \end{align*}
    for all $t\in\RR$. And, equality holds for all $t$ such that $s_n(t)\in\widehat{\mc L}_a$.
    \myitem{(c)} \label{CIlem:c} Both $\widehat A_{a,n},\widehat B_{a,n}$ are tight.
    \myitem{(d)} \label{CIlem:d} The vector of processes
    \begin{align*}
        (\widehat H^{\rm locmod}_{a,n},(\widehat H^{\rm locmod}_{a,n})^{(1)},(\widehat H^{\rm locmod}_{a,n})^{(2)},  \widehat G^{\rm locmod}_{a,n}, (\widehat H^{\rm locmod}_{a,n})^{(3)}, ( \widehat G^{\rm locmod}_{a,n})^{(1)})
    \end{align*}
    converges weakly in $(C[-T,T])^4\times (D[-T,T])^2$, endowed with the product topology induced by the uniform topology on the spaces $C[-T,T]$ and the $M_1$ Skorohod topology on the spaces $D[-T,T]$, to the process
    \begin{align*}
        (H_{\alpha_a,k,\sigma_a}, (H_{\alpha_a,k,\sigma_a})^{(1)}, H_{\alpha_a,k,\sigma_a}^{(2)}, Y_{\alpha_a,k,\sigma_a}, (H_{\alpha_a,k,\sigma_a})^{(3)}, (Y_{\alpha_a,k,\sigma_a})^{(1)}),
    \end{align*}
    where $H_{\alpha_a,k,\sigma_a}$ is the unique process on $\RR$ that satisfies
    \begin{align}\label{def:scaled.HandY.processes}
        \begin{cases}
            H_{\alpha_a,k,\sigma_a}(t) \le Y_{\alpha_a,k,\sigma_a}(t) & for~all~t\in\RR,\\
            \int(H_{\alpha_a,k,\sigma_a}(t)-Y_{\alpha_a,k,\sigma_a}(t))dH^{(3)}_{\alpha_a,k,\sigma_a}(t)=0, & \\
            H^{(2)}_{\alpha_a,k,\sigma_a} & is~concave.
        \end{cases}
    \end{align}
\end{enumerate}
\end{lemma}
To allow multiple jumps to approximate a single jump in $(\widehat H^{\rm locmod}_{a,n})^{(3)}$ and $ ( \widehat G^{\rm locmod}_{a,n})^{(1)}$, we use the $M_1$ topology on the space $D[-T,T]$ instead of $J_1$ topology (see Chapter 12 of \citealp{billingsley2013convergence} for the definition of the $J_1$ topology) which was used in Theorem 4.6 of \citet{BRW09}. The $M_1$ topology is defined in Section 12.3 of \citet{whitt2002stochastic}, and its separability and completeness are thoroughly proved in Lemma 8.22 and Proposition 8.23 of \citet{Doss.Wellner.2019}.
\begin{proof}
For the first part \ref{CIlem:a}, by a straightforward application of the proof of Theorem 4.6 in \citet{BRW09} in combination with the results from Lemma~\ref{lem:local-tightness-phi}, we have
\begin{align*}
    r_n \int_{s_0}^{s_n(t)} \int_{s_0}^v \widehat\varpi_{a,n}(u) du dv = \frac{[\varphi_a(s_0)]^k}{(k+2)!}t^{k+2}+o_p(1).
\end{align*}
Then, by Lemma \ref{lem:local-tightness-phi}, we have
\begin{align*}
    \widehat G^{\rm locmod}_{a,n}(t)&=\frac{r_n}{p_a(s_0)}\int_{s_0}^{s_n(t)}\Big(\widehat F^c_{a,n}(v)-\widehat F^c_{a,n}(s_0)-\int_{s_0}^v\sum_{j=0}^{k-1}\frac{p_a^{(j)}(s_0)}{j!}(u-s_0)^j du\Big)dv\\
    &\indent-\frac{[\varphi_a(s_0)]^k}{(k+2)!}t^{k+2}+o_p(1).
\end{align*}
In addition, the display above further yields
\begin{equation}\label{Glocmod} 
\begin{split}
\widehat G^{\rm locmod}_{a,n}(t) &=\frac{r_n}{p_a(s_0)}\int_{s_0}^{s_n(t)}\Big(\widehat F_{a,n}(v)-\widehat F_{a,n}(s_0)-\int_{s_0}^v\sum_{j=0}^{k-1}\frac{p_a^{(j)}(s_0)}{j!}(u-s_0)^j du\Big)dv\\
&\indent+\frac{r_n}{p_a(s_0)}\int_{s_0}^{s_n(t)}\Big((\widehat F^c_{a,n}(v)-\widehat F^c_{a,n}(s_0))-(\widehat F_{a,n}(v)-\widehat F_{a,n}(s_0))\Big)dv\\
 &\indent-\frac{[\varphi_a(s_0)]^k}{(k+2)!}t^{k+2}+o_p(1)\\
&=\frac{r_n}{p_a(s_0)}\int_{s_0}^{s_n(t)}\Big((\widehat F_{a,n}(v)-\widehat F_{a,n}(s_0))-(F_a(v)-F_a(s_0))\Big)dv\\
&\indent+\frac{p_a^{(k)}(s_0)}{(k+2)!p_a(s_0)}t^{k+2}-\frac{[\varphi_a(s_0)]^k}{(k+2)!}t^{k+2}+o_p(1),
\end{split}
\end{equation}
where the last equality arises from Lemma \ref{lem:integratedDiD}. Moreover, by Lemma \ref{techtermlemma}, we have
\begin{align*}
\frac{p_a^{(k)}(s_0)}{(k+2)!p_a(s_0)}t^{k+2}-\frac{[\varphi_a(s_0)]^k}{(k+2)!}t^{k+2}= \frac{\varphi_a^{(k)}(s_0)}{(k+2)!}t^{k+2}.  
\end{align*}
With the decomposition $\widehat F_{a,n}(s)-F_a(s)=\PP_nD_{a,\widehat{\theta}_a}(s)-P_*D_{a,\theta_{a,\infty}}(s)$, the preceding displays yield
\begin{align*}
&\widehat G^{\rm locmod}_{a,n}(t)\\
&=(\PP_n-P_*)\Big[\frac{r_n}{p_a(s_0)}\int_{s_0}^{s_n(t)}(D_{a,\widehat\theta_a}(v)-D_{a,\widehat\theta_a}(s_0))dv\Big]+\frac{\varphi_a^{(k)}(s_0)}{(k+2)!}t^{k+2}+o_p(1)\\
&\indent+P_*\Big[\frac{r_n}{p_a(s_0)}\int_{s_0}^{s_n(t)}\Big((D_{a,\widehat\theta_a}-D_{a,{\theta_{a,\infty}}})(v)-(D_{a,\hat\theta_a}-D_{a,{\theta_{a,\infty}}})(s_0)\Big)dv\Big].
\end{align*}
In addition, we verified
\begin{align*}
 P_*\Big[\frac{r_n}{p_a(s_0)}\int_{s_0}^{s_n(t)}\Big((D_{a,\hat\theta_a}-D_{a,{\theta_{a,\infty}}})(v)-(D_{a,\hat\theta_a}-D_{a,{\theta_{a,\infty}}})(s_0)\Big)dv\Big]=o_p(1),  
\end{align*}
in the proof of Lemma \ref{successiveknots} (see \eqref{reference.Remainder.1}, \eqref{reference.Remainder.2}). Thus, we finally obtain
\begin{align*}
&\widehat G^{\rm locmod}_{a,n}(t)\\
&=(\PP_n-P_*)\Big[\frac{r_n}{p_a(s_0)}\int_{s_0}^{s_n(t)}(D_{a,\widehat\theta_a}(v)-D_{a,\widehat\theta_a}(s_0))dv\Big]+\frac{\varphi_a^{(k)}(s_0)}{(k+2)!}t^{k+2}+o_p(1)\\
&\indent+P_*\Big[\frac{r_n}{p_a(s_0)}\int_{s_0}^{s_n(t)}\Big((D_{a,\widehat\theta_a}-D_{a,{\theta_{a,\infty}}})(v)-(D_{a,\hat\theta_a}-D_{a,{\theta_{a,\infty}}})(s_0)\Big)dv\Big]
\end{align*}
which is equal to 
\begin{equation*}
\begin{split}
&(\PP_n-P_*)\Big[\frac{r_n}{p_a(s_0)}\int_{s_0}^{s_n(t)}(D_{a,{\theta_{a,\infty}}}(v)-D_{a,{\theta_{a,\infty}}}(s_0))dv\Big]+\frac{\varphi_a^{(k)}(s_0)}{(k+2)!}t^{k+2}+o_p(1)\\
&\indent+(\PP_n-P_*)\Big[\frac{r_n}{p_a(s_0)}\int_{s_0}^{s_n(t)}(D_{a,\hat\theta_a}(v)-D_{a,\hat\theta_a}(s_0))-(D_{a,{\theta_{a,\infty}}}(v)-D_{a,{\theta_{a,\infty}}}(s_0))dv\Big] \\
&=(\PP_n-P_*)\Big[\frac{r_n}{p_a(s_0)}\int_{s_0}^{s_n(t)}(D_{a,{\theta_{a,\infty}}}(v)-D_{a,{\theta_{a,\infty}}}(s_0))dv\Big]+\frac{\varphi_a^{(k)}(s_0)}{(k+2)!}t^{k+2}+o_p(1),
\end{split}
\end{equation*}
where we used Lemma \ref{lem:empiricalprocDID} in the last equality.

Now, we study the convergence of the term
\begin{align*}
  (\PP_n-P_*)\Big[\frac{r_n}{p_a(s_0)}\int_{s_0}^{s_n(t)}(D_{a,{\theta_{a,\infty}}}(v)-D_{a,{\theta_{a,\infty}}}(s_0))dv\Big],  
\end{align*}
by Theorem 2.11.24 in \citet{vdvandW} (see Proposition \ref{prop:vdv.2.11.24}). To line up with the assumption that Proposition \ref{prop:vdv.2.11.24} requires, assuming $t\ge0$ without loss of generality, we first define the function class 
\begin{align*}
    \mc F_{n}:=\Big\{\sqrt{v_n^{-1}}\Big(D_{a,{\theta_{a,\infty}}}(s_0+v_nt)-D_{a,{\theta_{a,\infty}}}(s_0)\Big):t\in[-T,T]\Big\}.
\end{align*}
Since we have 
\begin{align*}
 &\sqrt{v_n^{-1}}\Big(D_{a,{\theta_{a,\infty}}}(s_0+v_nt)-D_{a,\theta_{a,\infty}}(s_0)\Big)\\
 &=\frac{I(A=a)}{\pi_{a,\infty}(\mathbf{X})}\left(I(Y\le s_0+v_nt)-I(Y\le s_0)\right)\\
 &\indent +\Big(1-\frac{I(A=a)}{\pi_{a,\infty}(\mathbf{X})}\Big)(\phi_{a,\infty}(s_0+v_nt|\mathbf{X})-\phi_{a,\infty}(s_0|\mathbf{X})),
\end{align*}
the class $\mc F_{n}$ has an envelope which is given by
\begin{align*}
F_{\rm env,n}:&=\sqrt{v_n^{-1}}\Big(KI(A=a)I(s_0-v_nT<Y\le s_0+v_nT)\\
&\indent +(K+1)(\phi_{a,\infty}(s_0+v_nT|\mathbf{X})-\phi_{a,\infty}(s_0-v_nT|\mathbf{X}))\Big),
\end{align*}
and so 
\begin{align*}
    &\EE F_{\rm env,n}^2=O(1),\\
    &\EE F_{\rm env,n}^2I(F_{\rm env}\ge \eta\sqrt{n})\rightarrow 0,
\end{align*}
for arbitrary $\eta>0$. In addition, we have
\begin{align*}
 &\Big(D_{a,{\theta_{a,\infty}}}(s_0+v_nt_1)-D_{a,{\theta_{a,\infty}}}(s_0)\Big)-\Big(D_{a,{\theta_{a,\infty}}}(s_0+v_nt_2)-D_{a,{\theta_{a,\infty}}}(s_0)\Big)\\
 &=D_{a,{\theta_{a,\infty}}}(s_0+v_nt_1)-D_{a,{\theta_{a,\infty}}}(s_0+v_nt_2)\\
 &=\Big(1-\frac{I(A=a)}{\pi_{a,\infty}(\mathbf{X})}\Big)(\phi_{a,\infty}(u_1|\mathbf{X})-\phi_{a,\infty}(u_2|\mathbf{X}))\\
 &\indent+\frac{I(A=a)}{\pi_{a,\infty}(\mathbf{X})}I(u_2<Y\le u_1),
\end{align*}
where $u_i=s_0+v_nt_i$ for $i\in\{1,2\}$. Then, for some $C_2>0$,
\begin{align*}
    &P_*\Big[\Big(D_{a,{\theta_{a,\infty}}}(s_0+v_nt_1)-D_{a,{\theta_{a,\infty}}}(s_0)\Big)-\Big(D_{a,{\theta_{a,\infty}}}(s_0+v_nt_2)-D_{a,{\theta_{a,\infty}}}(s_0)\Big)\Big]^2\\
    &\le 2P_* \Big[\frac{I(A=a)}{\pi^2_{a,\infty}(\mathbf{X})}I(u_2<Y\le u_1)\Big]\\
    &\indent+2P_* \Big[\Big(1-\frac{I(A=a)}{\pi_{a,\infty}(\mathbf{X})}\Big)^2(\phi_{a,\infty}(u_1|\mathbf{X})-\phi_{a,\infty}(u_2|\mathbf{X}))^2\Big]\\
    &\le C_2K^2v_n(t_1-t_2)+(K+1)^2(t_2-t_1)^{2}v_n^{2},
\end{align*}
where a similar reasoning to the proof of Theorem \ref{prop:Consistency} is applied with Assumption \ref{assm:phiregularity}. Thus, we have
\begin{align*}
    \sup_{|t_2-t_1|<\delta_n}v_n^{-1} &P_*\Big[\Big(D_{a,{\theta_{a,\infty}}}(s_0+v_nt_1)-D_{a,{\theta_{a,\infty}}}(s_0)\Big)-\Big(D_{a,{\theta_{a,\infty}}}(s_0+v_nt_2)-D_{a,{\theta_{a,\infty}}}(s_0)\Big)\Big]^2\\
    &\rightarrow 0,
\end{align*}
as $\delta_n\rightarrow 0$. Furthermore, following the  steps in the proof of Lemma \ref{successiveknots}, we have
\begin{align*}
    \sup_Q\int_0^{\delta_n}\sqrt{{\rm log}N(\epsilon\|F_{\rm env,n}\|_{Q,2},\mc F_n,L_2(Q))d\epsilon}\rightarrow 0,
\end{align*}
for every $\delta_n\rightarrow 0$. The last step to apply Proposition \ref{prop:vdv.2.11.24} is studying the limiting covariance structure. Indeed, for $t_1,t_2>s_0$ defined with $u_1,u_2$ as above,
\begin{align}
    &v_n^{-1}P_*(D_{a,{\theta_{a,\infty}}}(u_1)-D_{a,{\theta_{a,\infty}}}(s_0))(D_{a,{\theta_{a,\infty}}}(u_2)-D_{a,{\theta_{a,\infty}}}(s_0))\nonumber\\
    &=v_n^{-1}P_*\Big[\frac{I(A=a)}{\pi^2_{a,\infty}(\mathbf{X})}I(s_0<Y\le u_m)\label{cov:1}\\
    &\indent+\frac{I(A=a)}{\pi_{a,\infty}(\mathbf{X})}\Big(1-\frac{I(A=a)}{\pi_{a,\infty}(\mathbf{X})}\Big)I(s_0<Y\le u_1)(\phi_{a,\infty}(u_2|\mathbf{X})-\phi_{a,\infty}(s_0|\mathbf{X}))\label{cov:2}\\
    &\indent+\frac{I(A=a)}{\pi_{a,\infty}(\mathbf{X})}\Big(1-\frac{I(A=a)}{\pi_{a,\infty}(\mathbf{X})}\Big)I(s_0<Y\le u_2)(\phi_{a,\infty}(u_1|\mathbf{X})-\phi_{a,\infty}(s_0|\mathbf{X}))\label{cov:3}\\
    &\indent+\Big(1-\frac{I(A=a)}{\pi_{a,\infty}(\mathbf{X})}\Big)^2(\phi_{a,\infty}(u_1|\mathbf{X})-\phi_{a,\infty}(s_0|\mathbf{X}))(\phi_{a,\infty}(u_2|\mathbf{X})-\phi_{a,\infty}(s_0|\mathbf{X}))\Big],\label{cov:4}
\end{align}
where $u_m=\min \{u_1,u_2\}$. Let $t_m=\min \{t_1,t_2\}$. For the first term \eqref{cov:1}, we have
\begin{align*}
v_n^{-1}\EE\Big[\frac{I(A=a)}{\pi_{a,\infty}^2(\mathbf{X})}&I(s_0<Y\le u_m)\Big]\\
&=v_n^{-1}\EE\Big[\EE \Big(\frac{I(A=a)}{\pi^2_{a,\infty}(\mathbf{X})}I(s_0<Y\le u_m)\Big|\mathbf{X},A\Big)\Big] \\
&= v_n^{-1}\EE\Big[\frac{1}{\pi^2_{a,\infty}(\mathbf{X})}\EE \Big(I(A=a)I(s_0<Y^a\le s_0+v_nt_m)\Big|\mathbf{X},A\Big)\Big]\\
&\rightarrow t_m\EE \Big[\frac{\pi_{a,*}(\mathbf{X})}{\pi^2_{a,\infty}(\mathbf{X})}\eta_{a,*}(s_0|\mathbf{X})\Big],
\end{align*}
due to the Assumptions \ref{assm:consistency}, \ref{assm:noumconf}, \ref{assm:boundedpropscore}, and the Lebesgue Dominated convergence Theorem. Next, the second term \eqref{cov:2} converges to $0$, since
\begin{align*}
  v_n^{-1}\EE&\Big|\frac{I(A=a)}{\pi_{a,\infty}(\mathbf{X})}\Big(1-\frac{I(A=a)}{\pi_{a,\infty}(\mathbf{X})}\Big)I(s_0<Y\le u_1)(\phi_{a,\infty}(u_2|\mathbf{X})-\phi_{a,\infty}(s_0|\mathbf{X}))\Big|\\
  &\le v_n^{-1}K(K+1)\EE\Big|I(s_0<Y^a\le s_0+v_nt_1)(\phi_{a,\infty}(u_2|\mathbf{X})-\phi_{a,\infty}(s_0|\mathbf{X}))\Big|\\
  &\le v_n^{-1}K(K+1)\sqrt{\EE I(s_0<Y^a\le s_0+v_nt_1)} \sqrt{\EE R^2(\mathbf{X})}|v_nt_2|\\
  &\le K (K+1) |t_2| \sqrt{\EE R^2(\mathbf{X}) } \sqrt{\EE I(s_0<Y^a\le s_0+v_nt_1)} \rightarrow 0,
\end{align*}
by a similar reasoning as above with the Cauchy-Schwarz inequality and Assumption \ref{assm:phiregularity}. Identical procedure yields that the third term \eqref{cov:3} is also vanishing. Finally, again by Assumption \ref{assm:phiregularity}, we have
\begin{align*}
 v_n^{-1}\EE &\Big[\Big(1-\frac{I(A=a)}{\pi_{a,\infty}(\mathbf{X})}\Big)^2(\phi_{a,\infty}(u_1|\mathbf{X})-\phi_{a,\infty}(s_0|\mathbf{X}))(\phi_{a,\infty}(u_2|\mathbf{X})-\phi_{a,\infty}(s_0|\mathbf{X}))\Big] \\ 
 &\le  v_n|t_1t_2|(1+K)^2\EE R^2(\mathbf{X})\rightarrow 0,
\end{align*}
for the fourth term \eqref{cov:4}. On the other hand, when $t_1>s_0>t_2$, the first term \eqref{cov:1} is now
\begin{align*}
-v_n^{-1}P_*\left[\frac{I(A=a)}{\pi_{a,\infty}^2(\mathbf{X})}I(s_0<Y\le u_1)I(u_2< Y\le s_0)\right]=0.
\end{align*}
Obviously, the other terms \eqref{cov:2}, \eqref{cov:3}, \eqref{cov:4} are still vanishing.
Similarly, since
\begin{align}\label{eq:single.limit.mean.zero}
    v_n^{-1/2}P_*(D_{a,{\theta_{a,\infty}}}(s_0+v_nt)-D_{a,{\theta_{a,\infty}}}(s_0))\rightarrow 0
\end{align}
holds for each $t\in[-T,T]$,
it is straightforward that
\begin{align*}
  v_n^{-1}P_*(D_{a,{\theta_{a,\infty}}}(u_2)-D_{a,{\theta_{a,\infty}}}(s_0))P_*(D_{a,{\theta_{a,\infty}}}(u_1)-D_{a,{\theta_{a,\infty}}}(s_0)) \rightarrow 0.
\end{align*}
Thus, for each $T$, by Proposition \ref{prop:vdv.2.11.24},
\begin{align*}
(\PP_n-P_*)\Big[\frac{r_n}{p_a(s_0)}\int_{s_0}^{s_n(t)}(D_{a,{\theta_{a,\infty}}}(v)-D_{a,{\theta_{a,\infty}}}(s_0))dv\Big]  \rightarrow_d \alpha_a\int_0^t W(v)dv,
\end{align*}
in $C[-T,T]$ with $\alpha_a=\frac{1}{p_a(s_0)}\sqrt{\EE \Big[\frac{\pi_{a,*}(\mathbf{X})}{\pi^2_{a,\infty}(\mathbf{X})}\eta_{a,*}(s_0|\mathbf{X})\Big]}$. 

The part \ref{CIlem:b} can be proved identically to the proof for Lemma 4.6-(ii) in \cite{BRW09}. Furthermore, part \ref{CIlem:c} can easily be check with a slight modification of the proof Lemma 4.6-(iii) in \citet{BRW09}. Their $\widehat F_n, \FF_n, x_0$ are identified with our $\widehat G_{a,n}, \widehat F^c_{a,n}, s_0$, respectively. The terms $\widehat A_{n1},\widehat A_{n2}$ in \cite{BRW09} are technically the same, but our $\widehat A_{n3}$ is 
\begin{align*}
    \widehat A_{n3}:=r_ns_n\Big|\int_{\tau}^{s_0} d(\widehat F^c_{a,n}-F_a)\Big|.
\end{align*}
But, indeed, the perturbation of $\Delta(x)=I_{[\tau,s_0]}(x)$ also yields the same result as theirs with a similar reasoning to Lemma \ref{successiveknots} with \ref{emptermlemma}. 

Finally, the proof of the last part \ref{CIlem:d} follows identical steps to that of Theorem 6.2  in \citet{Groeneboom:2001jo} regarding the least squares estimator (or Lemma 8.24--Lemma 8.27 of \citet{Doss.Wellner.2019}).
\end{proof}

We state and prove the joint convergence result of the vectors of processes defined in \ref{CIlem:d} of Lemma \ref{lem:finallemmaforCI} for $a=0,1$ that serves as an intermediate result to show that asymptotic independence between $(\widehat{p}_{1,n}(s_0)-p_1(s_0))$ and $(\widehat{p}_{0,n}(s_0)-p_0(s_0))$.
\begin{lemma}\label{lem:finallemmaforCI.joint}
Let $T>0$ and
\begin{align}
&V_{a,n}:=(\widehat H^{\rm locmod}_{a,n},(\widehat H^{\rm locmod}_{a,n})^{(1)},(\widehat H^{\rm locmod}_{a,n})^{(2)},  \widehat G^{\rm locmod}_{a,n}, (\widehat H^{\rm locmod}_{a,n})^{(3)}, ( \widehat G^{\rm locmod}_{a,n})^{(1)}),\label{def:V.an.est.proc}\\
&V_a:=(H_{\alpha_a,k,\sigma_a}, (H_{\alpha_a,k,\sigma_a})^{(1)}, H_{\alpha_a,k,\sigma_a}^{(2)}, Y_{\alpha_a,k,\sigma_a}, (H_{\alpha_a,k,\sigma_a})^{(3)}, (Y_{\alpha_a,k,\sigma_a})^{(1)})\label{def:V.a.limit.proc},
\end{align}
where $H_{\alpha_a,k,\sigma_a}$ is the unique process defined in \eqref{def:scaled.HandY.processes}, for $a\in\{0,1\}$.
Then the vector of processes $(V_{1,n},V_{0,n})$ converges weakly in $(C[-T,T])^4\times (D[-T,T])^2\times (C[-T,T])^4\times (D[-T,T])^2$ to the process $(V_1,V_0)$, where $V_1$ and $V_0$ are independent vectors of processes.
\end{lemma}
\begin{proof}
Stacking the vectors of processes $V_{1,n}$ and $V_{0,n}$ and following the same derivation in Lemma \ref{lem:finallemmaforCI}, one can easily check the joint convergence holds.
Now it suffices to check the independence of two limit vectors of processes, $V_1$ and $V_0$.
Recall that $r_n=n^{(k+2)/(2k+1)}$.
In the proof of Lemma \ref{lem:finallemmaforCI}, we verified
\begin{align*}
&\widehat G^{\rm locmod}_{a,n}(t)\\
&=(\PP_n-P_*)\Big[\frac{r_n}{p_a(s_0)}\int_{s_0}^{s_n(t)}(D_{a,{\theta_{a,\infty}}}(v)-D_{a,{\theta_{a,\infty}}}(s_0))dv\Big]+\frac{\varphi_a^{(k)}(s_0)}{(k+2)!}t^{k+2}+o_p(1)\\
&\rightarrow_d \alpha_a\int_0^t W_a(v)dv+\frac{\varphi_a^{(k)}(s_0)}{(k+2)!}t^{k+2},
\end{align*}
for $a=0,\,1$.
To prove the asymptotic independence of $\widehat G^{\rm locmod}_{1,n}$ and $\widehat G^{\rm locmod}_{0,n}$, it suffices to show that the limiting covariance of $\widehat G^{\rm locmod}_{1,n}(t_1)$ and $\widehat G^{\rm locmod}_{0,n}(t_0)$ is zero for any $(t_1,\,t_0)\in[-T,T]^2$.
Recall that $v_n=n^{-1/(2k+1)}$. We now check, for any $(t_1,\,t_0)\in[-T,T]^2$,
\begin{align}
&v_n^{-1}P_*(D_{1,{\theta_{1,\infty}}}(u_1)-D_{1,{\theta_{1,\infty}}}(s_0))(D_{0,{\theta_{0,\infty}}}(u_0)-D_{0,{\theta_{0,\infty}}}(s_0))\rightarrow 0,\label{Joint.WTS1} \\
&v_n^{-1/2}P_*(D_{1,{\theta_{1,\infty}}}(u_1)-D_{1,{\theta_{1,\infty}}}(s_0))\rightarrow 0,\label{Joint.WTS2}\\
&v_n^{-1/2}P_*(D_{0,{\theta_{1,\infty}}}(u_0)-D_{0,{\theta_{0,\infty}}}(s_0))\rightarrow 0,\label{Joint.WTS3}
\end{align}
where $u_i=s_0+v_nt_i$, for $i=0,1$.
We already verified \eqref{Joint.WTS2} and \eqref{Joint.WTS3} hold in the proof of Lemma \ref{lem:finallemmaforCI} (see \eqref{eq:single.limit.mean.zero}).
Now we prove \eqref{Joint.WTS1}.
Without loss of generality we assume $t_1,\,t_0\geq 0$. We have
\begin{equation*}
\begin{split}
&v_n^{-1}P_*(D_{1,{\theta_{1,\infty}}}(u_1)-D_{1,{\theta_{1,\infty}}}(s_0))(D_{0,{\theta_{0,\infty}}}(u_0)-D_{0,{\theta_{0,\infty}}}(s_0))\nonumber\\
&=v_n^{-1}P_*\Big[\frac{I(A=1)}{\pi_{1,\infty}(\mathbf{X})}\Big(1-\frac{I(A=0)}{\pi_{0,\infty}(\mathbf{X})}\Big)I(s_0<Y\le u_1)(\phi_{0,\infty}(u_0|\mathbf{X})-\phi_{0,\infty}(s_0|\mathbf{X}))\\
&\indent+\frac{I(A=0)}{\pi_{0,\infty}(\mathbf{X})}\Big(1-\frac{I(A=1)}{\pi_{1,\infty}(\mathbf{X})}\Big)I(s_0<Y\le u_0)(\phi_{1,\infty}(u_1|\mathbf{X})-\phi_{a,\infty}(s_0|\mathbf{X}))\\
&\indent+\Big(1-\frac{I(A=1)}{\pi_{1,\infty}(\mathbf{X})}\Big)\Big(1-\frac{I(A=0)}{\pi_{0,\infty}(\mathbf{X})}\Big)(\phi_{1,\infty}(u_1|\mathbf{X})-\phi_{1,\infty}(s_0|\mathbf{X}))(\phi_{0,\infty}(u_0|\mathbf{X})-\phi_{0,\infty}(s_0|\mathbf{X}))\Big],
\end{split}
\end{equation*}
since $I(A=1)I(A=0)=0$.
The first term is $o(1)$, since 
\begin{align*}
v_n^{-1}\EE&\Big|\frac{I(A=1)}{\pi_{1,\infty}(\mathbf{X})}\Big(1-\frac{I(A=0)}{\pi_{0,\infty}(\mathbf{X})}\Big)I(s_0<Y\le u_1)(\phi_{0,\infty}(u_0|\mathbf{X})-\phi_{0,\infty}(s_0|\mathbf{X}))\Big|\\
&\le v_n^{-1}K(K+1)\sqrt{\EE I(s_0<Y^1\le s_0+v_nt_1)} \sqrt{\EE R^2(\mathbf{X})}|v_nt_0|\\
&\le K (K+1) |t_0| \sqrt{\EE R^2(\mathbf{X}) } \sqrt{\EE I(s_0<Y^1\le s_0+v_nt_1)} \rightarrow 0,
\end{align*} 
by Cauchy-Schwarz inequality and Assumption \ref{assm:phiregularity}.
Similar reasoning can be applied to show that the second term converges to $0$.
Lastly, we study the third term. Indeed, we have
\begin{align*}
v_n^{-1}\EE &\Big[\Big(1-\frac{I(A=1)}{\pi_{1,\infty}(\mathbf{X})}\Big)\Big(1-\frac{I(A=0)}{\pi_{0,\infty}(\mathbf{X})}\Big)(\phi_{1,\infty}(u_1|\mathbf{X})-\phi_{1,\infty}(s_0|\mathbf{X}))(\phi_{0,\infty}(u_0|\mathbf{X})-\phi_{0,\infty}(s_0|\mathbf{X}))\Big] \\ 
&\le  v_n|t_1t_0|(1+K)^2\EE R^2(\mathbf{X})\rightarrow 0.
\end{align*}
This verifies the asymptotic independence of $\widehat G_{1,n}^{\rm locmod}$ and $\widehat G_{0,n}^{\rm locmod}$.
In addition, their independent limit processes $Y_{\sigma_1,k,\sigma_1}$ and $Y_{\sigma_0,k,\sigma_0}$ uniquely characterize $H_{\sigma_1,k,\sigma_1}$ and $H_{\sigma_0,k,\sigma_0}$, respectively.
Thus, $V_1$ and $V_0$ are independent, and it concludes the statement of the lemma.
\end{proof}

We state and prove the following lemma to confirm that the isotonic correction of the one-step estimator has  a negligible impact on the limit distribution of the log-concave MLE $\widehat p_{a,n}$.  Lemma~\ref{lem:isotonicDID} is used in the proof of Lemma~\ref{lem:integratedDiD}, which in turn was used in the proof of Lemma~\ref{lem:finallemmaforCI} above.
\begin{lemma}\label{lem:isotonicDID}
Let $\mc E_{s_0,n}:=[s_0-Tv_n,s_0+Tv_n]\subset \mc I_{s_0,\omega}$ for sufficiently large $n$, arbitrary $T>0$, and $v_n=n^{-1/(2k+1)}$. Then under the same assumptions as in Lemma \ref{successiveknots}, we have
\begin{align*}
    \sup_{v_1,v_2\in\mc E_{s_0,n}}\Big|(\widehat F^{c}_{a,n}(v_2)-\widehat F_{a,n}(v_2))-(\widehat F^{c}_{a,n}(v_1)-\widehat F_{a,n}(v_1))\Big|=o_p(n^{-(k+1)/(2k+1)}),
\end{align*}
for $a \in \{0,1\}$.
\end{lemma}
\begin{proof}
Let $U_m:=[m,\infty)$ and $~L_m:=(-\infty,m]$ for any $m\in \mc S_n$. Then, similar to the proof of Lemma 3 in \citet{WvdlC2020}, we can define $L_m^*\in \argmin_{L\in \mc L_m} \Bar{\widehat F}_{a,n}(U_m \cap L),~U_m^*\in \argmax_{U\in \mc U_m} \Bar{\widehat F}_{a,n}(L_m \cap U)$, where $\mc L_m:=\{(-\infty,v]:v\ge m\},~\mc U_m:=\{[v,\infty):v\le m\}$, and $\Bar{\widehat F}_{a,n}(A):=|A|^{-1}\sum_{v\in A}\widehat F_{a,n}(v)$ for $A\subseteq \mc S_n$, and $|A|$ is the cardinality of $A$. Analogous to their proof, we have, for $t,s\in \mc S_n\cap I_{s_0,\omega}$,
\begin{align*}
 \Bar{\widehat F}_{a,n}(L_t^* \cap U_t) - \Bar{\widehat F}_{a,n}(L_s \cap U_s^*) \le  \widehat F^{c}_{a,n}(t)-\widehat F^{c}_{a,n}(s) \le \Bar{\widehat F}_{a,n}(L_t \cap U_t^*) - \Bar{\widehat F}_{a,n}(L_s^* \cap U_s),
\end{align*}
where the lengths of the intervals $L_v^* \cap U_v,~L_v \cap U_v^*$ for $v=s,t$ are bounded above by $\kappa_n$ (which is defined in the proof of Lemma \ref{lem:monotone.correction.ext} and is the supremum length of an interval on which $\widehat F_{a,n}$ is decreasing rather than increasing). Hence, letting $t_n^*:={\rm argmax}_{x\in L_t \cap U_t^*\cap \mc S_n} \widehat F_{a,n}(x)$ and $s_n^*:={\rm argmin}_{x\in L_s^* \cap U_s\cap \mc S_n} \widehat F_{a,n}(x)$ (note that the [finite] endpoints defining $U^*_m$ and $L^*_m$ are elements of $\mc S_n$ by the definition of $\Bar{\widehat F}_{a,n}$, so intersecting with $\mathcal{S}_n$ in the definitions of $t^*_n$ and $s^*_n$ does not change anything), we have
\begin{align*}
 \Big(\widehat F^{c}_{a,n}(t)-\widehat F^{c}_{a,n}(s)\Big)&- \Big(\widehat F_{a,n}(t)-\widehat F_{a,n}(s)\Big)\\
 & \le 
 \Bar{\widehat F}_{a,n}(L_t \cap U_t^*) - \Bar{\widehat F}_{a,n}(L_s^* \cap U_s)
  - \Big(\widehat F_{a,n}(t)-\widehat F_{a,n}(s)\Big)
\end{align*}
which is bounded above by 
\begin{align*}
 & \Big(\widehat F_{a,n}(t_n^*)-\widehat F_{a,n}(s_n^*)\Big)- \Big(\widehat F_{a,n}(t)-\widehat F_{a,n}(s)\Big)\\
 &=\Big\{\Big[\Big(\widehat F_{a,n}(t_n^*)-\widehat F_{a,n}(s_n^*)\Big)- \Big(F_a(t_n^*)-F_a(s_n^*)\Big)\Big]\\
 &\indent - \Big[\Big(\widehat F_{a,n}(t)-\widehat F_{a,n}(s)\Big)- \Big(F_a(t)-F_a(s)\Big)\Big]\Big\}\\
 &\indent +\Big(F_a(t_n^*)-F_a(s_n^*)\Big)- \Big(F_a(t)-F_a(s)\Big).
\end{align*}
Similarly, one can derive the analogous lower bound for the term
\begin{align*}
\Big(\Bar{\widehat F}_{a,n}(L_t^* \cap U_t) - \Bar{\widehat F}_{a,n}(L_s \cap U_s^*)\Big)-\Big(\widehat F_{a,n}(t)&-\widehat F_{a,n}(s)\Big).
\end{align*}
This further yields, when we define $\mc W:=\mc S_n \cap \mc E_{s_0,n}$ and $\mc I_v^{+}:=[v,v+\kappa_n]$, $\mc I_v^{-}:=[v-\kappa_n,v]$, $\mc I_v:=[v-\kappa_n,v+\kappa_n]$ for any $v\in I_{s_0,\omega}$, 
\begin{align}
&\sup_{s,t\in \mc W} \Big|\Big(\widehat F^{c}_{a,n}(t)-\widehat F_{a,n}(t)\Big)-\Big(\widehat F^{c}_{a,n}(s)-\widehat F_{a,n}(s)\Big)\Big|\nonumber\\
&\qquad\le 2\sup_{s,t\in\mc E_{s_0,n}{~\rm and~}t'\in\mc I_t,s'\in\mc I_s} \Big\{\Big|\Big[\Big(\widehat F_{a,n}(t')-F_{a}(t')\Big)-\Big(\widehat F_{a,n}(t)-F_{a}(t)\Big)\Big]\Big|\label{lemmaDiD-term0}\\
&\qquad\indent~~~~~~~~~~~~~~~~~~~~~~~~~~~~ -\Big[\Big(\widehat F_{a,n}(s')-F_{a}(s')\Big)-\Big(\widehat F_{a,n}(s)-F_{a}(s)\Big)\Big]\Big|\Big\}\label{lemmaDiD-term1}\\
&\qquad\indent + \sup_{s,t\in\mc E_{s_0,n}{~\rm and~}t'\in\mc I_t^{-},s'\in\mc I_s^{+}} \Big|\Big(F_{a}(t')-F_{a}(t)\Big)-\Big(F_{a}(s')- F_{a}(s)\Big)\Big|\label{lemmaDiD-term2}\\
&\qquad\indent + \sup_{s,t\in\mc E_{s_0,n}{~\rm and~}t'\in\mc I_t^{+},s'\in\mc I_s^{-}} \Big|\Big(F_{a}(t')-F_{a}(t)\Big)-\Big(F_{a}(s')- F_{a}(s)\Big)\Big|\label{lemmaDiD-term3}.
\end{align}
Furthermore, since $\varphi_a$ is at least twice continuously differentiable in $\mc I_{s_0,\omega}$, this implies that $F_a$ is increasing and at least three times continuously differentiable in $\mc I_{s_0,\omega}$, this yields the terms \eqref{lemmaDiD-term2}, \eqref{lemmaDiD-term3} above are of $O_p(n^{-(k+2)/(2k+1)})$, since $|t'-t|,|s'-s|\le \kappa_n=O_p(n^{-(k+1)/(2k+1)})$ and $\sup_{s,t\in \mc E_{s_0,n}}|t-s|=O_p(n^{-1/(2k+1)})$.

To analyze the term \eqref{lemmaDiD-term0}--\eqref{lemmaDiD-term1}, we exploit the following decomposition again:
\begin{align*}
    \widehat F_{a,n}(x)-F_a(x)&=\PP_n D_{a,\widehat\theta_a}(x)-P_* D_{a,{\theta_{a,\infty}}}(x)\\
    &=(\PP_n-P_*)D_{a,\widehat\theta_a}(x) + P_*(D_{a,\widehat\theta_a}(x)-D_{a,{\theta_{a,\infty}}}(x)),
\end{align*}
for any $x\in\RR$. And, we have
\begin{align}
&2\sup_{s,t\in\mc E_{s_0,n}{~\rm and~}t'\in\mc I_t,s'\in\mc I_s} \Big\{\Big|\Big[\Big(\widehat F_{a,n}(t')-F_{a}(t')\Big)-\Big(\widehat F_{a,n}(t)-F_{a}(t)\Big)\Big]\Big|\nonumber\\
&\indent~~~~~~~~~~~~~~~~~~~~~~~~~~~~ -\Big[\Big(\widehat F_{a,n}(s')-F_{a}(s')\Big)-\Big(\widehat F_{a,n}(s)-F_{a}(s)\Big)\Big]\Big|\Big\}\nonumber\\
&\qquad\le 4\sup_{t\in \mc E_{s_0,n}{~\rm and~}t'\in\mc I_t} \left|\left[\left(\widehat F_{a,n}(t')-F_{a}(t')\right)-\left(\widehat F_{a,n}(t)-F_{a}(t)\right)\right]\right|\label{lemmaDiD-term-last}.
\end{align}
Then for the second term (applied to the term \eqref{lemmaDiD-term-last}), we use the decomposition used in Lemma \ref{successiveknots} which is given by
\begin{align}
  P_*&\Big\{\Big[(D_{a,\widehat\theta_a}-D_{a,{\theta_{a,\infty}}})(t')-(D_{a,\widehat\theta_a}-D_{a,{\theta_{a,\infty}}})(t)\Big]\Big\}\nonumber\\
  &=P_*\Big\{\frac{\widehat\pi_a(\mathbf{X})-\pi_{\infty}(\mathbf{X})}{\widehat\pi_a(\mathbf{X})}\Big[\Big((\widehat\phi_a-\phi_{a,\infty})(t'|\mathbf{X})-(\widehat\phi_a-\phi_{a,\infty})(t|\mathbf{X})\Big)\Big]\nonumber\\
  &\indent+\frac{\pi_{a,*}(\mathbf{X})(\widehat\pi_a(\mathbf{X})-\pi_{a,\infty}(\mathbf{X}))}{\widehat\pi_a(\mathbf{X})\pi_{a,\infty}(\mathbf{X})}\Big[\Big((\phi_{a,\infty}-\phi_{a,*})(t'|\mathbf{X})-(\phi_{a,\infty}-\phi_{a,*})(t|\mathbf{X})\Big)\Big]\nonumber\\
  &\indent+\frac{\pi_{a,\infty}(\mathbf{X})-\pi_{a,*}(\mathbf{X})}{\widehat\pi_a(\mathbf{X})}\Big[\Big((\widehat\phi_a-\phi_{a,\infty})(t'|\mathbf{X})-(\widehat\phi_a-\phi_{a,\infty})(t|\mathbf{X})\Big)\Big]\Big\}\nonumber.
\end{align}
Thus, due to $\kappa_n=O_p(n^{-(k+1)/(2k+1)})$ and the Assumption \ref{assm:L2conditionsforCI}, we have
\begin{align}
 \sup_{t\in \mc E_{s_0,n}{~\rm and~}t'\in\mc I_t} \left|P_*\Big\{\Big[(D_{a,\widehat\theta_a}-D_{a,{\theta_{a,\infty}}})(t')-(D_{a,\widehat\theta_a}-D_{a,{\theta_{a,\infty}}})(t)\Big]\Big\}\right|=o_p(n^{-(k+1)/(2k+1)}).\label{result.in.DiD.remainder}
\end{align}
On the other hand, Lemma \ref{emptermlemma} with $l=k-1/2,t=1$ instead of $l=k$ yields, from the same reasoning for \eqref{empterm:identity.part} in the proof of Lemma \ref{successiveknots},
\begin{align}\label{empterm2:identity.part}
    \sup_{[r_1,r_2]\subset I_{s_0,\omega},r_1\le r_2\le r_1+\kappa_n} \Big|(\PP_n-P_*)\Big[\frac{I(A=a)}{\widehat\pi_a(\mathbf{X})}I(r_1<Y\le r_2)\Big|\le \kappa_n^{k+1/2}+B_3,
\end{align}
where $B_3$ is $O_p(n^{-(k+1/2)/(2k)})$ which is independent of $r_1,r_2$. Considering that $\kappa_n=O_p(n^{-(k+1)/(2k+1)})$, the right hand side of \eqref{empterm2:identity.part} is $o_p(n^{-(k+1)/(2k+1)})$. With \eqref{empterm2:identity.part} and another direct application of \eqref{empterm:phi.part} on $r_1,r_2$ which satisfy $[r_1,r_2]\subset I_{s_0,\omega},r_1\le r_2\le r_1+\kappa_n$ implies
\begin{align*}
    \sup_{[r_1,r_2]\subset I_{s_0,\omega},r_1\le r_2\le r_1+\kappa_n}&\Big|(\PP_n-P_*)\Big[\frac{I(A=a)}{\widehat\pi_a(\mathbf{X})}I(r_1<Y\le r_2)\\
    &+(\widehat\phi_a(r_2|\mathbf{X})-\widehat\phi_a(r_1|\mathbf{X}))\Big(1-\frac{I(A=a)}{\widehat\pi_a(\mathbf{X})}\Big)\Big]\Big|=o_p(n^{-(k+1)/(2k+1)}).
\end{align*}
This further yields,
\begin{align}
\sup_{t\in\mc E_{s_0,n}{~\rm and~}t'\in\mc I_t}\Big| (\PP_n-P_*)\Big\{D_{a,\widehat\theta_a}(t')-D_{a,\widehat\theta_a}(t)\Big\}\Big|=o_p(n^{-(k+1)/(2k+1)})\label{result.in.DiD.empterm}.
\end{align}
Combining \eqref{result.in.DiD.remainder} and \eqref{result.in.DiD.empterm}, we obtain
\begin{align}\label{result.gen.point.ongrid}
    \sup_{s,t\in\mc E_{s_0,n}} \Big|\Big(\widehat F^{c}_{a,n}(t)-\widehat F_{a,n}(t)\Big)-\Big(\widehat F^{c}_{a,n}(s)-\widehat F_{a,n}(s)\Big)\Big|=o_p(n^{-(k+1)/(2k+1)}).
\end{align}
For general points $v_1,v_2$ which are off-grid, we exploit
\begin{align}
&\Big|\Big(\widehat F^{c}_{a,n}(v_2)-\widehat F_{a,n}(v_2)\Big)-\Big(\widehat F^{c}_{a,n}(v_1)-\widehat F_{a,n}(v_1)\Big)\Big|\nonumber\\
&=\Big|\Big(\widehat F^{c}_{a,n}(v^*_2)-\widehat F_{a,n}(v_2)\Big)-\Big(\widehat F^{c}_{a,n}(v^*_1)-\widehat F_{a,n}(v_1)\Big)\Big|\nonumber\\
&\le \Big|\Big(\widehat F^{c}_{a,n}(v^*_2)-\widehat F_{a,n}(v^*_2)\Big)-\Big(\widehat F^{c}_{a,n}(v^*_1)-\widehat F_{a,n}(v^*_1)\Big)\Big|\nonumber\\
&\indent +\Big|\Big(\widehat F_{a,n}(v^*_2)-\widehat F_{a,n}(v_2)\Big)-\Big(\widehat F_{a,n}(v^*_1)-\widehat F_{a,n}(v_1)\Big)\Big|\label{result.gen.point.ontooff},
\end{align}
where $v^*_i:=\max\{v:v\in \mc S_n,v\le v_i\}$, for $i=1,2$. Since $\delta_n=O_p(n^{-(k+1)/(2k+1)})$, and from similar reasoning to \eqref{lemmaDiD-term1}-\eqref{lemmaDiD-term3}, \eqref{result.in.DiD.remainder} and \eqref{result.in.DiD.empterm}, one can control \eqref{result.gen.point.ontooff} by
\begin{align}\label{DiD.final.res1}
    \sup_{v_1,v_2\in\mc E_{s_0,n}}\Big|\Big(\widehat F_{a,n}(v^*_2)-\widehat F_{a,n}(v_2)\Big)-\Big(\widehat F_{a,n}(v^*_1)-\widehat F_{a,n}(v_1)\Big)\Big|=o_p(n^{-(k+1)/(2k+1)}).    
\end{align}
In addition, by \eqref{result.gen.point.ongrid}, we have
\begin{align}\label{DiD.final.res2}
   \sup_{v_1,v_2\in\mc E_{s_0,n}} \Big|\Big(\widehat F^{c}_{a,n}(v^*_2)-\widehat F_{a,n}(v^*_2)\Big)-\Big(\widehat F^{c}_{a,n}(v^*_1)-\widehat F_{a,n}(v^*_1)\Big)\Big|= o_p(n^{-(k+1)/(2k+1)}).
\end{align}
Combining \eqref{DiD.final.res1} with \eqref{DiD.final.res2} completes the proof.
\end{proof}
\begin{lemma}\label{lem:integratedDiD}
Under the assumptions as in Lemma \ref{successiveknots}, for $t\in[-K,K]$ for any $K>0$, the following holds for $a \in \{0,1\}$.
\begin{align*}
    \int_{s_0}^{s_n(t)}\Big((\widehat F^c_{a,n}(v)-\widehat F^c_{a,n}(s_0))-(\widehat F_{a,n}(v)-\widehat F_{a,n}(s_0))\Big)dv=o_p(n^{-(k+2)/(2k+1)}),
\end{align*}
where $s_n(t)=s_0+n^{-1/(2k+1)}t$.
\end{lemma}
\begin{proof}
Lemma \ref{lem:isotonicDID} directly concludes the proof.
\end{proof}
We state another lemma to control the empirical process term involving the integrated difference in localized terms between $D_{a,\widehat\theta_a}$ and $D_{a,\theta_{a,\infty}}$.
\begin{lemma}\label{lem:empiricalprocDID}
Under the same assumptions as in Lemma \ref{successiveknots}, for $a \in \{0,1\}$, we have
 \begin{align*}
    (\PP_n-P_*)\Big[\frac{r_n}{p_a(s_0)}\int_{s_0}^{s_n(t)}(D_{a,\hat\theta_a}(v)-D_{a,\hat\theta_a}(s_0))-(D_{a,{\theta_{a,\infty}}}(v)-D_{a,{\theta_{a,\infty}}}(s_0))dv\Big]=o_p(1).
\end{align*}   
\end{lemma}
\begin{proof}
Without loss of generality, we assume $t>0$.
First, from the definition of $D_{a,\hat\theta_a}$ and $D_{a,\theta_{a,\infty}}$, we have
\begin{align*}
(D_{a,\hat\theta_a}(v)-D_{a,\hat\theta_a}(s_0))&-(D_{a,{\theta_{a,\infty}}}(v)-D_{a,{\theta_{a,\infty}}}(s_0))\\
&=\frac{\pi_{a,\infty}(X)-\widehat\pi_a(X)}{\pi_{a,\infty}(X)\widehat\pi_a(X)}I(A=a)I(s_0< Y\le v)\\
&\indent + \Big(1-\frac{I(A=a)}{\widehat\pi_a(X)}\Big)\Big(\widehat\phi_a(v|X)-\widehat\phi_a(s_0|X)\Big)\\
&\indent - \Big(1-\frac{I(A=a)}{\pi_{a,\infty}(X)}\Big)\Big(\phi_{a,\infty}(v|X)-\phi_{a,\infty}(s_0|X)\Big).
\end{align*}
In the proof of Lemma \ref{successiveknots} (see \eqref{empterm:phi.part}), we showed that there exists $\epsilon>0$ such that 
\begin{align*}
   (\PP_n-P_*) \sup_{s\in[0,t]}\Big|\Big(1-\frac{I(A=a)}{\widehat\pi_a(X)}\Big)\Big(\widehat\phi_a(s_0+v_ns|X)&-\widehat\phi_a(s_0|X)\Big)\Big|\\
   &=O_p(n^{-(k+1+\epsilon)/(2k+1)}).
\end{align*}
Similarly to this, with Assumption \ref{assm:boundedpropscore},  \ref{assm:phiregularity}, one can show that
\begin{align*}
 (\PP_n-P_*) \sup_{s\in[0,t]}\Big|\Big(1-\frac{I(A=a)}{\pi_{a,\infty}(X)}\Big)\Big(\phi_{a,\infty}(s_0+v_ns|X)
 &-\phi_{a,\infty}(s_0|X)\Big)\Big|\\
 &=O_p(n^{-(k+3/2)/(2k+1)}).
\end{align*}
The two preceding convergence rates are $o_p(n^{-(k+1)/(2k+1)})$ and $\int_{s_0}^{s_n(t)} o_p(n^{-(k+1)/(2k+1)}) dv$ is $o_p(r_n)$. 
Thus, it suffices to show that 
\begin{align*}
(\PP_n-P_*) \sup_{s\in[0,t]}\Big|\frac{\pi_{a,\infty}(X)-\widehat\pi_a(X)}{\pi_{a,\infty}(X)\widehat\pi_a(X)}I(A=a)I(s_0< Y\le s_0+v_ns)\Big|    
\end{align*}
is $o_p(n^{-(k+1)/(2k+1)}).$
To prove this, define a function class $\mc F_{n,s_0,t}$ by
\begin{align*}
    \mc F_{n,s_0,t}:=\Big\{(s,f_{n,s_0,s,\pi}(x_1,\mathbf{x_2},x_3)):\pi \in\mc F_{\pi},x_1\in\{0,1\},\mathbf{x_2}\in\RR^d, x_3\in\RR, s\in[0,t]\Big\},
\end{align*}
where
\begin{align*}
f_{n,s_0,s,\pi}(x_1,\mathbf{x_2},x_3):=\sqrt{v_n^{-1}}&\frac{I_{a}(x_1)}{\pi(\mathbf{x_2})}I(s_0< x_3\le s_0+v_ns).
\end{align*}
Then we show $\rho$-equicontinuity of this class where the semi-metric $\rho$ is a product metric of Euclidean distance in $\RR$ and $L_2$ norm in $\mc F_{\pi}$ (eg. $\rho$ is the sum of two metrics).

We will show the four conditions for concluding $\rho$-equicontinuity of Theorem 2.11.24 of \citet{vdvandW} (which we have provided as Proposition \ref{prop:vdv.2.11.24} in the Appendix for completeness). First, due to the Assumption \ref{assm:boundedpropscore}, the class $\mc F_{n,s_0,t}$ admits an envelope $F_{n,s_0,t}:=K\sqrt{v_n^{-1}}I(A=a)I(s_0< Y\le s_0+v_nt)$. And, by the log-concavity of the distribution of $Y^a$, there exists a constant $C_5>0$ such that
\begin{align*}
P_* F_{n,s_0,t}^2=K^2 v_n^{-1}\PP(s_0< Y^a\le s_0+v_nt) \le K^2C_5t.    
\end{align*}
Furthermore, for any $\vartheta>0$, since  $|F_{n,s_0,t}|\lesssim \sqrt{v_n^{-1}}$, it is obvious that 
\begin{align*}
    P_* F_{n,s_0,t}^2I( F_{n,s_0,t}>\vartheta\sqrt{n})\rightarrow 0.
\end{align*}
Next, we prove 
\begin{align*}
\sup_{|t_1-t_2|\le\zeta_n,\|\pi_1-\pi_2\|\le\zeta_n}P_*(f_{n,s_0,t_1,\pi_1}-f_{n,s_0,t_2,\pi_2})^2 \rightarrow 0,  
\end{align*}
for any $\zeta_n\rightarrow 0$. Indeed, since $\zeta_n\rightarrow 0$, 
\begin{align*}
P_*(f_{n,s_0,t_1,\pi_1}-f_{n,s_0,t_2,\pi_2})^2&\le P_*(f_{n,s_0,t_1,\pi_1}-f_{n,s_0,t_2,\pi_1})^2 + P_*(f_{n,s_0,t_2,\pi_1}-f_{n,s_0,t_2,\pi_2})^2 \\
&\le K^2C_5|t_1-t_2|+K^4C_5tP_*\|\pi_1-\pi_2\|^2\\
&\le K^2C_5\zeta_n+K^4C_5t\zeta_n^2 \rightarrow 0.
\end{align*}
We already showed the bracketing entropy condition of the class $\mc F_{t1}$ in the proof of Lemma \ref{successiveknots}. Since, for arbitrary probability measure $Q$ and $\epsilon>0$,
\begin{align*}
 \sup_Q{\rm log}N(\epsilon,\mc F_{n,s_0,t},L_2(Q))   \asymp \sup_Q{\rm log}N(\epsilon,\mc F_{t1},L_2(Q)) \lesssim \epsilon^{-V}-{\rm log}(\epsilon),
\end{align*}
 the function class $\mc F_{n,s_0,t}$ satisfies the uniform entropy integral condition. Combining the preceding results with $L_2(P_*)$ convergence of $\widehat\pi_a$ to $\pi_{a,\infty}$ which is given in \ref{assm:nuisanceconvergence} and $v_nt\rightarrow0$, we have $(\PP_n-P_*)\sup_{(t,f)\in\mc F_{n,s_0,t}}f=o_p(n^{-1/2})$ by $\rho$-equicontinuity. This further implies $$(\PP_n-P_*)\sup_{s\in[0,t]}\left|\frac{\pi_{a,\infty}(X)-\widehat\pi_a(X)}{\pi_{a,\infty}(X)\widehat\pi_a(X)}I(A=a)I(s_0< Y\le v_ns)\right|=o_p(n^{-(k+1)/(2k+1)}).$$
And, again since $\int_{s_0}^{s_n(t)} o_p(n^{-(k+1)/(2k+1)}) dv=o_p(r_n)$, this concludes the lemma.
\end{proof}

\subsubsection{Proof of the main Theorem}\label{subsec:CI.proof.section}
In this section we give the proof of Theorem \ref{prop:CI} based on the preceding lemmas in Section \ref{subsec:lemmas.for.CI}.

\begin{proof}
First, we find the two constants $\gamma_{1a},\gamma_{2a}$ for each $a\in\{0,1\}$ which satisfy
\begin{align*}
   \gamma_{1a} G_{\alpha_a,k,\sigma_a}(\gamma_{2a} t)=_d Y_k(t),
\end{align*}
where $Y_k$ is the integrated Gaussian process defined in \eqref{integrgaussprocess}. Due to the scaling property of Brownian motion, we have,
\begin{align*}
\gamma_{1a}\gamma_{2a}^{3/2}=\alpha_a^{-1},~~~\gamma_{1a}\gamma_{2a}^{k+2}=\sigma_a^{-1},
\end{align*}
where $\alpha_a,\sigma_a$ are defined in Lemma \ref{lem:finallemmaforCI}-(a).
The solution of the above system of equations is 
\begin{align*}
    \gamma_{1a}&=\Bigg(\frac{|\varphi^{(k)}_a(s_0)|}{(k+2)!}\Bigg)^{3/(2k+1)}\Bigg(\frac{p_a(s_0)^2}{\chi_{\theta_a}}\Bigg)^{(k+2)/(2k+1)},\\
    \gamma_{2a}&=\Bigg(\frac{|\varphi^{(k)}_a(s_0)|}{(k+2)!}\Bigg)^{-2/(2k+1)}\Bigg(\frac{p_a(s_0)^2}{\chi_{\theta_a}}\Bigg)^{-1/(2k+1)}.
\end{align*}
Next, since
\begin{align}
\gamma_{1a} \gamma_{2a}^2 \widehat H^{\rm locmod,(2)}_{a,n}\rightarrow H_k^{(2)}(t),\\
\gamma_{1a} \gamma_{2a}^3 \widehat H^{\rm locmod,(3)}_{a,n}\rightarrow H_k^{(3)}(t),
\end{align}
by Lemma \ref{lem:finallemmaforCI} and \ref{lem:finallemmaforCI.joint}, the preceding displays imply
\begin{align*}
\begin{pmatrix}
n^{k/(2k+1)}(\widehat\varphi_{1,n}(s_0)-\varphi_1(s_0))\\
n^{(k-1)/(2k+1)}(\widehat\varphi'_{1,n}(s_0)-\varphi'_1(s_0))\\
n^{k/(2k+1)}(\widehat\varphi_{0,n}(s_0)-\varphi_0(s_0))\\
n^{(k-1)/(2k+1)}(\widehat\varphi'_{0,n}(s_0)-\varphi'_0(s_0))\\
\end{pmatrix}
\rightarrow_d 
\begin{pmatrix}
C_k(s_0,\varphi_1) H_{1k}^{(2)}(0)\\
D_k(s_0,\varphi_1) H_{1k}^{(3)}(0)\\
C_k(s_0,\varphi_0) H_{0k}^{(2)}(0)\\
D_k(s_0,\varphi_0) H_{0k}^{(3)}(0)\\
\end{pmatrix},
\end{align*}
where $H_{1k}$ and $H_{0k}$ are independent copies of $H_k$ defined in \eqref{H-k:definition}; $C_k(s_0,\varphi_a)=(\gamma_{1a}\gamma_{2a}^2)^{-1}$ and $D_k(s_0,\varphi_a)=(\gamma_{1a}\gamma_{2a}^3)^{-1}$, $a\in\{0,1\}$.

Now, plugging in the exact values of $\gamma_{1a},\gamma_{2a}$, we get the exact values of $C_k(s_0,\varphi_a)$, and $D_k(s_0,\varphi_a)$ which are given by,
\begin{align*}
&C_k(s_0,\varphi_a)=\Bigg(\frac{|\varphi^{(k)}_a(s_0)|}{(k+2)!}\Bigg)^{1/(2k+1)}\Bigg(\frac{p_a(s_0)^2}{\chi_{\theta_a}}\Bigg)^{-k/(2k+1)},\\
&D_k(s_0,\varphi_a)=\Bigg(\frac{|\varphi^{(k)}_a(s_0)|}{(k+2)!}\Bigg)^{3/(2k+1)}\Bigg(\frac{p_a(s_0)^2}{\chi_{\theta_a}}\Bigg)^{-(k-1)/(2k+1)}.
\end{align*}
Next, \eqref{jointfororiginal} follows directly from the delta method. This completes the proof.
\end{proof}

\subsection{Proof of Theorem \ref{prop:CI-construction}}\label{subsec:proof.CI.constr}
We give the proof of Theorem \ref{prop:CI-construction} as follows.

\begin{proof}
For part (a), based on the result from Lemma \ref{lem:finallemmaforCI}-(d), the entire proof of both Theorem 2.4 and 3.2 in \citet{deng2022inference} can be directly applied to the joint process $(V_{1,n},V_{0,n})$, which converges weakly to $(V_1,V_0)$ (see \eqref{def:V.an.est.proc}--\eqref{def:V.a.limit.proc} for the definition of the joint processes).
Consequently, this yields
\begin{align}\label{eq:pitoval.joint.log.level}
\begin{pmatrix}
\sqrt{n(\tau_n^{+}(s_0;1)-\tau_n^{-}(s_0;1))}(\widehat\varphi_{1,n}(s_0)-\varphi_1(s_0))\\
\sqrt{n(\tau_n^{+}(s_0;1)-\tau_n^{-}(s_0;1))^3}(\widehat\varphi'_{1,n}(s_0)-\varphi'_1(s_0))\\
\sqrt{n(\tau_n^{+}(s_0;0)-\tau_n^{-}(s_0;0))}(\widehat\varphi_{0,n}(s_0)-\varphi_0(s_0))\\
\sqrt{n(\tau_n^{+}(s_0;0)-\tau_n^{-}(s_0;0))^3}(\widehat\varphi'_{0,n}(s_0)-\varphi'_0(s_0))\\
\end{pmatrix}
\rightarrow_d 
\begin{pmatrix}
-\alpha_1\LL_{1k}^{(0)}\\
-\alpha_1\LL_{1k}^{(1)}\\
-\alpha_0\LL_{0k}^{(0)}\\
-\alpha_0\LL_{0k}^{(1)}\\
\end{pmatrix},   
\end{align}
where $\alpha_a=\frac{1}{p_a(s_0)}\sqrt{\EE \Big[\frac{\pi_{a,*}(\mathbf{X})}{\pi^2_{a,\infty}(\mathbf{X})}\eta_{a,*}(s_0|\mathbf{X})\Big]}$, for $a\in\{0,1\}$, is defined in Lemma \ref{lem:finallemmaforCI}-(a). 
Recall $h_{k;-}^*$ and $h_{k;+}^*$ are characterized as the absolute values of the location of the first touch points of the pair $(H_k,Y_k)$.
Since $V_1$ and $V_0$ converge weakly to independent joint processes that are of the form \eqref{def:V.a.limit.proc}, it is obvious that $\LL_{1k}^{(i)}$ and $\LL_{0k}^{(i)}$ are independent, for $i=0,1$.
The distributional result for $\widehat p_{a,n}$ and $\widehat p'_{a,n}$ follow directly by the delta method.

Next, the part (b) can be obtained directly from the part (a), similarly to the proof of Theorem 2.6 in \citet{deng2022inference} which is directly concluded by Theorem 2.4 therein.
\end{proof}

\subsection{Proof of Lemma \ref{lemma:tuning.consistency}}\label{subsec:tuning.proof}
\begin{proof}
Without loss of generality, we show that the following holds for a positive deterministic sequence $h_n$.
\begin{align}
&h^{-1}P_* \left\{ D_{a,{\theta_{a,\infty}}}(s_0+h)-D_{a,{\theta_{a,\infty}}}(s_0)\right\}^2 \rightarrow  \chi_{\theta_a},\label{WTS.chi.1}\\
&h^{-1}\left[P_*\left\{ D_{a,{\theta_{a,\infty}}}(s_0+h)-D_{a,{\theta_{a,\infty}}}(s_0) \right\}^2-P_*\left\{D_{a,{\widehat\theta_{a}}}(s_0+h)-D_{a,{\widehat\theta_{a}}}(s_0)\right\}^2\right] \rightarrow  0,\label{WTS.chi.2}\\
&h^{-1}(\PP_n-P_*)\left\{  D_{a,\widehat\theta_a}(s_0+h) - D_{a,\widehat\theta_a}(s_0)\right\}^2\rightarrow_p 0,\label{WTS.chi.3}
\end{align}	
since combining \eqref{WTS.chi.1}--\eqref{WTS.chi.3} yields the main result.
First, the same procedure to check the limiting behavior of \eqref{cov:1}--\eqref{cov:4} can be directly applied to show \eqref{WTS.chi.1}. Hence, we omit the proof.
Secondly, for \eqref{WTS.chi.2}, since
\begin{align}
&\left|P_*\left\{ D_{a,{\theta_{a,\infty}}}(s_0+h)-D_{a,{\theta_{a,\infty}}}(s_0) \right\}^2-P_*\left\{D_{a,{\widehat\theta_{a}}}(s_0+h)-D_{a,{\widehat\theta_{a}}}(s_0)\right\}^2\right|\nonumber\\
&\leq 4(K+1)P_*\left|\left( D_{a,{\theta_{a,\infty}}}(s_0+h)-D_{a,{\theta_{a,\infty}}}(s_0) \right)-\left(D_{a,{\widehat\theta_{a}}}(s_0+h)-D_{a,{\widehat\theta_{a}}}(s_0)\right)\right|,\label{ineq.chi.lemma}
\end{align}
a similar derivation step to check \eqref{result.in.DiD.remainder} (see the decomposition in the preceding paragraph therein) can be applied to \eqref{WTS.chi.2}. In \eqref{ineq.chi.lemma}, we used the fact that $|D_{a,{\widehat\theta_{a}}}|,|D_{a,{\theta_{a,\infty}}}| \le K+1$ by conditions \ref{assm:boundedpropscore} and \ref{assm:properCDF}.
Lastly, we now check \eqref{WTS.chi.3}. As in \eqref{empterm1}, we have
\begin{align*}
& (\PP_n-P_*)\left\{D_{a,\widehat\theta_a}(s_2)-D_{a,\widehat\theta_a}(s_1)\right\}^2\\
 &=(\PP_n-P_*)\left[\frac{I(A=a)}{\widehat\pi^2_a(\mathbf{X})}I(s_1<Y\le s_2)\right.\\
 &\qquad\qquad \qquad+2I(s_1<Y\le s_2)(\widehat\phi_a(s_2|\mathbf{X})-\widehat\phi_a(s_1|\mathbf{X}))\Big(1-\frac{I(A=a)}{\widehat\pi_a(\mathbf{X})}\Big)\frac{I(A=a)}{\widehat\pi_a(\mathbf{X})}\\
&\qquad\qquad\qquad\left.+(\widehat\phi_a(s_2|\mathbf{X})-\widehat\phi_a(s_1|\mathbf{X}))^2\Big(1-\frac{I(A=a)}{\widehat\pi_a(\mathbf{X})}\Big)^2\right].
\end{align*}
By an analogous derivation step used in Lemma \ref{successiveknots} to show \eqref{empterm:identity.part}, one can easily verify that
\begin{align}\label{empterm:identity.part.chi}	\sup_{[s_1,s_2]\subset I_{s_0,\omega},s_1\le s_2\le s_1+R} \Big|(\PP_n-P_*)\Big[\frac{I(A=a)}{\widehat\pi^2_a(\mathbf{X})}I(s_1<Y\le s_2)\Big|\le \epsilon|s_2-s_1|^{k+1}+B'_1,
\end{align}
where a random variable $B'_1$ has order of $O_p(n^{-(k+1)/(2k+1)})$ and is independent of $s_1,s_2$.
And, following a similar step to prove \eqref{empterm:phi.part}, one can further check that
\begin{align}\label{empterm:phi.part.chi}
	\sup_{[s_1,s_2]\subset I_{s_0,\omega},s_1\le s_2\le s_1+R} \Big|(\widehat\phi_a(s_2|\mathbf{X})-\widehat\phi_a(s_1|\mathbf{X}))^2\Big(1-\frac{I(A=a)}{\widehat\pi_a(\mathbf{X})}\Big)^2\Big|\le \epsilon|s_2-s_1|^{k+t}+B'_2,   
\end{align}
knowing that this class of functions allows an envelope $(1+K)^2R^2_1(\mathbf{X})\omega^{2\alpha}$ where $\alpha\in(1/2,1]$ by condition \ref{assm:boundedpropscore} and \ref{assm:phiestholder},
where a random variable $B_2'$ has order of $O_p(n^{-(k+t)/(2k+1)})$ with some $t>3/2$, and is independent of $s_1,s_2$.
Thus, the following holds
\begin{align*}
&(\PP_n-P_*)\left[\frac{1}{h}\left|\frac{I(A=a)}{\widehat\pi^2_a(\mathbf{X})}I(s_0<Y\le s_0+h)\right|\right]\rightarrow_p 0,\\
&(\PP_n-P_*)\left[\frac{1}{h}\left|(\widehat\phi_a(s_0+h|\mathbf{X})-\widehat\phi_a(s_0|\mathbf{X}))^2\Big(1-\frac{I(A=a)}{\widehat\pi_a(\mathbf{X})}\Big)^2\right|\right]\rightarrow_p 0,
\end{align*}
as long as $h^{-1}=O(\sqrt{n})$ regardless of an even integer $k$.
Thus, the proof now reduces to show that
\begin{align*}
(\PP_n-P_*)\left[\frac{1}{h}\left| I(s_1<Y\le s_2)(\widehat\phi_a(s_2|\mathbf{X})-\widehat\phi_a(s_1|\mathbf{X}))\Big(1-\frac{I(A=a)}{\widehat\pi_a(\mathbf{X})}\Big)\frac{I(A=a)}{\widehat\pi_a(\mathbf{X})}\right|\right]\rightarrow_p 0.
\end{align*}
We follow the same derivation steps used in Lemma \ref{successiveknots}. Recall that $\mathcal{F}_1=\{I_{(s_1,s_2](\cdot)}:[s_1,s_2]\subseteq I_{s_0,\omega}, s_1\le s_2\le s_1+R\}$ satisfied
\begin{align*}
	\sup_Q N(\epsilon,\mathcal{F}_1,L_2(Q))\lesssim -\log(\varepsilon),
\end{align*}
and $\mathcal{F}_2:=\Big\{\frac{I_{\{a\}}(x_1)}{\pi_a(\mathbf{x_2})}: \pi_a \in \mathcal{F}_{\pi}, x_1\in \{0,1\}, \mathbf{x_2}\in \mathbb{R}^d\Big\}$ satisfied
\begin{align*}
	\sup_Q N(\epsilon, \mathcal{F}_2, L_2(Q)) \lesssim \varepsilon^{-V},
\end{align*}
and $\mathcal{F}_3:=\{\phi_a(s_2|\cdot)-\phi_a(s_1|\cdot): [s_1,s_2]\subseteq [s_0-\delta,s_0+\delta], s_1\le s_2\le s_1+R, \phi_a \in \mathcal{F}_{\phi}\}$ satisfied
\begin{align*}
	\sup_Q{\rm log} N(\epsilon, \mathcal{F}_3, L_2(Q)) \lesssim  \epsilon^{-V}, 
\end{align*}
for any probability measure $Q$ with $V\in[0,2)$.
Thus, the whole function class $\mathcal{F}_{0}:=\mathcal{F}_1\cdot \mathcal{F}_2\cdot \mc F_3 \cdot (1-\mc F_2)$ has finite uniform entropy integral, by Assumptions \ref{assm:pi-bracketing} and Lemma 5.1 from \citet{vdvvl06},
\begin{align*} 
	\sup_Q {\rm log} N(\epsilon , \mathcal{F}_{0}, L_2(Q)) \lesssim \epsilon^{-V}-{\rm log}(\epsilon) 
\end{align*}
up to a constant, with $V\in[0,2)$. Furthermore, the class allows an envelope $\omega^{\alpha}K(K+1)I(s_0-\omega<Y\le s_0+\omega)I(A=a)R_1(\mathbf{X})$ which satisfies
\begin{align*}
	&\mathbb{E}\Big(\omega^{\alpha}K(K+1)I(s_0-\omega<Y\le s_0+\omega)I(A=a)R_1(\mathbf{X})\Big)^2\\
	&= \omega^{2\alpha}K^2(K+1)^2P_*(s_0-\omega<Y^a\le s_0+\omega) \mathbb{E}R^2_1(\mathbf{X})\\
	&\le C_0\omega^{2\alpha+1}K^2 \mathbb{E}R^2_1(\mathbf{X})
\end{align*}
for some constant $C_0>0$.
Hence, by Lemma \ref{emptermlemma}, for each $\epsilon>0$, there exist a random variable $B_0$ which has order of $O_p(n^{-(k+t)/(2k+1)})$ with some $t>3/2$ and is independent of $s_1,s_2$ such that
\begin{align*}
	\sup_{[s_1,s_2]\subset I_{s_0,\omega},s_1\le s_2\le s_1+R}& \left|I(s_1<Y \le s_2) (\widehat\phi_a(s_2|\mathbf{X})-\widehat\phi_a(s_1|\mathbf{X}))\Big(1-\frac{I(A=a)}{\widehat\pi_a(\mathbf{X})}\Big)\frac{I(A=a)}{\widehat\pi_a(\mathbf{X})}\right|\\
	& \le \epsilon|s_2-s_1|^{k+t}+B_0.   
\end{align*}
This implies that \eqref{WTS.chi.3} when $h^{-1}=O(\sqrt{n})$, and it completes the proof.
\end{proof}

\subsection{Proof of Theorem \ref{prop:Cons-split}}\label{subsec:proof.cons.split}
\begin{proof}
We denote the empirical process over the $i$-th fold (or subgroup) $\mc V_{n,i}$ by $\mathbb{G}^i_n=\sqrt{N}(\PP^i_n-P_*)$ where $\PP^i_n$ is the empirical measure on the same subgroup.
As discussed in the proof of Theorem \ref{prop:Consistency}, we have to check the following:
\begin{align}
&\sup_{s\in\RR}\left|\widehat F^{K_0}_{a,n}(s)-F_a(s)\right|=o_p(1),\label{split.consist.WTS0}\\
&\left|\int_{\RR}|s|dF^{K_0}_{a,n}(s)-\int_{\RR} |s|dF_a(s)\right|=o_p(1).\label{split.consist.WTS.integral}
\end{align}
Now, we prove \eqref{split.consist.WTS0}, which is equivalent to show that 
\begin{align}
&\sup_{s\in(-\infty,0]}\left|\widehat F^{K_0}_{a,n}(s)-F_a(s)\right|=o_p(1), \label{split.consist.WTS1}   \\
&\sup_{s\in[0,\infty)}\left|(1-\widehat F^{K_0}_{a,n}(s))-(1-F_a(s))\right|=o_p(1). \label{split.consist.WTS2} 
\end{align}
We start with showing \eqref{split.consist.WTS1}.
First, by Assumption \ref{assm:DR}, we have
\begin{equation}
\begin{split}
&\widehat F^{K_0}_{a,n}(s)-F_a(s)\\
&=\frac{1}{K_0}\sum_{i=1}^{K_0}\left[\PP^i_n\left\{D_{a,\widehat\theta_{a,-i}}(s)\right\}\right]-P_*D_{a,\theta_{a,\infty}}(s)\\
&=\frac{1}{K_0}\sum_{i=1}^{K_0}\left[\PP^i_n\left\{D_{a,\widehat\theta_{a,-i}}(s)\right\}\right]-\PP_n D_{a,\theta_{a,\infty}}(s)+(\PP_n-P_*)D_{a,\theta_{a,\infty}}(s),\\
\end{split}
\end{equation}
where $\PP_n^i D_{a,\widehat\theta_{a,-i}}$ is the estimated centered efficient influence function evaluated with validation sample in which the nuisance estimators are constructed upon only the observations from the training set $\mc T_{n,i}$.
Since $\sum_{i=1}^{K_0}\PP^i_n D_{a,\theta_{a,\infty}}(s)=\sum_{i=1}^{K_0}\PP_nD_{a,\theta_{a,\infty}}(s)$, the preceding display equals
\begin{align}
&\frac{1}{K_0}\sum_{i=1}^{K_0}\left[(\PP^i_n-P_*)\left\{D_{a,\widehat\theta_{a,-i}}(s)-D_{a,\theta_{a,\infty}}(s)\right\}\right]+ \frac{1}{K_0}\sum_{i=1}^{K_0}\left[P_*\left\{D_{a,\widehat\theta_{a,-i}}(s)-D_{a,\theta_{a,\infty}}(s)\right\}\right]\nonumber\\
&\qquad + (\PP_n-P_*)D_{a,\theta_{a,\infty}}(s)\nonumber\\
&=: R_{1n}(s)+R_{2n}(s)+(\PP_n-P_*)D_{a,\theta_{a,\infty}}(s).\label{sample.split.decomposition}
\end{align}

We show the three terms in \eqref{sample.split.decomposition} are negligible.  
We start with the last summand, $(\PP_n-P_*)D_{a,\theta_{a,\infty}}(s)$.  
First, by a similar reasoning used in the proof of Lemma \ref{lem:vdvW2.7.3.ext}, we note that one can show that
\begin{align}\label{eq:below.zero.GC}
 N_{[]}(\epsilon, \mc F_{\phi_\infty}, L_2(P_*)) <\infty,    
\end{align} 
for any $\epsilon>0$, where $\mc F_{\phi_\infty}:=\{\mathbf{X}\mapsto{\phi}_{a,\infty}(s|\mathbf{X}):s\in(-\infty,0]\}$.
This implies that $\mc F_{\phi_\infty}$ is a Glivenko-Cantelli class (see Theorem 2.4.1 of \citet{vdvandW}). Thus, we have
\begin{align*}
 &\sup_{s\in(-\infty,0]}\Big|(\PP_n-P_*)D_{a,{\theta_{a,\infty}}}(s)\Big|\rightarrow 0~~a.s.   
\end{align*}
Moreover, a similar procedure used to prove that \eqref{indist.rem.pt1}--\eqref{indist.rem.pt3} are $o_p(1)$ in the proof of Theorem \ref{prop:Consistency} can be applied to show that $\sup_{s\in(-\infty,0]} \left|R_{2n}\right|$ is $o_p(1)$, since $K=O_p(1)$. 
Next, it suffices to study the term $R_{1n}(s)$ in \eqref{sample.split.decomposition}.
By the tower property of expectation, 
\begin{align*}
\EE\left[\sup_{f\in \mc F_{n,-i}}\left|\mathbb{G}^i_n f\right|\right]=\EE\left\{\EE\left[\sup_{f\in \mc F_{n,-i}}\left|\mathbb{G}^i_n f\right|\Big| \mc T_{n,i}\right]\right\},
\end{align*}
where $\mc F_{n,-i}:=\left\{D_{a,\widehat\theta_{a,-i}}(s)-D_{a,\theta_{a,\infty}}(s):s\in(-\infty,0] \right\}$. We further note that
\begin{equation}\label{split.emp.proc.decom}
\begin{split}
&\left|D_{a,\widehat\theta_{a,-i}}(s)-D_{a,\theta_{a,\infty}}(s)\right|\\
&=\left|I(A=a)\left(\frac{1}{\widehat\pi_{a,-i}(\mathbf{X})}-\frac{1}{\pi_{a,\infty}(\mathbf{X})}\right)\left[I(Y \le s)-\phi_{a,\infty}(s|\mathbf{X})\right]\right|\\
&\indent + \left|\left(1-\frac{I(A=a)}{\widehat\pi_{a,-i}(\mathbf{X})}\right)
\left(\widehat\phi_{a,-i}(s|\mathbf{X})-\phi_{a,\infty}(s|\mathbf{X})\right)\right|,
\end{split}
\end{equation}
which is bounded by 
\begin{align}\label{envlope.function}
H_{n,i}(\mathbf{z})=K^2 \left|\widehat\pi_{a,-i}(\mathbf{x})-\pi_{a,\infty}(\mathbf{x})\right|+(K+1)\sup_{s\in(-\infty,0]}\left|\widehat\phi_{a,-i}(s|\mathbf{x})-\phi_{a,\infty}(s|\mathbf{x})\right|.
\end{align}
By Theorem 2.14.16 in \citet{vdvandW} (see Proposition \ref{prop:vdv.2.14.16} in our Appendix), for sufficiently large $n$, we have
\begin{align*}
\EE\left[\sup_{f\in \mc F_{n,-i}}\left|\mathbb{G}^i_n f\right|\Big| \mc T_{n,i}\right]\le C_*\|H_{n,i}\|_{P_*,2}J_{[]}(1,\mc F_{n,-i},L_2(P_*)),    
\end{align*}
for a universal constant $C_*>0$, where $J_{[]}(1,\mc F_{n,-i},L_2(P_*))$ is the bracketing integral of the function class $\mc F_{n,-i}$ (see Proposition \ref{prop:vdv.2.14.16} for its definition).
We have uniformly bounded $J_{[]}(1,\mc F_{n,-i},L_2(P_*))$ for all $n$ and $i$, since $\{ Y \mapsto I(Y \leq s) : |s|\le M \}$ is a VS subgraph class, and Lemma \ref{lem:vdvW2.7.3.ext} can be applied to a function class ${\mc F}_{\widehat\phi_{a,-i}}:=\{\mathbf{X}\mapsto \widehat\phi_{a,-i}(s|\mathbf{X})ds:s\in(-\infty,0]\}$ by Assumption \ref{assm:properCDF.split}. Hence, it suffices to show that $\max_{1\le i\le K_0} \EE\|H_{n,i}\| \rightarrow 0$. Indeed, by Assumption \ref{assm:nuisanceconvergence.split},
\begin{align*}
\max_{1\le i\le K_0} \|H_{n,i}\|_{P_*,2}&\le K^2 \max_{1\le i\le K_0} \left\| \widehat{\pi}_{a,-i}(\mathbf{X})-{\pi}_{a,\infty}(\mathbf{X})\right\|_{P_*,2} \\
&\qquad + (K+1) \max_{1\le i\le K_0} \left\|\sup_{s\in(-\infty,0]}\left(\widehat{\phi}_{a,-i}(s|\mathbf{X})-{\phi}_{a,\infty}(s|\mathbf{X})\right)\right\|_{P_*,2}\rightarrow 0,    
\end{align*}
since the convergence of the second term on the right hand side of the preceding display can be easily obtained by the similar steps used in Theorem \ref{prop:Consistency}.
Similar reasoning can be applied to show \eqref{split.consist.WTS2}, since we have
\begin{equation*}
\begin{split}
&(1-\widehat F^{K_0}_{a,n}(s))-(1-F_a(s))\\
&=\frac{1}{K_0}\sum_{i=1}^{K_0}\left[\PP^i_n\left\{1-D_{a,\widehat\theta_{a,-i}}(s)\right\}\right]-P_*\left(1-D_{a,\theta_{a,\infty}}(s)\right)\\
&=\frac{1}{K_0}\sum_{i=1}^{K_0}\left[\PP^i_n\left\{1-D_{a,\widehat\theta_{a,-i}}(s)\right\}\right]-\PP_n \left(1-D_{a,\theta_{a,\infty}}(s)\right)+(\PP_n-P_*)\left(1-D_{a,\theta_{a,\infty}}(s)\right),\\
\end{split}
\end{equation*}
and, further since, analogously to Lemma \ref{lem:vdvW2.7.3.ext}, we obtain
\begin{align}\label{eq:below.zero.GC}
 N_{[]}(\epsilon, \mc F'_{\phi_\infty}, L_2(P_*)) <\infty,    
\end{align} 
for any $\epsilon>0$, where $\mc F'_{\phi_\infty}:=\{\mathbf{X}\mapsto (1-{\phi}_{a,\infty}(s|\mathbf{X})):s\in[0,\infty)\}$.
Hence, we omit the detailed derivation.

Now, we prove \eqref{split.consist.WTS.integral}.
Since one can write $\widehat F^{K_0}_{a,n}(s)=\widehat F^{K_0}_{a,n,1}(s)-\widehat F^{K_0}_{a,n,2}(s)$ for
\begin{align*}
&\widehat F^{K_0}_{a,n,1}(s)=\frac{1}{K_0}\sum_{k=1}^{K_0}\frac{1}{N}\sum_{i\in\mc V_{n,k}} \left\{\frac{I(A_i=a)}{\widehat\pi_{a,-k}(\mathbf{X}_i)}I(Y_i\le s) +\widehat\phi_{a,-k} (s|\mathbf{X}_i)\right\}, \\
&\widehat F^{K_0}_{a,n,2}(s)=\frac{1}{K_0}\sum_{k=1}^{K_0}\frac{1}{N}\sum_{i\in\mc V_{n,k}} \frac{I(A_i=a)}{\widehat\pi_{a,-k}(\mathbf{X}_i)}\widehat\phi_{a,-k}(s|\mathbf{X}_i), 
\end{align*}
where $\widehat F^{K_0}_{a,n,1}(s)$ and $\widehat F^{K_0}_{a,n,2}(s)$ are uniformly bounded monotone functions (by Assumption \ref{assm:boundedpropscore.split} and \ref{assm:properCDF.split}), 
we can still proceed with the derivation used in the proof of Theorem \ref{prop:Consistency}.
It suffices to check
\begin{align}
&\int_{-\infty}^0 \left|\widehat F^{K_0}_{a,n}(s)\right|ds<\infty \quad \mathbf{X}-a.e.,\label{eq:const.split.part2.wts1}\\
&\sup_{t\leq 0}\left|\int_{-\infty}^t \left[\widehat F^{K_0}(s)-F_a(s)\right]ds\right|=o_p(1),\label{eq:const.split.part2.wts2}
\end{align}
since one can derive the following analogously:
\begin{align*}
&\int_0^\infty \left|\widehat 1-F^{K_0}_{a,n}(s)\right|ds<\infty \quad \mathbf{X}-a.e.,\\
&\sup_{0\leq t}\left|\int_t^{\infty} \left[\left\{1-\widehat F^{K_0}(s)\right\}-\left\{1-F_a(s)\right\}\right]ds\right|=o_p(1).
\end{align*}
We first check \eqref{eq:const.split.part2.wts1}.
By Assumption \ref{assm:boundedpropscore.split} and \ref{assm:properCDF.split}, we have
\begin{equation*}
\begin{split}
\EE\int_{-\infty}^0 \left|\widehat F^{K_0}_{a,n}(s)\right|ds&=\EE\left[\int_{-\infty}^0\left|\frac{1}{K_0}\sum_{i=1}^{K_0}\PP_n^i D_a,\hat\theta_{a,-i}(s)\right|ds\right]\\
&\leq \EE\left\{\frac{1}{K_0}\sum_{i=1}^{K_0}\PP_n^i\int_{-\infty}^0 \left|D_a,\hat\theta_{a,-i}(s)\right|ds\right\}\\
&\leq \EE\left\{\frac{1}{K_0}\sum_{k=1}^{K_0}\frac{1}{N}\sum_{i\in\mc V_{n,k}} \frac{I(A_i=a)}{\widehat\pi_{a,-k}(\mathbf{X}_i)} I(Y_i\leq 0)(-Y_i)\right\}\\
&\quad + \EE\left\{\frac{1}{K_0}\sum_{k=1}^{K_0}\frac{1}{N} \left|\left[1-\frac{I(A_i=a)}{\widehat\pi_{a,-k}(\mathbf{X}_i)}\right]\int_{-\infty}^0 \widehat\phi_{a,-k}(s|\mathbf{X}_i) ds\right|\right\} \\
&\leq K\EE|Y^a| + (1+K)\EE h(\mathbf{X})<\infty.
\end{split}
\end{equation*}
Thus, this implies \eqref{eq:const.split.part2.wts1}.
Now we check \eqref{eq:const.split.part2.wts2}.
As in \eqref{sample.split.decomposition}, we have
\begin{equation*}
\begin{split}
&\int_{-\infty}^t \widehat F^{K_0}_{a,n}(s)-F_a(s) ds\\
&=\int_{-\infty}^t\frac{1}{K_0}\sum_{i=1}^{K_0}\left[(\PP^i_n-P_*)\left\{D_{a,\widehat\theta_{a,-i}}(s)-D_{a,\theta_{a,\infty}}(s)\right\}\right] ds\\
&\indent + \int_{-\infty}^t \frac{1}{K_0}\sum_{i=1}^{K_0} \left[P_* \left\{D_{a,\widehat\theta_{a,-i}}(s)-D_{a,\theta_{a,\infty}}(s)\right\}\right]ds\\
&\indent + \int_{-\infty}^t (\PP_n-P_*)D_{a,\theta_{a,\infty}}(s)ds\\
&=:\Tilde R_{1n}(t)+\Tilde R_{2n}(t)+ \int_{-\infty}^t (\PP_n-P_*)D_{a,\theta_{a,\infty}}(s)ds.
\end{split}
\end{equation*}
A similar technique used to prove that \eqref{first-second-pt2}--\eqref{third-second-pt2} (see the proof of Theorem \ref{prop:Consistency}) can be applied to show that $\sup_{t\leq 0}|\Tilde R_{2n}(t)|=o_p(1)$.
For $\Tilde R_{1n}$, one can apply Lemma \ref{lem:vdvW2.7.3.ext} to a class of functions $\Tilde{\mc F}_{\phi_\infty}:=\{\mathbf{X}\mapsto \int_{-\infty}^t {\phi}_{a,\infty}(s|\mathbf{X})ds:t\in(-\infty,0]\}$,
since $\int_{-\infty}^0 {\phi}_{a,\infty}(s|\mathbf{X})ds \leq \int_{\RR} |s|{\phi}_{a,\infty}(s|\mathbf{X})\in L_2(P_*)$ (see the proof of Theorem \ref{prop:Consistency}).
This implies
\begin{align}\label{eq:below.zero.GC.integrated}
 N_{[]}(\epsilon, \Tilde{\mc F}_{\phi_\infty}, L_2(P_*)) <\infty,    
\end{align} 
for any $\epsilon>0$, which in turn indicates $\Tilde{\mc F}_{\phi_\infty}$ is a Glivenko-Cantelli class by Theorem 2.4.1 of \citet{vdvandW}. 
Thus, we have
\begin{align*}
 &\sup_{t\in(-\infty,0]}\Big|\int_{-\infty}^t(\PP_n-P_*)D_{a,{\theta_{a,\infty}}}(s)ds\Big|\rightarrow 0~~a.s..   
\end{align*}
Lastly, we study the term $\Tilde{R}_{1n}$.
By the tower property of expectation, we have
\begin{align*}
\EE\left[\sup_{f\in \Tilde{\mc F}_{n,-i}}\left|\mathbb{G}^i_n f\right|\right]=\EE\left\{\EE\left[\sup_{f\in \Tilde{\mc F}_{n,-i}}\left|\mathbb{G}^i_n f\right|\Big| \mc T_{n,i}\right]\right\},
\end{align*}
where $\Tilde{\mc F}_{n,-i}:=\left\{\int_{-\infty}^t \left(D_{a,\widehat\theta_{a,-i}}(s)-D_{a,\theta_{a,\infty}}(s)\right)ds:t\in(-\infty,0] \right\}$. 
We further note that
\begin{equation}\label{split.emp.proc.decom.integrated}
\begin{split}
&\left|\int_{-\infty}^t \left(D_{a,\widehat\theta_{a,-i}}(s)-D_{a,\theta_{a,\infty}}(s)\right)ds\right|\\
&=\left|I(A=a)\left(\frac{1}{\widehat\pi_{a,-i}(\mathbf{X})}-\frac{1}{\pi_{a,\infty}(\mathbf{X})}\right)\left[I(Y \le t)(t-Y)-\int_{-\infty}^t\phi_{a,\infty}(s|\mathbf{X})ds\right]\right|\\
&\indent + \left|\left(1-\frac{I(A=a)}{\widehat\pi_{a,-i}(\mathbf{X})}\right)
\int_{-\infty}^t\left(\widehat\phi_{a,-i}(s|\mathbf{X})-\phi_{a,\infty}(s|\mathbf{X})\right)ds\right|,
\end{split}
\end{equation}
which admits an envelope function of
\begin{equation}\label{envlope.function}
\begin{split}
\Tilde H_{n,i}(\mathbf{z})&=K^2 \left|\widehat\pi_{a,-i}(\mathbf{x})-\pi_{a,\infty}(\mathbf{x})\right|(I(A=a)|Y|+h(\mathbf{X}))\\
&\indent+(K+1)\int_{-\infty}^{\infty}\left|\widehat\phi_{a,-i}(s|\mathbf{x})-\phi_{a,\infty}(s|\mathbf{x})\right|ds,
\end{split}
\end{equation}
by Assumption \ref{assm:boundedpropscore.split} and \ref{assm:properCDF.split}.
By Theorem 2.14.16 in \citet{vdvandW} (see Proposition \ref{prop:vdv.2.14.16} in our Appendix), for sufficiently large $n$, we have
\begin{align*}
\EE\left[\sup_{f\in \Tilde{\mc F}_{n,-i}}\left|\mathbb{G}^i_n f\right|\Big| \mc T_{n,i}\right]\le C'_*\|\Tilde H_{n,i}\|_{P_*,2}J_{[]}(1,\Tilde{\mc F}_{n,-i},L_2(P_*)),    
\end{align*}
for a universal constant $C'_*>0$, where $J_{[]}(1,\Tilde{\mc F}_{n,-i},L_2(P_*))$ is the bracketing integral of the function class $\Tilde{\mc F}_{n,-i}$ (see Proposition \ref{prop:vdv.2.14.16} for its definition).
We have uniformly bounded $J_{[]}(1,\mc F_{n,-i},L_2(P_*))$ for all sufficiently large $n$ and $i$, since $\{ Y \mapsto I(Y \leq t)(t-Y) : |t|\le M \}$ has a polynomial covering number (discussed in the proof of Theorem \ref{prop:Consistency}), and Lemma \ref{lem:vdvW2.7.3.ext} can be applied to a class of functions $\Tilde{\mc F}_{\widehat\phi_{a,-i}}:=\{\mathbf{X}\mapsto \int_{-\infty}^t \widehat\phi_{a,-i}(s|\mathbf{X})ds:t\in(-\infty,0]\}$ by Assumption \ref{assm:properCDF.split}. 
Hence, it suffices to show that $\max_{1\le i\le K_0} \EE\|\Tilde H_{n,i}\| \rightarrow 0$. 
By Assumption \ref{assm:nuisanceconvergence.split}, we have
\begin{align*}
\max_{1\le i\le K_0} \|\Tilde H_{n,i}\|_{P_*,2}&\le 2K^2 \max_{1\le i\le K_0} \left\| \widehat{\pi}_{a,-i}(\mathbf{X})-{\pi}_{a,\infty}(\mathbf{X})\right\|_{P_*,2} (\|Y^a\|_{P_*,2}+\|h\|_{P_*,2})\\
&\qquad + (K+1) \max_{1\le i\le K_0} \left\|\int_{-\infty}^{\infty}\left|\widehat\phi_{a,-i}(s|\mathbf{x})-\phi_{a,\infty}(s|\mathbf{x})\right|ds\right\|_{P_*,2}\rightarrow 0.  \end{align*}
This completes the proof.
\end{proof}

This is an extension of the Lemma 4 of \citet{westling2024inference}.
\begin{lemma}\label{lem:vdvW2.7.3.ext}
Assume that a fixed function $f$ satisfies $0\leq f_s(x)\leq f_t(x)$ for all $s\leq t$ and $x\in \mc X$, and $F:=f_0$ satisfies $F\in L_2(P)$.
Then $N_{[]}(\epsilon \|F\|_{P,2},\mc F,L_2(P))\leq 4/\epsilon^2$ for all $\epsilon\in (0,1]$, where $\mc F:=\{x\mapsto f_t(x):t\in (-\infty,0]\}$.
\end{lemma}
\begin{proof}
When $f_0=0$, then there is nothing to show. 
Hence, without loss of generality, we assume $\|F\|_{P,2}>0$.
We denote $f_{-\infty}=0$, and $h(s)=\int_\RR f_s(x)^2 dP(x)$ which is a monotone function from $(-\infty,0]$ to $[0,1]$.
And, we have $P(f_t-f_s)^2\leq Pf_t^2-Pf_s^2=h(t)-h(s)$, for all $-\infty<s\leq t\leq 0$.
For a given $\epsilon\in (0,1]$, we define $\eta=\epsilon \|F\|_{P,2}$.
We follow the same procedure used in the proof of Lemma 4 in \citet{westling2024inference} to construct a set of $\eta$-brackets (in $L_2(P)$) which cover $\mc F$.
We denote $t_0=-\infty$ and $\ell_1=f_{-\infty}=0$. 
Then we define $t_1:=\sup \{t\leq 0: h(t)-h(-\infty)\leq \eta^2\}$, where $h(-\infty)=0$.
It implies $h(t_1-)\leq \eta^2$ and $h(t_1+)\geq \eta^2$.
If $h(t_1)\leq \eta^2$, then we define $I_1=(-\infty,t_1]$ and $u_1:=f_{t_1}$. 
On the other hand, if $h(t_1)>\eta^2$, we define $I_1=(-\infty,t_1)$ and $u_1:x\mapsto sup_{t<t_1}f_t(x)$. Then, the bracket $[\ell_1,u_1]$ covers $\{f_t:t\in I_1\}$ and satisfies $P(u_1-\ell_1)^2\leq \eta^2$.
The rest of the proof follows exactly same steps in the proof of aforementioned Lemma 4 of \citet{westling2024inference}, except that we define $r_1:=h(t_1)$ when $I_1=(-\infty,t_1)$, and we define $r_1:=h(t_1+)$ when $I_1=(-\infty,t_1]$.
We refer the readers to the proof therein for detailed derivation.
\end{proof}

\subsection{Proof of Theorem \ref{prop:CI.split}}\label{subsec:proof.CI.split}
\begin{proof}
The only differences compared with the proof of Theorem \ref{prop:CI} arise from Lemma \ref{successiveknots}, \ref{lem:isotonicDID} and the limiting behavior of $\widehat G^{\rm locmod}_{a,n}$ that is defined in Lemma \ref{lem:finallemmaforCI}.
We start by checking the conclusion of Lemma \ref{successiveknots} still holds, namely,
\begin{align*}
\tau^{+}_n-\tau^{-}_n=O_p(n^{-1/(2k+1)}),
\end{align*}
where $\tau^+_n, \tau^-_n$ are estimator knots as in Lemma \ref{successiveknots}.
We first show 
\begin{align*}
&\indent\left|\frac{\tau_n^{+}-\tau_n^{-}}{4}\Big[(\widehat F^{K_0}_{a,n}-F_a)(\tau_n^{+})-(\widehat F^{K_0}_{a,n}-F_a)(\tau_n^{-})\Big]\right| \\
&\le \epsilon(\tau^{+}_n-\tau^{-}_n)^{k+2}+(\tau^{+}_n-\tau^{-}_n)O_p(n^{-(k+1)/(2k+1)})+(\tau^{+}_n-\tau^{-}_n)^2o_p(n^{-k/(2k+1)}),
\end{align*}
for sufficiently small $\epsilon>0$.

Since we verified the following decomposition,
\begin{equation}\label{Reference.decom}
\begin{split}
&\widehat F^{K_0}_{a,n}(s)-F_a(s)\\
&=\frac{1}{K_0}\sum_{i=1}^{K_0}\left[(\PP^i_n-P_*)\left\{D_{a,\widehat\theta_{a,-i}}(s)-D_{a,\theta_{a,\infty}}(s)\right\}\right]+ \frac{1}{K_0}\sum_{i=1}^{K_0}\left[P_*\left\{D_{a,\widehat\theta_{a,-i}}(s)-D_{a,\theta_{a,\infty}}(s)\right\}\right]\\
&\qquad + (\PP_n-P_*)D_{a,\theta_{a,\infty}}(s)
\end{split}
\end{equation}
in the proof of Theorem \ref{prop:Consistency}, this yields
\begin{align*}
&\left(\widehat F^{K_0}_{a,n}(\tau_n^{+})-F_a(\tau_n^{+})\right)-\left(\widehat F^{K_0}_{a,n}(\tau_n^{-})-F_a(\tau_n^{-})\right)\nonumber\\
&=\frac{1}{K_0}\sum_{i=1}^{K_0}\left[(\PP^i_n-P_*)\left\{\left(D_{a,\widehat\theta_{a,-i}}(\tau_n^{+})-D_{a,\theta_{a,\infty}}(\tau_n^{+})\right)-\left(D_{a,\widehat\theta_{a,-i}}(\tau_n^{-})-D_{a,\theta_{a,\infty}}(\tau_n^{-})\right)\right\}\right]\\
&\indent + \frac{1}{K_0}\sum_{i=1}^{K_0}\left[P_*\left\{\left(D_{a,\widehat\theta_{a,-i}}(\tau_n^{+})-D_{a,\theta_{a,\infty}}(\tau_n^{+})\right)-\left(D_{a,\widehat\theta_{a,-i}}(\tau_n^{-})-D_{a,\theta_{a,\infty}}(\tau_n^{-})\right)\right\}\right]\\
&\indent + (\PP_n-P_*)\left[D_{a,\theta_{a,\infty}}(\tau_n^{+})-D_{a,\theta_{a,\infty}}(\tau_n^{-})\right]\\
&=:R_{5n}+R_{6n}+R_{7n},
\end{align*}
where $D_{a,\widehat\theta_{a,-i}}$ is defined in the proof of Theorem \ref{prop:Cons-split}.
The terms $R_{6n}$ and $R_{7n}$ do not involve empirical processes indexed by the nuisance functions, and so their analysis, which we present briefly next, is similar to the analogous analysis done previously in the proof of Theorem~\ref{prop:CI}.

By similar reasoning to $(\PP_n-P_*)D_{a,\widehat\theta_a}$ in \eqref{empterm:identity.part} and \eqref{empterm:phi.part} with Assumption \ref{assm:phiinfholder}, we have
\begin{equation}\label{Lemma.knots.check}
\begin{split}
(\tau_n^{+}-\tau_n^{-})\left|R_{7n}\right|&=(\tau_n^{+}-\tau_n^{-})\left|(\PP_n-P_*)\left[\frac{I(A=a)}{\pi_{a,\infty}(\mathbf{X})}I(\tau_n^{-}<Y\le \tau_n^{+})\right.\right.\\
&\indent +\left.\left. (\phi_{a,\infty}(\tau_n^{+}|\mathbf{X})-\phi_{a,\infty}(\tau_n^{-}|\mathbf{X}))\Big(1-\frac{I(A=a)}{\pi_{a,\infty}(\mathbf{X})}\Big)\right]\right|\\
&\le \epsilon(\tau_n^{+}-\tau_n^{-})^{k+2}+(\tau_n^{+}-\tau_n^{-})O_p\Big(n^{-(k+1)/(2k+1)}\Big),
\end{split}
\end{equation}
for sufficiently small $\epsilon>0$.
By relying on similar steps to check terms \eqref{remainder1}--\eqref{remainder3} in the proof of Lemma \ref{successiveknots}, one can easily prove
\begin{align*}
(\tau_n^{+}-\tau_n^{-})|R_{6n}| \le (\tau_n^{+}-\tau_n^{-})^2o_p(n^{-k/(2k+1)}),
\end{align*}
since $K_0=O_p(1)$.

Now we consider $R_{5n}$.
By applying the following decomposition,
\begin{equation}\label{Did.decomposition}
\begin{split}
&\Big((D_{a,\widehat\theta_{a,-i}}(s_2)-D_{a,\widehat\theta_{a,-i}}(s_1))-(D_{a,\theta_{a,\infty}}(s_2)-D_{a,\theta_{a,\infty}}(s_1))\Big)\\
&=\Big(\frac{1}{\widehat\pi_{a,-i}(\mathbf{X})}-\frac{1}{\pi_{a,\infty}(\mathbf{X})}\Big)I(A=a)\Big(I(s_1<Y\le s_2)-(\phi_{a,\infty}(s_2|\mathbf{X})-\phi_{a,\infty}(s_1|\mathbf{X}))\Big)\\
&\indent + \Big(1-\frac{I(A=a)}{\widehat\pi_{a,-i}(\mathbf{X})}\Big)
\Big((\widehat\phi_{a,-i}(s_2|\mathbf{X})-\widehat\phi_{a,-i}(s_1|\mathbf{X})-(\phi_{a,\infty}(s_2|\mathbf{X})-\phi_{a,\infty}(s_1|\mathbf{X}))\Big),
\end{split}
\end{equation}
for $s_0\in[s_1,s_2]$, and further
by Lemma C.3 in \citet{kim2018causal} (see Lemma \ref{kkklem}) and Assumption \ref{assm:logconcv}, \ref{assm:boundedpropscore}, \ref{assm:samplesplit}, \ref{assm:L2conditionsforCI.split} (recall that the $L_2(P_*)$ norm is denoted by $\|\cdot\|$), we have
\begin{equation*}
\begin{split}
&\sqrt{n} P_*\left|(\tau_n^{+}-\tau_n^{-})^{-1}R_{5n}\right|\\
&\asymp \sqrt{N} P_*\left|(\tau_n^{+}-\tau_n^{-})^{-1}R_{5n}\right|\\
&\le \max_i P_*\left|(\tau_n^{+}-\tau_n^{-})^{-1}\GG^i_n\Big((D_{a,\widehat\theta_{a,-i}}(\tau_n^{+})-D_{a,\widehat\theta_{a,-i}}(\tau_n^{-}))-(D_{a,\theta_{a,\infty}}(\tau_n^{+})-D_{a,\theta_{a,\infty}}(\tau_n^{-}))\Big)\right|\\
&\le \max_i \Big\{ K^2\|\widehat\pi_{a,-i}(\mathbf{X})-\pi_{a,\infty}(\mathbf{X})\|\\
&\qquad \qquad \times (\tau_n^{+}-\tau_n^{-})^{-1}\left(F_a(\tau_n^+)-F_a(\tau_n^-)+\|\phi_{a,\infty}(\tau_n^{+}|\mathbf{X})-\phi_{a,\infty}(\tau_n^{-}|\mathbf{X})\|\right)\\
&\indent + (K+1)\|(\tau_n^{+}-\tau_n^{-})^{-1}(\widehat\phi_{a,-i}(\tau_n^{+}|\mathbf{X})-\widehat\phi_{a,-i}(\tau_n^{-}|\mathbf{X})-(\phi_{a,\infty}(\tau_n^{+}|\mathbf{X})-\phi_{a,\infty}(\tau_n^{-}|\mathbf{X}))\|\Big\}\\
&\indent \rightarrow 0,
\end{split}
\end{equation*}
by Cauchy-Schwarz inequality, Assumptions \ref{assm:phiregularity}, \ref{assm:pregularity}, \ref{assm:L2conditionsforCI.split}, $(\tau_n^{+}-\tau_n^{-})=o_p(1)$ and $K_0=O_p(1)$. This implies that
\begin{align*}
&(\tau_n^{+}-\tau_n^{-})|R_{5n}| = (\tau_n^{+}-\tau_n^{-})^2o_p(n^{-1/2}),
\end{align*}
which is further $(\tau_n^{+}-\tau_n^{-})^2O_p(n^{-k/(2k+1)})$.
Furthermore, since the conclusion for the other term
\begin{align*}
&\left|\int_{\Bar{\tau}}^{\tau_n^{+}} \left[(\widehat F^{K_0}_{a,n}-F_a)(y)-(\widehat F^{K_0}_{a,n}-F_a)(2\Bar{\tau}-y)\right]dy\right|\\
&\le \int_{\Bar{\tau}}^{\tau_n^{+}} \left|(\widehat F^{K_0}_{a,n}-F_a)(y)-(\widehat F^{K_0}_{a,n}-F_a)(2\Bar{\tau}-y)\right|dy\\
&\asymp \left|\left(\tau_n^{+}-\tau_n^{-}\right)\left[(\widehat F^{K_0}_{a,n}-F_a)(\tau_n^{+})-(\widehat F^{K_0}_{a,n}-F_a)(\tau_n^{-})\right]\right| 
\end{align*}
follows by the same reasoning to \eqref{pt2.decom.emp}, \eqref{reference.Remainder.2}, and the preceding derivation above,  
thus we have shown that the conclusion of Lemma \ref{successiveknots} still holds (i.e., $\tau^+_n - \tau^-_n = O_p(n^{-1/(2k+1)})$).

Next, we verify that the conclusion of Lemma \ref{lem:isotonicDID} holds (in the sample splitting setting). It suffices to check that the term  \eqref{lemmaDiD-term0}--\eqref{lemmaDiD-term1} (for the sample splitting estimator $\widehat F_{a,n}$) is $o_p(n^{-(k+1)/(2k+1)})$ (the terms \eqref{lemmaDiD-term2}--\eqref{lemmaDiD-term3}
are unchanged). Indeed, by exploiting the decomposition \eqref{Reference.decom} again, we have
\begin{align*}
&\Big[\Big(\widehat F^{K_0}_{a,n}(t')-F_{a}(t')\Big)-\Big(\widehat F^{K_0}_{a,n}(t)-F_{a}(t)\Big)\Big]\\
&\qquad=\frac{1}{K_0}\sum_{i=1}^{K_0}\left[(\PP^i_n-P_*)\left((D_{a,\widehat\theta_{a,-i}}-D_{a,\theta_{a,\infty}})(t')-(D_{a,\widehat\theta_{a,-i}}-D_{a,\theta_{a,\infty}})(t)\right)\right]\\
&\qquad+\frac{1}{K_0}\sum_{i=1}^{K_0}\left[P_*\left((D_{a,\widehat\theta_{a,-i}}-D_{a,\theta_{a,\infty}})(t')-(D_{a,\widehat\theta_{a,-i}}-D_{a,\theta_{a,\infty}})(t)\right)\right]\\
&\indent+(\PP_n-P_*)\left[D_{a,\theta_{a,\infty}}(t')-D_{a,\theta_{a,\infty}}(t)\right]\\
&:=R_{8n}+R_{9n}+R_{10n}.
\end{align*}
Handling the terms $R_{9n}$ and $R_{10n}$ is done quite similarly as in Lemma~\ref{lem:isotonicDID}; $R_{8n}$ requires some modifications.  To check $R_{10n}$, following  the steps used to prove \eqref{result.in.DiD.empterm} in the proof of Lemma \ref{lem:isotonicDID}, it can be easily shown that
\begin{align*}
 \sup_{t\in\mc W{~\rm and~}t'\in\mc I_t^{K_0}}\Big| &(\PP_n-P_*)\Big(D_{a,\theta_{a,\infty}}(t')-D_{a,\theta_{a,\infty}}(t)\Big)\Big|=o_p(n^{-(k+1)/(2k+1)}), 
\end{align*}
where $\mc W:=\mc S_n \cap \mc E_{s_0,n}$, $\mc I^{K_0}_v:=[v-\kappa^{K_0}_n,v+\kappa^{K_0}_n]$ in which $\kappa^{K_0}_n$ is defined as 
\begin{align*}
\kappa^{K_0}_n:=\sup\{|t-s|:t,s\in[L_n,U_n],s\le t, \widehat F^{K_0}_{a,n}(t)\le \widehat F^{K_0}_{a,n}(s)\},   
\end{align*}
for any $v\in I_{s_0,\omega}$, and $I_{s_0,\omega}$, $\mc E_{s_0,n}$ are defined in Assumption \ref{assm:pregularity}, Lemma \ref{lem:isotonicDID}, respectively.
For $R_{9n}$, one can check the following,
\begin{equation}\label{R12n.reference}
\begin{split}
&\sup_{t\in\mc W{~\rm and~}t'\in\mc I_t^{K_0}}\Bigg|\frac{1}{K_0}\sum_{i=1}^{K_0}\left[P_*\left\{(D_{a,\widehat\theta_{a,-i}}-D_{a,\theta_{a,\infty}})(t')-(D_{a,\widehat\theta_{a,-i}}-D_{a,\theta_{a,\infty}})(t)\right\}\right]\Bigg|\\
&\qquad =o_p(n^{-(k+1)/(2k+1)}),
\end{split}
\end{equation}
by similar reasoning to \eqref{result.in.DiD.remainder} in the proof of Lemma \ref{lem:isotonicDID}, since $K_0=O_p(1)$.

We now check $R_{8n}$. Indeed, from the proof of Lemma \ref{lem:isotonicDID} (see \eqref{result.in.DiD.empterm}) it suffices to show
\begin{align*}
\sup_{t\in\mc W{~\rm and~}t'\in\mc I_t^{K_0}}&\Bigg|\frac{1}{K_0}\sum_{i=1}^{K_0}\left[(\PP^i_n-P_*)\left\{(D_{a,\widehat\theta_{a,-i}}-D_{a,\theta_{a,\infty}})(t')-(D_{a,\widehat\theta_{a,-i}}-D_{a,\theta_{a,\infty}})(t)\right\}\right]\Bigg|\\
&=o_p(n^{-(k+1)/(2k+1)}).
\end{align*}
Multiplying by $\sqrt{n}$, this is equivalent to showing (since $K_0=O_p(1)$) that
\begin{align*}
\sup_{t\in\mc W{~\rm and~}t'\in\mc I_t^{K_0}}&\Bigg|\frac{1}{K_0}\sum_{i=1}^{K_0}\left[\GG_n^i\left\{(D_{a,\widehat\theta_{a,-i}}-D_{a,\theta_{a,\infty}})(t')-(D_{a,\widehat\theta_{a,-i}}-D_{a,\theta_{a,\infty}})(t)\right\}\right]\Bigg|\\
&=o_p(n^{-1/2(2k+1)}).
\end{align*}
Indeed, by the tower property of expectation, 
\begin{align*}
\EE\left[\sup_{f\in \mc F^1_{n,-i}} \left|\mathbb{G}^i_n f\right|\right]=\EE\left\{\EE\left[\sup_{f\in \mc F^1_{n,-i}} \left|\mathbb{G}^i_n f\right|\Big| \mc T_{n,i}\right]\right\},
\end{align*}
where
\begin{align*}
\mc F^1_{n,-i}:=\Bigg\{\Big(D_{a,\widehat\theta_{a,-i}}(t')-D_{a,\theta_{a,\infty}}(t')&-(D_{a,\widehat\theta_{a,-i}}(t)-D_{a,\theta_{a,\infty}}(t)\Big)\\
&:t\in[s_0-Tv_n,s_0+Tv_n],\, |t'-t|\le \kappa_n \Bigg\},    
\end{align*}
which admits an envelope function
\begin{align}\label{envlope.function2}
&H^1_{n,i}(\mathbf{z})\\
&=2K^2\sup_{(t_1,t_2)\in \mc W^1}\Bigg\{ \left|\widehat\pi_{a,-i}(\mathbf{x})-\pi_{a,\infty}(\mathbf{x})\right|\nonumber\\
&\qquad \qquad\qquad\qquad \indent\times I(A=a)\left|I(t_1<Y\le t_2)-(\phi_{a,\infty}(t_2|\mathbf{X})-\phi_{a,\infty}(t_1|\mathbf{X}))\right|\nonumber\\
&\qquad \qquad\qquad\qquad + \left|\left(\widehat\phi_{a,-i}(t_2|\mathbf{X})-\widehat\phi_{a,-i}(t_2|\mathbf{X})\right)-\left(\phi_{a,\infty}(t_1|\mathbf{X})-\phi_{a,\infty}(t_1|\mathbf{X})\right)\right|\Bigg\}\nonumber,
\end{align}
where $\mc W^1:\{(x,y):|y-s_0|\le Tv_n,\,|y-x|\le \kappa_n^{K_0}\}$.
By Theorem 2.14.1 in \citet{vdvandW} (see Proposition \ref{prop:vdv.2.14.1}), for sufficiently large $n$, we have
\begin{align*}
\EE\left[n^{1/2(2k+1)}\sup_{f\in \mc F^1_{n,-i}}\left|\mathbb{G}^i_n f\right|\Big| \mc T_{n,i}\right]\le C_*n^{1/2(2k+1)}\|H^1_{n,i}\|_{P_*,2}J(1,\mc F^1_{n,-i}).    
\end{align*}
Similarly to $J(1,\mc F_{n-i})$ in the proof of Theorem \ref{prop:Cons-split}, since here  $\{ Y \mapsto I(t_1<Y \le t_2) : (t_1,t_2)\in \mc W^1 \}$ is VC, we have uniformly bounded $J(1,\mc F^1_{n,-i})$ for all $n$ and $-i$. Furthermore, we have
\begin{align*}
&n^{1/2(2k+1)}\max _i \|H^1_{n,i}\| _{P_*,2}\\
&\lesssim \max_{i}\Bigg\{ \| \widehat{\pi}_{a,-i}(\mathbf{X})-{\pi}_{a,\infty}(\mathbf{X})\|_{P_*,2}\\
&\qquad \times n^{1/2(2k+1)}\left\|\sup_{(t_1,t_2)\in \mc W^1}\Big[(\phi_{a,*}(t_2|\mathbf{X})-\phi_{a,*}(t_1|\mathbf{X}))-(\phi_{a,\infty}(t_2|\mathbf{X})-\phi_{a,\infty}(t_1|\mathbf{X}))\Big]\right\|_{P_*,2}\\
&\indent + n^{1/2(2k+1)}  \left\|\sup_{(t_1,t_2)\in \mc W^1}\Big[(\widehat\phi_{a,-i}(t_2|\mathbf{X})-\widehat\phi_{a,-i}(t_2|\mathbf{X}))-(\phi_{a,\infty}(t_1|\mathbf{X})-\phi_{a,\infty}(t_1|\mathbf{X}))\Big]\right\|_{P_*,2}\Bigg\}\\
&\qquad \rightarrow0,
\end{align*}
since Assumption \ref{assm:L2conditionsforCI.split} holds, and one can check $\kappa_n^{K_0}=O_p(n^{-(k+1)/(2k+1)})$ following the same reasoning in the proof of Lemma \ref{lem:monotone.correction.ext}. Hence we have shown that the conclusion of Lemma \ref{lem:isotonicDID} continues to hold in the sample splitting setting.

Lastly, we check that the term $\widehat G^{\rm locmod}_{a,n}$ has the same asymptotics as in the non sample splitting case. Since it was already shown that
\begin{align*}
\widehat G^{\rm locmod}_{a,n}(t)&=\frac{r_n}{p_a(s_0)}\int_{s_0}^{s_n(t)}\Big((\widehat F^{K_0}_{a,n}(v)-\widehat F^{K_0}_{a,n}(s_0))-(F_a(v)-F_a(s_0))\Big)dv \\
&\indent +\frac{\varphi_a^{(k)}(s_0)}{(k+2)!}t^{k+2}+o_p(1),
\end{align*}
from \eqref{Glocmod} and the sentence (and display) that follows,  again by the decomposition \eqref{Reference.decom} we have
\begin{align*}
&\int_{s_0}^{s_n(t)}\Big((\widehat F^{K_0}_{a,n}(v)-\widehat F^{K_0}_{a,n}(s_0))-(F_a(v)-F_a(s_0))\Big)dv\\   
&=\frac{1}{K_0}\sum_{i=1}^{K_0}\int_{s_0}^{s_n(t)}\left[(\PP^i_n-P_*)\left\{\left(D_{a,\widehat\theta_{a,-i}}(v)-D_{a,\theta_{a,\infty}}(v)\right)-\left(D_{a,\widehat\theta_{a,-i}}(s_0)-D_{a,\theta_{a,\infty}}(s_0)\right)\right\}\right]dv\\
&\indent + \frac{1}{K_0}\sum_{i=1}^{K_0}\int_{s_0}^{s_n(t)}\left[P_*\left\{\left(D_{a,\widehat\theta_{a,-i}}(v)-D_{a,\theta_{a,\infty}}(v)\right)-\left(D_{a,\widehat\theta_{a,-i}}(s_0)-D_{a,\theta_{a,\infty}}(s_0)\right)\right\}\right]dv\\
&\indent + \frac{r_n}{p_a(s_0)}\int_{s_0}^{s_n(t)} (\PP_n-P_*)\left[D_{a,\theta_{a,\infty}}(v)-D_{a,\theta_{a,\infty}}(s_0)\right]dv\\
&:=R_{11n}+R_{12n}+R_{13n}.
\end{align*}
We already verified the limiting behavior of $R_{13n}$ in the proof of Lemma \ref{lem:finallemmaforCI}. For the term $R_{12n}$, we can see
\begin{equation}
\begin{split}
\sup_{v\in[s_0-Tv_n+s_0+Tv_n]}&\Bigg|\frac{1}{K_0}\sum_{i=1}^{K_0}\left[P_*\left\{\left((D_{a,\widehat\theta_{a,-i}}-D_{a,\theta_{a,\infty}})(t')-(D_{a,\widehat\theta_{a,-i}}-D_{a,\theta_{a,\infty}})(t)\right)\right.\right.\\
&\qquad \indent \left.\left. -\left((D_{a,\widehat\theta_{a,-i}}-D_{a,\theta_{a,\infty}})(s')-(D_{a,\widehat\theta_{a,-i}}-D_{a,\theta_{a,\infty}})(s)\right)\right\}\right]\Bigg|\\
&\qquad =o_p(n^{-(k+1)/(2k+1)})
\end{split}
\end{equation}
by an identical argument to that leading to \eqref{R12n.reference},
since $v_n=n^{-1/(2k+1)}$.  Thus, it suffices to show that 
\begin{align*}
\sup_{v\in[s_0-Tv_n+s_0+Tv_n]}\left|R_{11n}\right| =o_p(n^{-(k+1)/(2k+1)}).   
\end{align*}
The same argument we used above to show $R_{8n} = o_p(n^{-(k+1)/(2k+1)})$ can be applied here to $R_{11n}$, except that we use $v_n$ instead of $\kappa_n^{K_0}$ for the width of the shrinking interval.  This concludes the proof.
\end{proof}

\section{Further results and materials for the simulations}
\subsection{Explicit forms of counterfactual density $p_a$ in Section \ref{Section:Simul}}\label{section:pas}
Here we give the exact density, $ir_4(\cdot)$, of the sum of four i.i.d. standard uniform random variables as follows.
\begin{align}
ir_4(x)=\begin{cases}
    0 & x<0\\
    \frac{x^3}{6} & 0\le x<1\\
\frac{1}{6} \left(-3x^3+12x^2-12x+4\right)& 1\le x<2 \\
\frac{1}{6}\left(3x^3-24x^2+60x-44\right) & 2\le x<3\\
 \frac{(4-x)^3}{6} & 3\le x<4 \\
 0 & x\ge 4.
\end{cases}
\end{align}
Then, $p_1$ and $p_0$ are given by
\begin{align}
& p_1(y)=I_{(4,12)}(y)\int_{(y-4)/2}^{4} \frac{1}{2x} ir_4(x)dx ,\label{p1form}\\
& p_0(y)=I_{(0,8)}(y)\int_{(8-y)/2}^{4} \frac{1}{2x} ir_4(x)dx\label{p0form},
\end{align}
hence, the support of $p_1$ and $p_0$ are $(4,12)$ and $(0,8)$, respectively. Furthermore, we have, for a random variable $S$ which follows the density $ir_4$,
\begin{equation}\label{mean.p1andp0}
\begin{split}
\EE\left[Y^1\right]=\EE\left[S+4\right]=6,\\
\EE\left[Y^0\right]=\EE\left[8-S\right]=6.
\end{split}   
\end{equation}

\subsection{Estimating nuisance functions}\label{subsec:est.nuisances}
To construct a well-specified estimator  $\widehat\pi_a$, we use a correctly specified logistic regression model. 
To construct a mis-specified $\widehat\pi_a$, we omit $(X_{i1},X_{i3})$ from the regression model.
To construct a well-specified $\widehat\phi_1$, we first estimate a correctly specified linear regressions of $Y_i$ on $\mathbf{X}_i$ among $i$ with $A_i = 1$, and we then set  $\hat\mu_1(\mathbf{x})$ equal to the maximum of the prediction from this regression at $\mathbf{x}$ and 4 for any $\mathbf{x}$. We then set $\widehat\phi_1(\cdot | \mathbf{x})$ as the CDF of $U[4,2\hat\mu_1(\mathbf{x})-4]$. Similarly, we construct a well-specified $\widehat\phi_0$ by estimating a linear regression of $Y_i$ on $\mathbf{X}_i$ among $i$ with $A_i = 0$, setting  $\hat\mu_0(\mathbf{x})$ equal to the minimum of the prediction from this regression at $\mathbf{x}$ and 8, and setting $\widehat\phi_0(\cdot | \mathbf{x})$ as the CDF of $U[2\hat\mu_0(\mathbf{x})-8,8]$.  To construct mis-specified $\widehat\phi_a$, we use the same procedure as above, but omit $(X_{i1},X_{i4})$ from the linear regression steps.

\subsection{Details for constructing projection estimator by \citet{KBWcounterfactualdensity}}\label{subsec:est.KBW}
To facilitate a fair comparison with our estimators, we estimate $\hat\mu_a$ in \citet{KBWcounterfactualdensity} with $\int b(t) \widehat\phi_a(t|\mathbf{x})dt$ for $\widehat\phi_a$ defined above for each basis function $b$. 
We truncate negative density values to $0$ and normalized the function to integrate to $1$ on its support $[\min_i\{Y_i\},\max_i\{Y_i\}]$. The method also requires selecting the number of basis functions to use.
We select this tuning parameter in an oracle fashion: we select the number of basis functions 
as the one among the set $\{2,3,\ldots,30\}$ that achieved the lowest average $L_1$ distance between the estimated and true counterfactual densities for each sample size $n$.

\subsection{Generating quantiles used in difference/log-ratio CI's}\label{subsec:quantile.generation}
Note that $k=2$ in our simulations.
We use the same scaling factor estimates as those employed in constructing the pointwise CI's.
Here we outline the exact procedure for obtaining the upper $1-\alpha$ quantiles used in the difference CI's (see \eqref{def:contrast.CI} and \eqref{def:contrast.CI.split}) as well as log-ratio CI's (see \eqref{def:ratio.CI} and \eqref{def:ratio.CI.split}).
The methodology follows the approach described in Appendix B.1 of \citet{deng2022inference}. 
Specifically, we simulate $\LL_2^{(0)}$ using convex regression, with the true function $f_0(x)=12(x-1/2)^2$ and the target point set to $x_0=1/2$.
The design points are evenly spaced as $x_i=i/n$ for $1\leq i \leq n$, with $n=10^6$.
We conducted $B=2\times 10^4$ independent replications to sample $\LL_2^{(0)}$.
And, we split the samples equally into two subsets: the first half for $\LL_{12}^{(0)}$ and the second half for $\LL_{02}^{(0)}$, each with a length of $B/2$.
This fixed pair $\LL_{12}^{(0)}$ and $\LL_{02}^{(0)}$ was used consistently across all simulations.
With the estimated scaling factors, we then numerically compute a vector of length $B/2$, given by:
\begin{align*}
|\sqrt{\hat \chi_{\theta_{0}}/(n\Delta\tau_{0,n})}\LL_{02}^{(0)}-\sqrt{\hat \chi_{\theta_{1}}/(n\Delta\tau_{1,n})}\LL_{12}^{(0)}|,
\end{align*}
or
\begin{align*}
|\sqrt{\hat \chi^{K_0}_{\theta_{0}}/(n\Delta\tau^{K_0}_{0,n})}\LL_{02}^{(0)}-\sqrt{\hat \chi^{K_0}_{\theta_{1}}/(n\Delta\tau^{K_0}_{1,n})}\LL_{12}^{(0)}|,
\end{align*}
for each method with and without sample splitting.
Finally the upper $1-\alpha$ quantile of this computed vector is used for each difference CI's.
The log-ratio CI's can be constructed by using upper $1-\alpha$ quantile of 
\begin{align*}
\left|\left(\sqrt{\hat \chi_{\theta_{0}}/(n\Delta\tau_{0,n})}/\widehat p_{0,n}(s_0)\right)\LL_{02}^{(0)}-\left(\sqrt{\hat \chi_{\theta_{1}}/(n\Delta\tau_{1,n})}/\widehat p_{1,n}(s_0)\right)\LL_{12}^{(0)}\right|,
\end{align*}
or
\begin{align*}
\left|\left(\sqrt{\hat \chi^{K_0}_{\theta_{0}}/(n\Delta\tau^{K_0}_{0,n})}/\widehat p_{0,n}^{K_0}(s_0)\right)\LL_{02}^{(0)}-\left(\sqrt{\hat \chi^{K_0}_{\theta_{1}}/(n\Delta\tau^{K_0}_{1,n})}/\widehat p_{1,n}^{K_0}(s_0)\right)\LL_{12}^{(0)}\right|.
\end{align*}

\subsection{Table for projection method in Section \ref{Section:Simul}}\label{appd:basis.table}
\begin{table}[h!]
\centering
\resizebox{\columnwidth}{!}{%
\begin{tabular}{ |p{2cm}||p{2cm}|p{2cm}|p{2cm}||p{2cm}|p{2cm}|p{2cm}|  }
 \hline
 \multicolumn{7}{|c|}{Number of basis functions selected by the projection basis method} \\
 \hline
 Sample size& Case 1; $p_1$ & Case 2; $p_1$ &Case 3; $p_1$ & Case 1; $p_0$ & Case 2; $p_0$ &Case 3; $p_0$\\
 \hline
 $n=500$ & 8 & 9 & 8 & 10 & 10 & 10\\
 $n=1000$ & 13 & 13 & 13 & 14 & 14 & 14\\
 $n=2500$ & 19 & 19 & 19 & 21 & 21 & 21\\
 $n=4000$ & 26 & 24 & 23 & 27 & 28 & 28\\
 $n=6000$ & 30 & 29 & 30 & 30 & 30 & 30\\
 $n=8000$ & 30 & 30 & 30 & 30 & 30 & 30\\
 \hline
\end{tabular}
}
\caption{The above table exhibits the selected (oracle) number of basis functions for the projection basis method by \citet{KBWcounterfactualdensity} for each setting.}
\label{table:basis.number}
\end{table}

\clearpage

\subsection{Tuning parameter selection in Section \ref{Section:Simul}}\label{appd:tuning.param.section}
For the 95\% CI construction, we considered the tuning parameter selection from the set $b\in\{1/25,1/20,1/15,1/10,1/5\}$ for scalar factor estimation procedure (see Section \ref{subsec:tuning.param}). 
When both nuisance functions are well-specified (Case 1), since the true value of $\chi_{\theta_a}$ is available through the true limit nuisance functions, we also compute the oracle 95\% confidence intervals with the true scalar factor.

In Figure \ref{fig:kplots.case.1.p1}--\ref{fig:kplots.case.3.p0} below, we display the coverage probabilities from our 95\% pointwise confidence intervals constructed from \eqref{chi.theta.hat}, \eqref{chi.theta.hat.split} for each original and cross-fitted estimator with aforementioned five different tuning parameters.
In general, within our estimators, coverage probabilities in each tail for $p_1,\,p_0$  where each mode locates are relatively high compared with the other points for all $n$ values.
However, as $n$ grows, the high coverage tendency reduces to relatively low levels.
The oracle estimator's coverage was the highest among all candidates, in general.
Among the possible tuning parameters in the set $\{1/25,1/20,1/15,1/10,1/5\}$, $b=1/5$ performed relatively worse than the others at the other side of the tail where each density is close to $0$ in low sample sizes such as $n=500,1000$. 
As discussed in Section \ref{Section:Simul}, conservative coverage of oracle tuned CI attributed to the discovery in \citet{deng2022inference} (see Figure \ref{fig:kplots.case.1.p1}, \ref{fig:kplots.case.1.p0}).
We suggest $b=1/10$ for our tuning parameter selection.

\begin{figure}
    \centering
    \begin{subfigure}[$p_1$; C1; $n=500$]{\label{500WW1}\includegraphics[width=50mm]{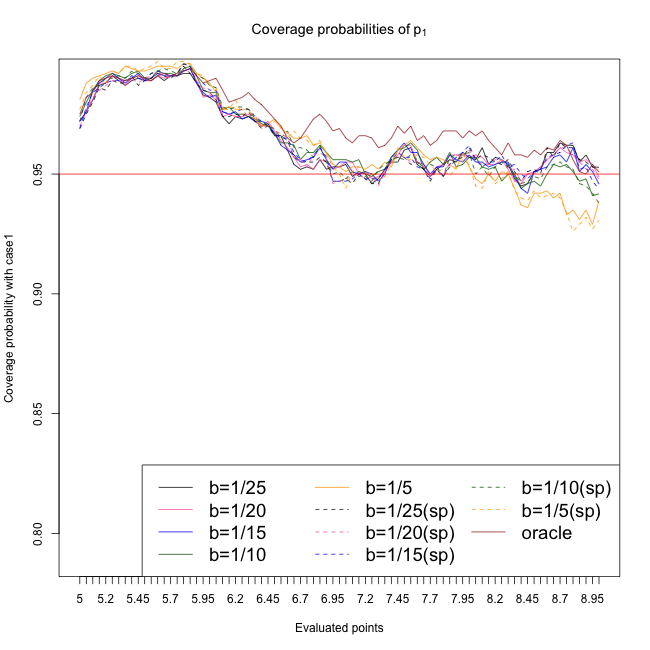}}
   \end{subfigure}
    \hspace{0.45cm}
    \begin{subfigure}[$p_1$; C1; $n=1000$]
    {\label{1000WW1}\includegraphics[width=50mm]{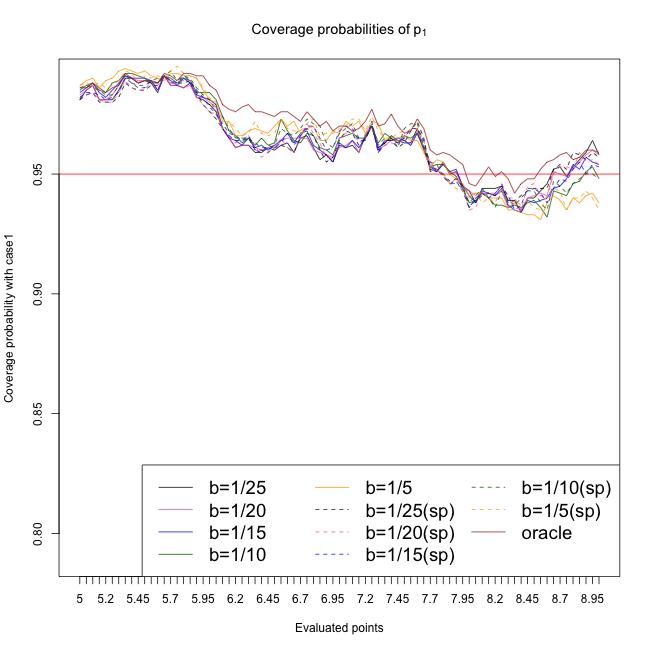}}
    \end{subfigure}\\
    \begin{subfigure}[$p_1$; C1; $n=2500$]
    {\label{2500WW1}\includegraphics[width=50mm]{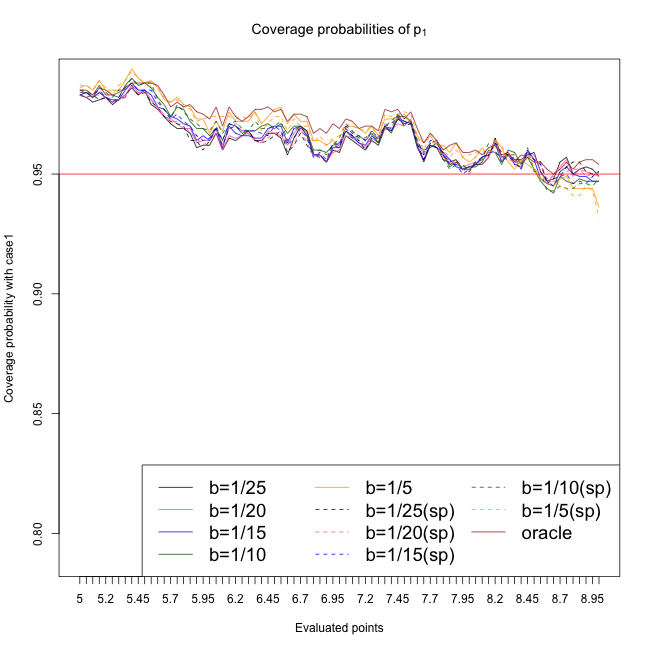}}
    \end{subfigure}
    \hspace{0.45cm}
    \begin{subfigure}[$p_1$; C1; $n=4000$]{\label{4000WW1}\includegraphics[width=50mm]{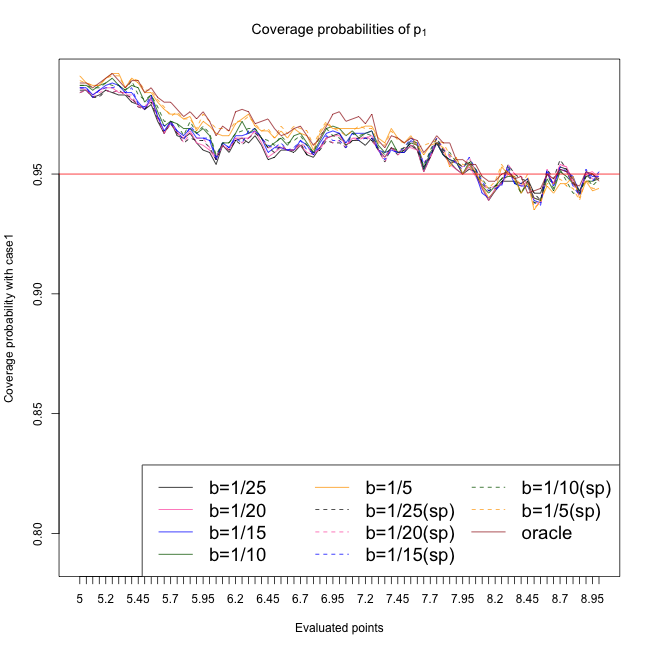}}
   \end{subfigure}\\
    \begin{subfigure}[$p_1$; C1; $n=6000$]
    {\label{6000WW1}\includegraphics[width=50mm]{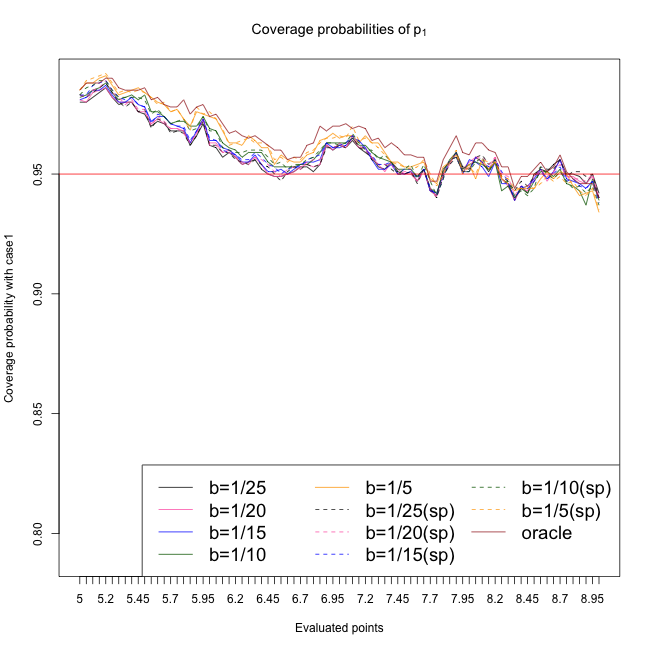}}
    \end{subfigure}
    \hspace{0.45cm}
    \begin{subfigure}[$p_1$; C1; $n=8000$]
    {\label{8000WW1}\includegraphics[width=50mm]{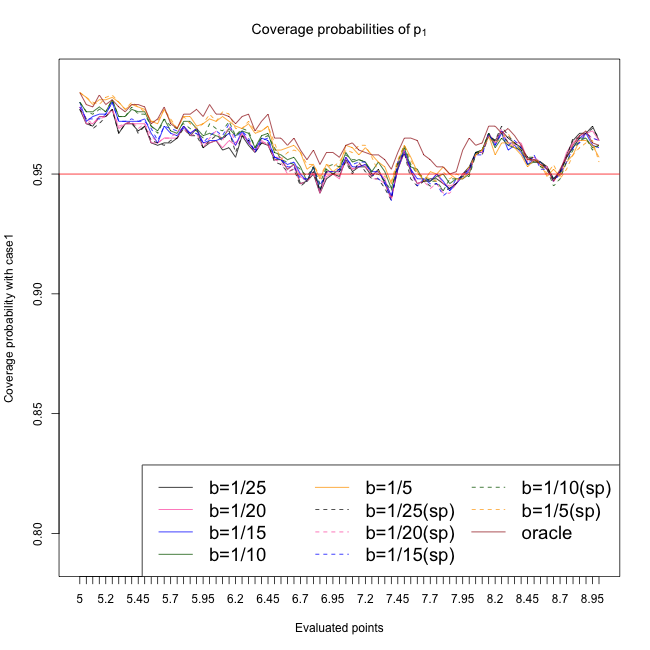}}
    \end{subfigure}
    \caption{\label{fig:kplots.case.1.p1}The above displays are coverage probabilities of our proposed log-concave projection estimator's 95\% CI's (labeled as $b=1/m$ where each $b$ means that we used $h_n=n^{-b}$ for approximating $\chi_{\theta_a}$ (see Section \ref{subsec:tuning.param}).
    For the Case 1 where both nuisance functions are well-specified, we also use true value of $\chi_{\theta_a}$ to construct the oracle 95\% CI which is labeled as oracle in the displays.
    ``sp" stands for our sample splitting based estimator (see \eqref{crossfitted.onestep}).
    The coverage probabilities are measured in 81 equally spaced points in each domain. 
    Each subcaption describes the estimation target ($p_1$ or $p_0$), the sample size, and each case of nuisance estimations (Case 1, 2, or 3 abbreviated to C1, C2, and C3, respectively).}
\end{figure}

\begin{figure}
    \centering
    \begin{subfigure}[$p_1$; C2; $n=500$]{\label{500WM1}\includegraphics[width=50mm]{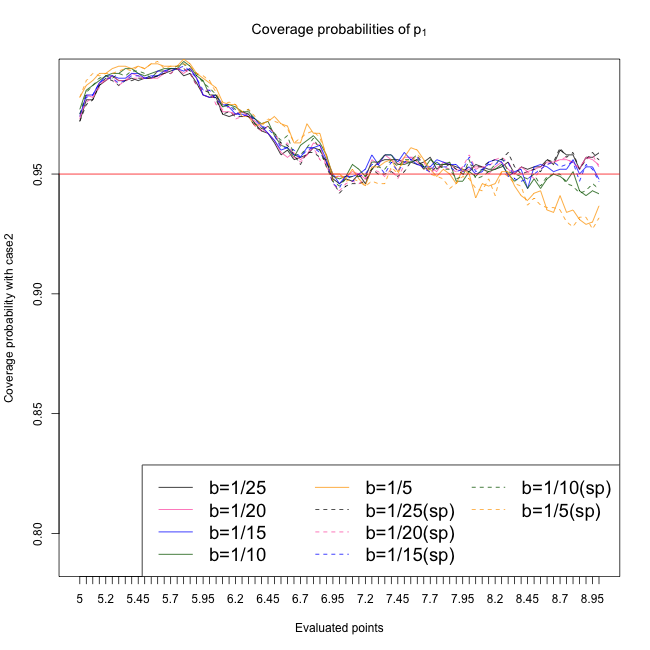}}
   \end{subfigure}
    \hspace{0.45cm}
    \begin{subfigure}[$p_1$; C2; $n=1000$]
    {\label{1000WM1}\includegraphics[width=50mm]{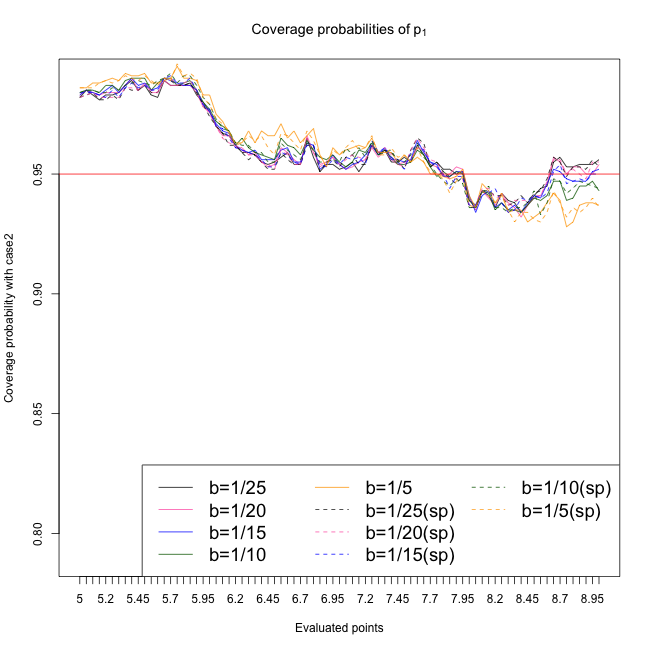}}
    \end{subfigure}\\
    \begin{subfigure}[$p_1$; C2; $n=2500$]
    {\label{2500WM1}\includegraphics[width=50mm]{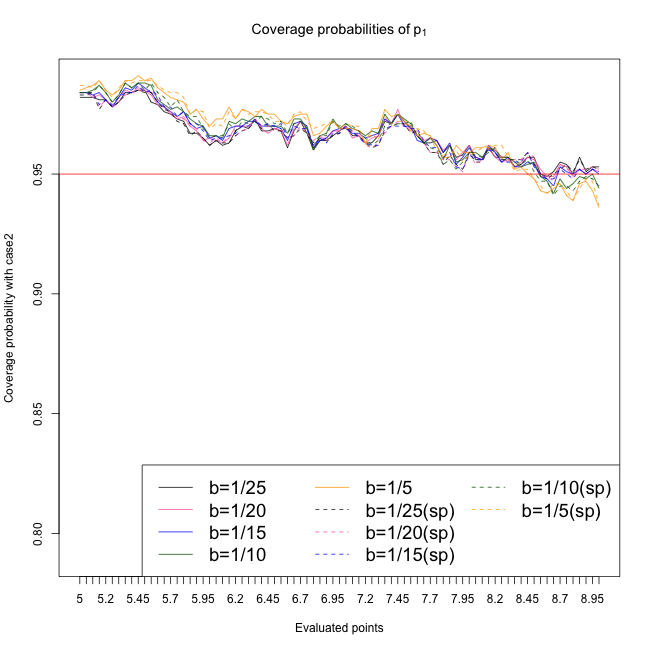}}
    \end{subfigure}
    \hspace{0.45cm}
    \begin{subfigure}[$p_1$; C2; $n=4000$]{\label{4000WM1}\includegraphics[width=50mm]{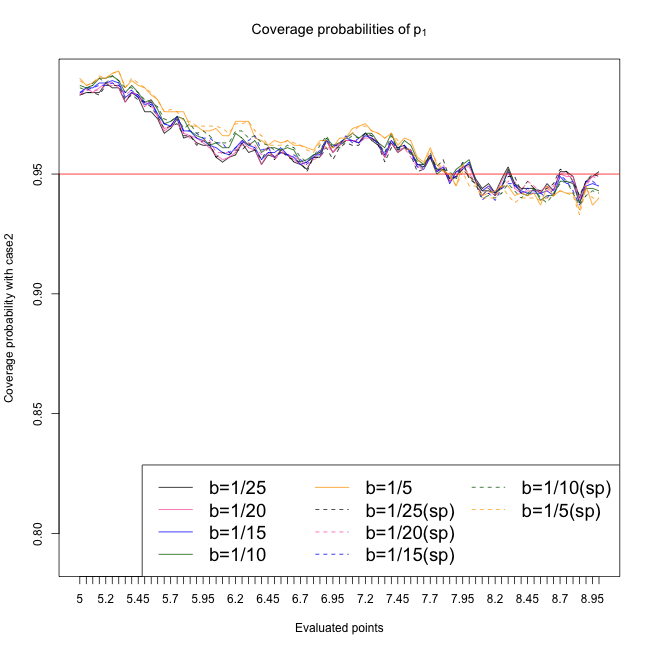}}
   \end{subfigure}\\
    \begin{subfigure}[$p_1$; C2; $n=6000$]
    {\label{6000WM1}\includegraphics[width=50mm]{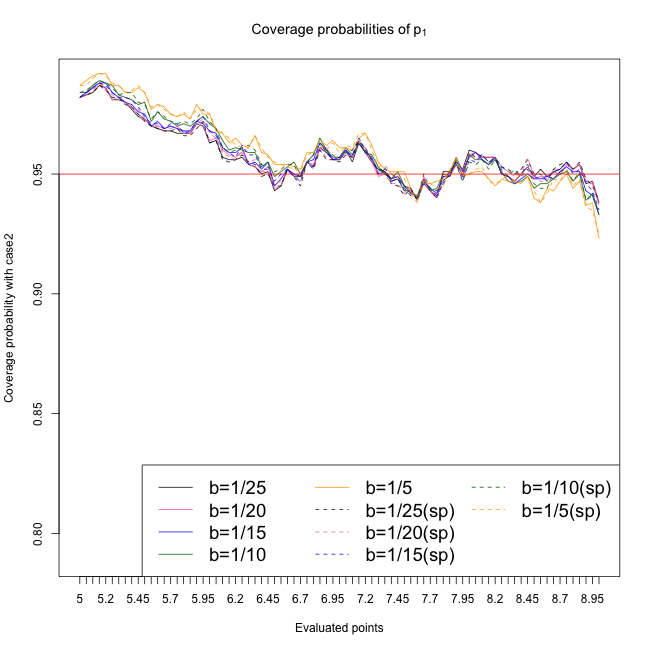}}
    \end{subfigure}
    \hspace{0.45cm}
    \begin{subfigure}[$p_1$; C2; $n=8000$]
    {\label{8000WM1}\includegraphics[width=50mm]{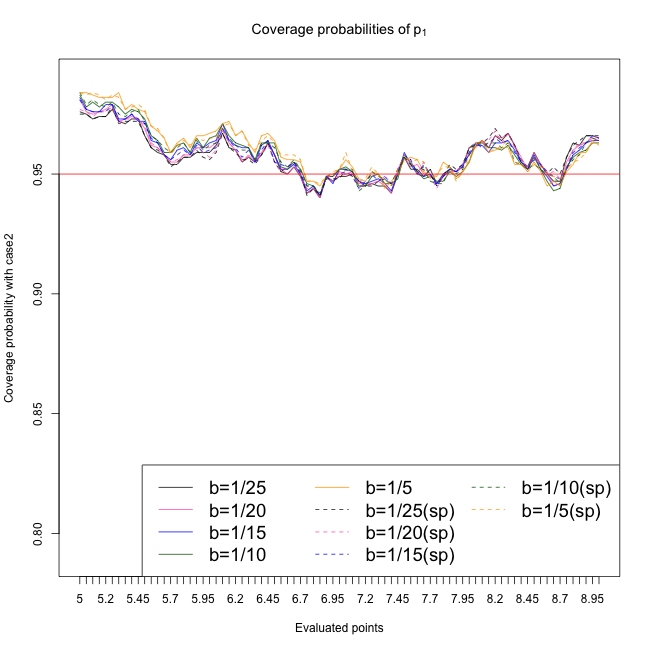}}
    \end{subfigure}
    \caption{\label{fig:kplots.case.2.p1} Notational details can be found in Figure \ref{fig:kplots.case.1.p1}.}
\end{figure}

\begin{figure}
    \centering
    \begin{subfigure}[$p_1$; C3; $n=500$]{\label{500MW1}\includegraphics[width=50mm]{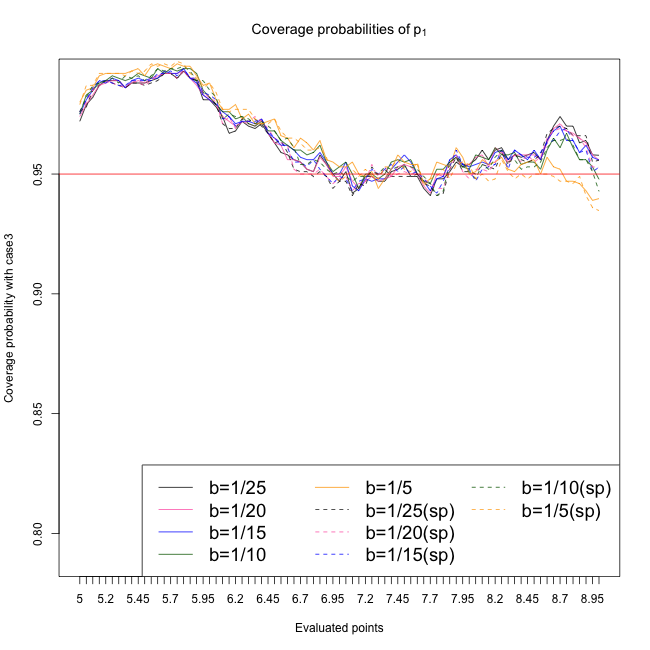}}
   \end{subfigure}
    \hspace{0.45cm}
    \begin{subfigure}[$p_1$; C3; $n=1000$]
    {\label{1000MW1}\includegraphics[width=50mm]{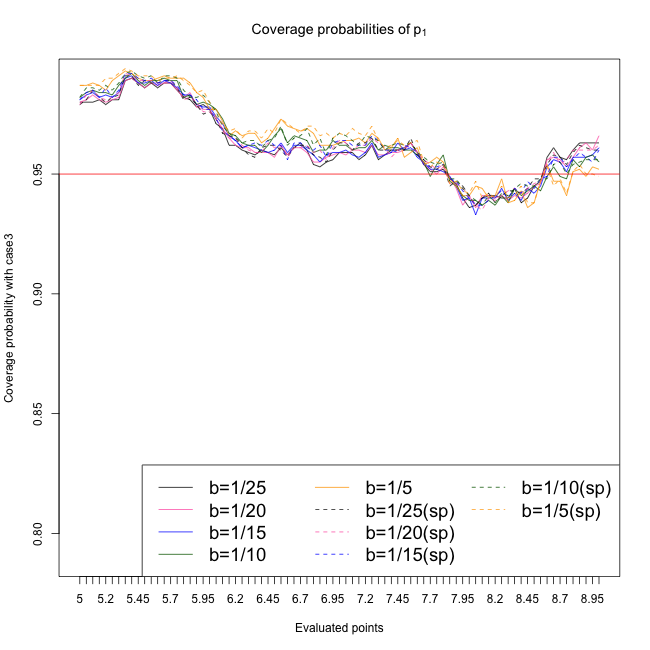}}
    \end{subfigure}\\
    \begin{subfigure}[$p_1$; C3; $n=2500$]
    {\label{2500MW1}\includegraphics[width=50mm]{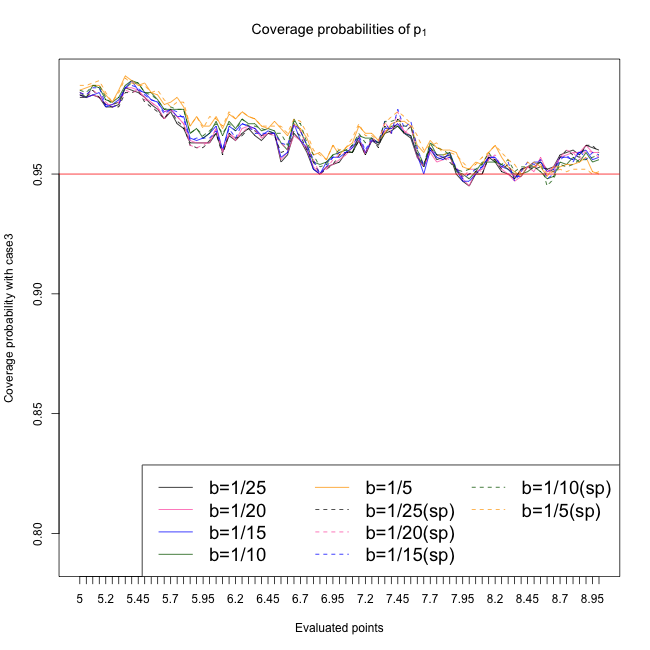}}
    \end{subfigure}
    \hspace{0.45cm}
    \begin{subfigure}[$p_1$; C3; $n=4000$]{\label{4000MW1}\includegraphics[width=50mm]{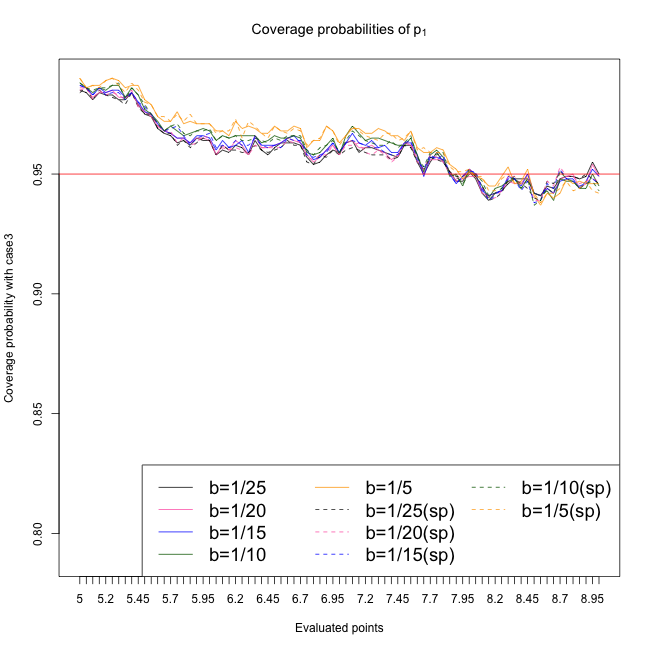}}
   \end{subfigure}\\
    \begin{subfigure}[$p_1$; C3; $n=6000$]
    {\label{6000MW1}\includegraphics[width=50mm]{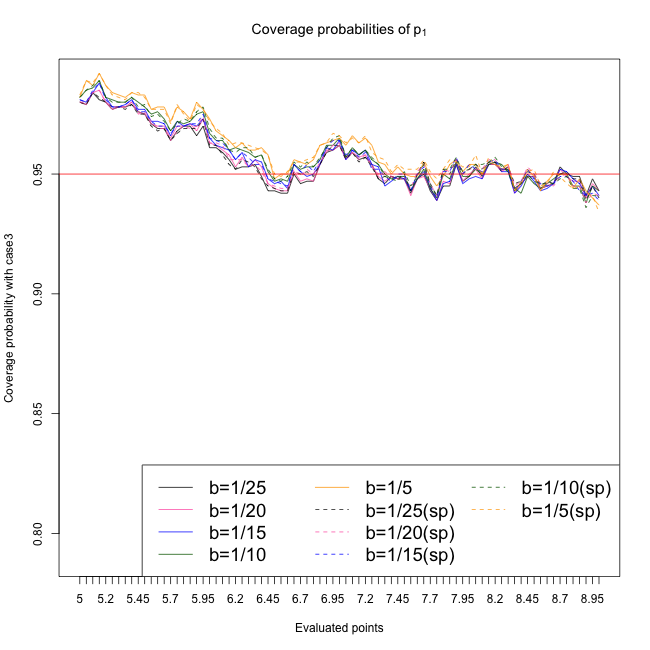}}
    \end{subfigure}
    \hspace{0.45cm}
    \begin{subfigure}[$p_1$; C3; $n=8000$]
    {\label{8000MW1}\includegraphics[width=50mm]{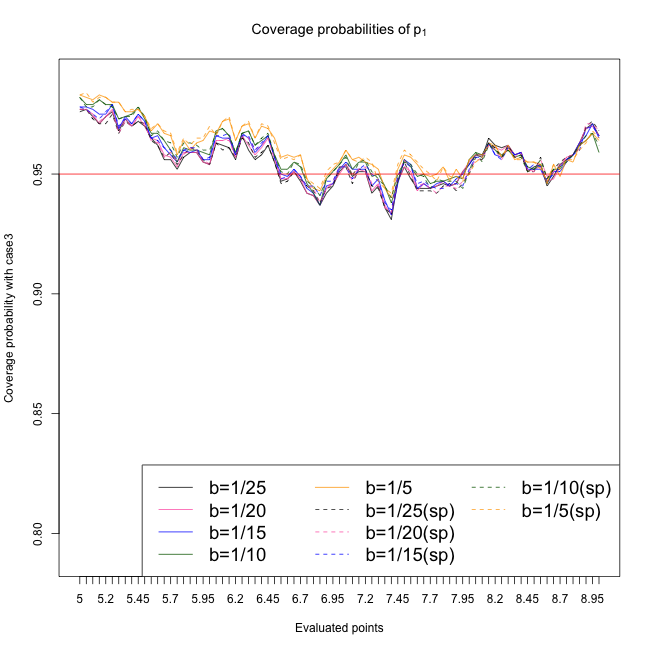}}
    \end{subfigure}
    \caption{\label{fig:kplots.case.3.p1}  Notational details can be found in Figure \ref{fig:kplots.case.1.p1}.}
\end{figure}

\begin{figure}
    \centering
    \begin{subfigure}[$p_0$; C1; $n=500$]{\label{500WW0}\includegraphics[width=50mm]{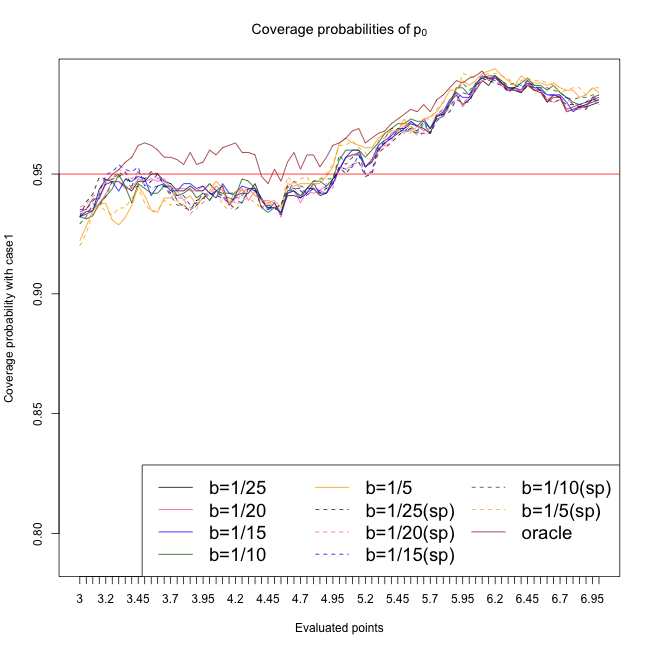}}
   \end{subfigure}
    \hspace{0.45cm}
    \begin{subfigure}[$p_0$; C1; $n=1000$]
    {\label{1000WW0}\includegraphics[width=50mm]{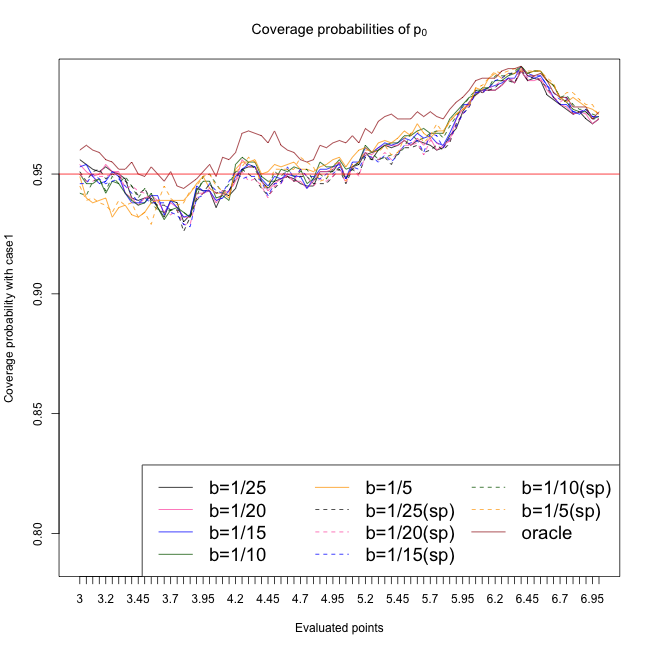}}
    \end{subfigure}\\
    \begin{subfigure}[$p_0$; C1; $n=2500$]
    {\label{2500WW0}\includegraphics[width=50mm]{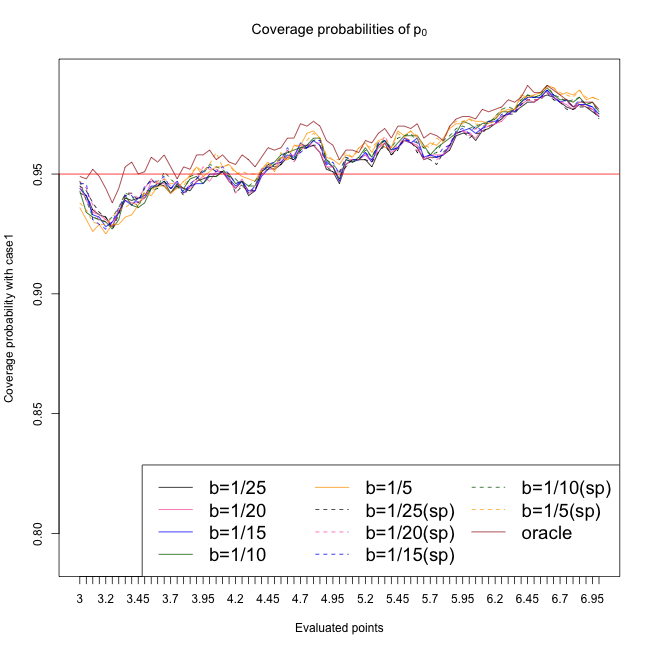}}
    \end{subfigure}
    \hspace{0.45cm}    
    \begin{subfigure}[$p_0$; C1; $n=4000$]{\label{4000WW0}\includegraphics[width=50mm]{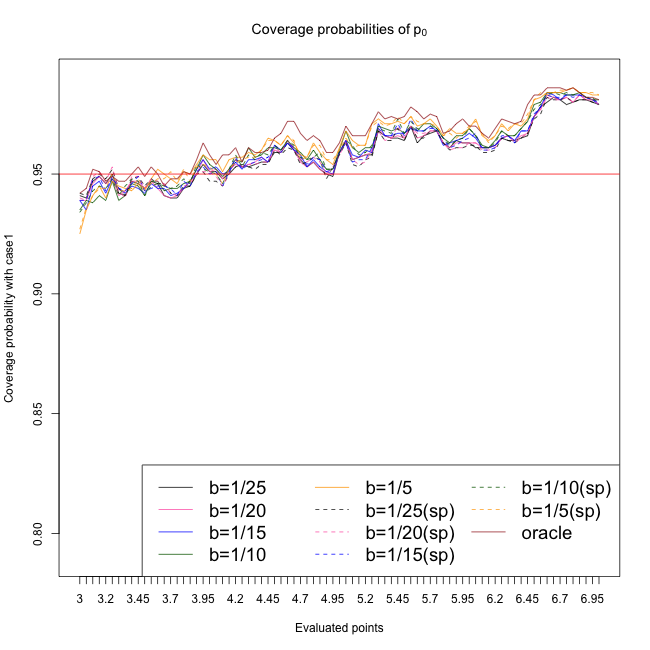}}
   \end{subfigure}\\
    \begin{subfigure}[$p_0$; C1; $n=6000$]
    {\label{6000WW0}\includegraphics[width=50mm]{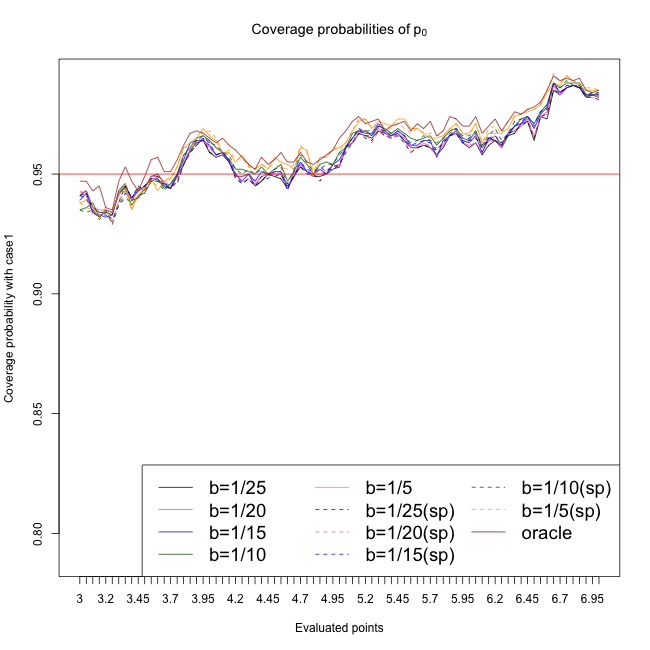}}
    \end{subfigure}
    \hspace{0.45cm}
    \begin{subfigure}[$p_0$; C1; $n=8000$]
    {\label{8000WW0}\includegraphics[width=50mm]{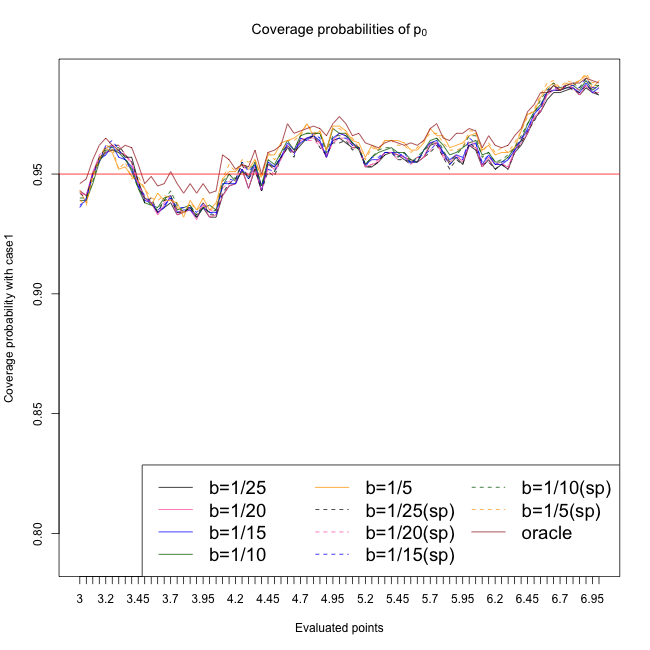}}
    \end{subfigure}
    \caption{\label{fig:kplots.case.1.p0} Notational details can be found in Figure \ref{fig:kplots.case.1.p1}.}
\end{figure}

\begin{figure}
    \centering
    \begin{subfigure}[$p_0$; C2; $n=500$]{\label{500WM0}\includegraphics[width=50mm]{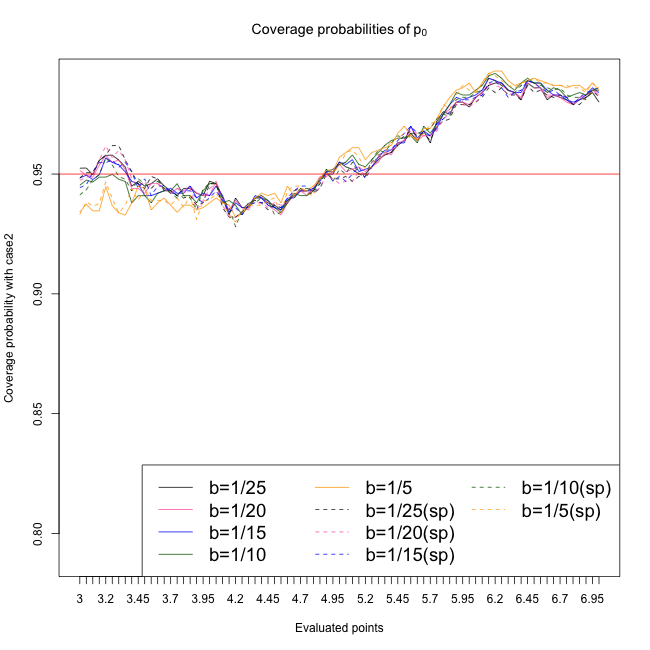}}
   \end{subfigure}
    \hspace{0.45cm}
    \begin{subfigure}[$p_0$; C2; $n=1000$]
    {\label{1000WM0}\includegraphics[width=50mm]{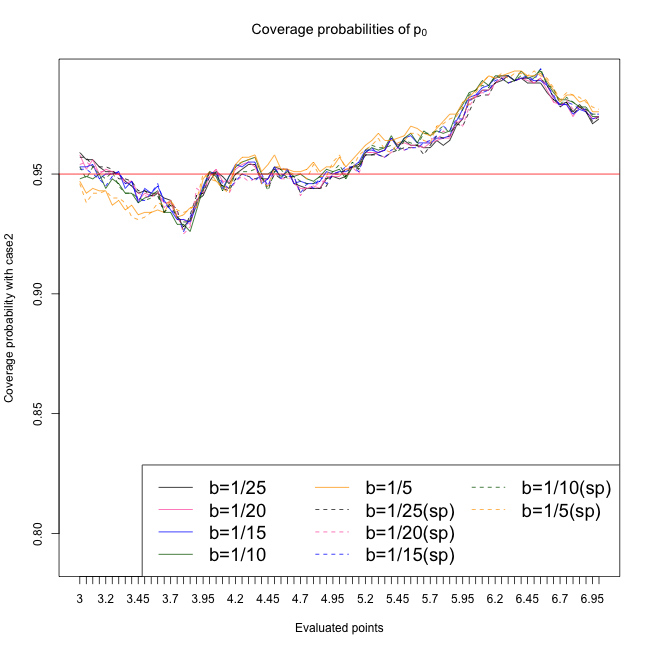}}
    \end{subfigure}\\
    \begin{subfigure}[$p_0$; C2; $n=2500$]
    {\label{2500WM0}\includegraphics[width=50mm]{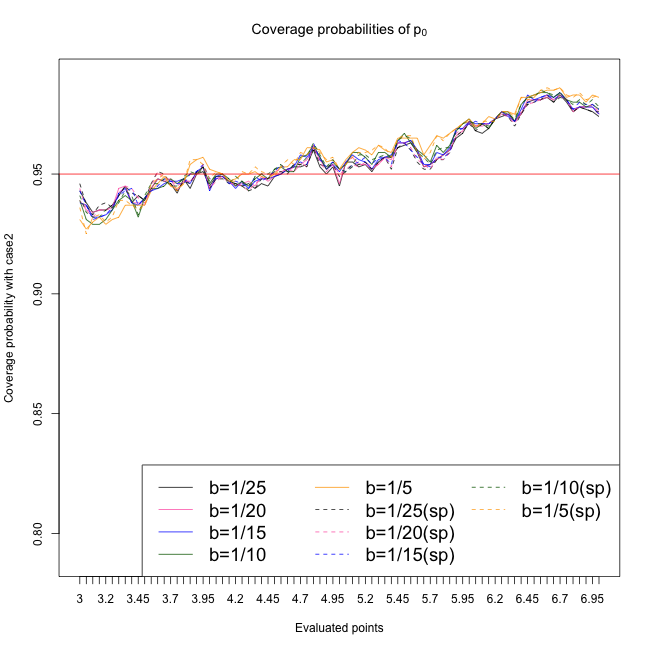}}
    \end{subfigure}
    \hspace{0.45cm}    
    \begin{subfigure}[$p_0$; C2; $n=4000$]{\label{4000WM0}\includegraphics[width=50mm]{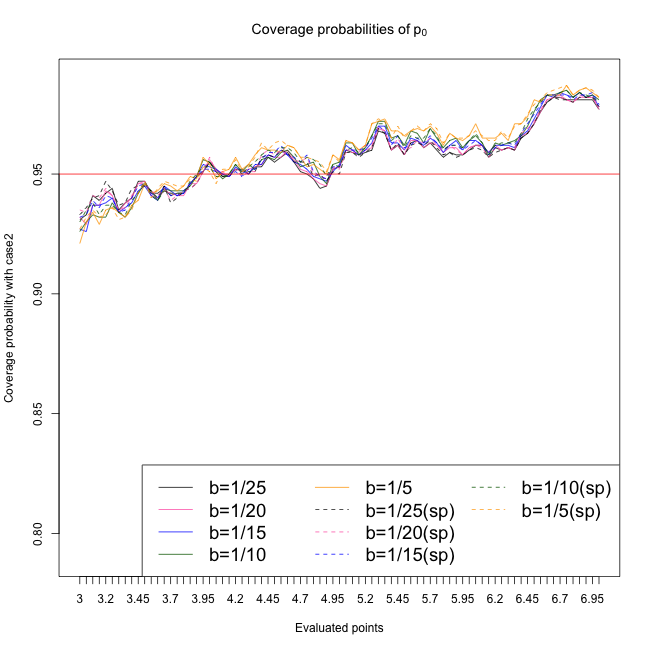}}
   \end{subfigure}\\
    \begin{subfigure}[$p_0$; C2; $n=6000$]
    {\label{6000WM0}\includegraphics[width=50mm]{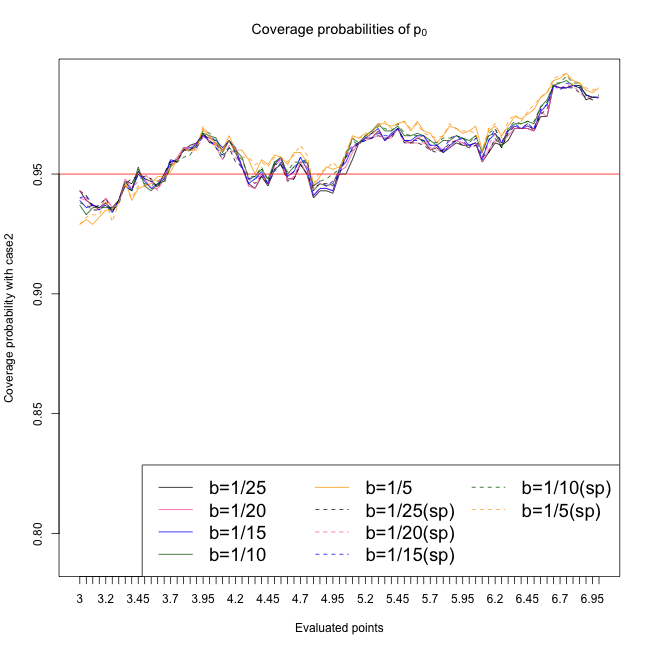}}
    \end{subfigure}
    \hspace{0.45cm}
    \begin{subfigure}[$p_0$; C2; $n=8000$]
    {\label{8000WM0}\includegraphics[width=50mm]{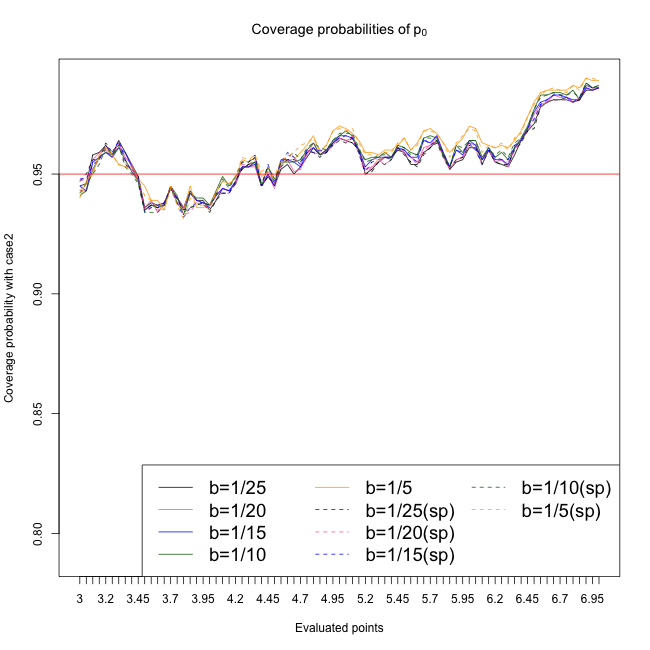}}
    \end{subfigure}
    \caption{\label{fig:kplots.case.2.p0} Notational details can be found in Figure \ref{fig:kplots.case.1.p1}.}
\end{figure}

\begin{figure}
    \centering
    \begin{subfigure}[$p_0$; C3; $n=500$]{\label{500MW0}\includegraphics[width=50mm]{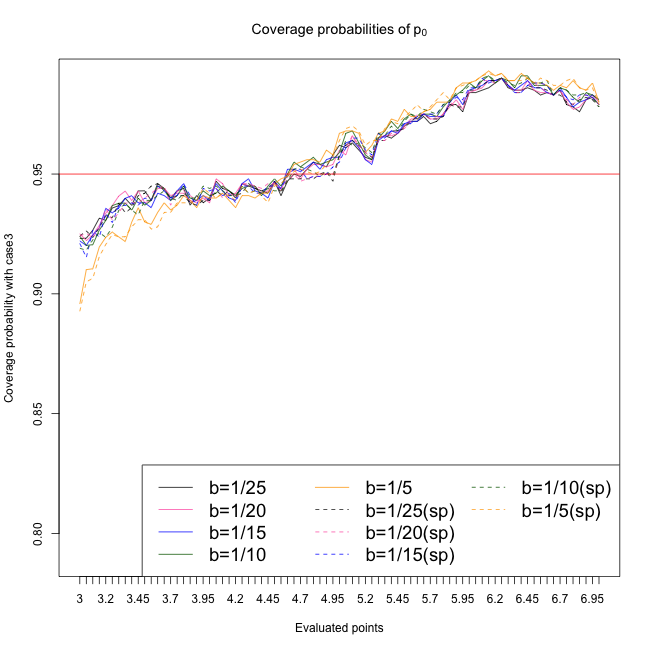}}
   \end{subfigure}
    \hspace{0.45cm}
    \begin{subfigure}[$p_0$; C3; $n=1000$]
    {\label{1000MW0}\includegraphics[width=50mm]{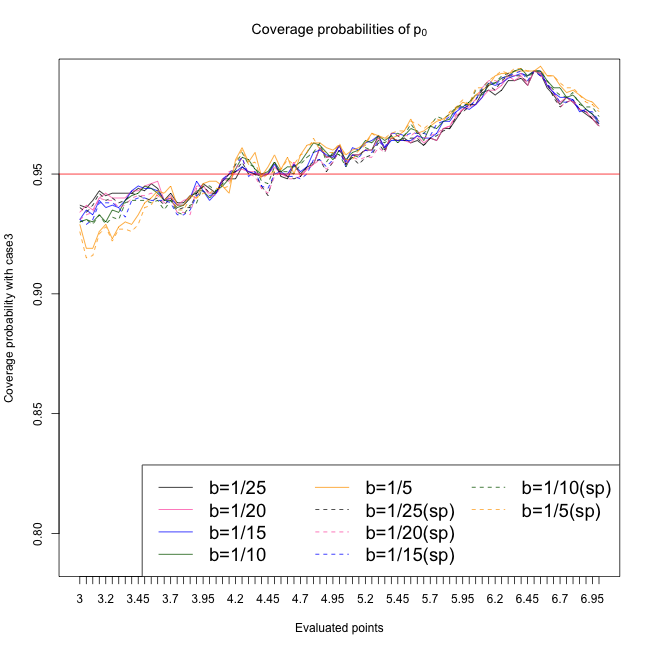}}
    \end{subfigure}\\
    \begin{subfigure}[$p_0$; C3; $n=2500$]
    {\label{2500MW0}\includegraphics[width=50mm]{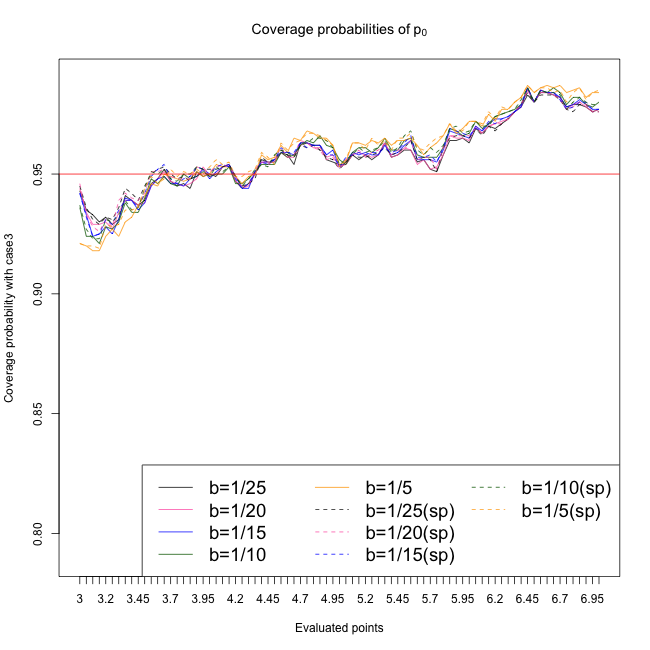}}
    \end{subfigure}
    \hspace{0.45cm}
    \begin{subfigure}[$p_0$; C3; $n=4000$]{\label{4000MW0}\includegraphics[width=50mm]{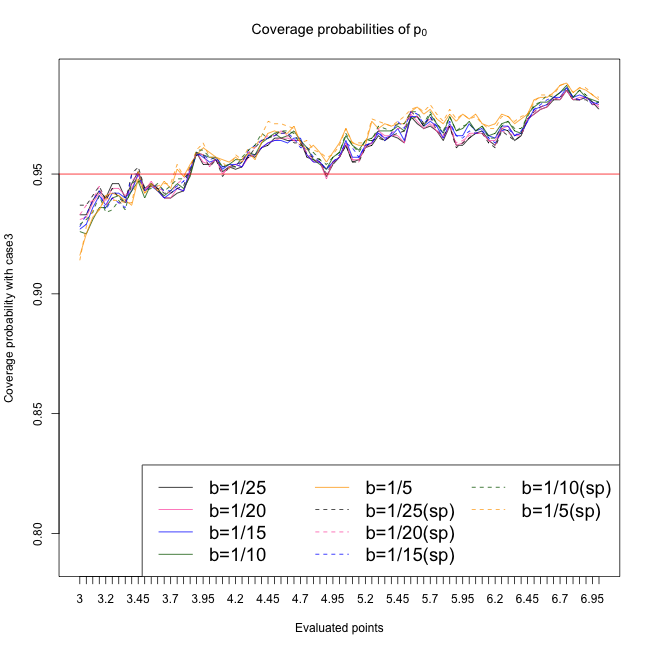}}
   \end{subfigure}\\
    \begin{subfigure}[$p_0$; C3; $n=6000$]
    {\label{6000MW0}\includegraphics[width=50mm]{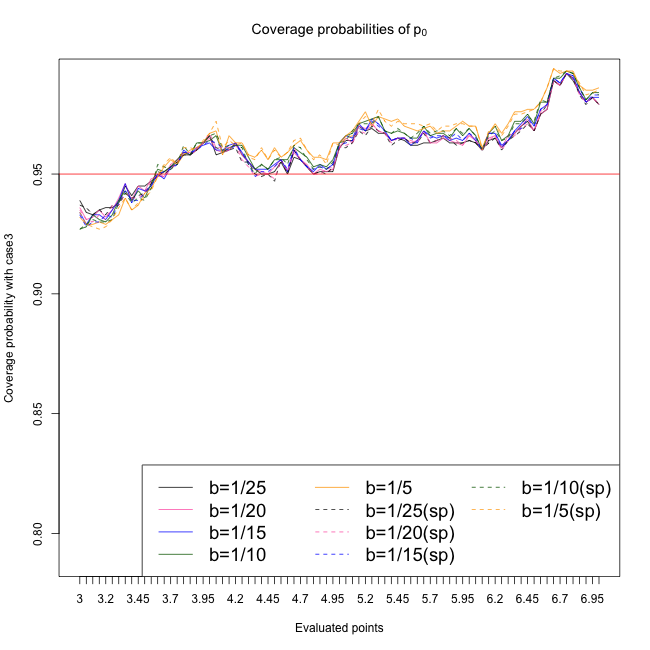}}
    \end{subfigure}
    \hspace{0.45cm}
    \begin{subfigure}[$p_0$; C3; $n=8000$]
    {\label{8000MW0}\includegraphics[width=50mm]{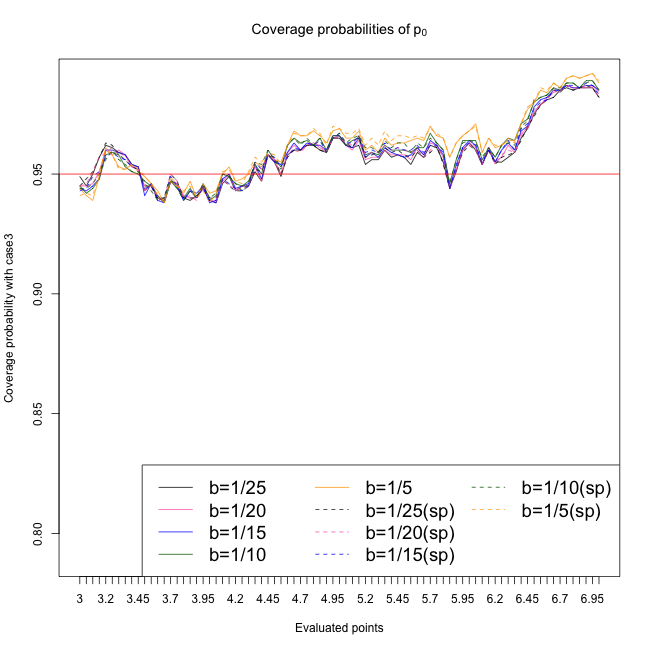}}
    \end{subfigure}
    \caption{\label{fig:kplots.case.3.p0}  Notational details can be found in Figure \ref{fig:kplots.case.1.p1}.}
\end{figure}

\clearpage

\subsection{Additional plots for Section \ref{Section:Simul}}\label{appd:additional.kplots}

\begin{figure}
	\centering
	\begin{subfigure}[Estimation of $p_0$; C1] {\label{P0C1}\includegraphics[width=50mm]{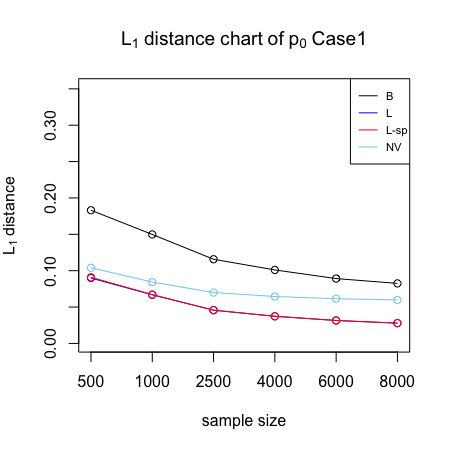}}
	\end{subfigure}
	\hspace{0.45cm}
	\begin{subfigure}[Estimation of $p_0$; C2]
		{\label{P0C2}\includegraphics[width=50mm]{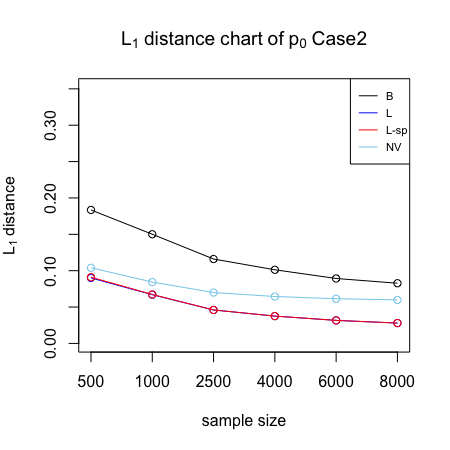}}
	\end{subfigure}
	\hspace{0.45cm}
	\begin{subfigure}[Estimation of $p_0$; C3]
		{\label{P0C3}\includegraphics[width=50mm]{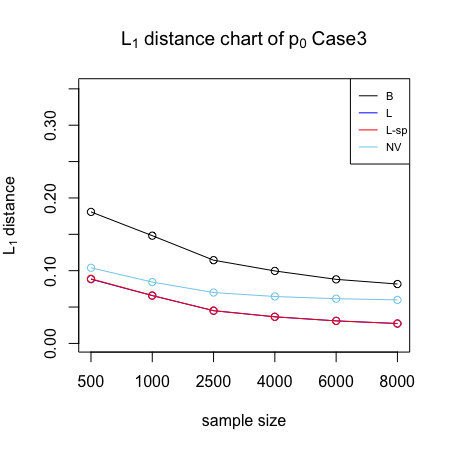}}
	\end{subfigure}
	\caption{\label{fig:Hellinger.p0.plots} Average $L_1$ distance between the estimators and truth for $p_0$. ``B", ``L", and ``NV" stand for the basis estimator of \cite{KBWcounterfactualdensity}, the log-concave estimator, and the naive log-concave MLE, respectively. ``sp" stands for sample splitting.
		Case 1 (both nuisance functions well-specified), Case 2 (only the propensity score is well-specified) and Case 3 (only the conditional CDF is well-specified), respectively. 
		The lines for L and L-sp cannot be visually distinguished.}
\end{figure}

\begin{figure}
	\centering
	\begin{subfigure}[$p_1$; C1; $n=500$]{\label{cmp500WW1}\includegraphics[width=50mm]{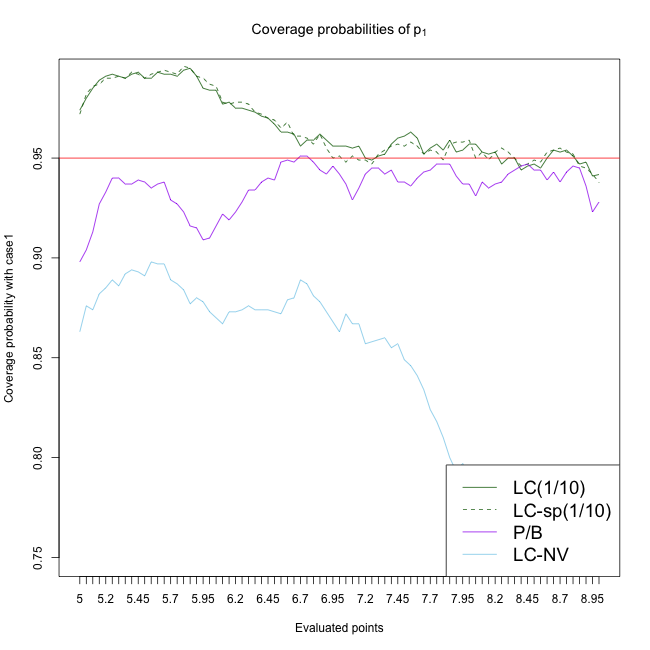}}
	\end{subfigure}
	\hspace{0.45cm}
	\begin{subfigure}[$p_1$; C1; $n=1000$]
		{\label{cmp1000WW1}\includegraphics[width=50mm]{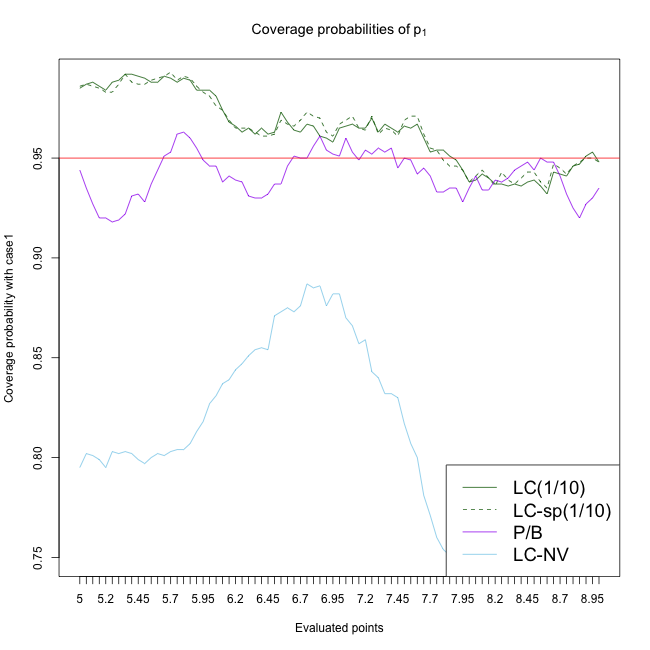}}
	\end{subfigure}\\
	\begin{subfigure}[$p_1$; C1; $n=2500$]
		{\label{cmp2500WW1}\includegraphics[width=50mm]{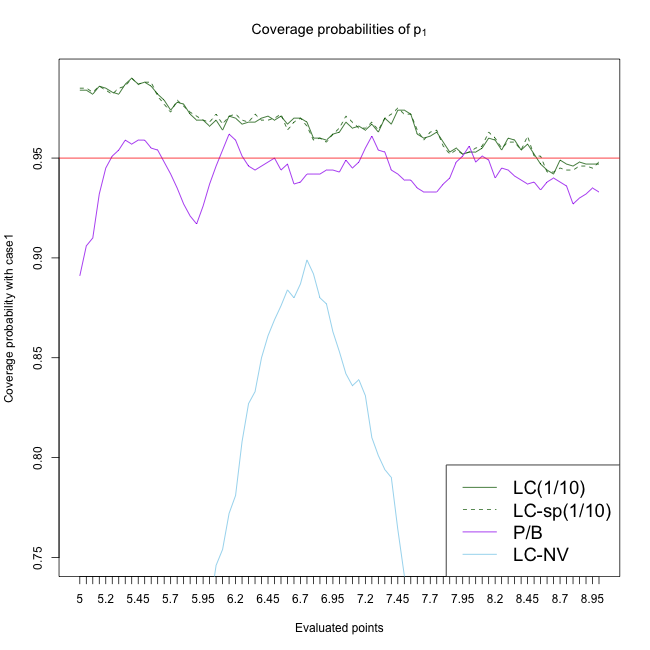}}
	\end{subfigure}
	\hspace{0.45cm}
	\begin{subfigure}[$p_1$; C1; $n=4000$]{\label{cmp4000WW1}\includegraphics[width=50mm]{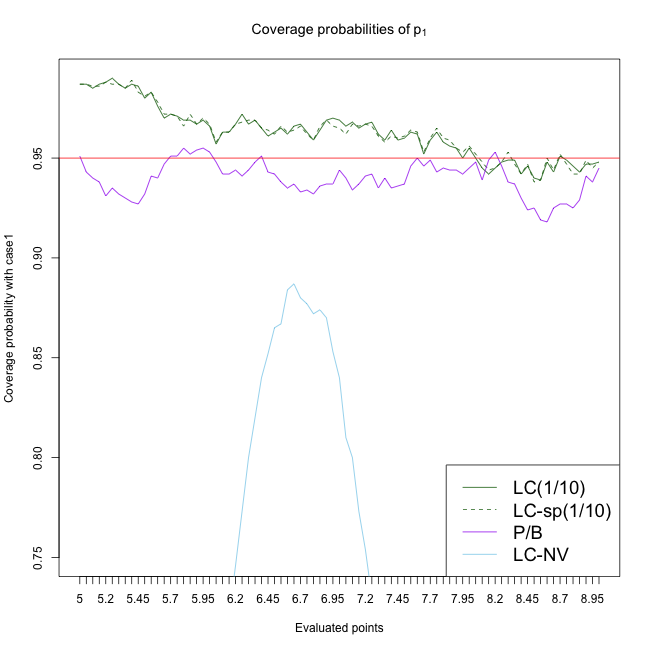}}
	\end{subfigure}\\
	\begin{subfigure}[$p_1$; C1; $n=6000$]
		{\label{cmp6000WW1}\includegraphics[width=50mm]{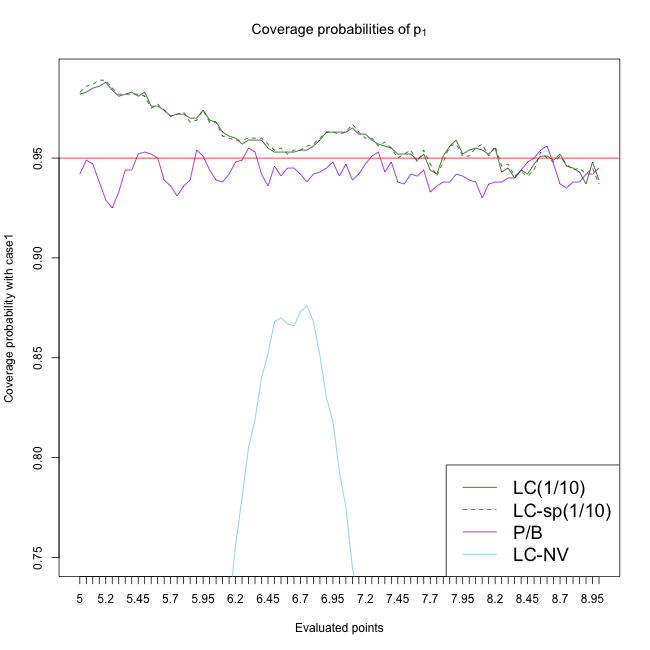}}
	\end{subfigure}
	\hspace{0.45cm}
	\begin{subfigure}[$p_1$; C1; $n=8000$]
		{\label{cmp8000WW1}\includegraphics[width=50mm]{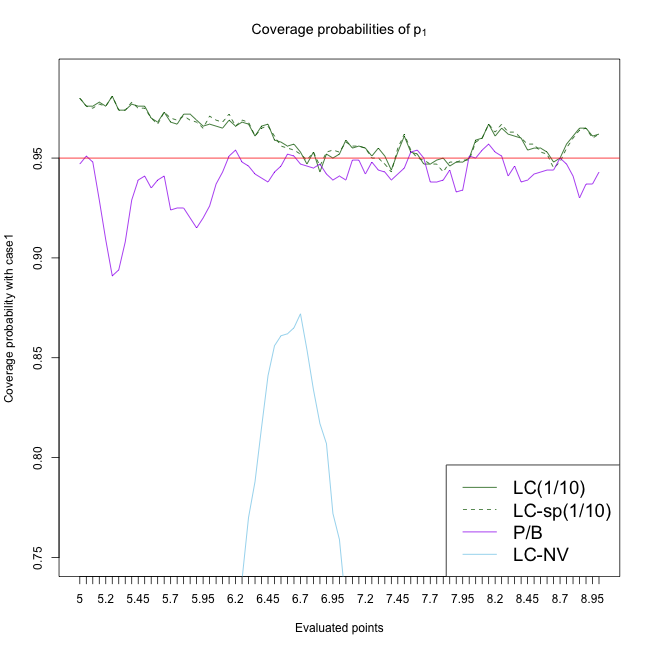}}
	\end{subfigure}
	\caption{\label{fig:kplot2s.case.1.p1}The above displays are coverage probabilities of our proposed log-concave projection estimator's 95\% CI's with the suggested tuning parameter $b=1/10$ which is labeled as LC (see Section \ref{subsec:tuning.param}), and the coverage probabilities from 95\% CI of projection basis method by \citet{KBWcounterfactualdensity} which is denoted by P and B.
		For the projection basis method, the number of basis functions is the oracle number of functions which achieved the lowest average $L_1$ distance in each setting. 
		``sp" stands for our sample splitting based estimator (see \eqref{crossfitted.onestep}).
		And, ``NV" denotes the naive log-concave MLE.
		The coverage probabilities are measured in 81 equally spaced points in the domain. Each subcaption describes the estimand ($p_1$ or $p_0$), the sample size, and each case of nuisance estimations (Case 1, 2, or 3 abbreviated to C1, C2, and C3, respectively).}
\end{figure}

\begin{figure}
	\centering
	\begin{subfigure}[$p_1$; C2; $n=500$]{\label{cmp500WM1}\includegraphics[width=50mm]{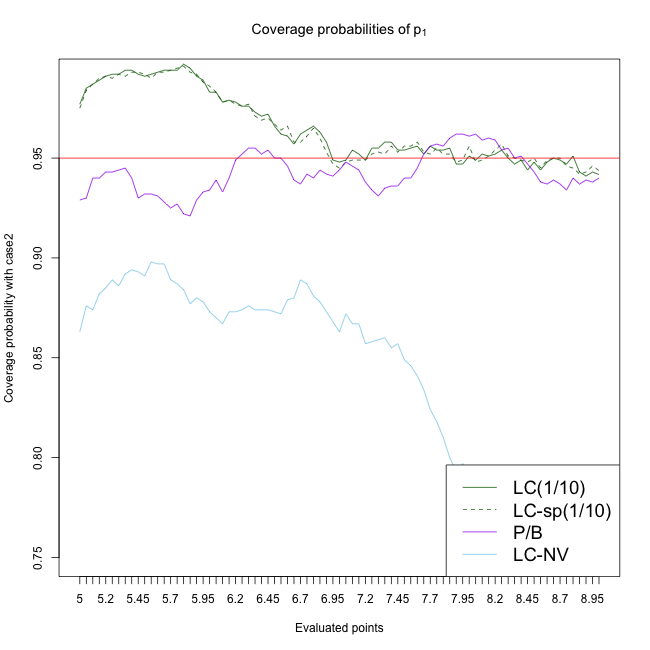}}
	\end{subfigure}
	\hspace{0.45cm}
	\begin{subfigure}[$p_1$; C2; $n=1000$]
		{\label{cmp1000WM1}\includegraphics[width=50mm]{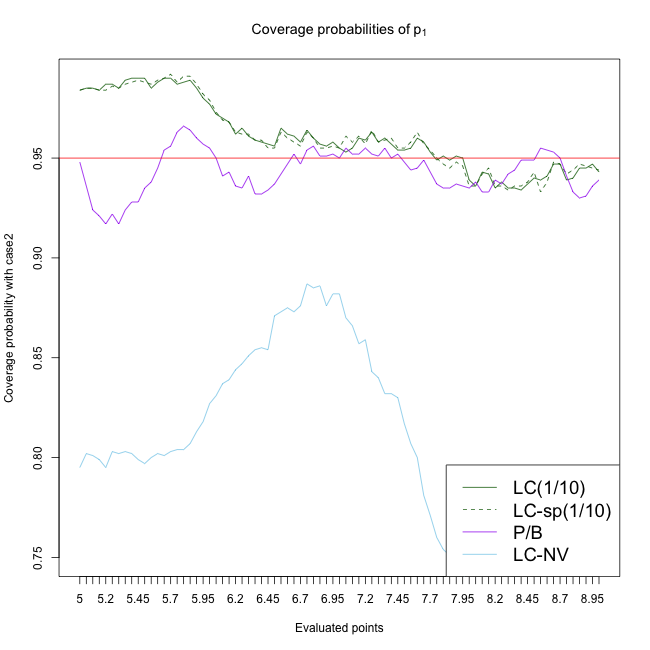}}
	\end{subfigure}\\
	\begin{subfigure}[$p_1$; C2; $n=2500$]
		{\label{cmp2500WM1}\includegraphics[width=50mm]{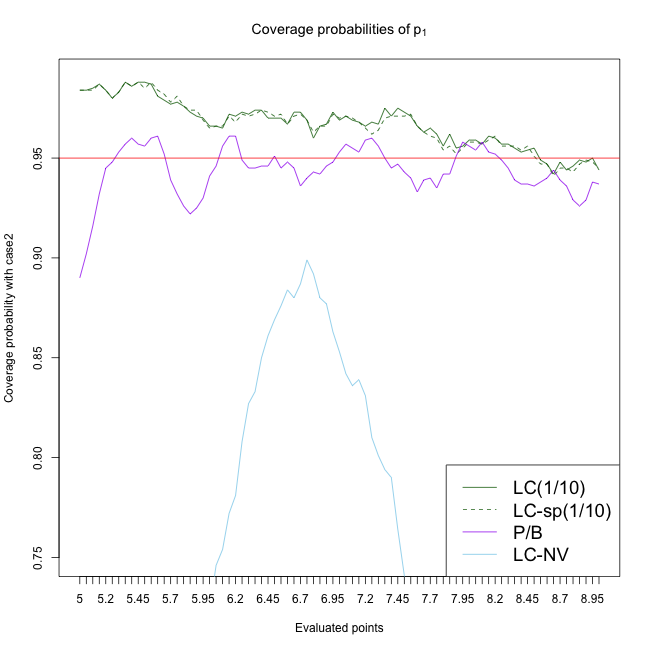}}
	\end{subfigure}
	\hspace{0.45cm}
	\begin{subfigure}[$p_1$; C2; $n=4000$]{\label{cmp4000WM1}\includegraphics[width=50mm]{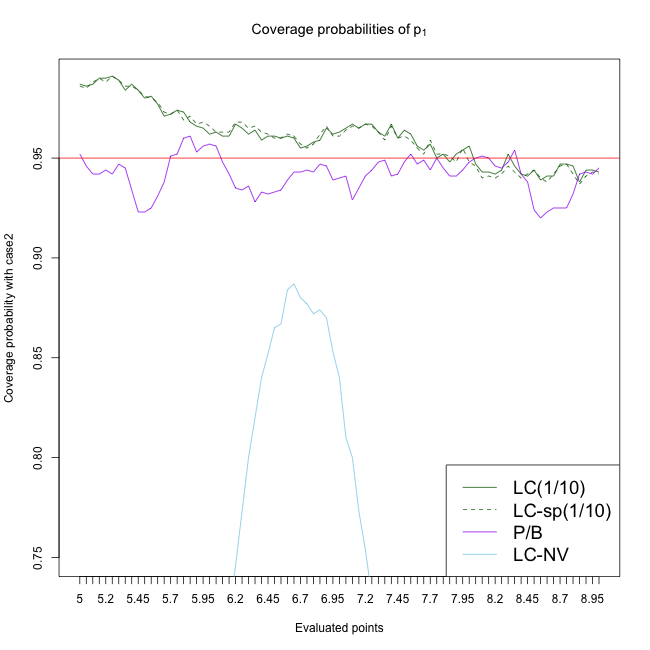}}
	\end{subfigure}\\
	\begin{subfigure}[$p_1$; C2; $n=6000$]
		{\label{cmp6000WM1}\includegraphics[width=50mm]{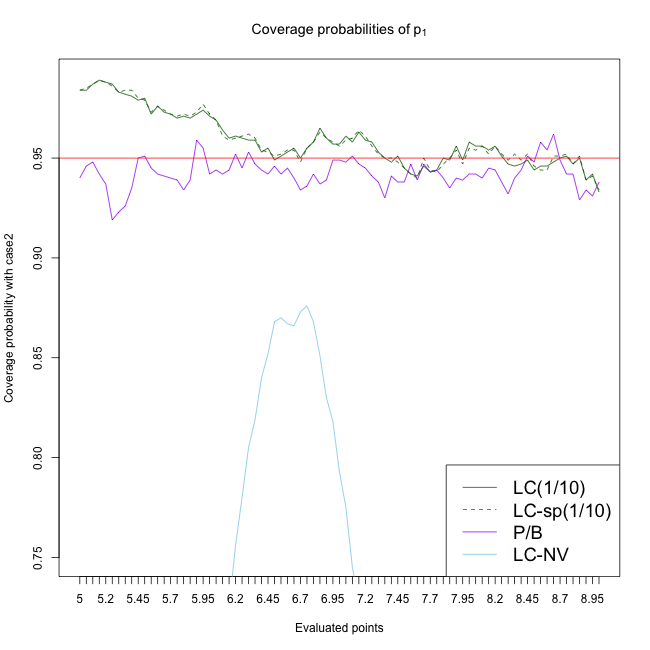}}
	\end{subfigure}
	\hspace{0.45cm}
	\begin{subfigure}[$p_1$; C2; $n=8000$]
		{\label{cmp8000WM1}\includegraphics[width=50mm]{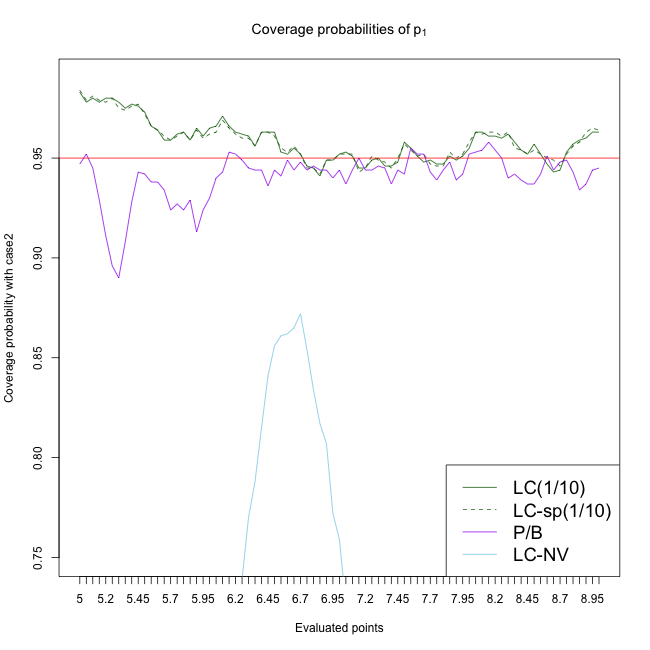}}
	\end{subfigure}
	\caption{\label{fig:kplot2s.case.2.p1} Notational details can be found in Figure \ref{fig:kplot2s.case.1.p1}.}
\end{figure}

\begin{figure}
	\centering
	\begin{subfigure}[$p_1$; C3; $n=500$]{\label{cmp500MW1}\includegraphics[width=50mm]{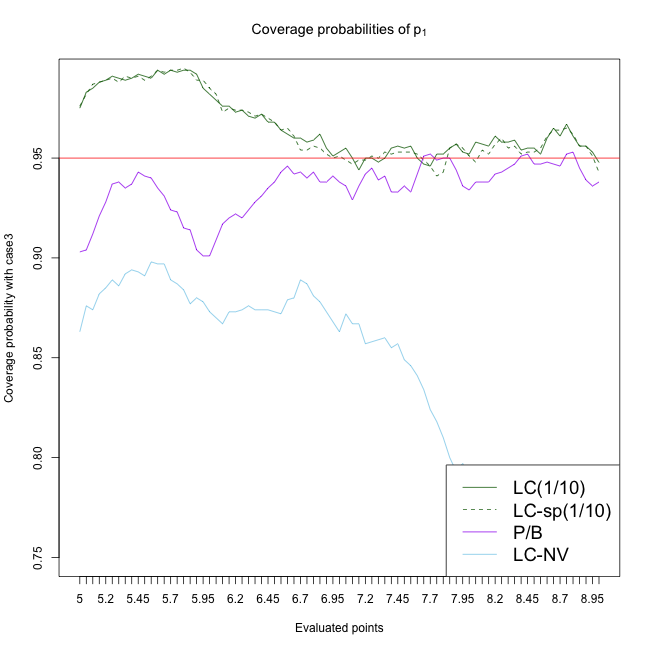}}
	\end{subfigure}
	\hspace{0.45cm}
	\begin{subfigure}[$p_1$; C3; $n=1000$]
		{\label{cmp1000MW1}\includegraphics[width=50mm]{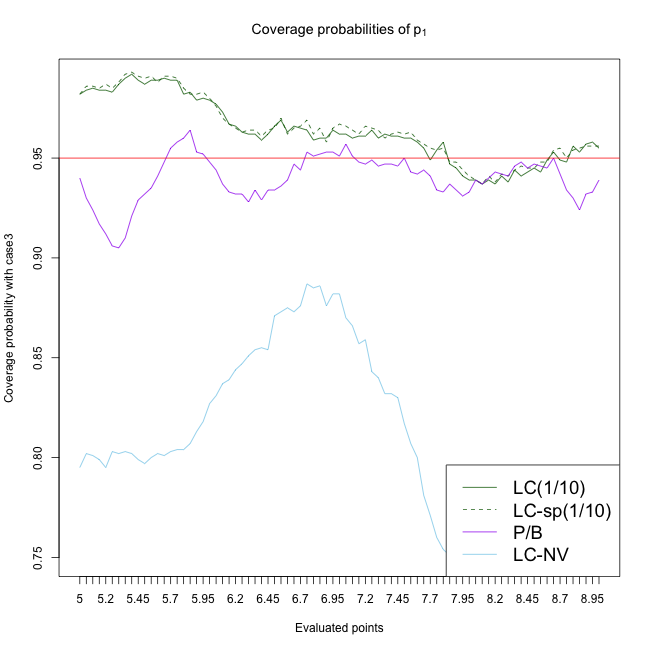}}
	\end{subfigure}\\
	\begin{subfigure}[$p_1$; C3; $n=2500$]
		{\label{cmp2500MW1}\includegraphics[width=50mm]{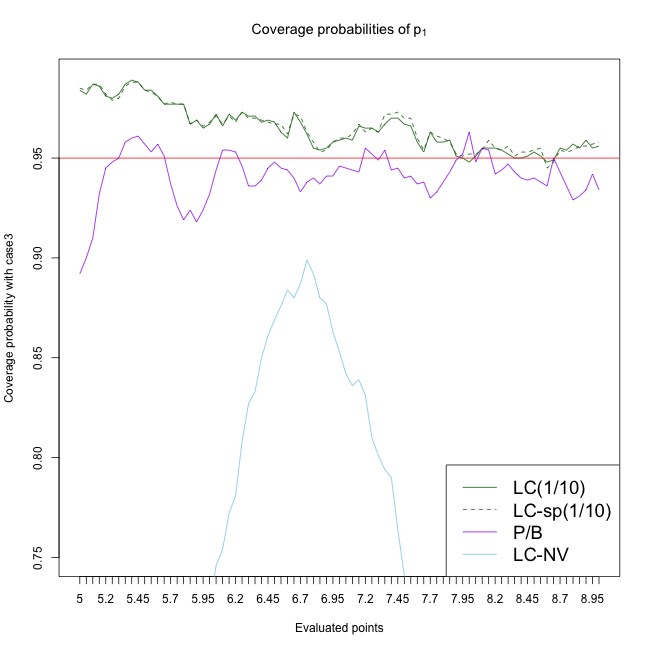}}
	\end{subfigure}
	\hspace{0.45cm}
	\begin{subfigure}[$p_1$; C3; $n=4000$]{\label{cmp4000MW1}\includegraphics[width=50mm]{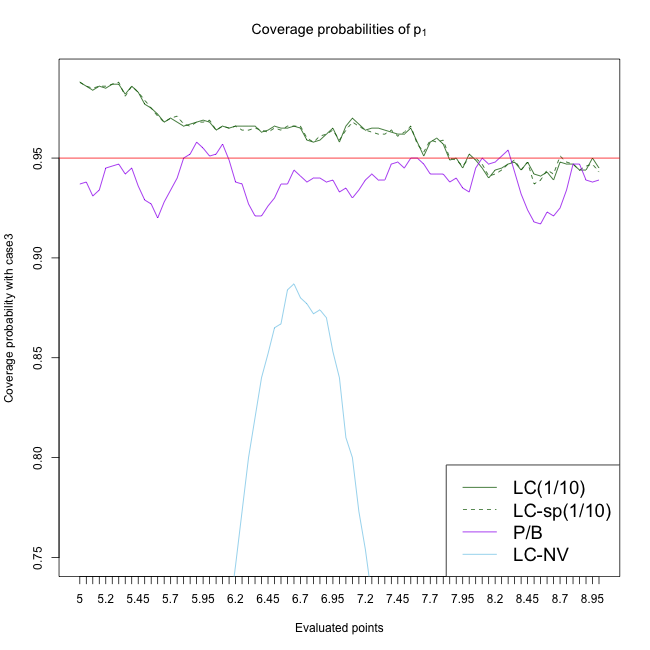}}
	\end{subfigure}\\
	\begin{subfigure}[$p_1$; C3; $n=6000$]
		{\label{cmp6000MW1}\includegraphics[width=50mm]{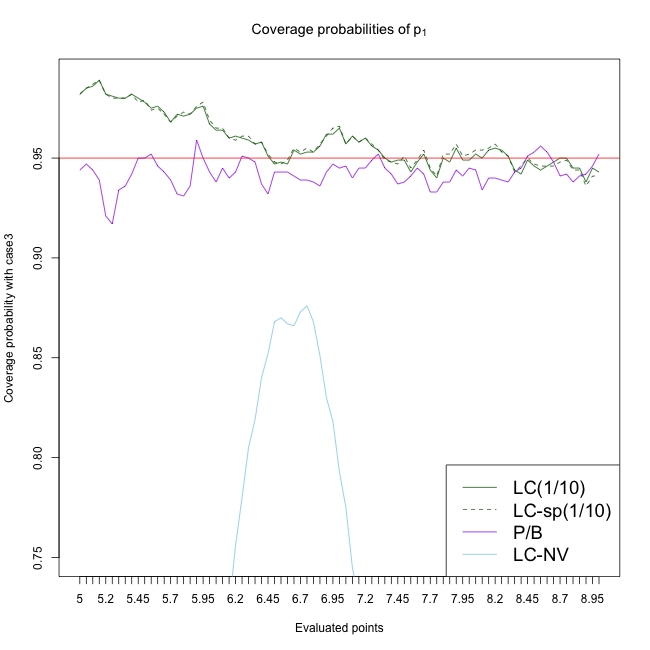}}
	\end{subfigure}
	\hspace{0.45cm}
	\begin{subfigure}[$p_1$; C3; $n=8000$]
		{\label{cmp8000MW1}\includegraphics[width=50mm]{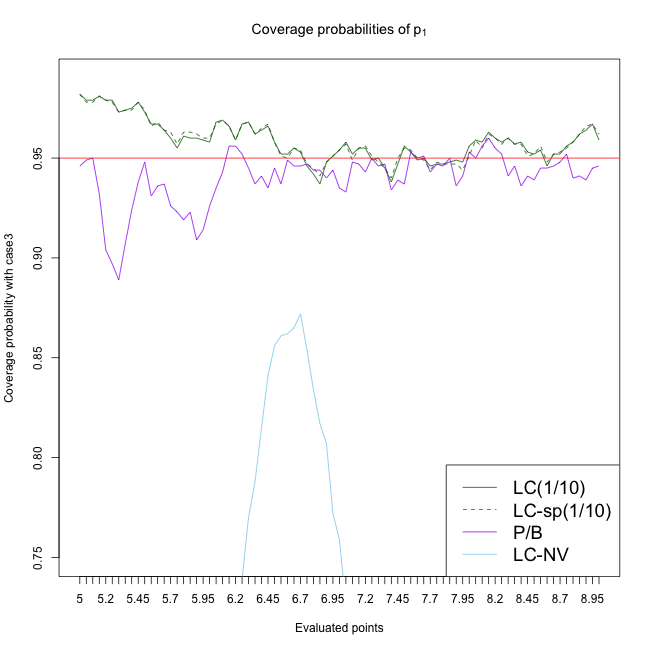}}
	\end{subfigure}
	\caption{\label{fig:kplot2s.case.3.p1}  Notational details can be found in Figure \ref{fig:kplot2s.case.1.p1}.}
\end{figure}

\begin{figure}
	\centering
	\begin{subfigure}[$p_0$; C1; $n=500$]{\label{cmp500WW0}\includegraphics[width=50mm]{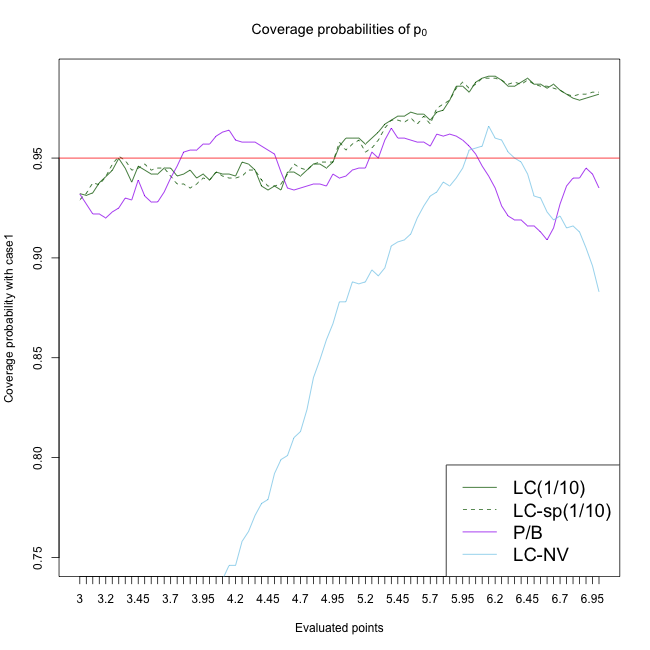}}
	\end{subfigure}
	\hspace{0.45cm}
	\begin{subfigure}[$p_0$; C1; $n=1000$]
		{\label{cmp1000WW0}\includegraphics[width=50mm]{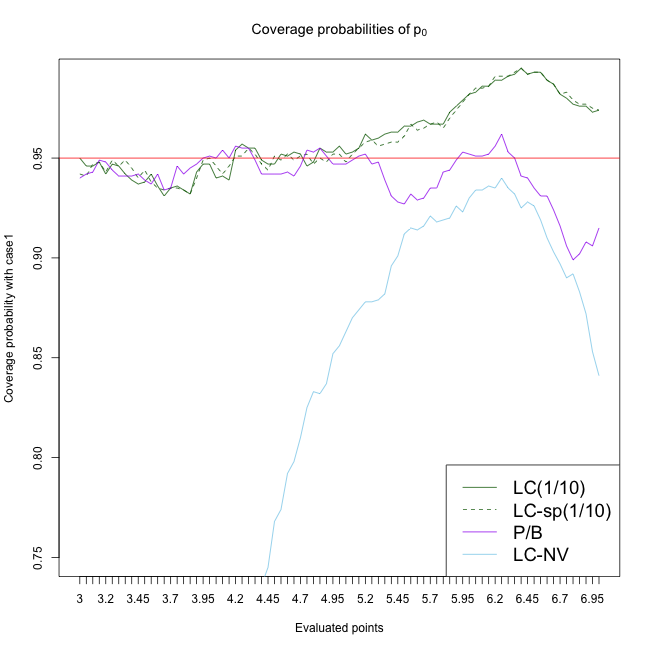}}
	\end{subfigure}\\
	\begin{subfigure}[$p_0$; C1; $n=2500$]
		{\label{cmp2500WW0}\includegraphics[width=50mm]{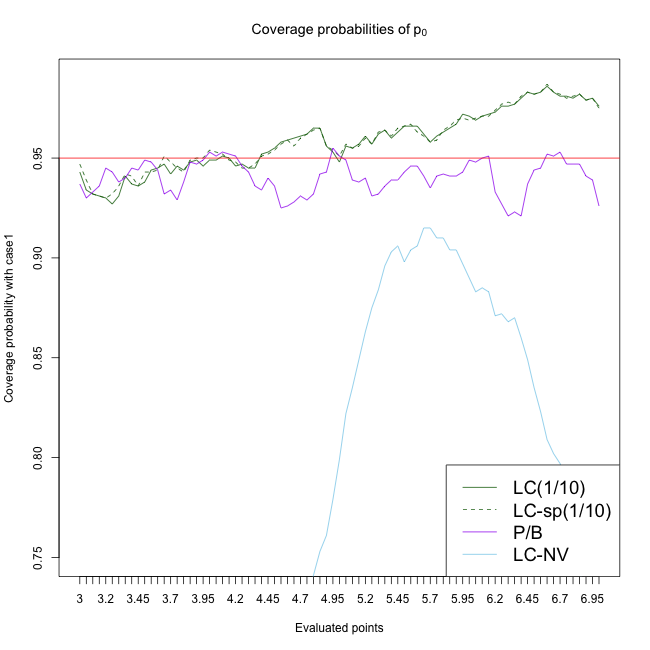}}
	\end{subfigure}
	\hspace{0.45cm}    
	\begin{subfigure}[$p_0$; C1; $n=4000$]{\label{cmp4000WW0}\includegraphics[width=50mm]{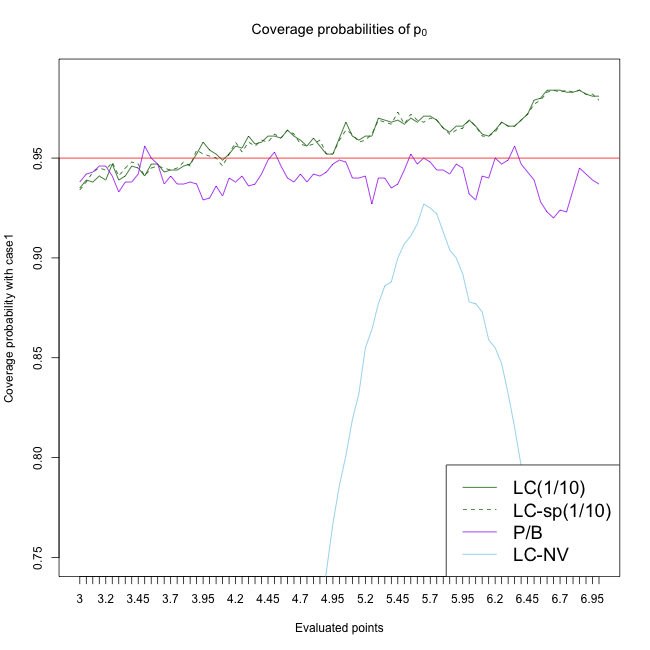}}
	\end{subfigure}\\
	\begin{subfigure}[$p_0$; C1; $n=6000$]
		{\label{cmp6000WW0}\includegraphics[width=50mm]{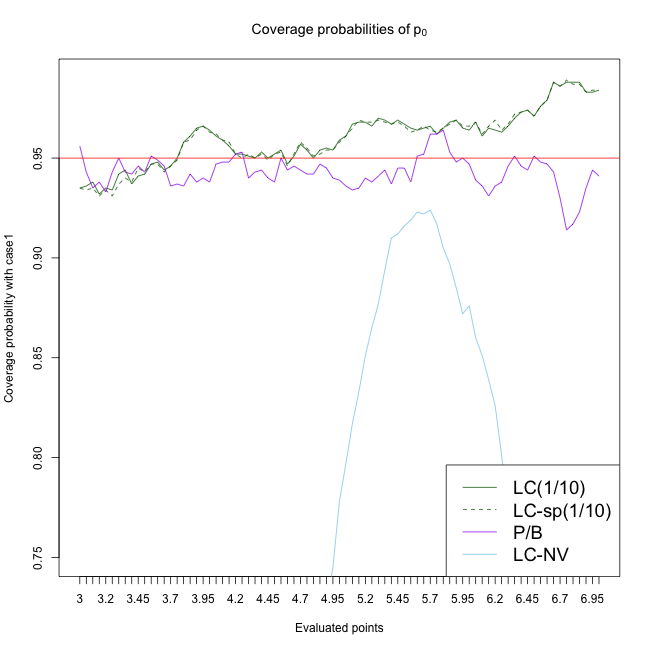}}
	\end{subfigure}
	\hspace{0.45cm}
	\begin{subfigure}[$p_0$; C1; $n=8000$]
		{\label{cmp8000WW0}\includegraphics[width=50mm]{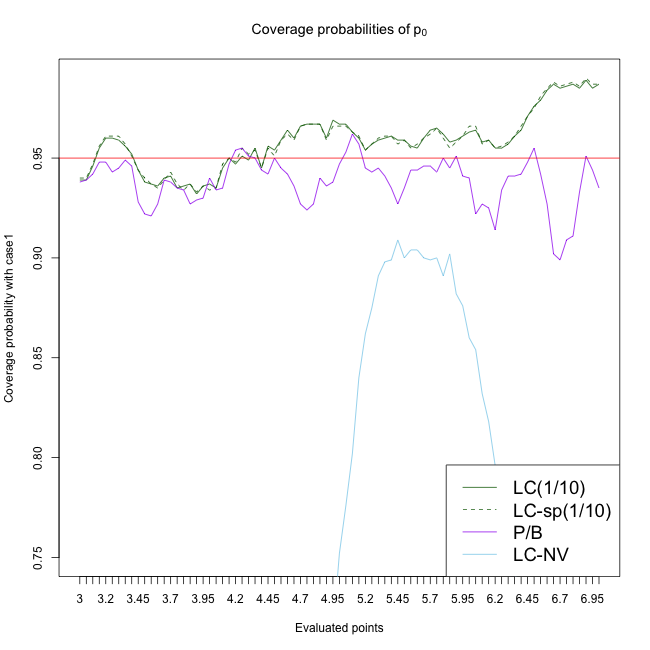}}
	\end{subfigure}
	\caption{\label{fig:kplot2s.case.1.p0} Notational details can be found in Figure \ref{fig:kplot2s.case.1.p1}.}
\end{figure}

\begin{figure}
	\centering
	\begin{subfigure}[$p_0$; C2; $n=500$]{\label{cmp500WM0}\includegraphics[width=50mm]{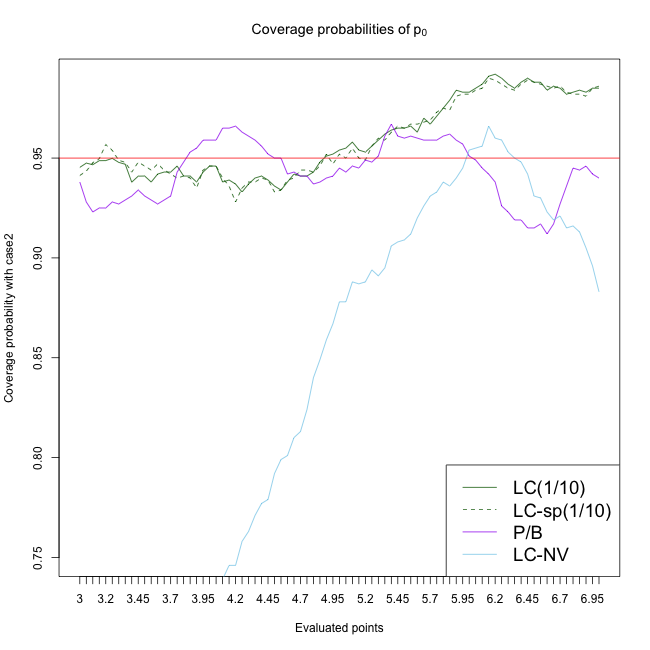}}
	\end{subfigure}
	\hspace{0.45cm}
	\begin{subfigure}[$p_0$; C2; $n=1000$]
		{\label{cmp1000WM0}\includegraphics[width=50mm]{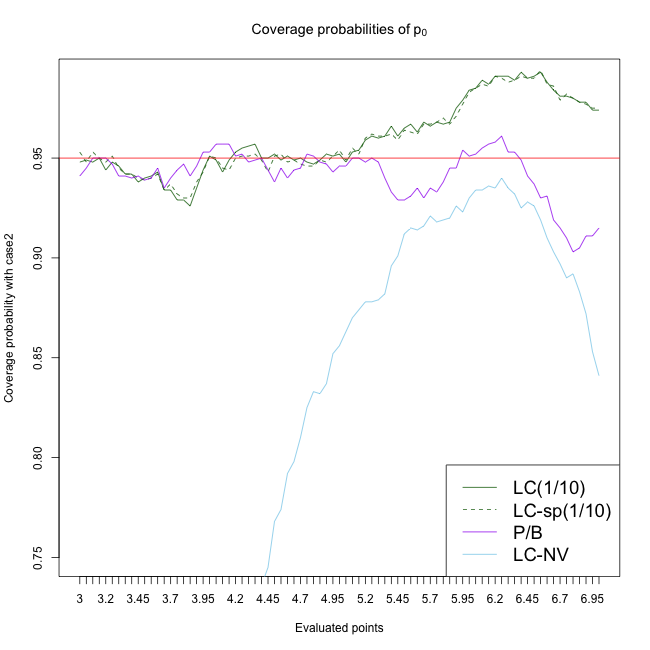}}
	\end{subfigure}\\
	\begin{subfigure}[$p_0$; C2; $n=2500$]
		{\label{cmp2500WM0}\includegraphics[width=50mm]{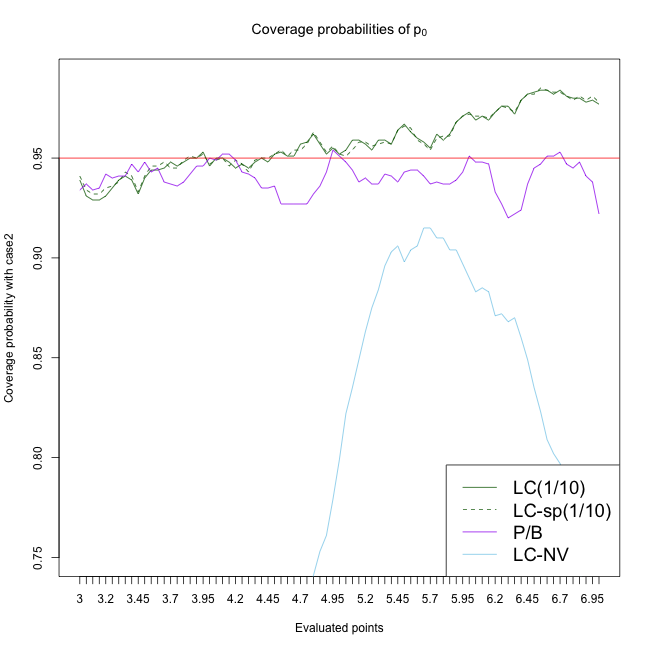}}
	\end{subfigure}
	\hspace{0.45cm}    
	\begin{subfigure}[$p_0$; C2; $n=4000$]{\label{cmp4000WM0}\includegraphics[width=50mm]{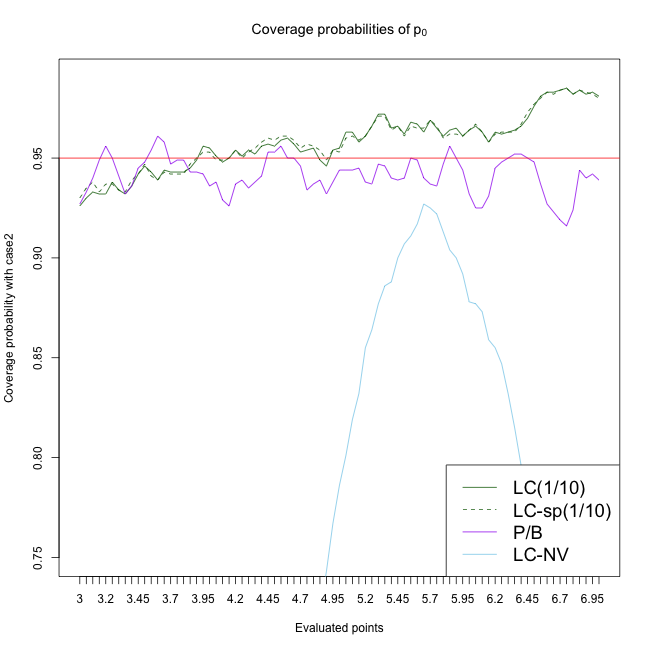}}
	\end{subfigure}\\
	\begin{subfigure}[$p_0$; C2; $n=6000$]
		{\label{cmp6000WM0}\includegraphics[width=50mm]{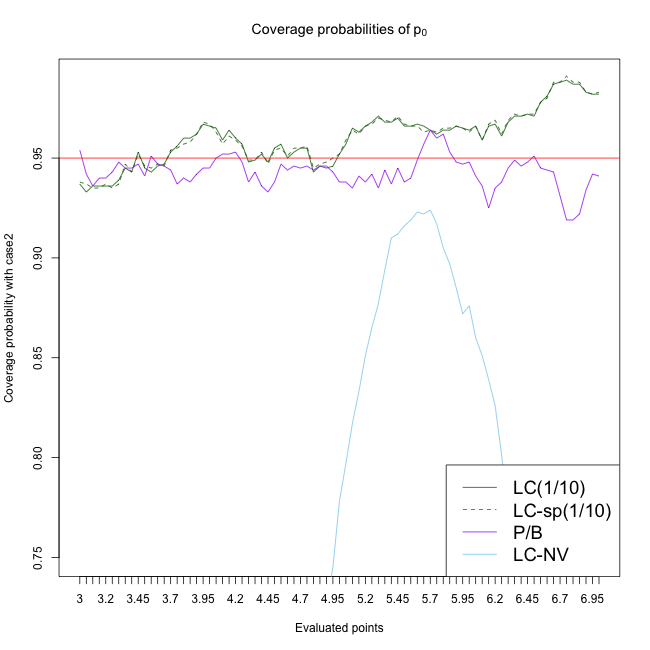}}
	\end{subfigure}
	\hspace{0.45cm}
	\begin{subfigure}[$p_0$; C2; $n=8000$]
		{\label{cmp8000WM0}\includegraphics[width=50mm]{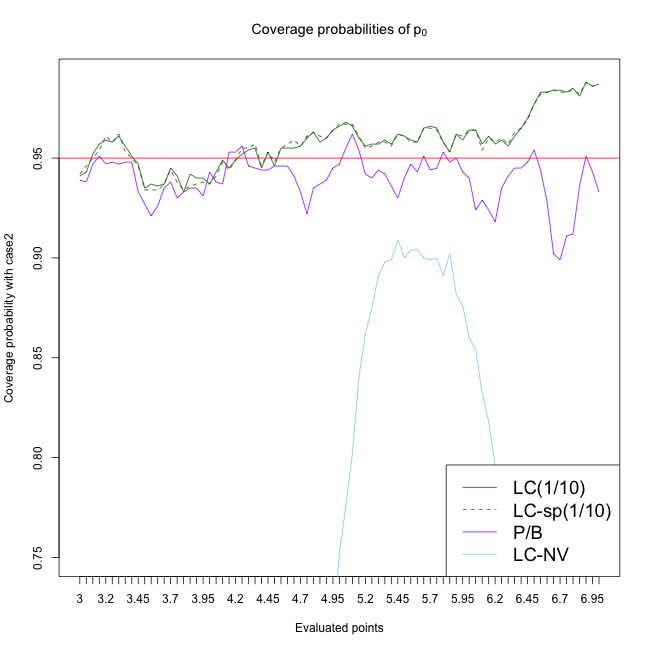}}
	\end{subfigure}
	\caption{\label{fig:kplot2s.case.2.p0} Notational details can be found in Figure \ref{fig:kplot2s.case.1.p1}.}
\end{figure}

\begin{figure}
	\centering
	\begin{subfigure}[$p_0$; C3; $n=500$]{\label{cmp500MW0}\includegraphics[width=50mm]{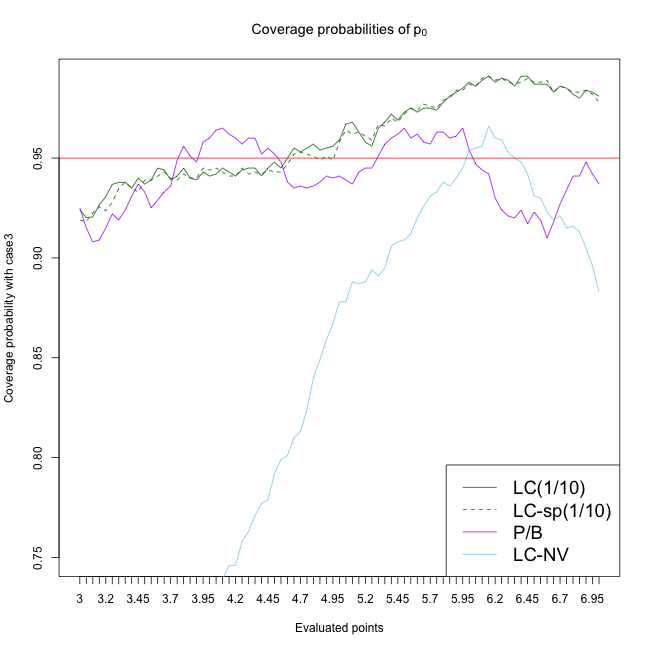}}
	\end{subfigure}
	\hspace{0.45cm}
	\begin{subfigure}[$p_0$; C3; $n=1000$]
		{\label{cmp1000MW0}\includegraphics[width=50mm]{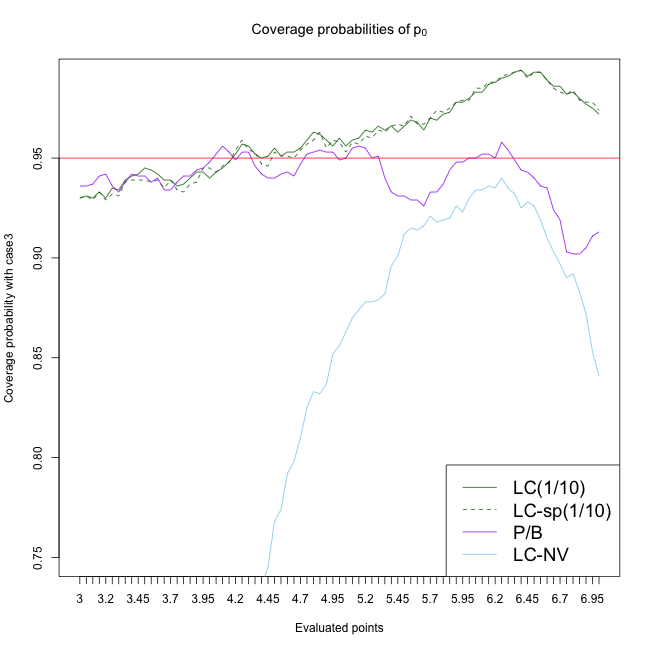}}
	\end{subfigure}\\
	\begin{subfigure}[$p_0$; C3; $n=2500$]
		{\label{cmp2500MW0}\includegraphics[width=50mm]{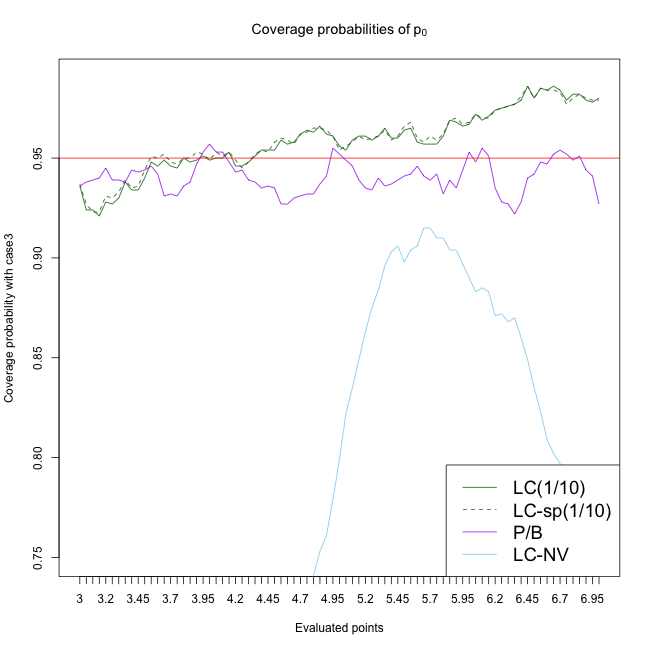}}
	\end{subfigure}
	\hspace{0.45cm}
	\begin{subfigure}[$p_0$; C3; $n=4000$]{\label{cmp4000MW0}\includegraphics[width=50mm]{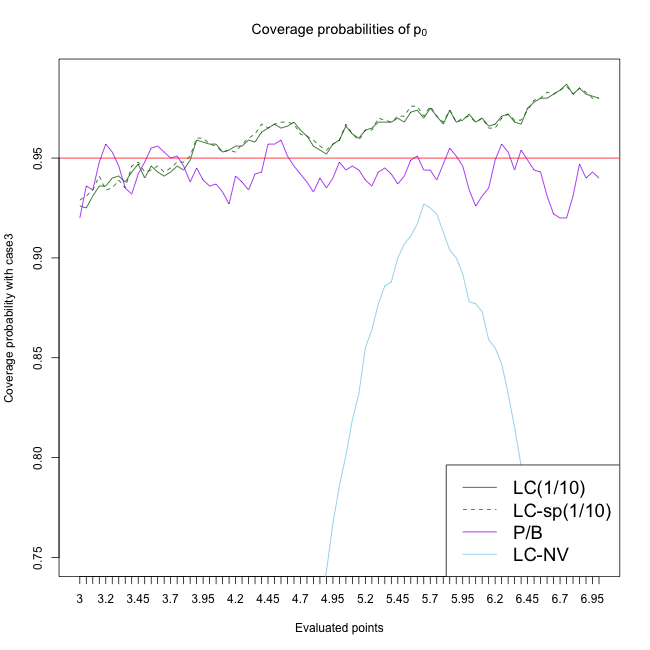}}
	\end{subfigure}\\
	\begin{subfigure}[$p_0$; C3; $n=6000$]
		{\label{cmp6000MW0}\includegraphics[width=50mm]{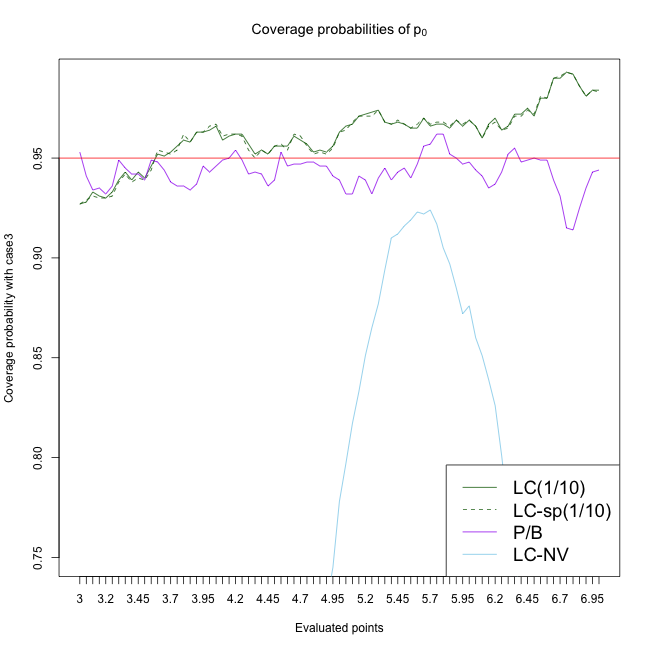}}
	\end{subfigure}
	\hspace{0.45cm}
	\begin{subfigure}[$p_0$; C3; $n=8000$]
		{\label{cmp8000MW0}\includegraphics[width=50mm]{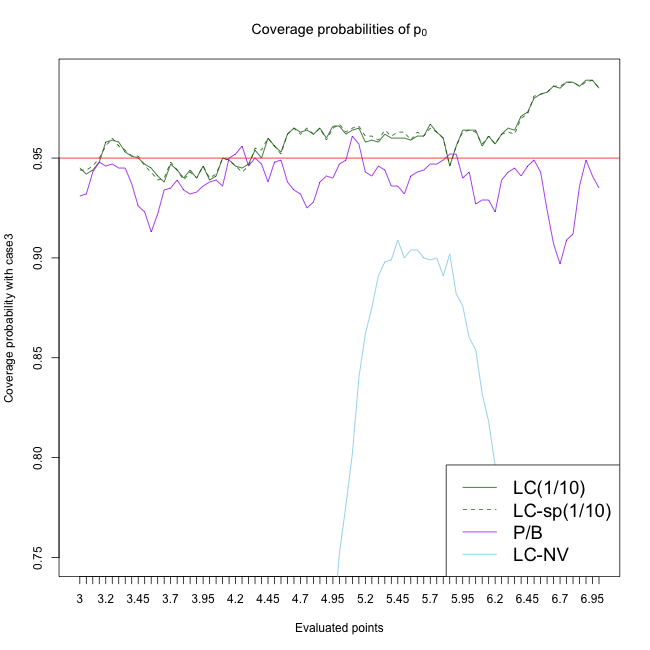}}
	\end{subfigure}
	\caption{\label{fig:kplot2s.case.3.p0}  Notational details can be found in Figure \ref{fig:kplot2s.case.1.p1}.}
\end{figure}

\begin{figure}
	\centering
	\begin{subfigure}[$p_1$; C1; $n=500$]{\label{wd500WW1}\includegraphics[width=50mm]{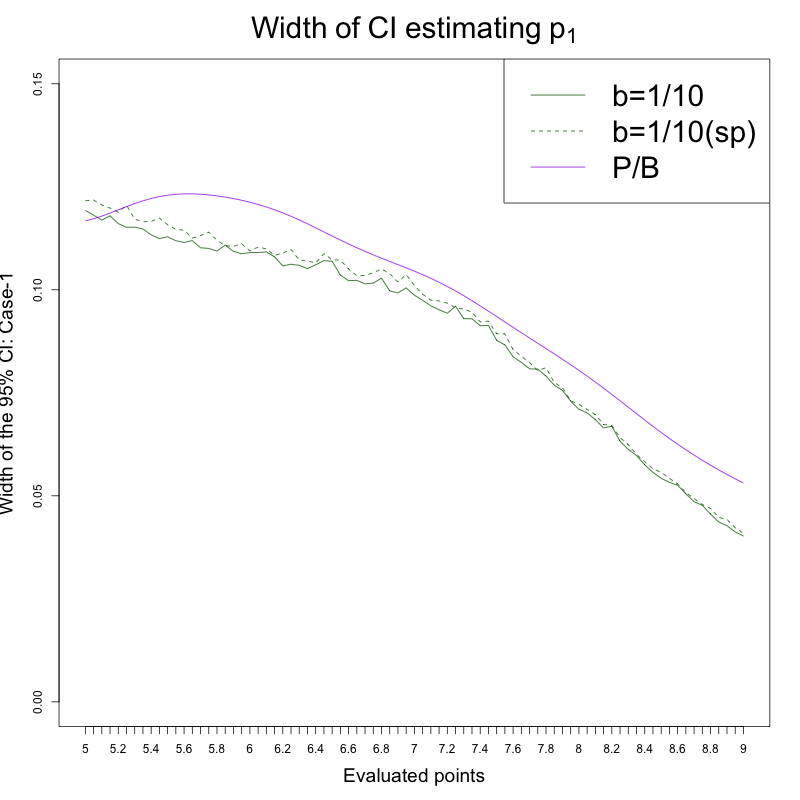}}
	\end{subfigure}
	\hspace{0.45cm}
	\begin{subfigure}[$p_1$; C1; $n=1000$]
		{\label{wd1000WW1}\includegraphics[width=50mm]{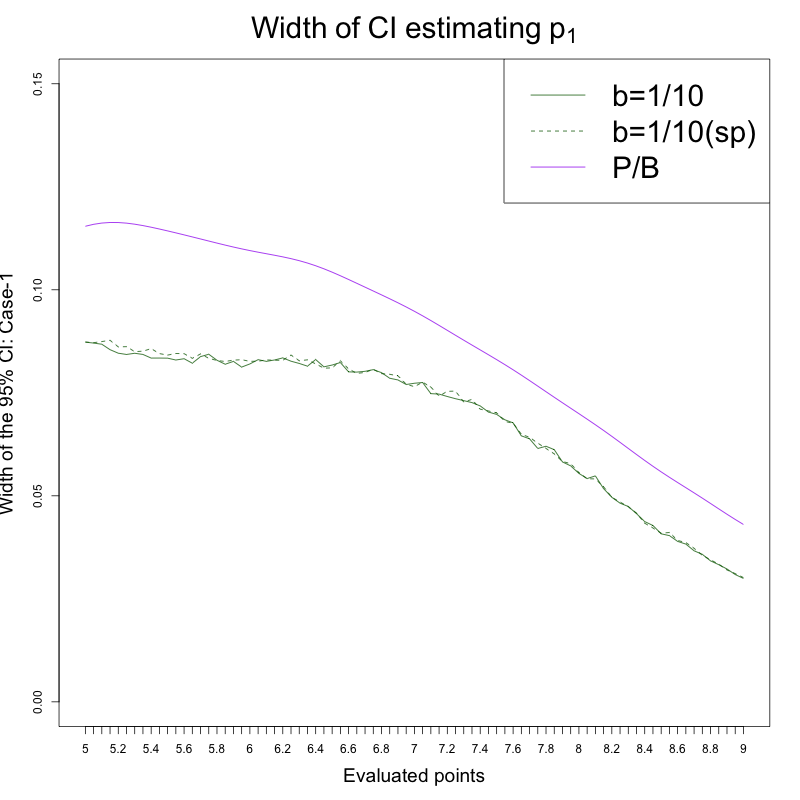}}
	\end{subfigure}\\
	\begin{subfigure}[$p_1$; C1; $n=2500$]
		{\label{wd2500WW1}\includegraphics[width=50mm]{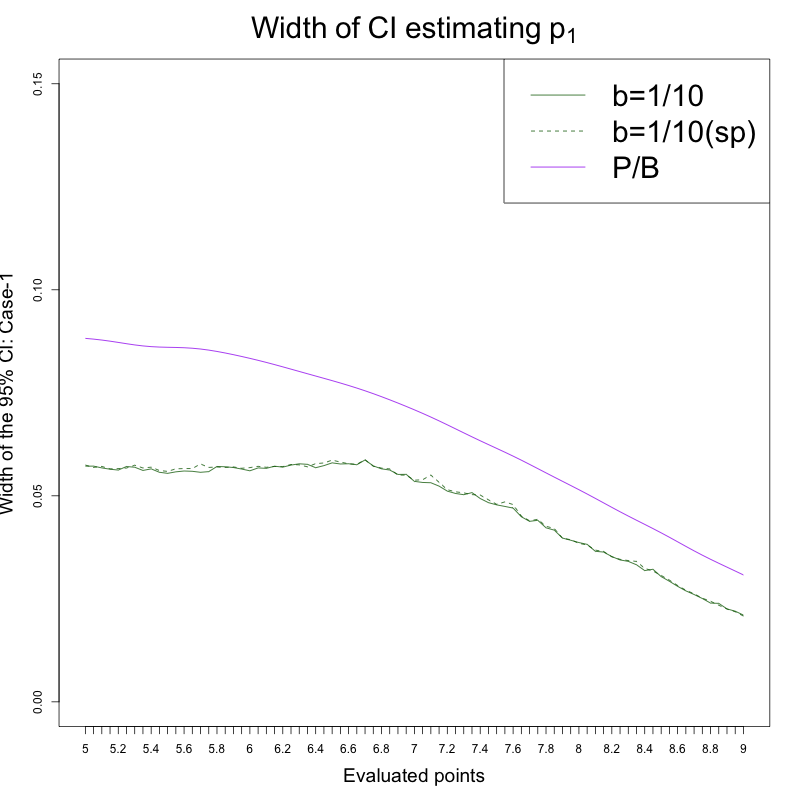}}
	\end{subfigure}
	\hspace{0.45cm}
	\begin{subfigure}[$p_1$; C1; $n=4000$]{\label{wd4000WW1}\includegraphics[width=50mm]{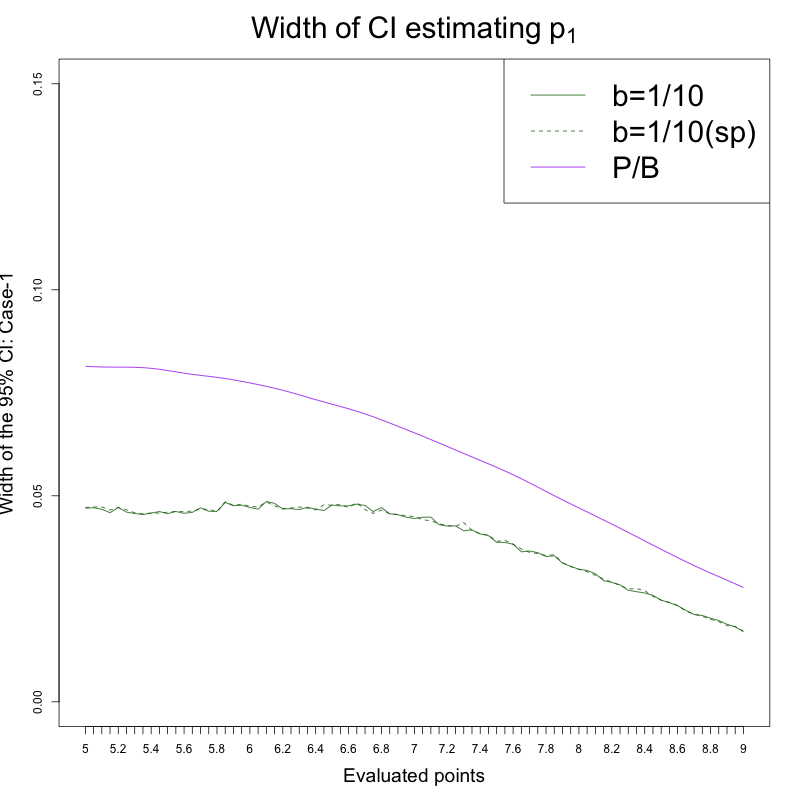}}
	\end{subfigure}\\
	\begin{subfigure}[$p_1$; C1; $n=6000$]
		{\label{wd6000WW1}\includegraphics[width=50mm]{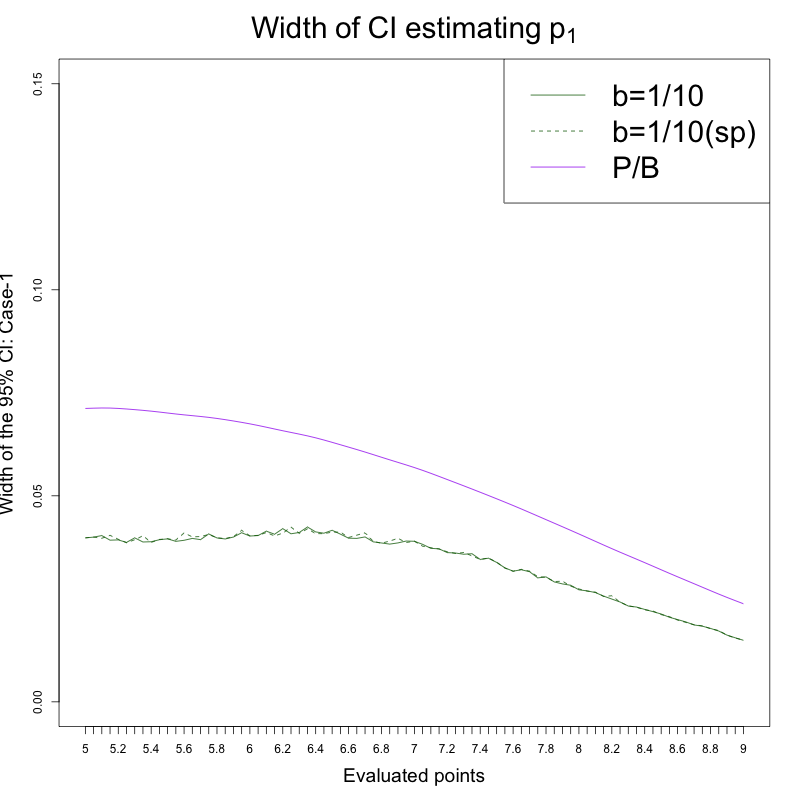}}
	\end{subfigure}
	\hspace{0.45cm}
	\begin{subfigure}[$p_1$; C1; $n=8000$]
		{\label{wd8000WW1}\includegraphics[width=50mm]{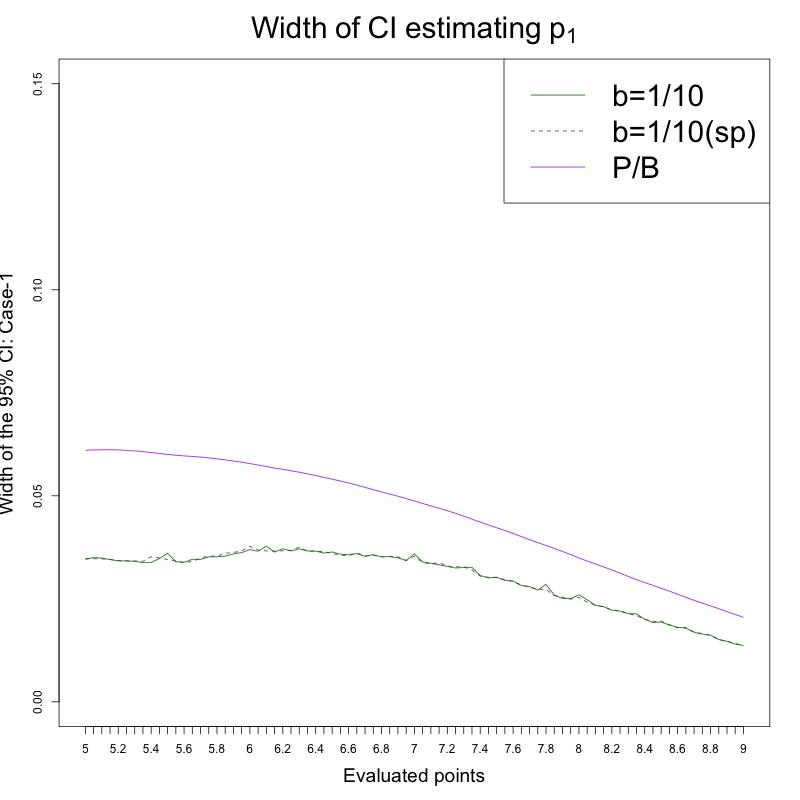}}
	\end{subfigure}
	\caption{\label{fig:wplots.case.1.p1}The above displays are widths of our proposed log-concave projection estimator's 95\% CI's with the suggested tuning parameter $b=1/10$ which is labeled as LC (see Section \ref{subsec:tuning.param}), and the widths from 95\% CI of projection basis method by \citet{KBWcounterfactualdensity} which is denoted by P and B.
		For the projection basis method, the number of basis functions is the oracle number of functions which achieved the lowest average $L_1$ distance in each setting. 
		``sp" stands for our sample splitting based estimator (see \eqref{crossfitted.onestep}).
		The widths are measured in 81 equally spaced points in the domain. Each subcaption describes the estimand ($p_1$ or $p_0$), the sample size, and each case of nuisance estimations (Case 1, 2, or 3 abbreviated to C1, C2, and C3, respectively).
	Since the difference between the cross-fitted procedure and the non-sample splitting procedure is insignificant, it is not visually distinguishable. }
\end{figure}

\begin{figure}
	\centering
	\begin{subfigure}[$p_1$; C2; $n=500$]{\label{wd500WM1}\includegraphics[width=50mm]{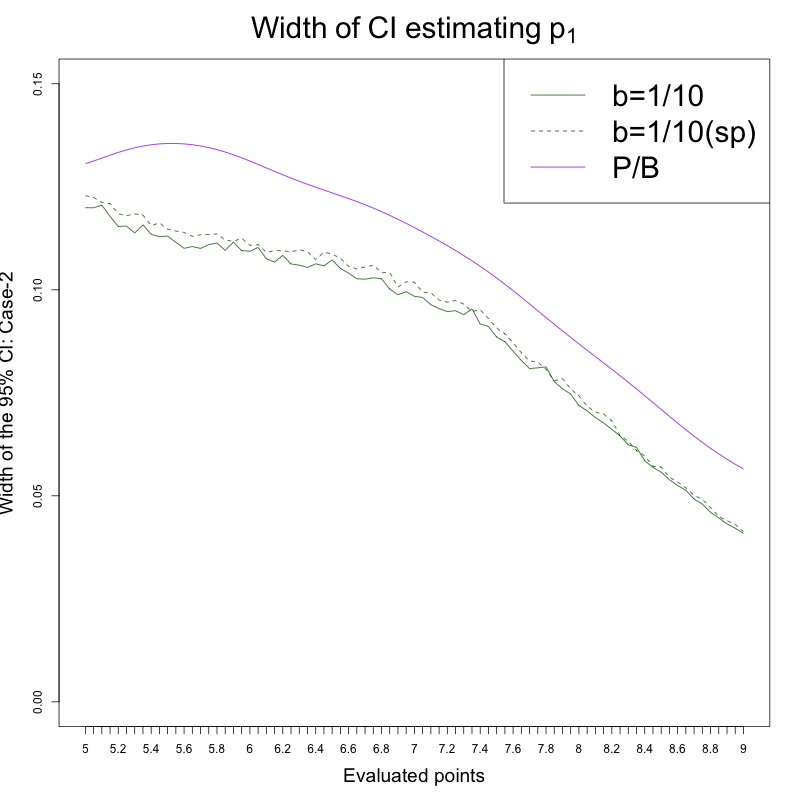}}
	\end{subfigure}
	\hspace{0.45cm}
	\begin{subfigure}[$p_1$; C2; $n=1000$]
		{\label{wd1000WM1}\includegraphics[width=50mm]{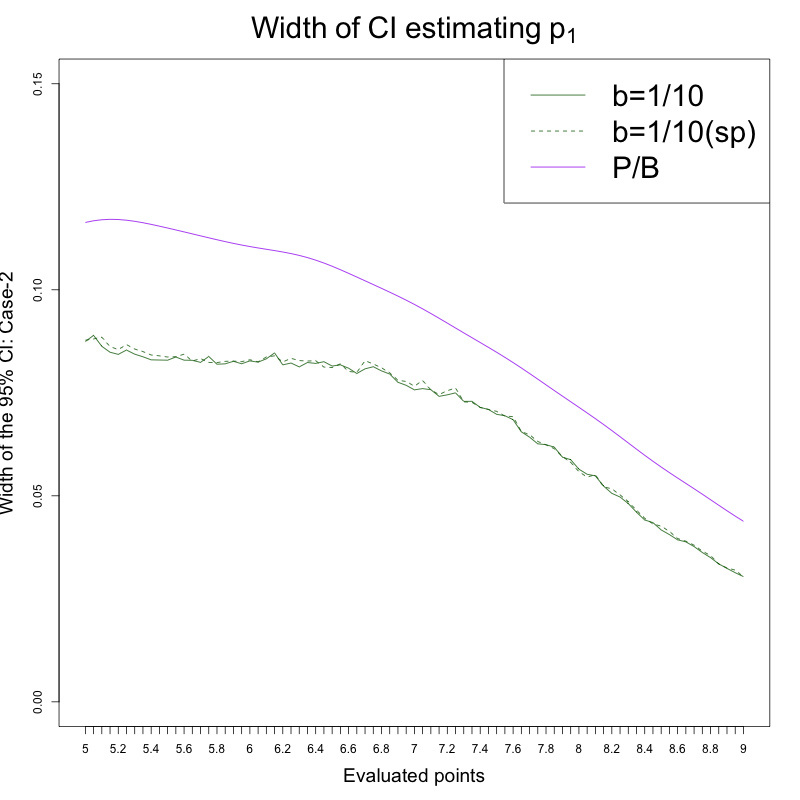}}
	\end{subfigure}\\
	\begin{subfigure}[$p_1$; C2; $n=2500$]
		{\label{wd2500WM1}\includegraphics[width=50mm]{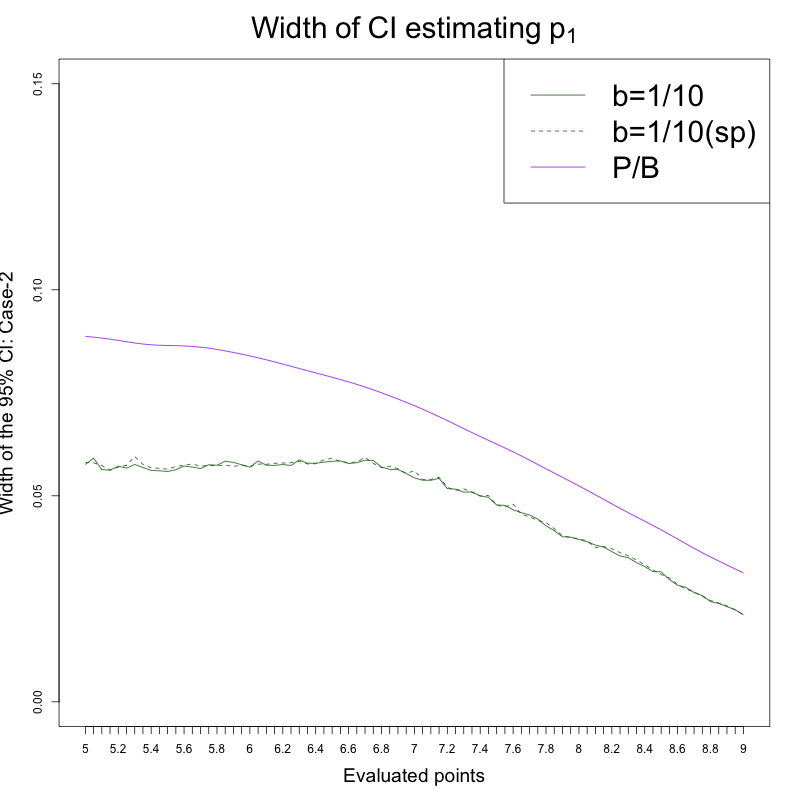}}
	\end{subfigure}
	\hspace{0.45cm}
	\begin{subfigure}[$p_1$; C2; $n=4000$]{\label{wd4000WM1}\includegraphics[width=50mm]{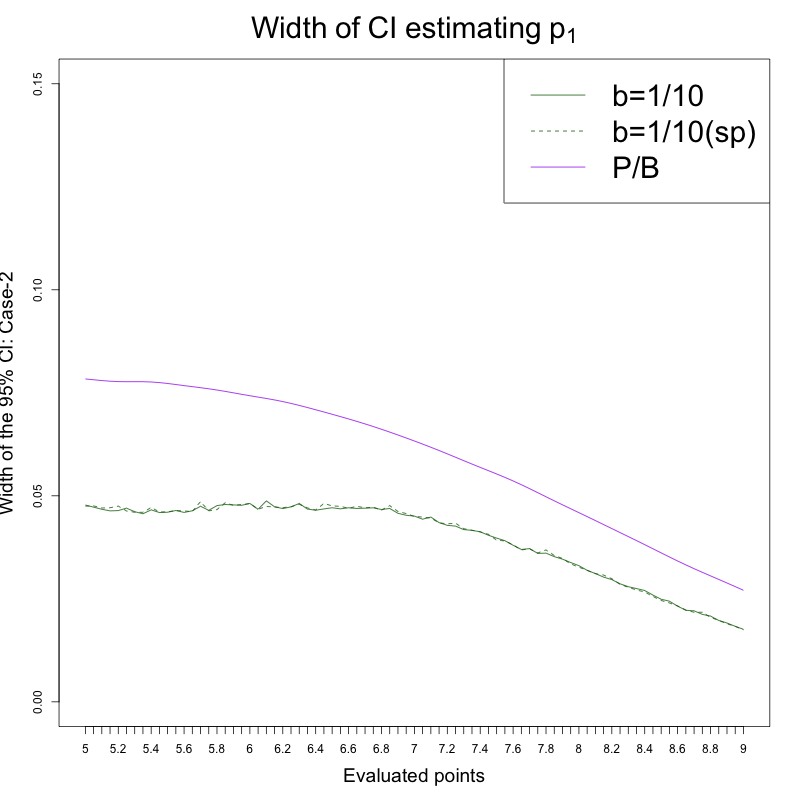}}
	\end{subfigure}\\
	\begin{subfigure}[$p_1$; C2; $n=6000$]
		{\label{wd6000WM1}\includegraphics[width=50mm]{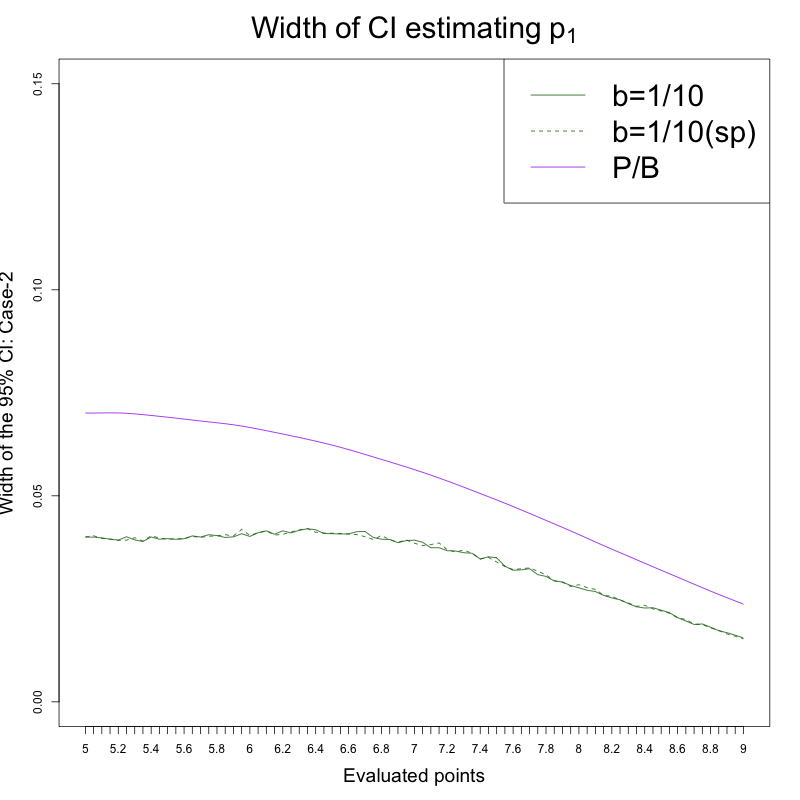}}
	\end{subfigure}
	\hspace{0.45cm}
	\begin{subfigure}[$p_1$; C2; $n=8000$]
		{\label{wd8000WM1}\includegraphics[width=50mm]{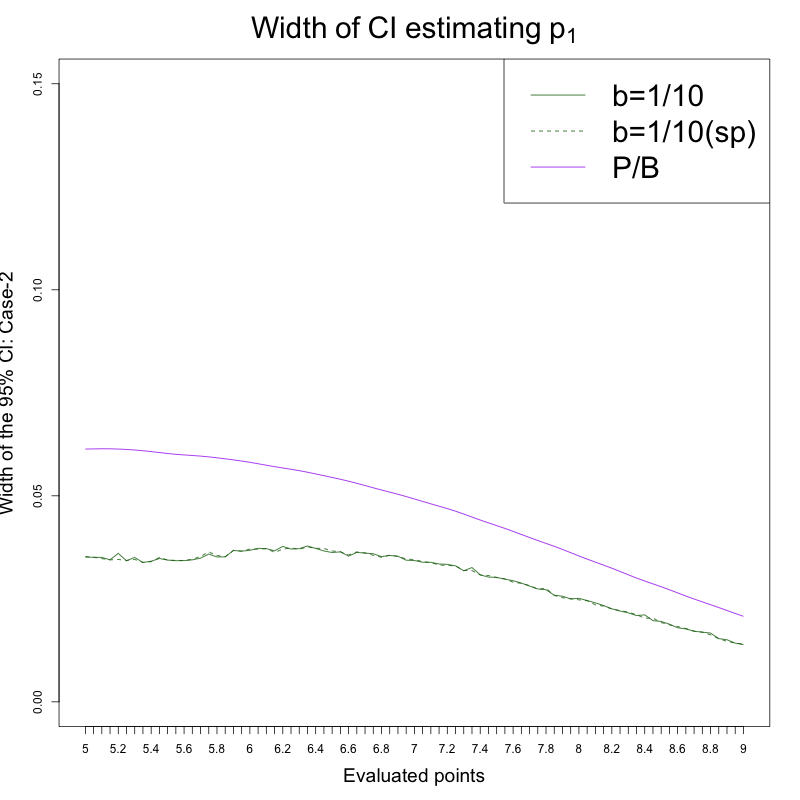}}
	\end{subfigure}
	\caption{\label{fig:wplots.case.2.p1} Notational details can be found in Figure \ref{fig:wplots.case.1.p1}.}
\end{figure}

\begin{figure}
	\centering
	\begin{subfigure}[$p_1$; C3; $n=500$]{\label{wd500MW1}\includegraphics[width=50mm]{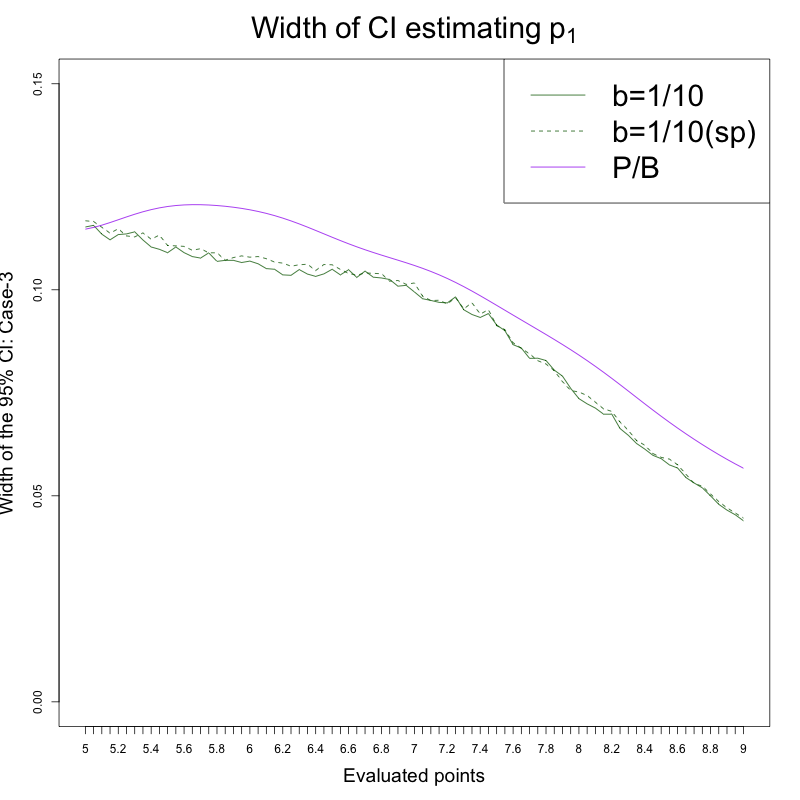}}
	\end{subfigure}
	\hspace{0.45cm}
	\begin{subfigure}[$p_1$; C3; $n=1000$]
		{\label{wd1000MW1}\includegraphics[width=50mm]{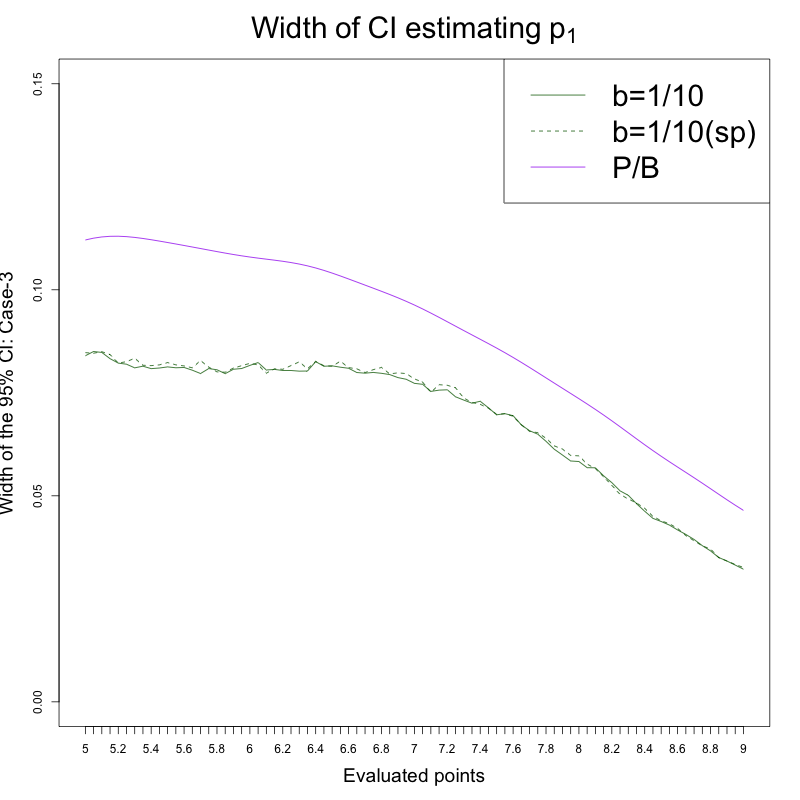}}
	\end{subfigure}\\
	\begin{subfigure}[$p_1$; C3; $n=2500$]
		{\label{wd2500MW1}\includegraphics[width=50mm]{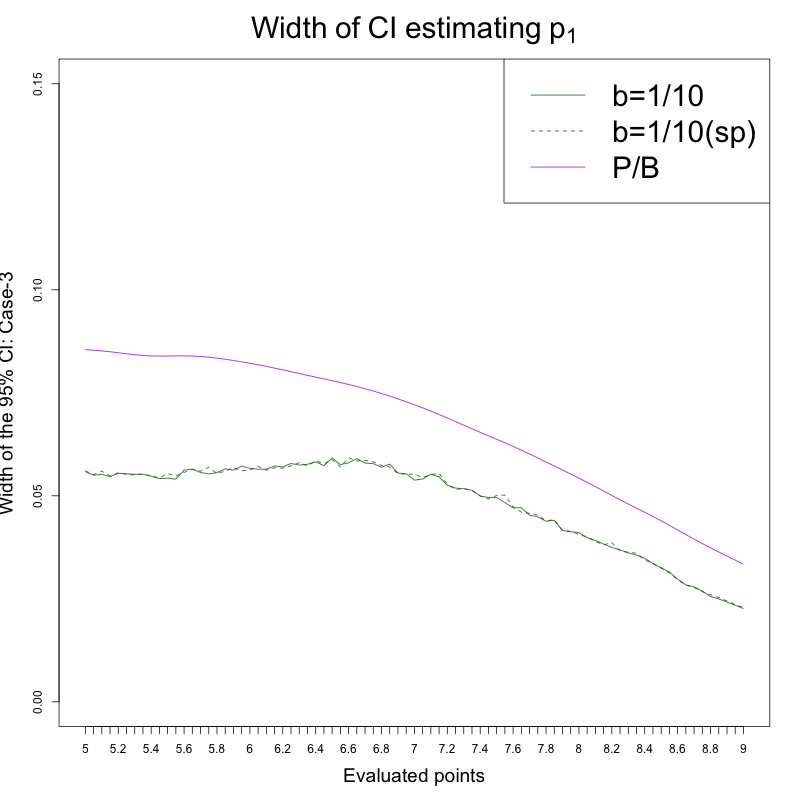}}
	\end{subfigure}
	\hspace{0.45cm}
	\begin{subfigure}[$p_1$; C3; $n=4000$]{\label{wd4000MW1}\includegraphics[width=50mm]{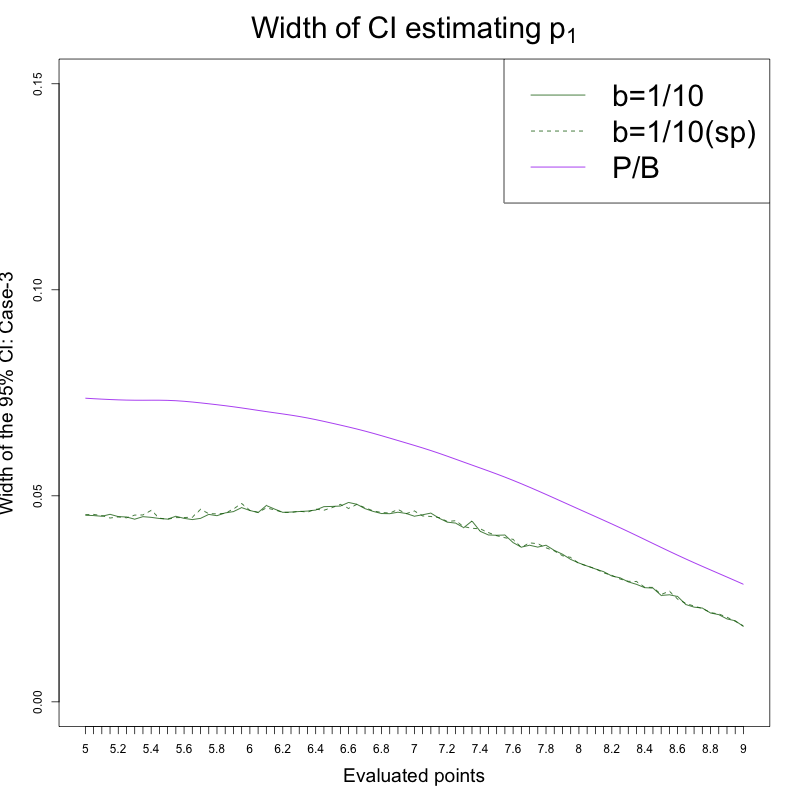}}
	\end{subfigure}\\
	\begin{subfigure}[$p_1$; C3; $n=6000$]
		{\label{wd6000MW1}\includegraphics[width=50mm]{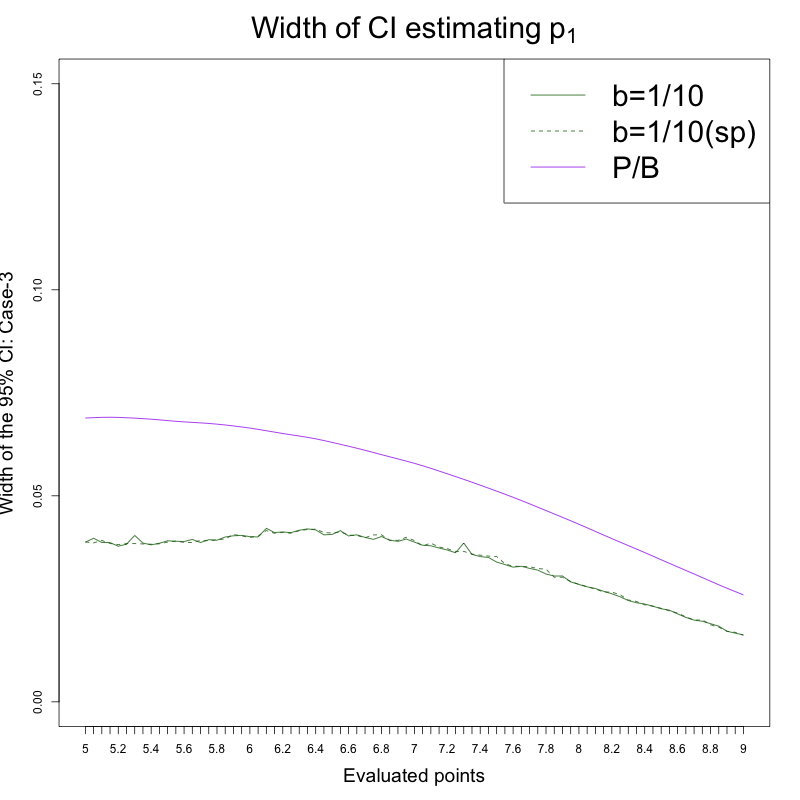}}
	\end{subfigure}
	\hspace{0.45cm}
	\begin{subfigure}[$p_1$; C3; $n=8000$]
		{\label{wd8000MW1}\includegraphics[width=50mm]{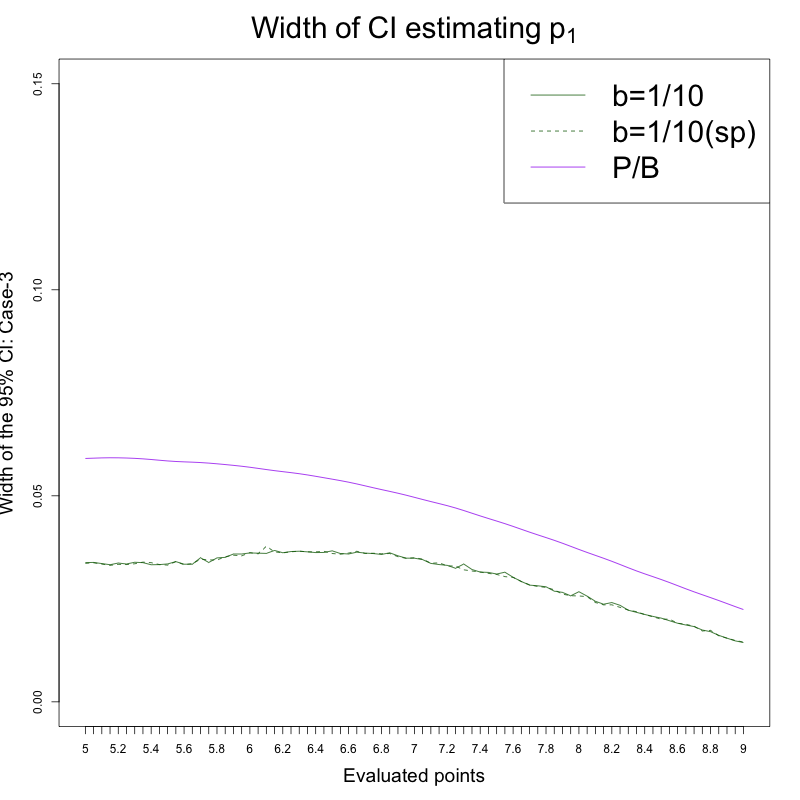}}
	\end{subfigure}
	\caption{\label{fig:wplots.case.3.p1}  Notational details can be found in Figure \ref{fig:wplots.case.1.p1}.}
\end{figure}

\begin{figure}
	\centering
	\begin{subfigure}[$p_0$; C1; $n=500$]{\label{wd500WW0}\includegraphics[width=50mm]{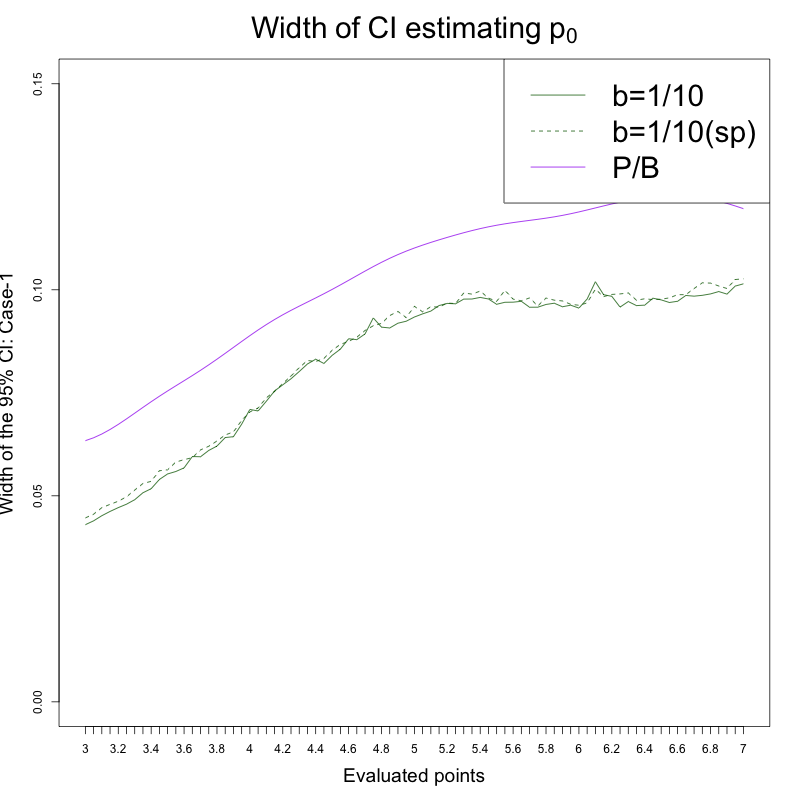}}
	\end{subfigure}
	\hspace{0.45cm}
	\begin{subfigure}[$p_0$; C1; $n=1000$]
		{\label{wd1000WW0}\includegraphics[width=50mm]{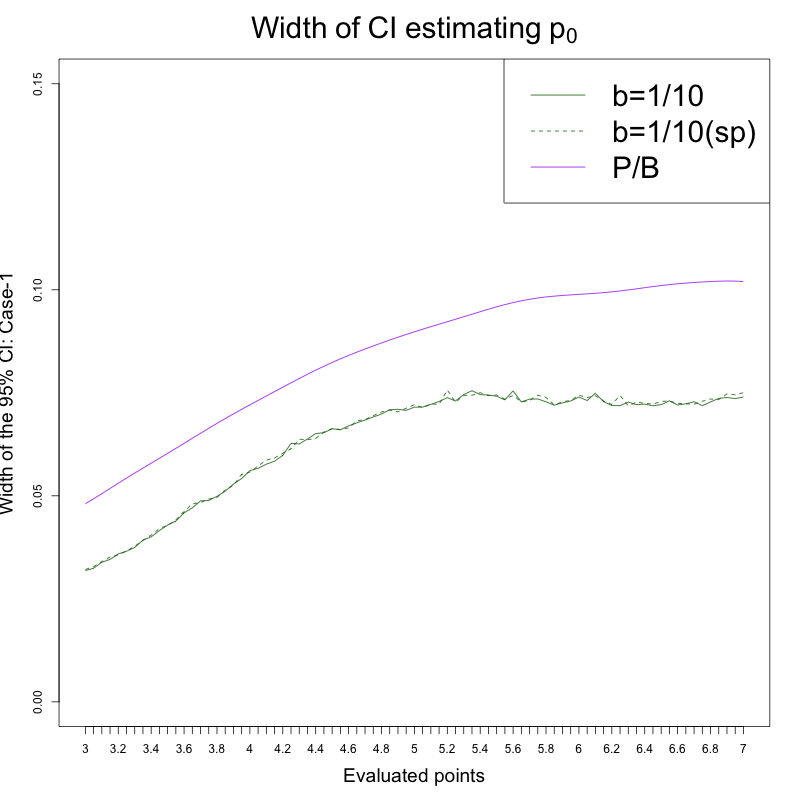}}
	\end{subfigure}\\
	\begin{subfigure}[$p_0$; C1; $n=2500$]
		{\label{wd2500WW0}\includegraphics[width=50mm]{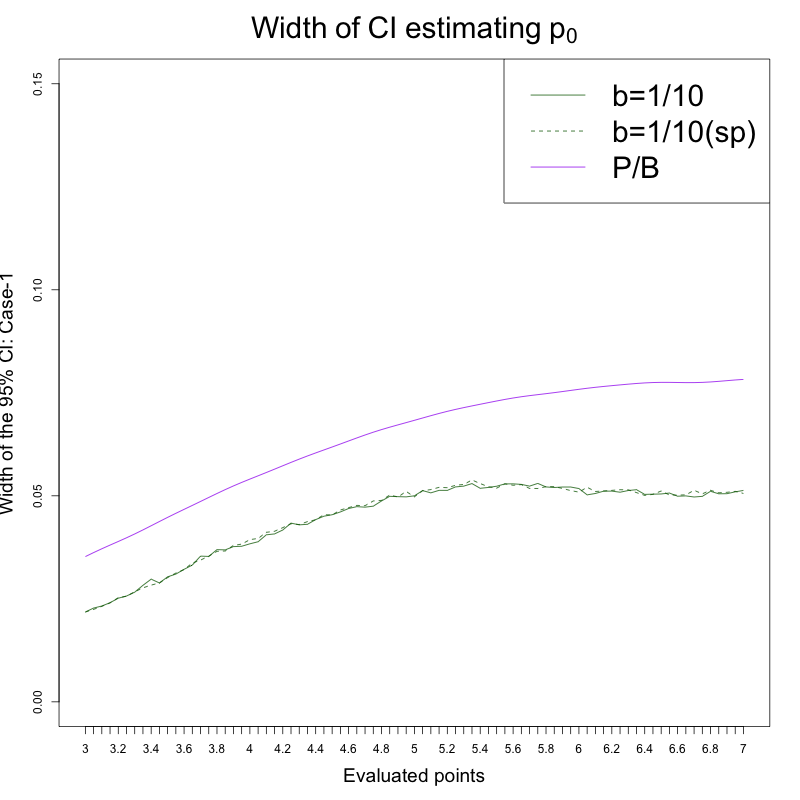}}
	\end{subfigure}
	\hspace{0.45cm}    
	\begin{subfigure}[$p_0$; C1; $n=4000$]{\label{wd4000WW0}\includegraphics[width=50mm]{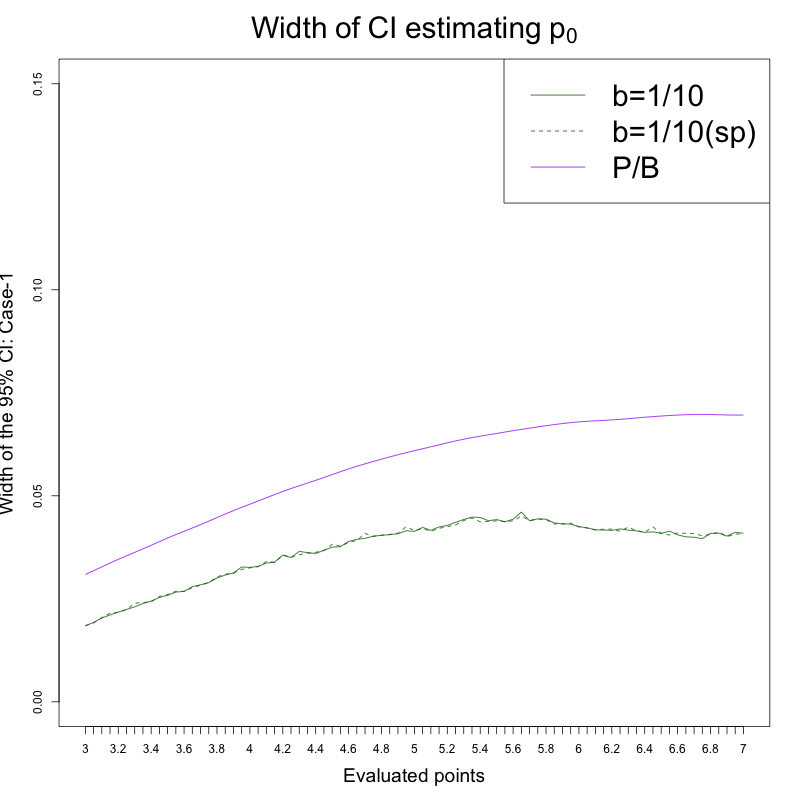}}
	\end{subfigure}\\
	\begin{subfigure}[$p_0$; C1; $n=6000$]
		{\label{wd6000WW0}\includegraphics[width=50mm]{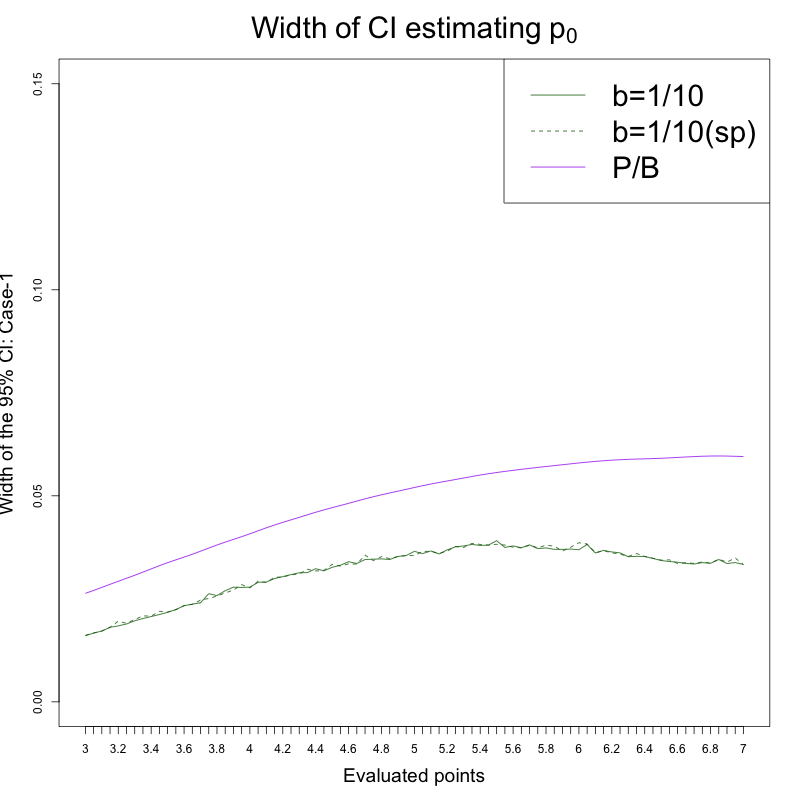}}
	\end{subfigure}
	\hspace{0.45cm}
	\begin{subfigure}[$p_0$; C1; $n=8000$]
		{\label{wd8000WW0}\includegraphics[width=50mm]{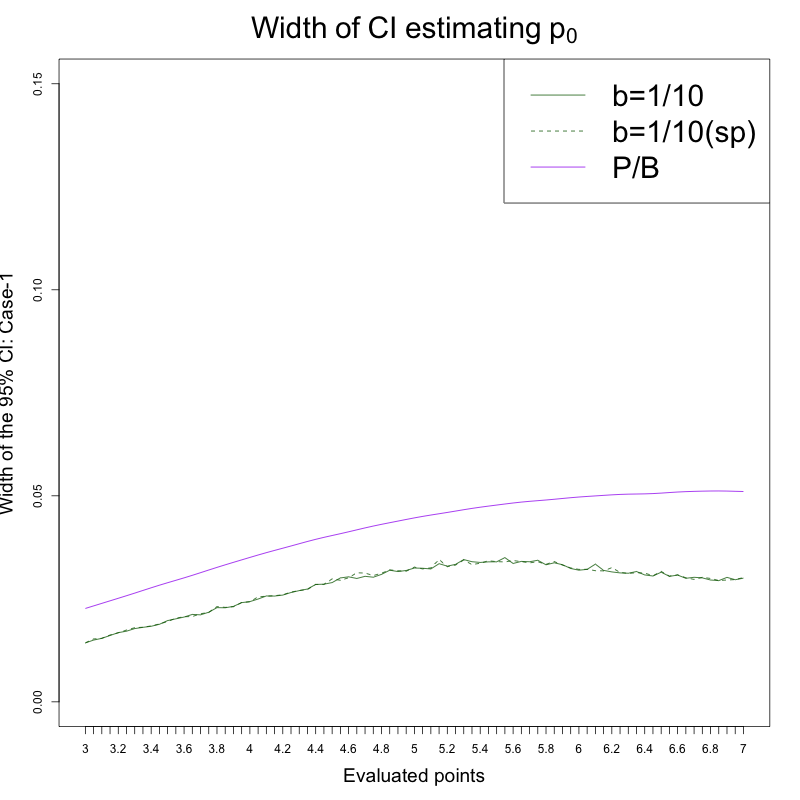}}
	\end{subfigure}
	\caption{\label{fig:wplots.case.1.p0} Notational details can be found in Figure \ref{fig:wplots.case.1.p1}.}
\end{figure}

\begin{figure}
	\centering
	\begin{subfigure}[$p_0$; C2; $n=500$]{\label{wd500WM0}\includegraphics[width=50mm]{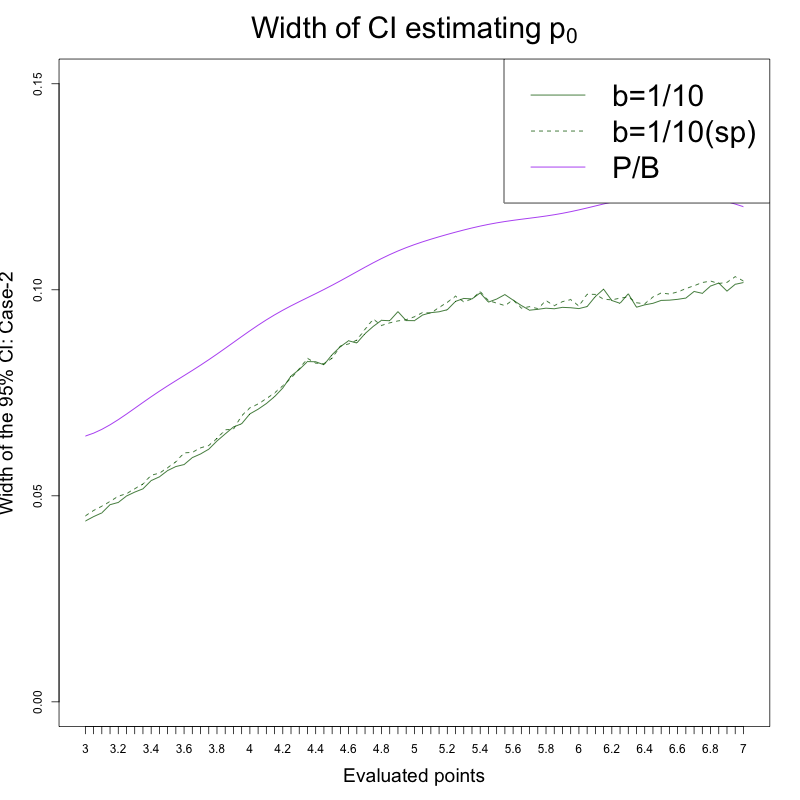}}
	\end{subfigure}
	\hspace{0.45cm}
	\begin{subfigure}[$p_0$; C2; $n=1000$]
		{\label{wd1000WM0}\includegraphics[width=50mm]{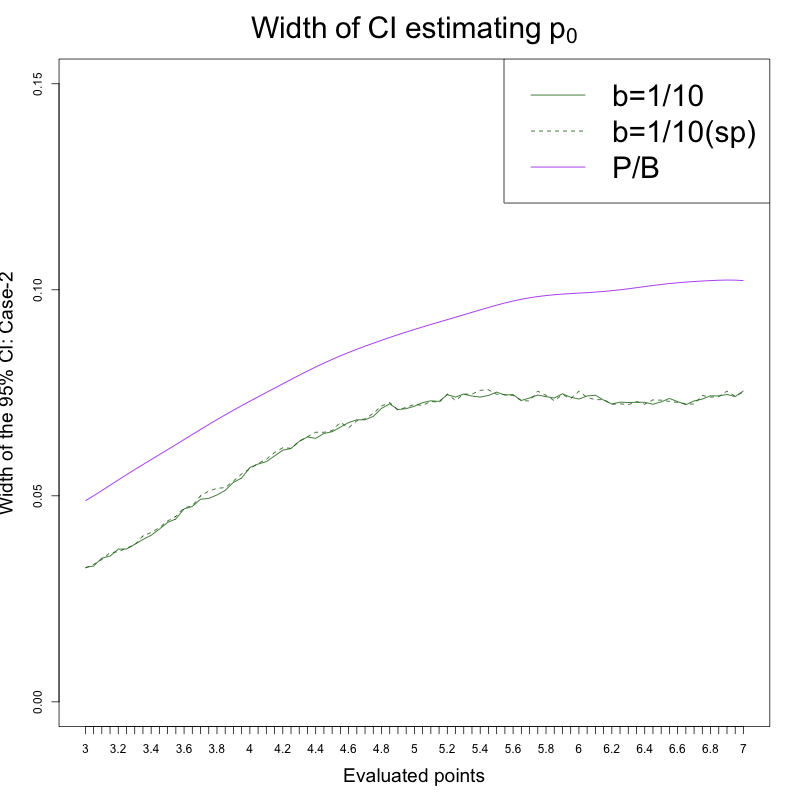}}
	\end{subfigure}\\
	\begin{subfigure}[$p_0$; C2; $n=2500$]
		{\label{wd2500WM0}\includegraphics[width=50mm]{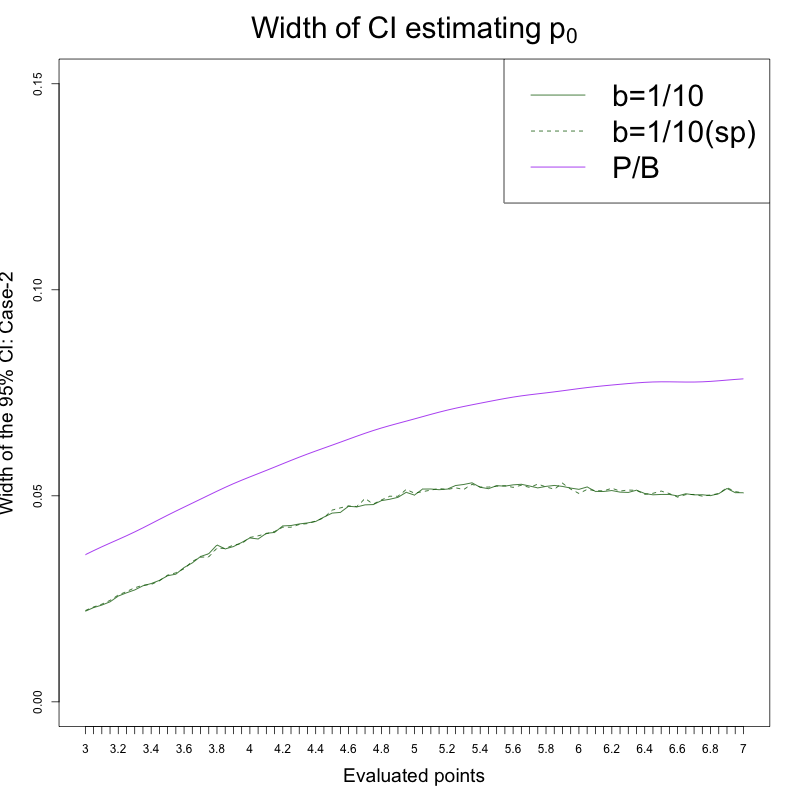}}
	\end{subfigure}
	\hspace{0.45cm}    
	\begin{subfigure}[$p_0$; C2; $n=4000$]{\label{wd4000WM0}\includegraphics[width=50mm]{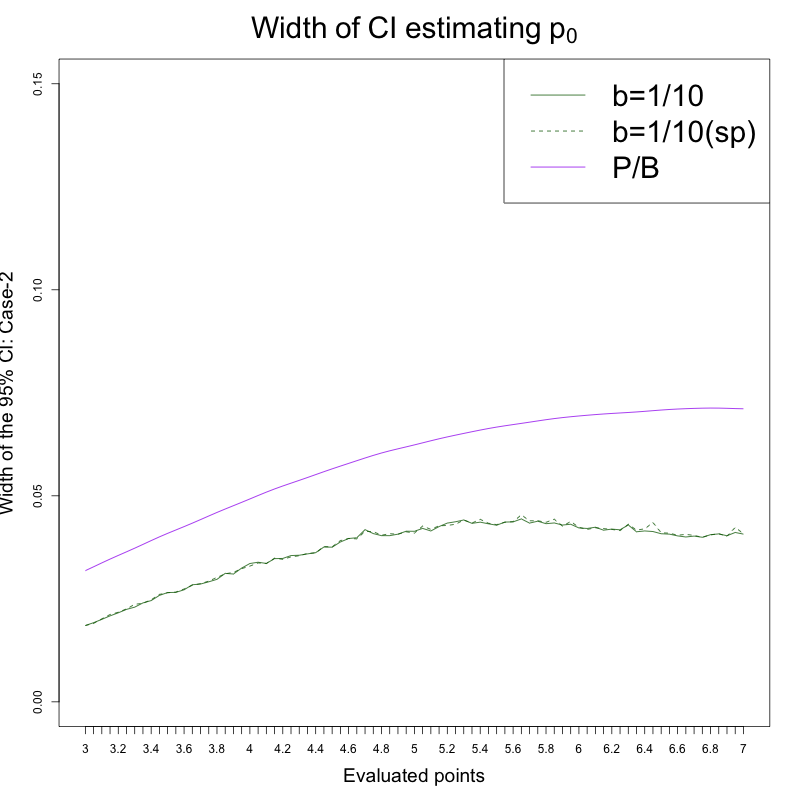}}
	\end{subfigure}\\
	\begin{subfigure}[$p_0$; C2; $n=6000$]
		{\label{wd6000WM0}\includegraphics[width=50mm]{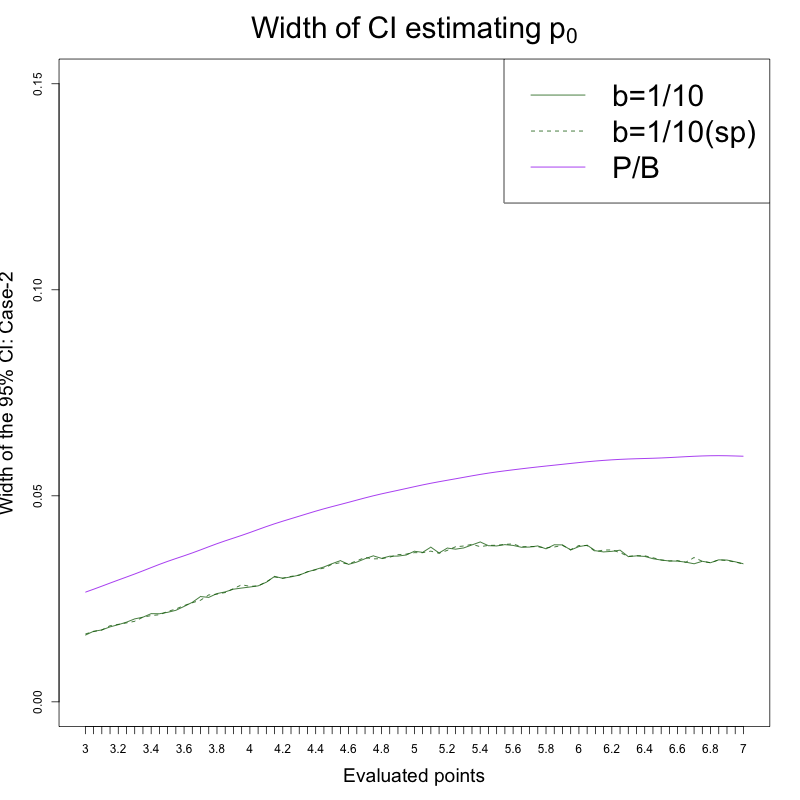}}
	\end{subfigure}
	\hspace{0.45cm}
	\begin{subfigure}[$p_0$; C2; $n=8000$]
		{\label{wd8000WM0}\includegraphics[width=50mm]{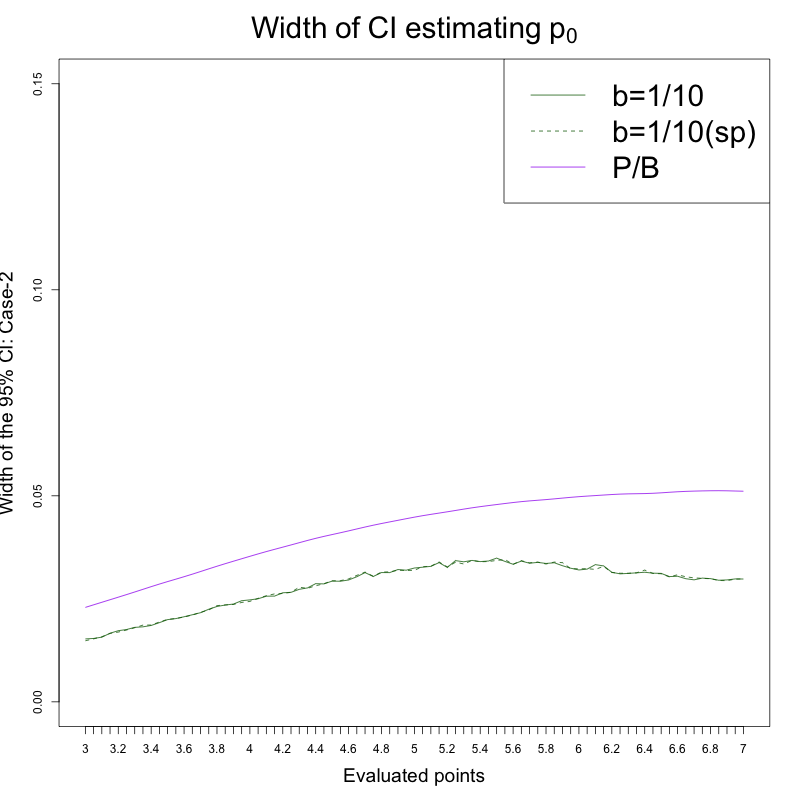}}
	\end{subfigure}
	\caption{\label{fig:wplots.case.2.p0} Notational details can be found in Figure \ref{fig:wplots.case.1.p1}.}
\end{figure}

\begin{figure}
	\centering
	\begin{subfigure}[$p_0$; C3; $n=500$]{\label{wd500MW0}\includegraphics[width=50mm]{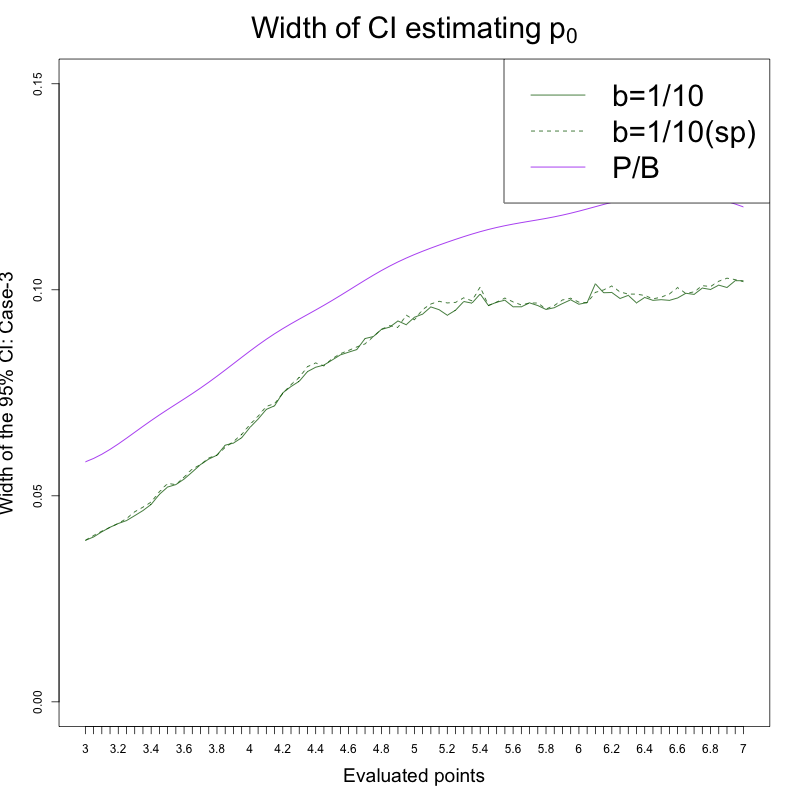}}
	\end{subfigure}
	\hspace{0.45cm}
	\begin{subfigure}[$p_0$; C3; $n=1000$]
		{\label{wd1000MW0}\includegraphics[width=50mm]{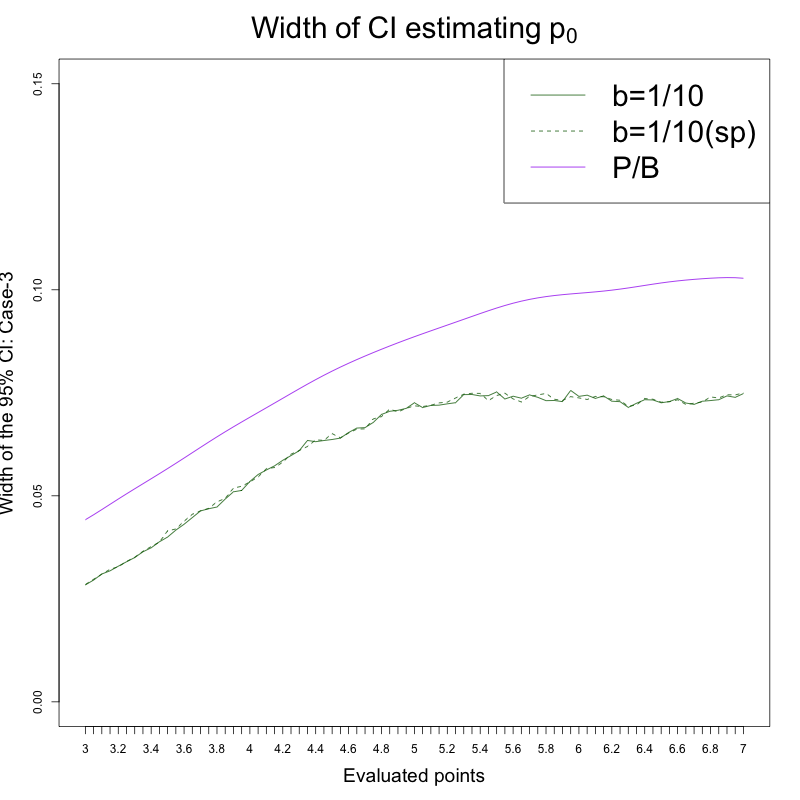}}
	\end{subfigure}\\
	\begin{subfigure}[$p_0$; C3; $n=2500$]
		{\label{wd2500MW0}\includegraphics[width=50mm]{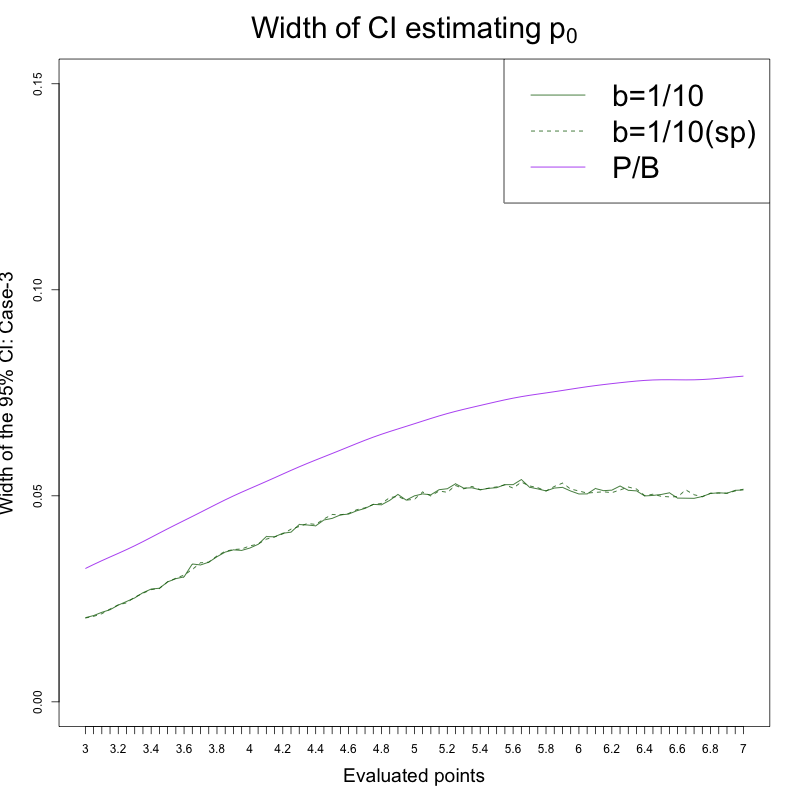}}
	\end{subfigure}
	\hspace{0.45cm}
	\begin{subfigure}[$p_0$; C3; $n=4000$]{\label{wd4000MW0}\includegraphics[width=50mm]{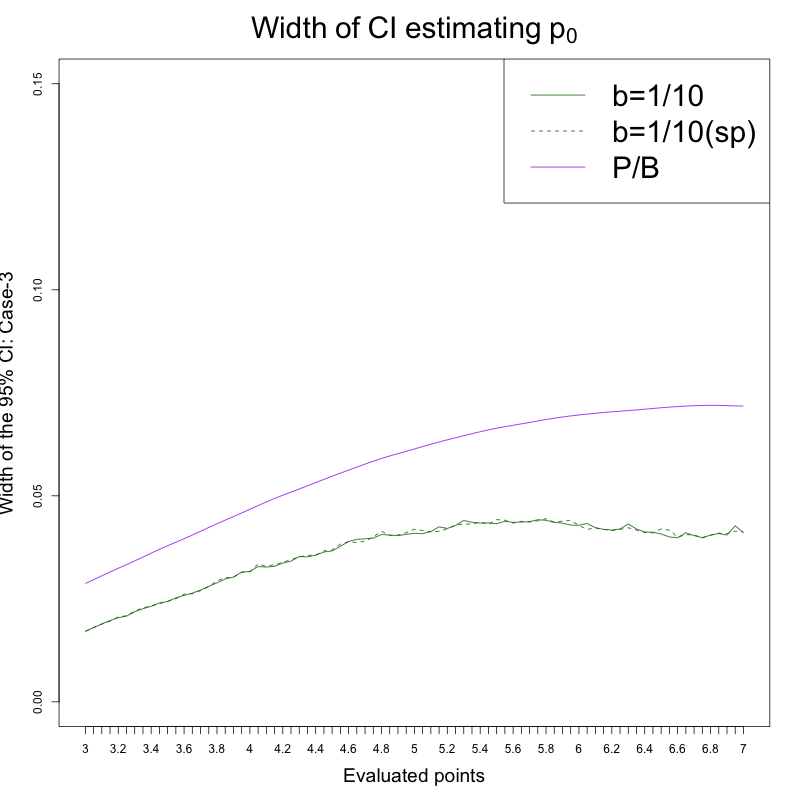}}
	\end{subfigure}\\
	\begin{subfigure}[$p_0$; C3; $n=6000$]
		{\label{wd6000MW0}\includegraphics[width=50mm]{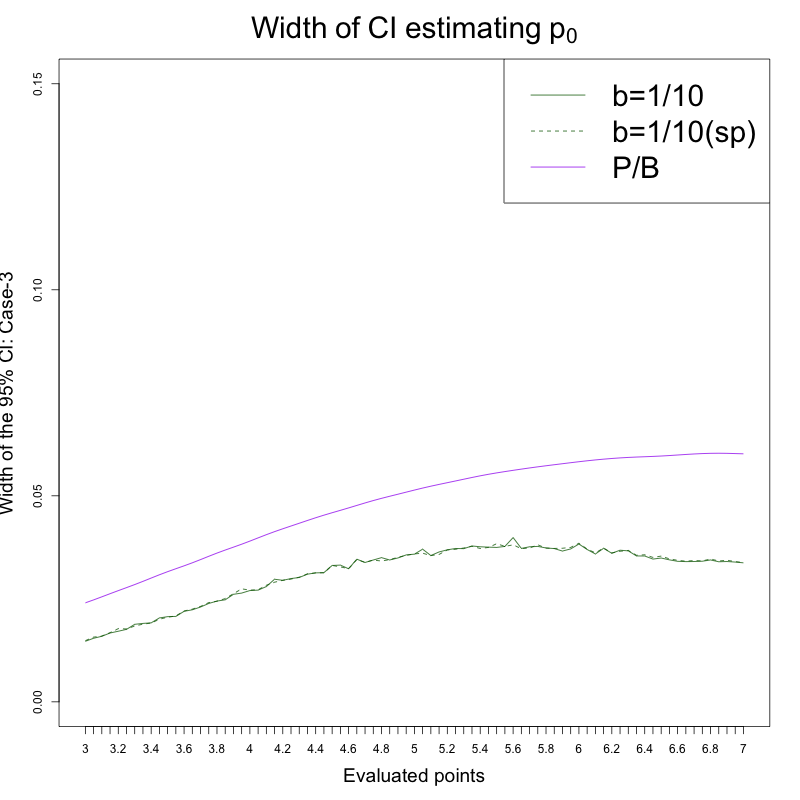}}
	\end{subfigure}
	\hspace{0.45cm}
	\begin{subfigure}[$p_0$; C3; $n=8000$]
		{\label{wd8000MW0}\includegraphics[width=50mm]{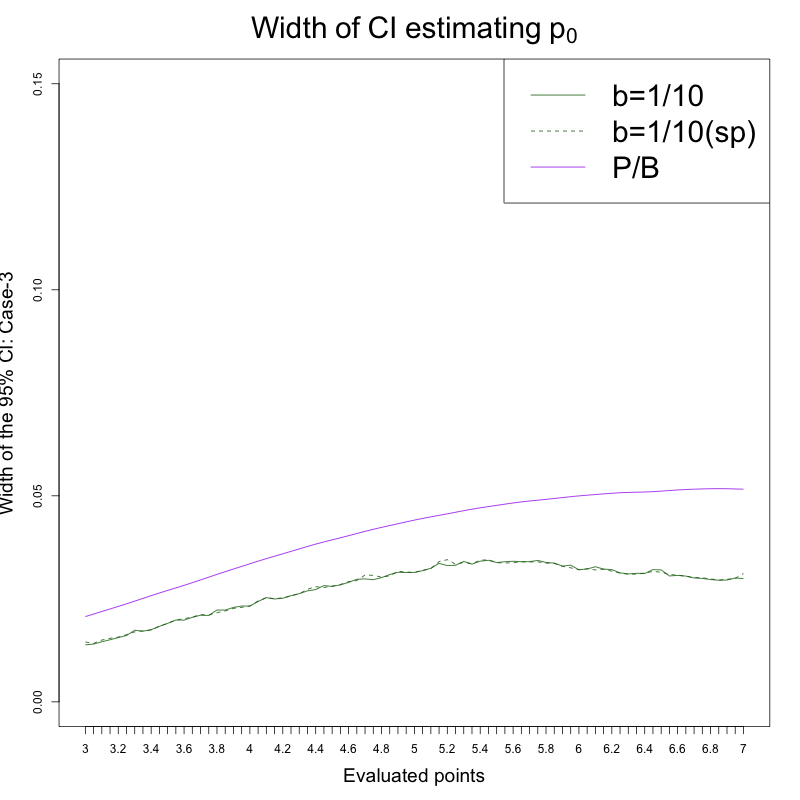}}
	\end{subfigure}
	\caption{\label{fig:wplots.case.3.p0}  Notational details can be found in Figure \ref{fig:wplots.case.1.p1}.}
\end{figure}

\begin{figure}
	\centering
	\begin{subfigure}[$p_1-p_0$; C1; $n=500$]{\label{cont500ww}\includegraphics[width=70mm]{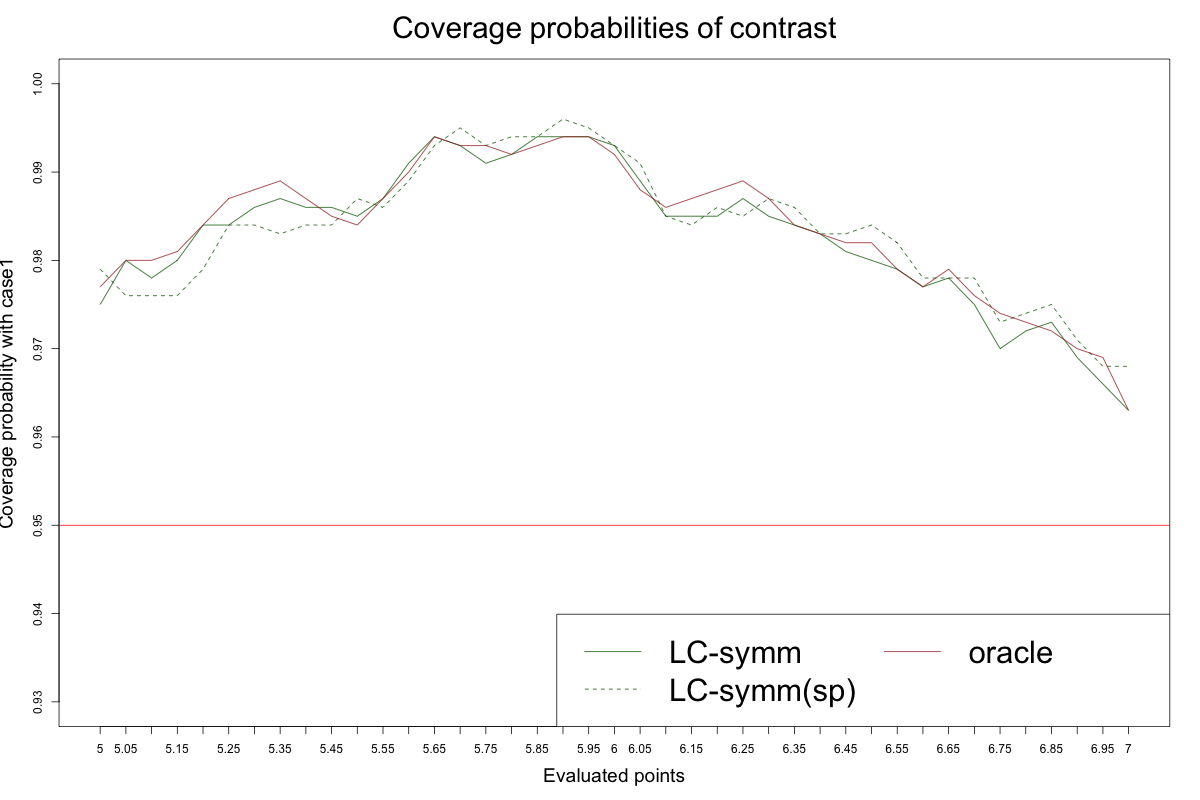}}
	\end{subfigure}
	\hspace{0.45cm}
	\begin{subfigure}[$p_1-p_0$; C1; $n=1000$]
		{\label{cont1000ww}\includegraphics[width=70mm]{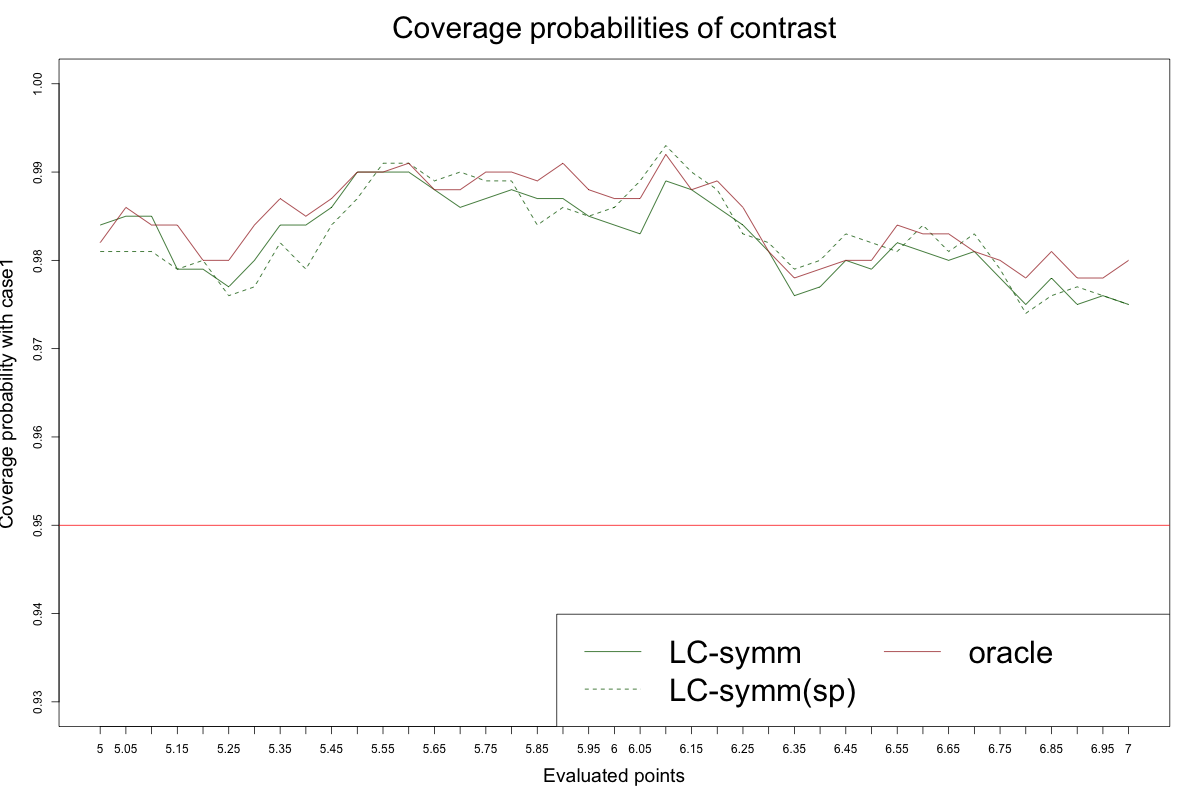}}
	\end{subfigure}\\
	\begin{subfigure}[$p_1-p_0$; C1; $n=2500$]
		{\label{cont2500ww}\includegraphics[width=70mm]{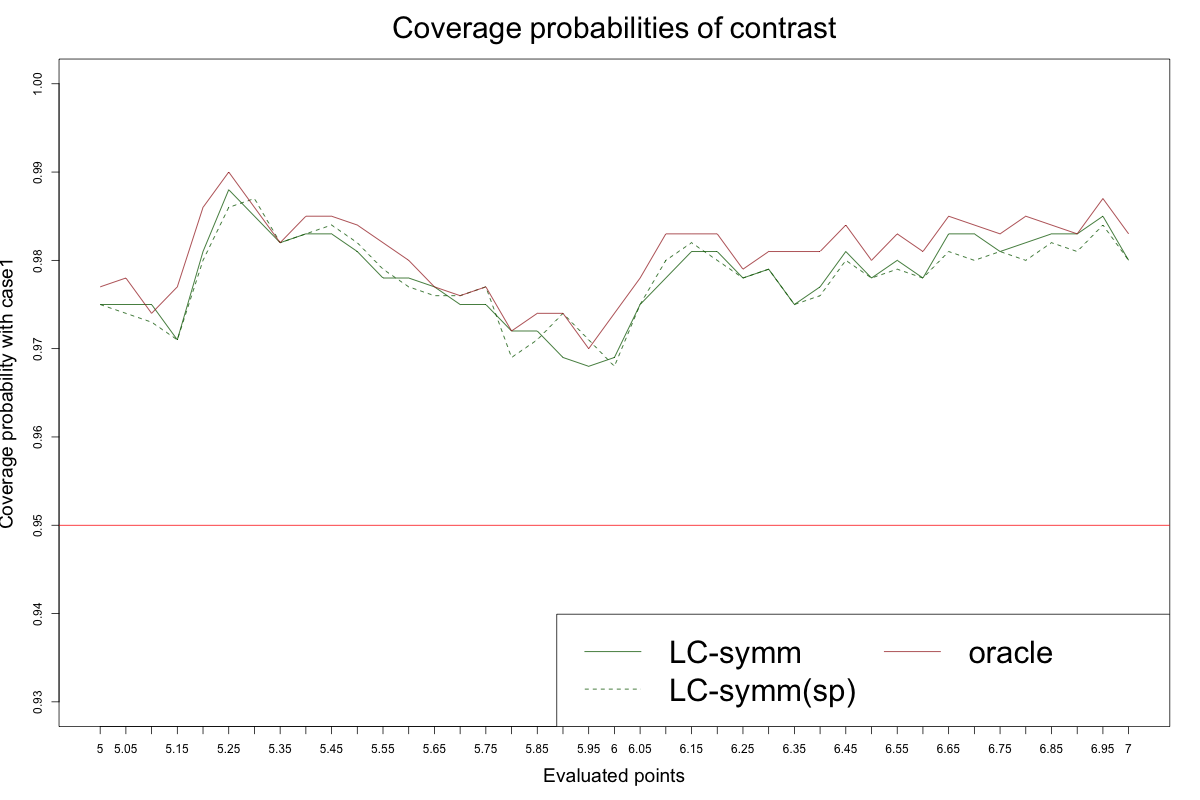}}
	\end{subfigure}
	\hspace{0.45cm}
	\begin{subfigure}[$p_1-p_0$; C1; $n=4000$]{\label{cont4000ww}\includegraphics[width=70mm]{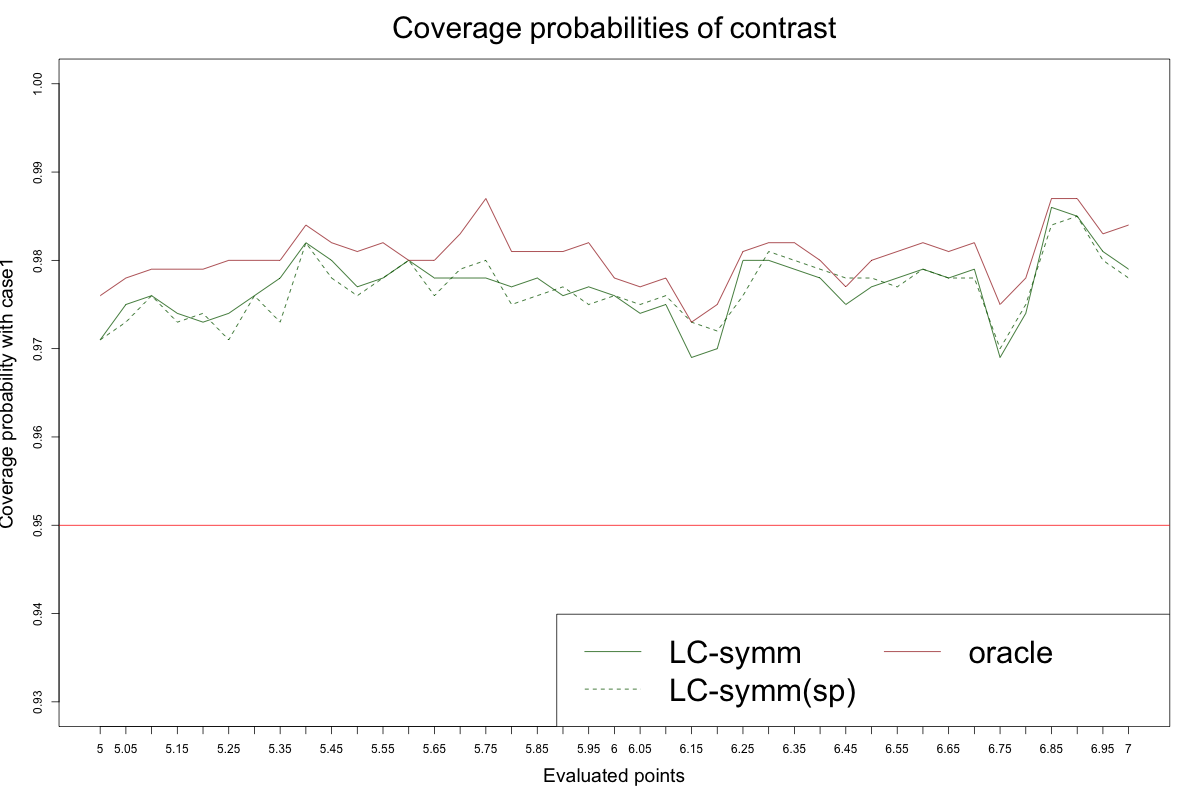}}
	\end{subfigure}\\
	\begin{subfigure}[$p_1-p_0$; C1; $n=6000$]
		{\label{cont6000ww}\includegraphics[width=70mm]{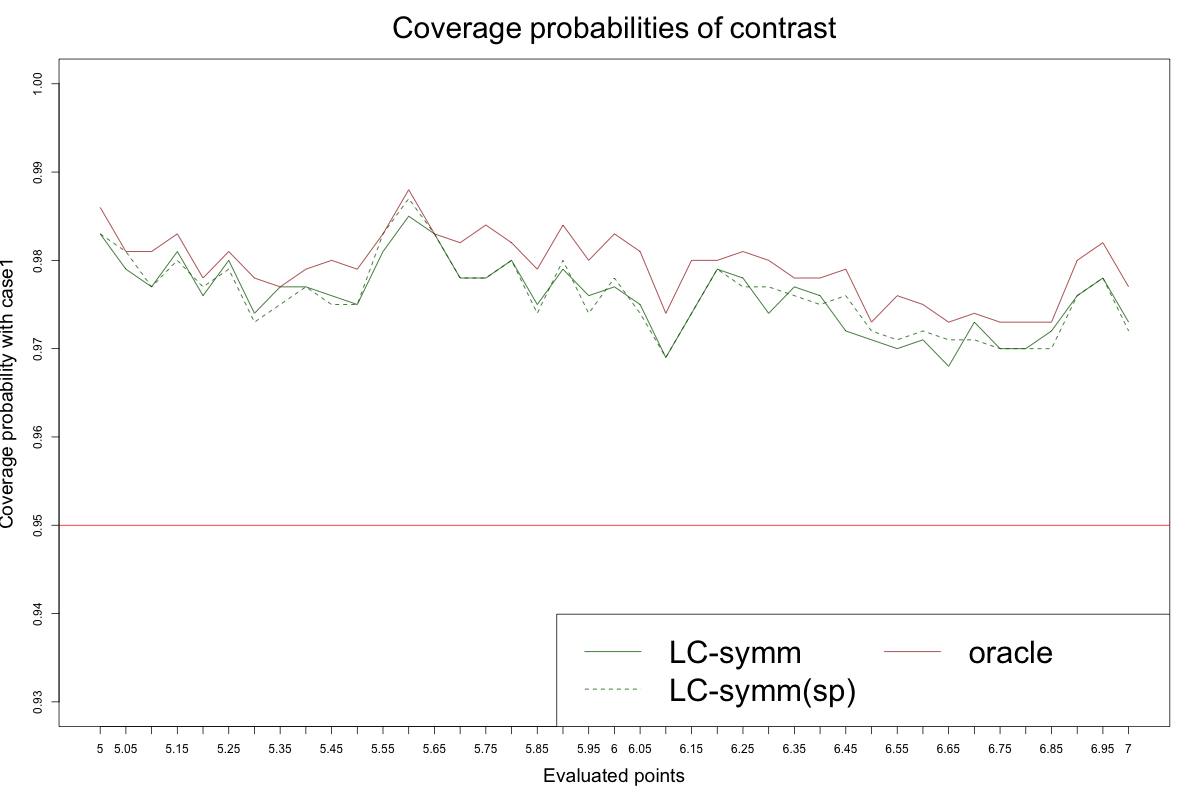}}
	\end{subfigure}
	\hspace{0.45cm}
	\begin{subfigure}[$p_1-p_0$; C1; $n=8000$]
		{\label{cont8000ww}\includegraphics[width=70mm]{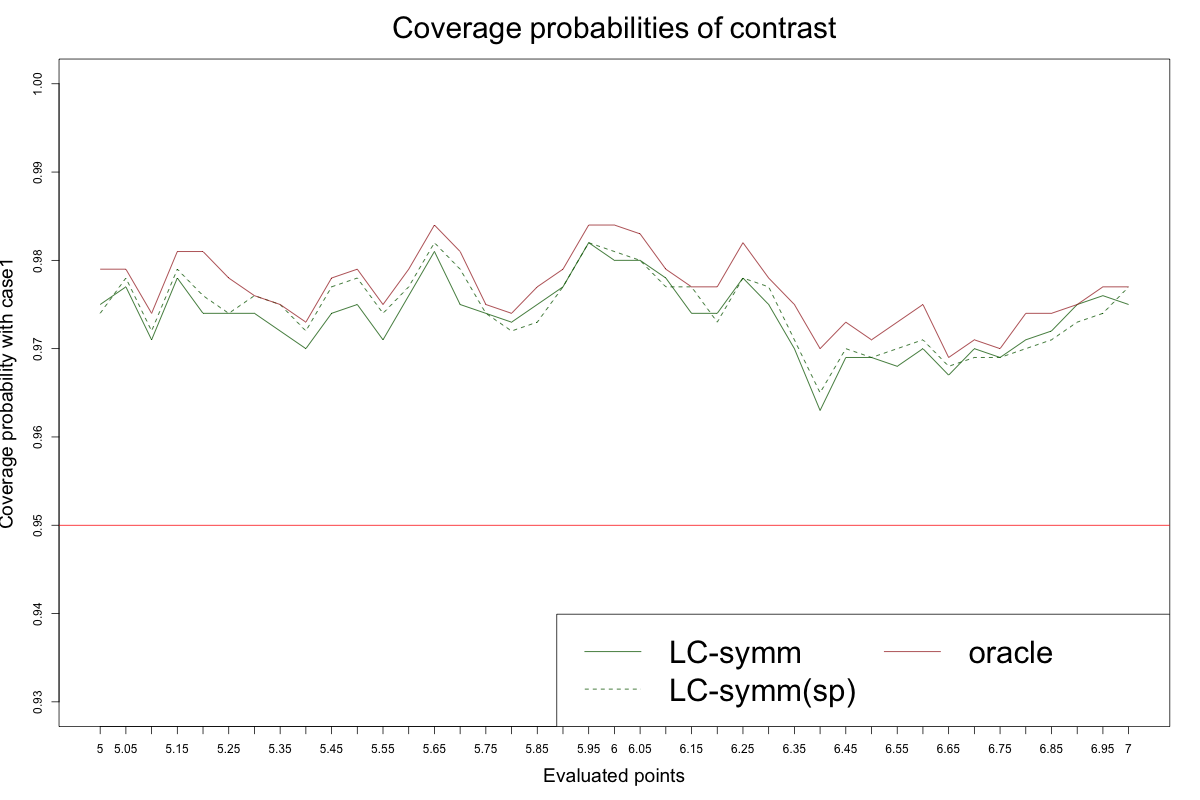}}
	\end{subfigure}
	\caption{\label{fig:contplots.ww} The above displays are coverage probabilities (for $p_1-p_0$ at 41 equally spaced point) with our proposed log-concave projection estimators' 95\% difference CI's with the suggested tuning parameter $b=1/10$ which is labeled as LC (see Section \ref{subsec:tuning.param}).
	For the Case 1 where both nuisance functions are well-specified, we also use true value of $\chi_{\theta_a}$ to construct the oracle 95\% difference CI which is labeled as oracle in the displays.
 Each subcaption describes the sample size, and each case of nuisance estimations (Case 1, 2, or 3 abbreviated to C1, C2, and C3, respectively).}
\end{figure}

\begin{figure}
	\centering
	\begin{subfigure}[$p_1-p_0$; C2; $n=500$]{\label{cont500wm}\includegraphics[width=70mm]{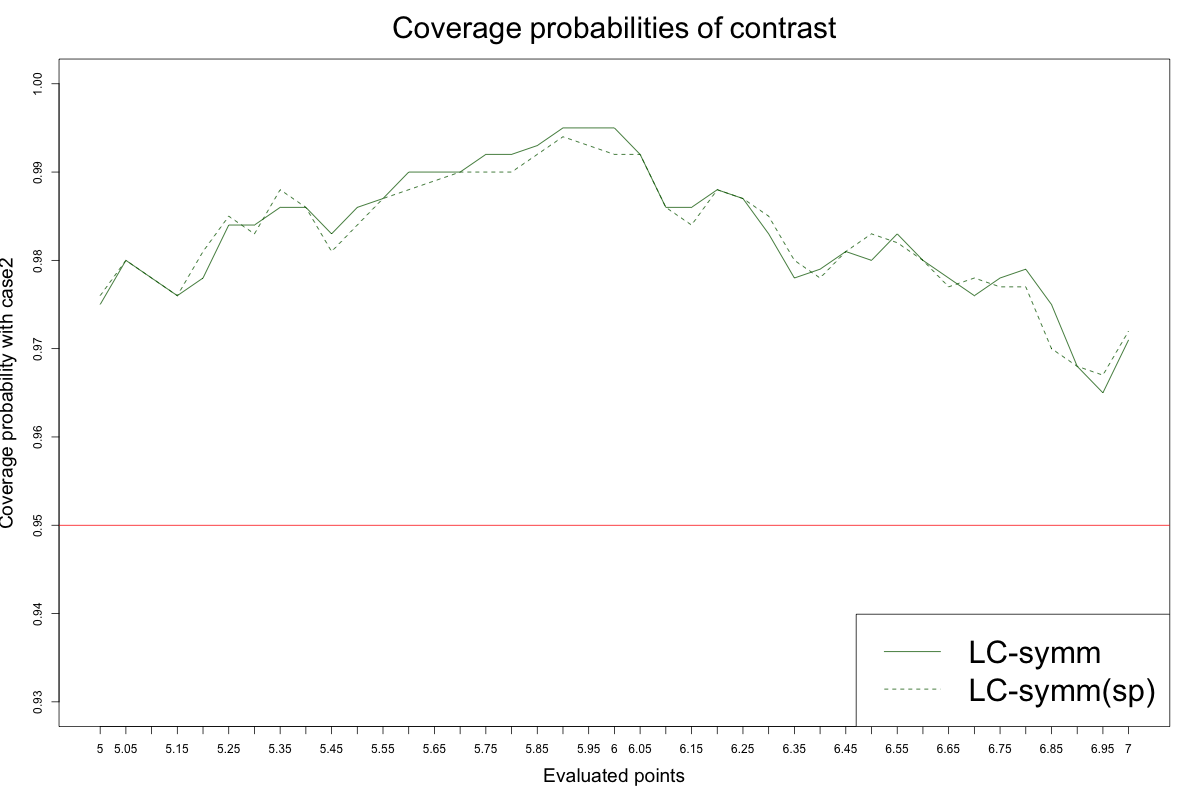}}
	\end{subfigure}
	\hspace{0.45cm}
	\begin{subfigure}[$p_1-p_0$; C2; $n=1000$]
		{\label{cont1000wm}\includegraphics[width=70mm]{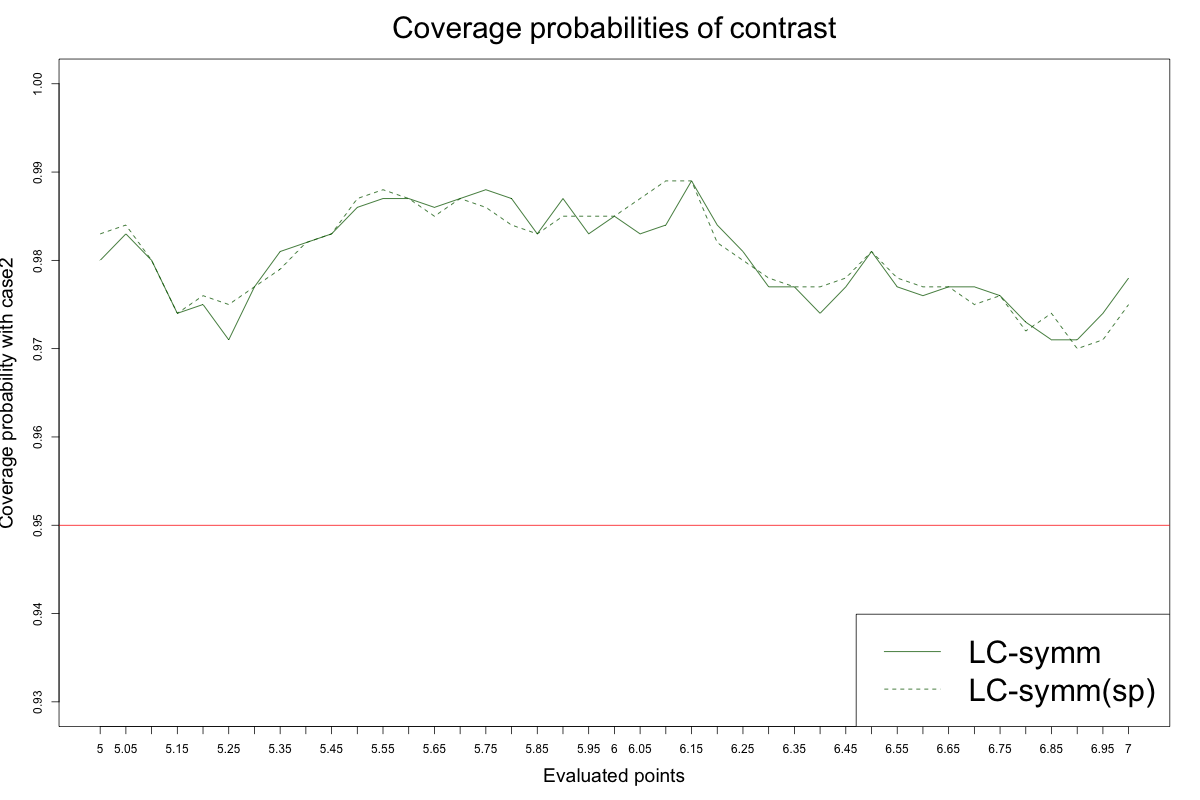}}
	\end{subfigure}\\
	\begin{subfigure}[$p_1-p_0$; C2; $n=2500$]
		{\label{cont2500wm}\includegraphics[width=70mm]{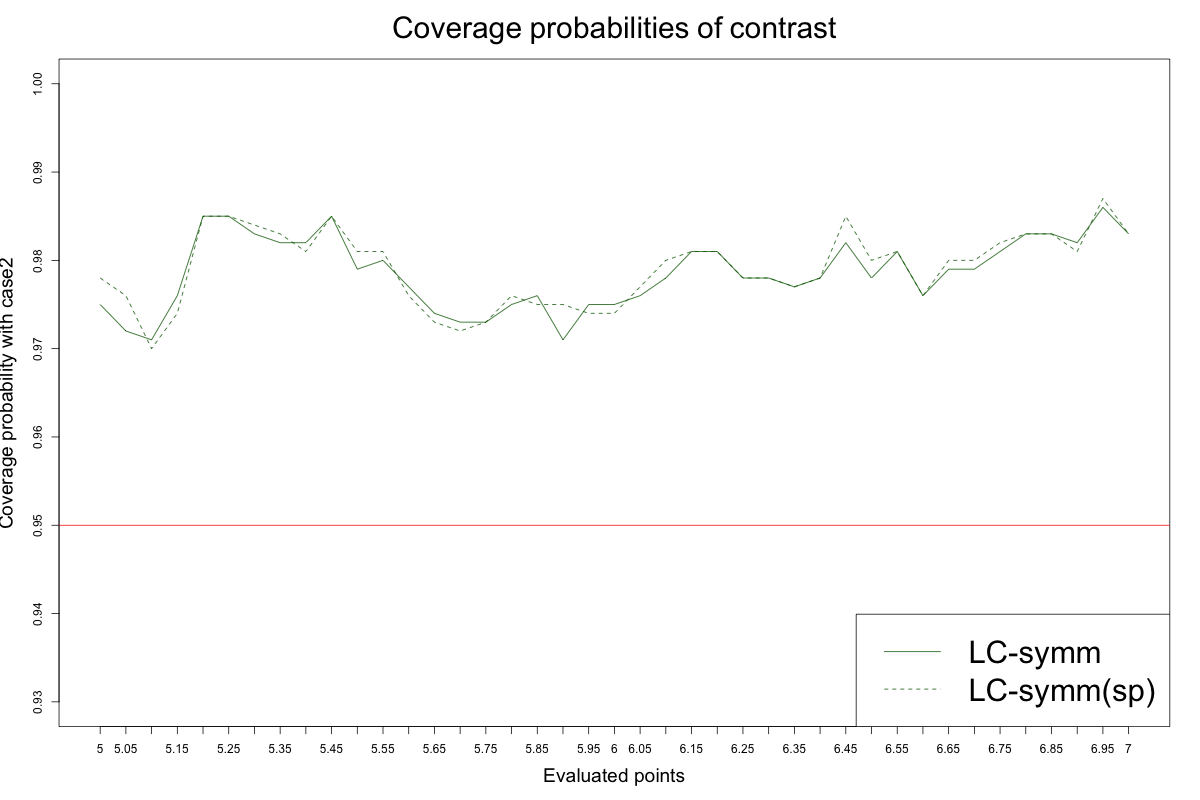}}
	\end{subfigure}
	\hspace{0.45cm}
	\begin{subfigure}[$p_1-p_0$; C2; $n=4000$]{\label{cont4000wm}\includegraphics[width=70mm]{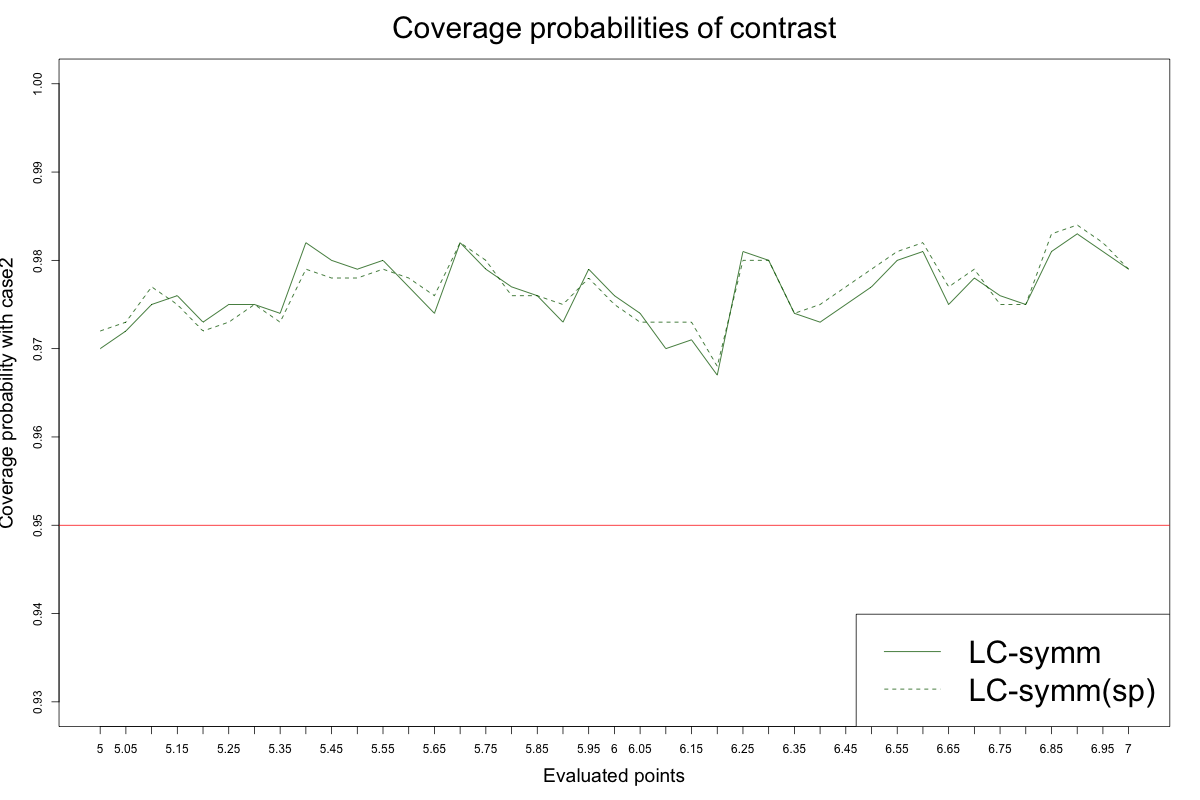}}
	\end{subfigure}\\
	\begin{subfigure}[$p_1-p_0$; C2; $n=6000$]
		{\label{cont6000wm}\includegraphics[width=70mm]{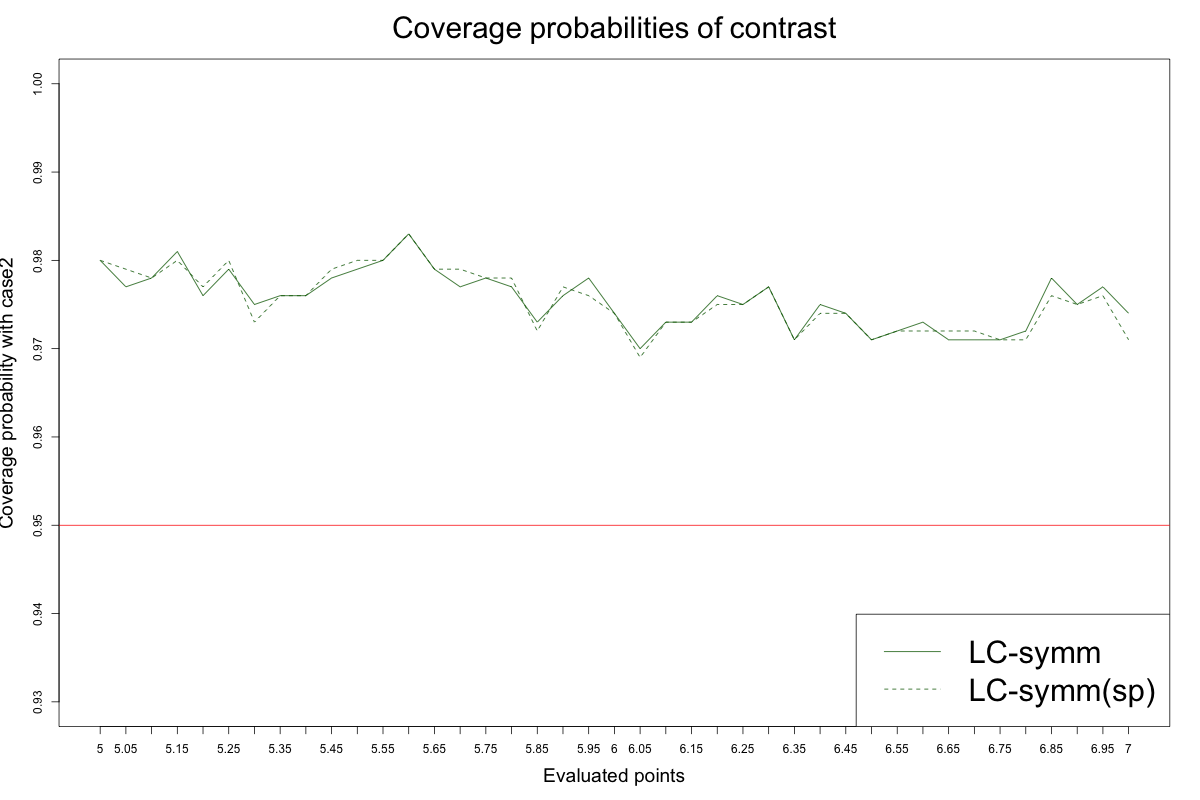}}
	\end{subfigure}
	\hspace{0.45cm}
	\begin{subfigure}[$p_1-p_0$; C2; $n=8000$]
		{\label{cont8000wm}\includegraphics[width=70mm]{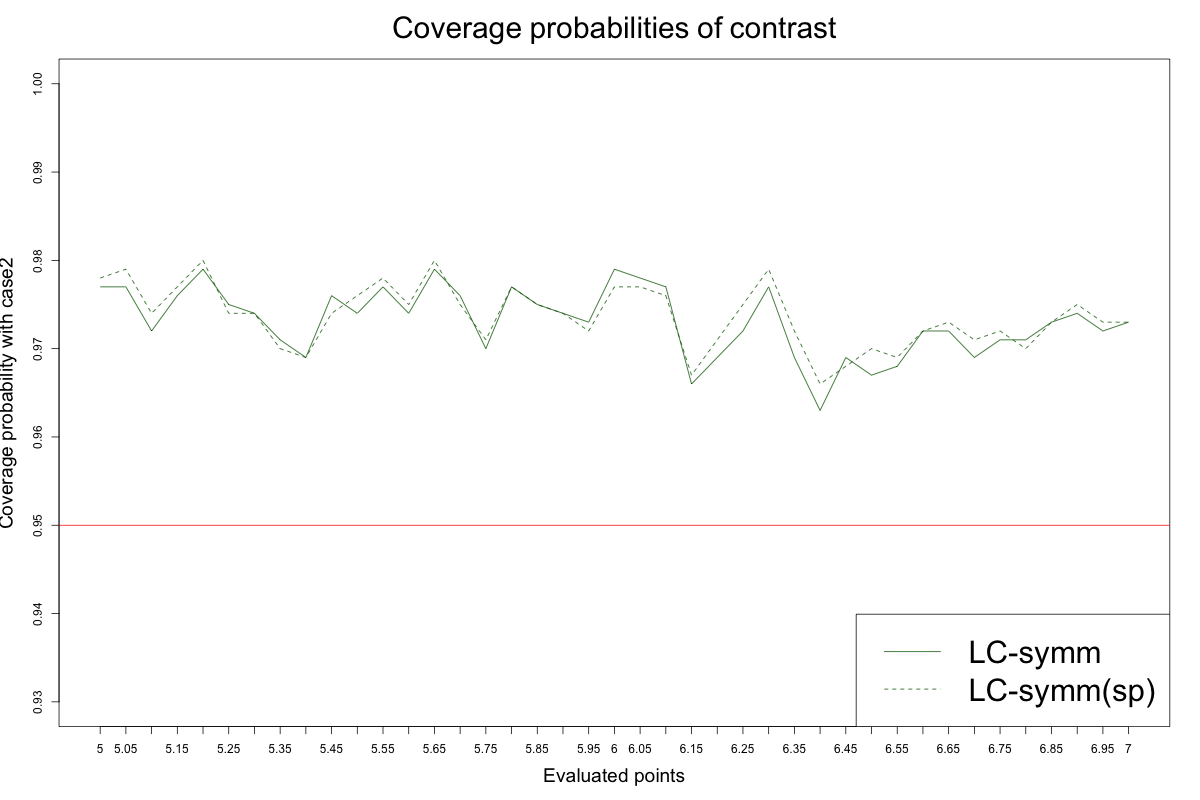}}
	\end{subfigure}
	\caption{\label{fig:contplots.wm} Notational details can be found in Figure \ref{fig:contplots.ww}.}
\end{figure}

\begin{figure}
	\centering
	\begin{subfigure}[$p_1-p_0$; C3; $n=500$]{\label{cont500mw}\includegraphics[width=70mm]{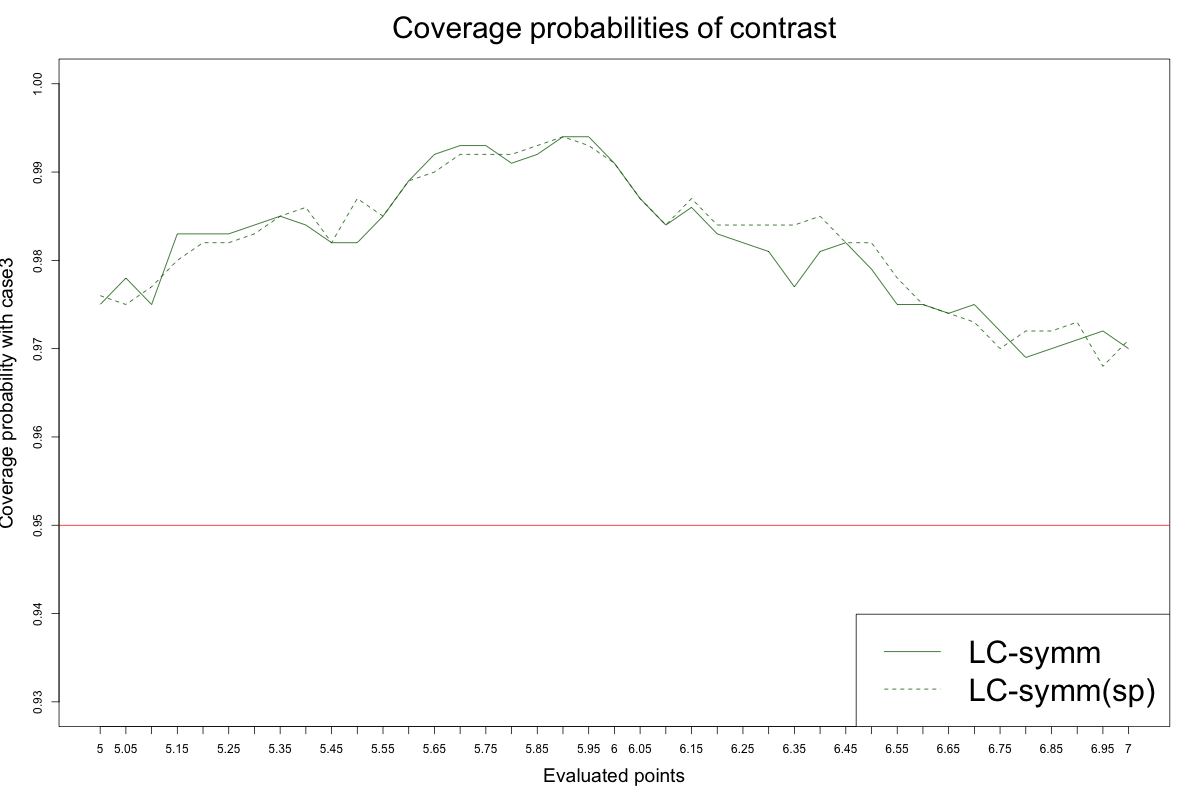}}
	\end{subfigure}
	\hspace{0.45cm}
	\begin{subfigure}[$p_1-p_0$; C3; $n=1000$]
		{\label{cont1000mw}\includegraphics[width=70mm]{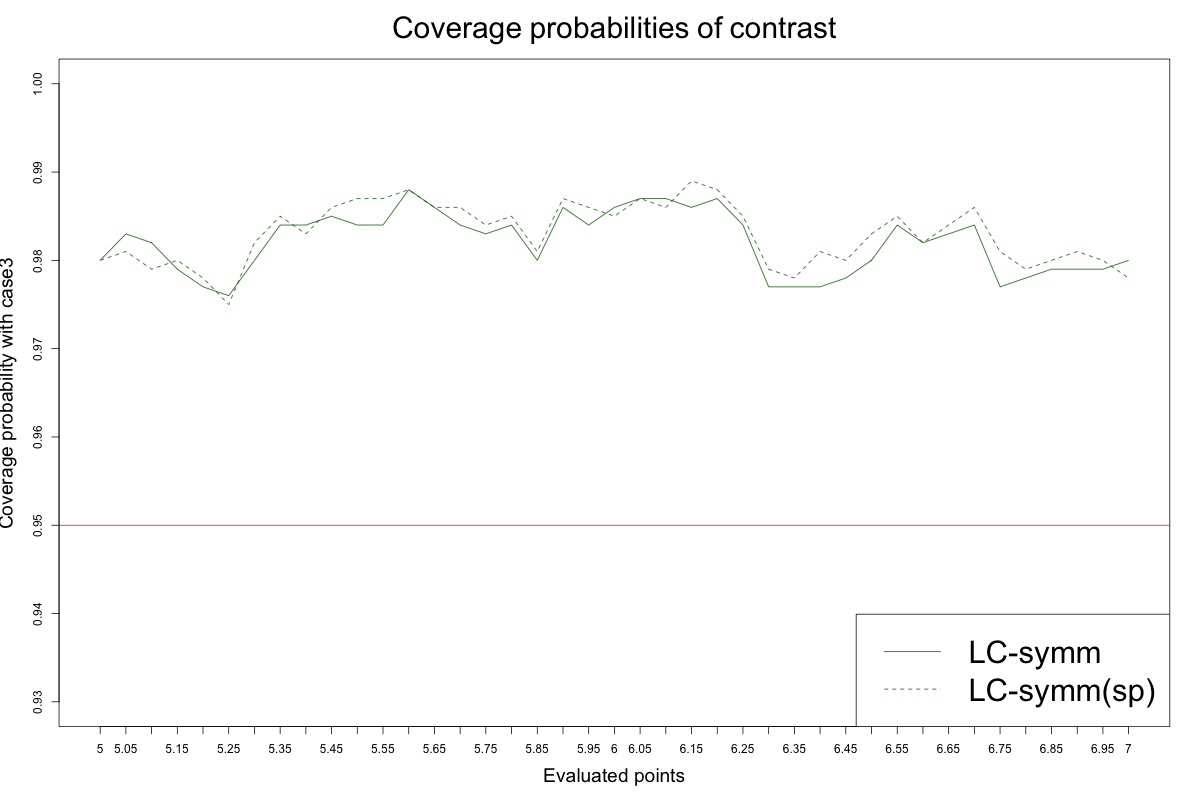}}
	\end{subfigure}\\
	\begin{subfigure}[$p_1-p_0$; C3; $n=2500$]
		{\label{cont2500mw}\includegraphics[width=70mm]{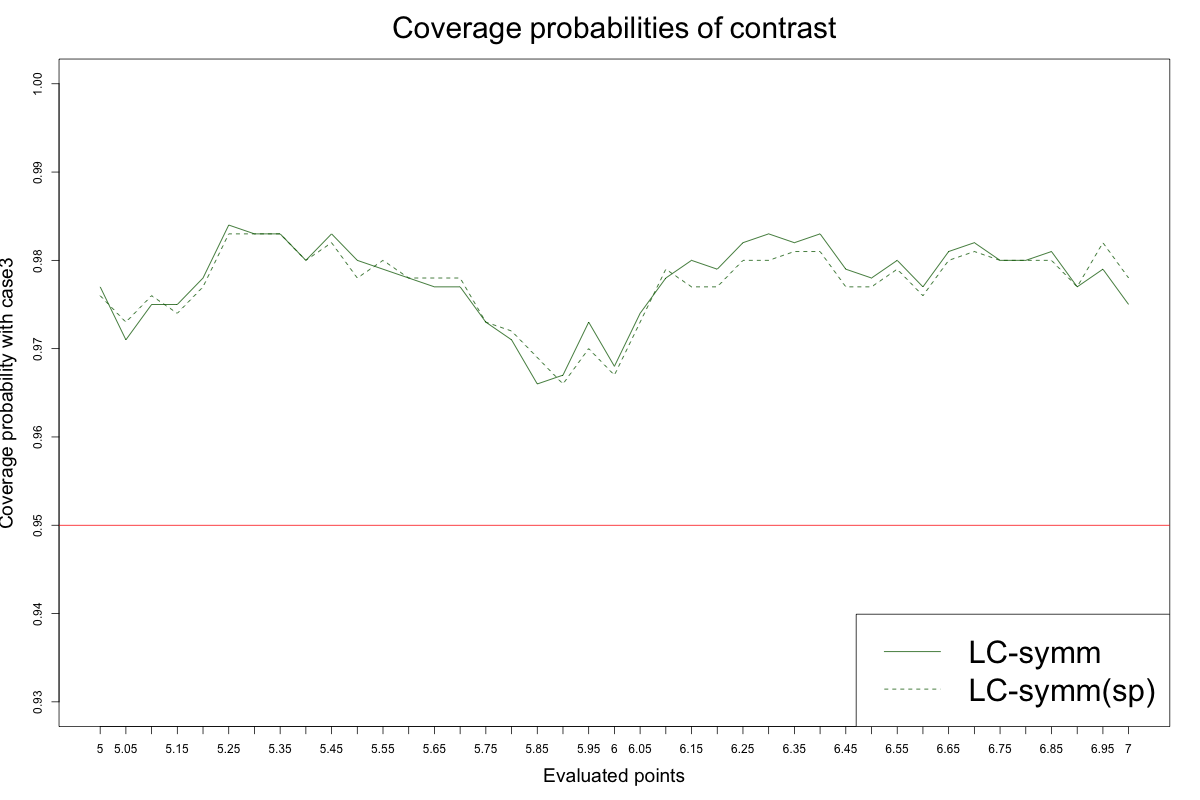}}
	\end{subfigure}
	\hspace{0.45cm}
	\begin{subfigure}[$p_1-p_0$; C3; $n=4000$]{\label{cont4000mw}\includegraphics[width=70mm]{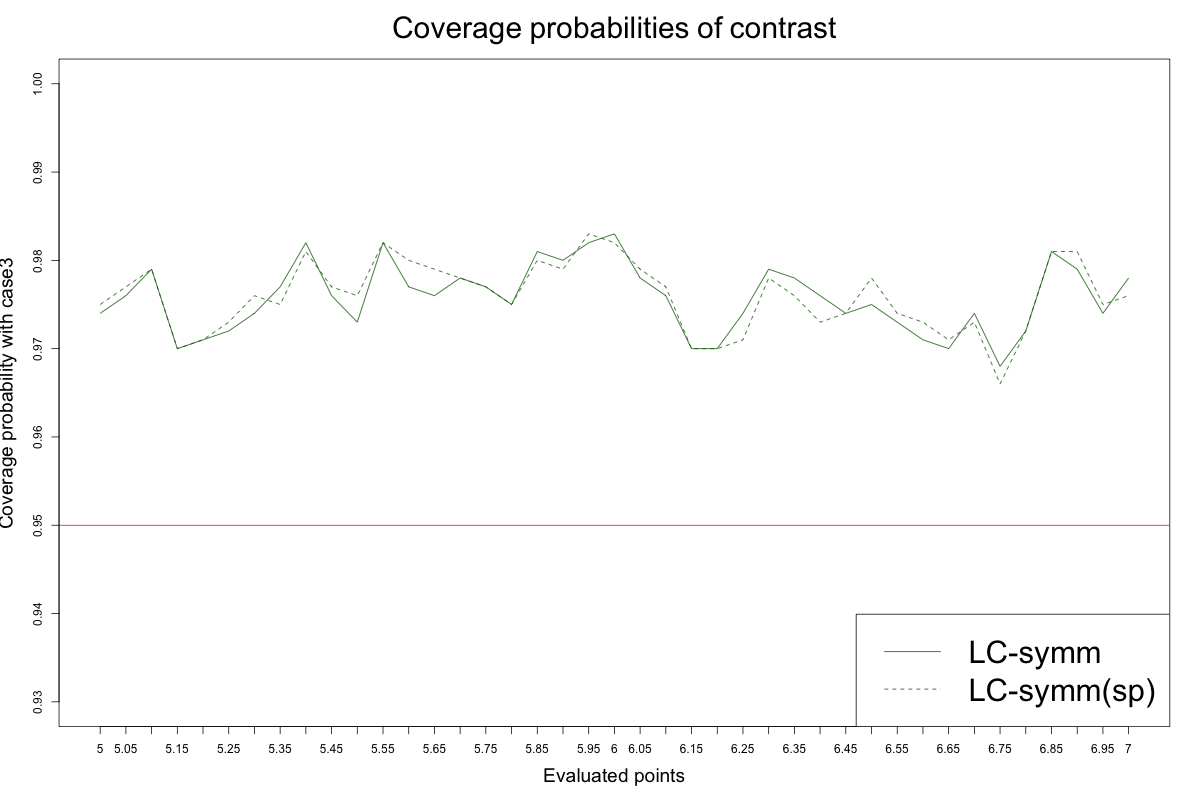}}
	\end{subfigure}\\
	\begin{subfigure}[$p_1-p_0$; C3; $n=6000$]
		{\label{cont6000mw}\includegraphics[width=70mm]{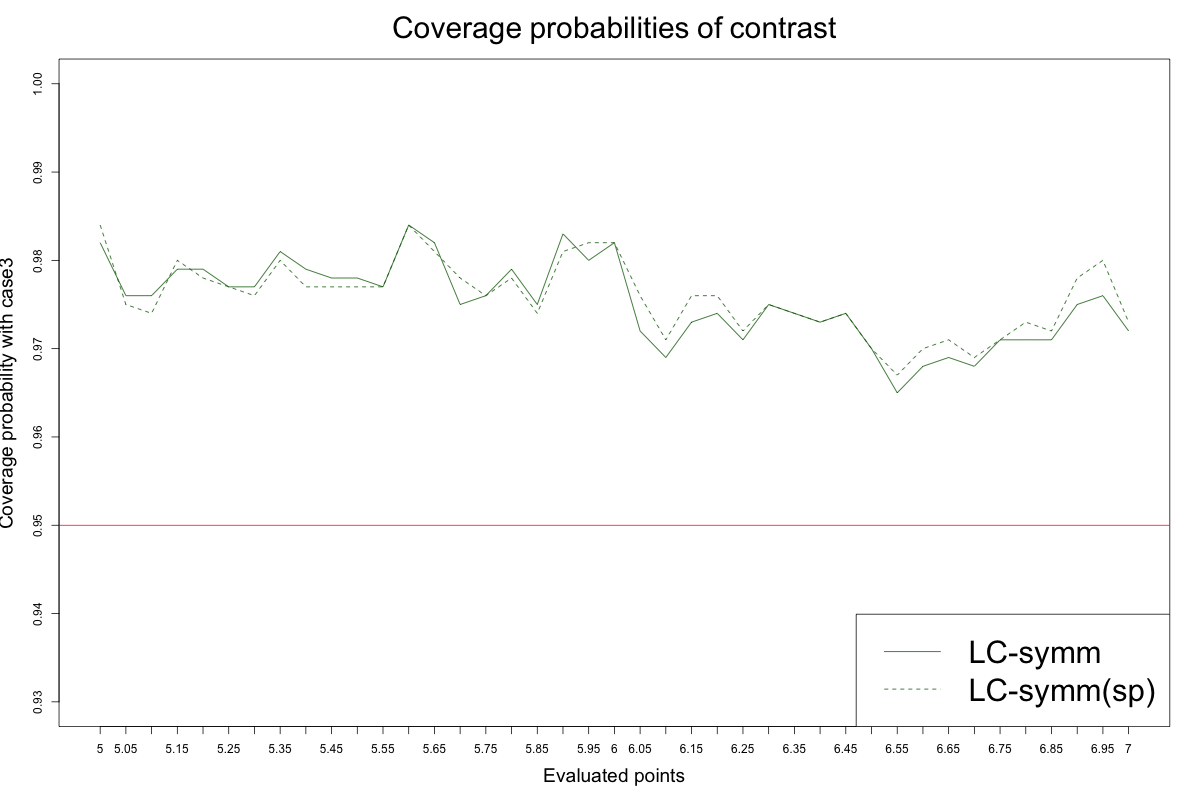}}
	\end{subfigure}
	\hspace{0.45cm}
	\begin{subfigure}[$p_1-p_0$; C3; $n=8000$]
		{\label{cont8000mw}\includegraphics[width=70mm]{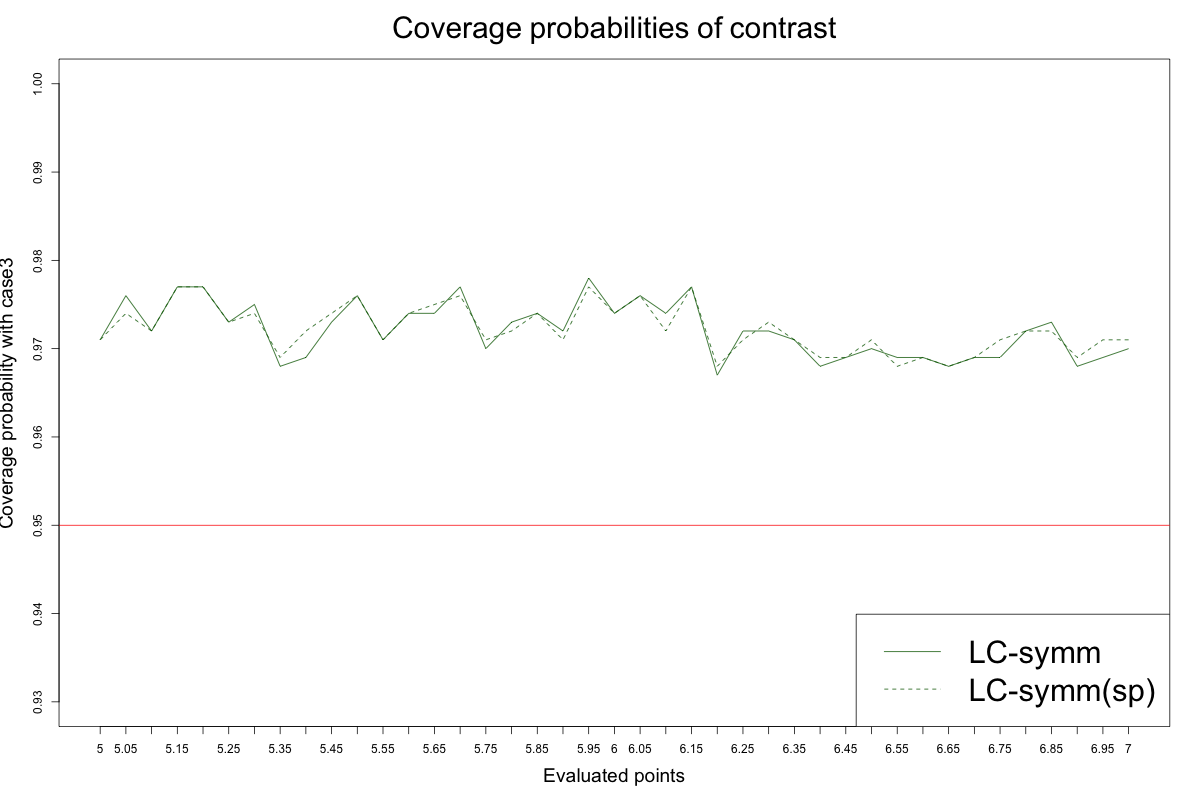}}
	\end{subfigure}
	\caption{\label{fig:contplots.mw}  Notational details can be found in Figure \ref{fig:contplots.ww}.}
\end{figure}

\begin{figure}
	\centering
	\begin{subfigure}[$\log(p_1/p_0)$; C1; $n=500$]{\label{cont500ww2}\includegraphics[width=70mm]{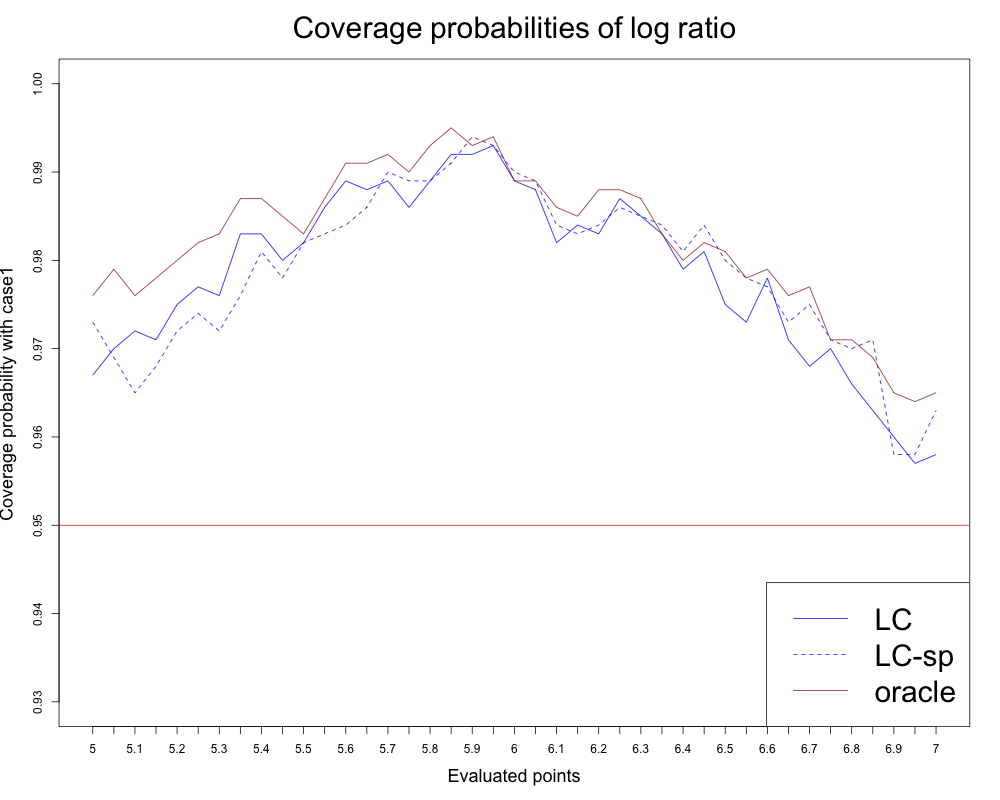}}
	\end{subfigure}
	\hspace{0.45cm}
	\begin{subfigure}[$\log(p_1/p_0)$; C1; $n=1000$]
		{\label{cont1000ww2}\includegraphics[width=70mm]{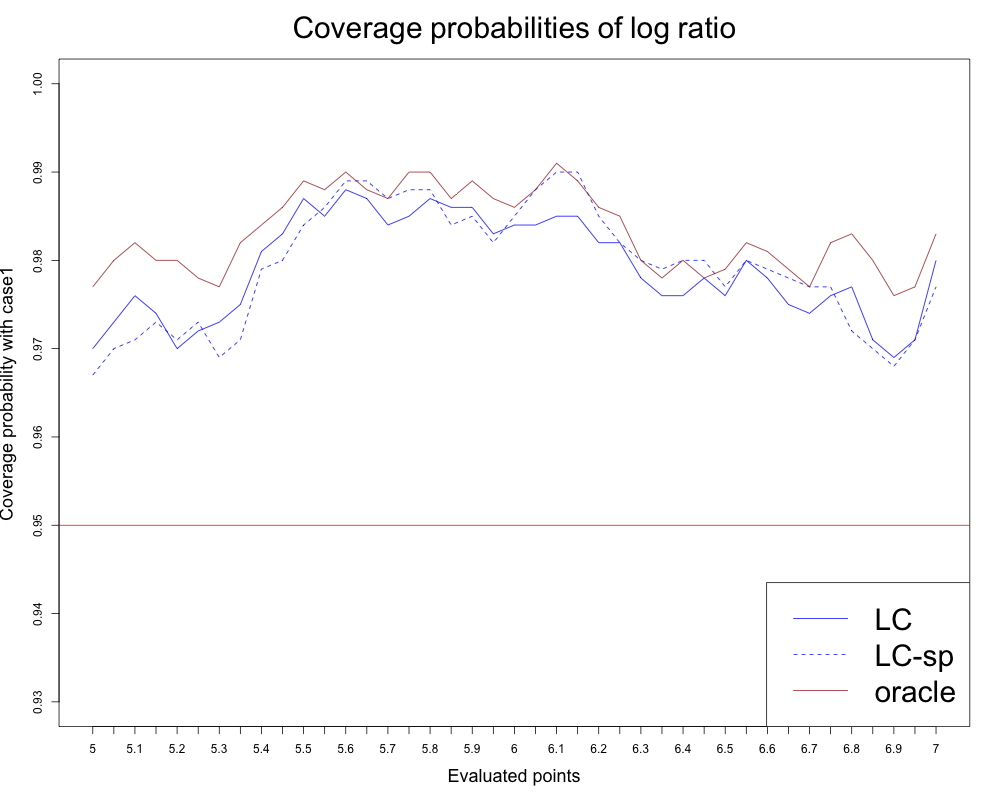}}
	\end{subfigure}\\
	\begin{subfigure}[$\log(p_1/p_0)$; C1; $n=2500$]
		{\label{cont2500ww2}\includegraphics[width=70mm]{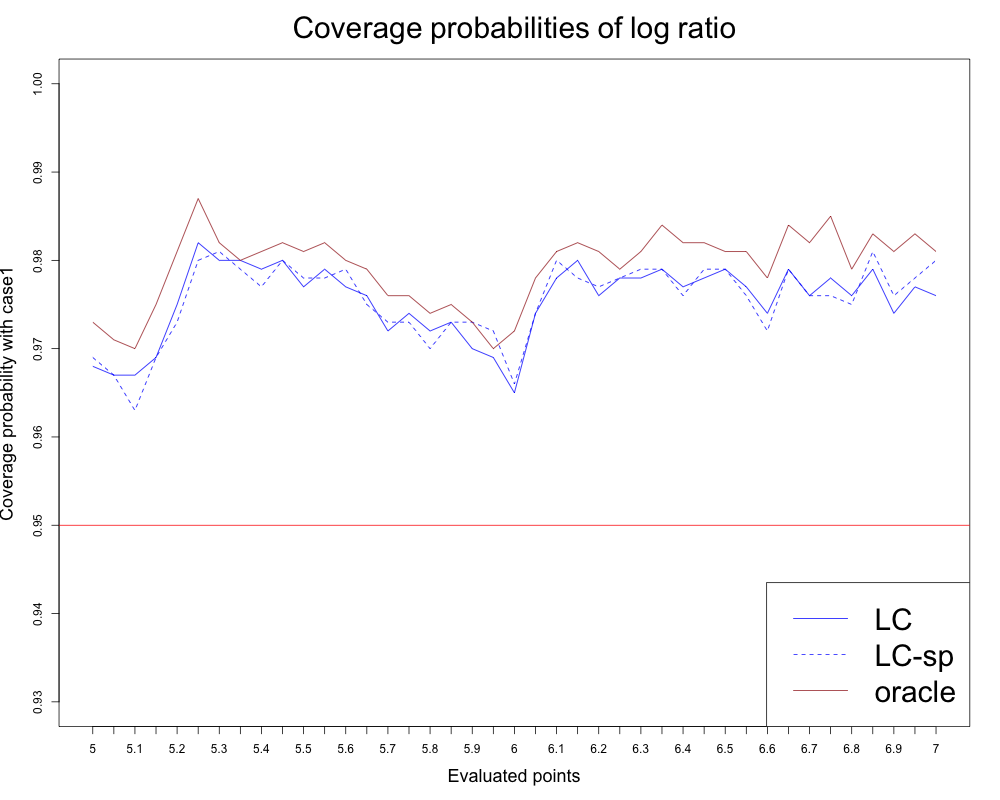}}
	\end{subfigure}
	\hspace{0.45cm}
	\begin{subfigure}[$\log(p_1/p_0)$; C1; $n=4000$]{\label{cont4000ww2}\includegraphics[width=70mm]{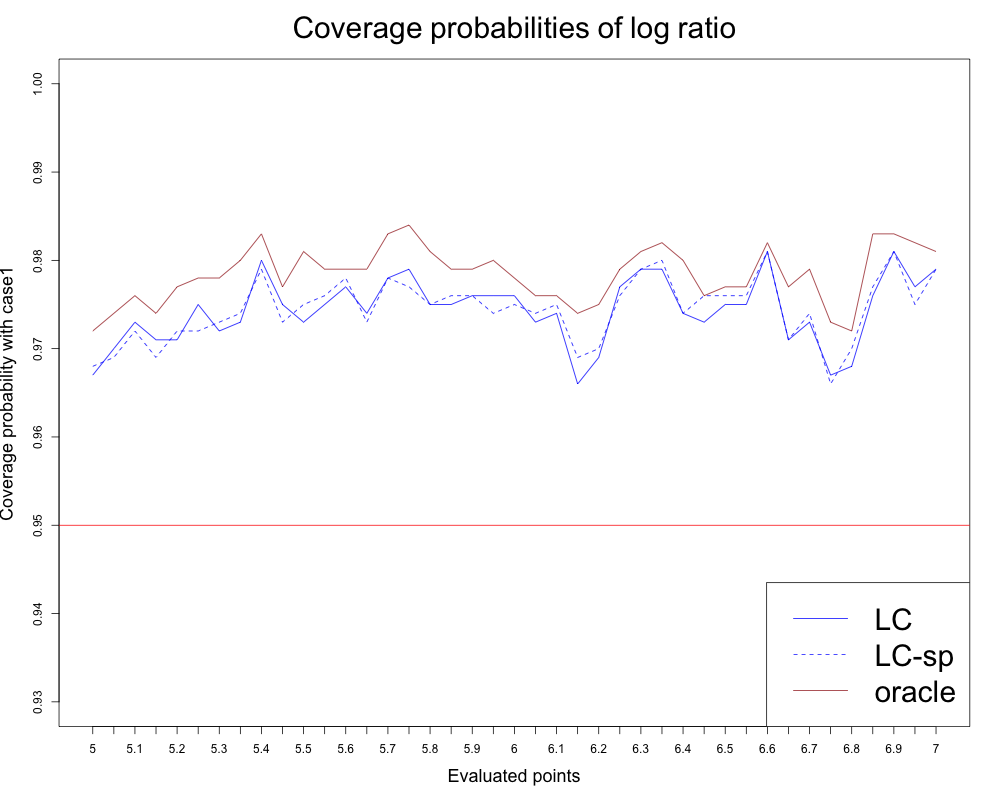}}
	\end{subfigure}\\
	\begin{subfigure}[$\log(p_1/p_0)$; C1; $n=6000$]
		{\label{cont6000ww2}\includegraphics[width=70mm]{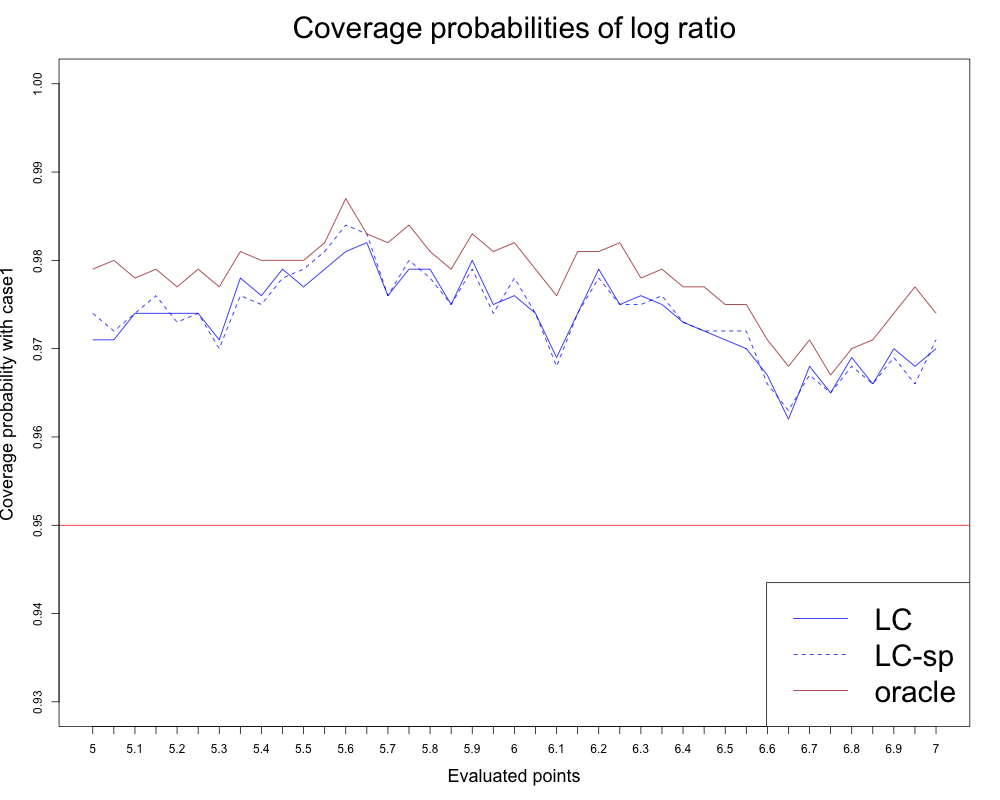}}
	\end{subfigure}
	\hspace{0.45cm}
	\begin{subfigure}[$\log(p_1/p_0)$; C1; $n=8000$]
		{\label{cont8000ww2}\includegraphics[width=70mm]{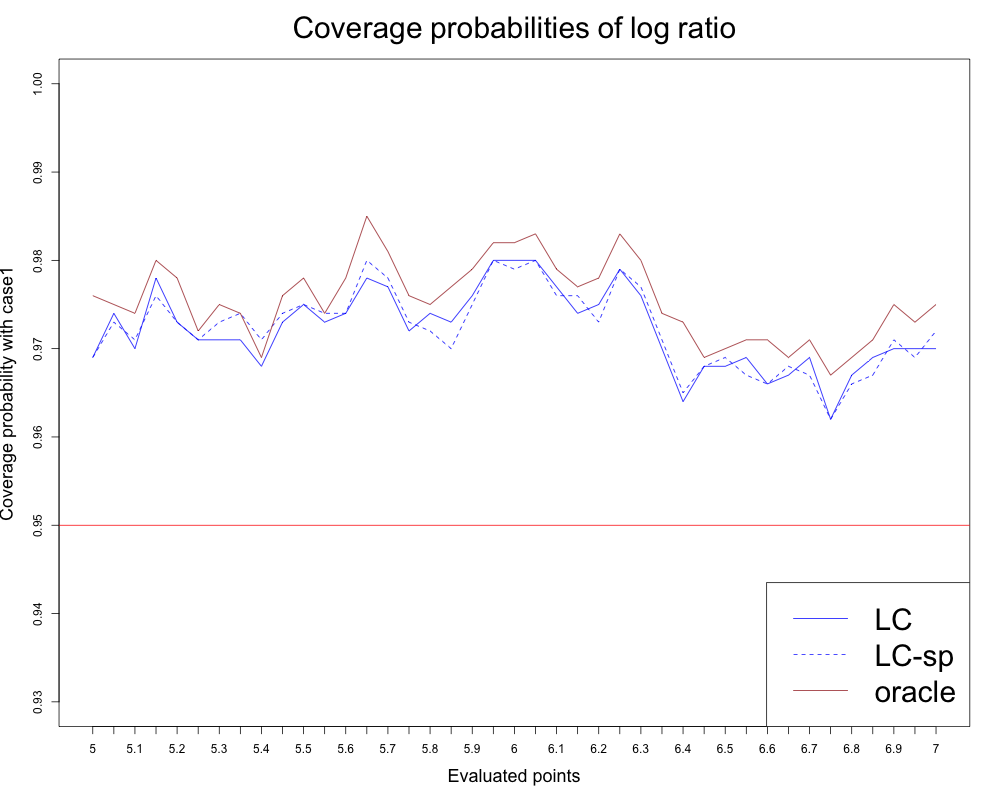}}
	\end{subfigure}
	\caption{\label{fig:contplot2s.ww} The above displays are coverage probabilities (for $\log(p_1/p_0)$ at 41 equally spaced point) with our proposed log-concave projection estimators' 95\% log-ratio CI's with the suggested tuning parameter $b=1/10$ which is labeled as LC (see Section \ref{subsec:tuning.param}).
	For the Case 1 where both nuisance functions are well-specified, we also use true value of $\chi_{\theta_a}$ to construct the oracle 95\% log-ratio CI which is labeled as oracle in the displays.
 Each subcaption describes the sample size, and each case of nuisance estimations (Case 1, 2, or 3 abbreviated to C1, C2, and C3, respectively).}
\end{figure}

\begin{figure}
	\centering
	\begin{subfigure}[$\log(p_1/p_0)$; C2; $n=500$]{\label{cont500wm2}\includegraphics[width=70mm]{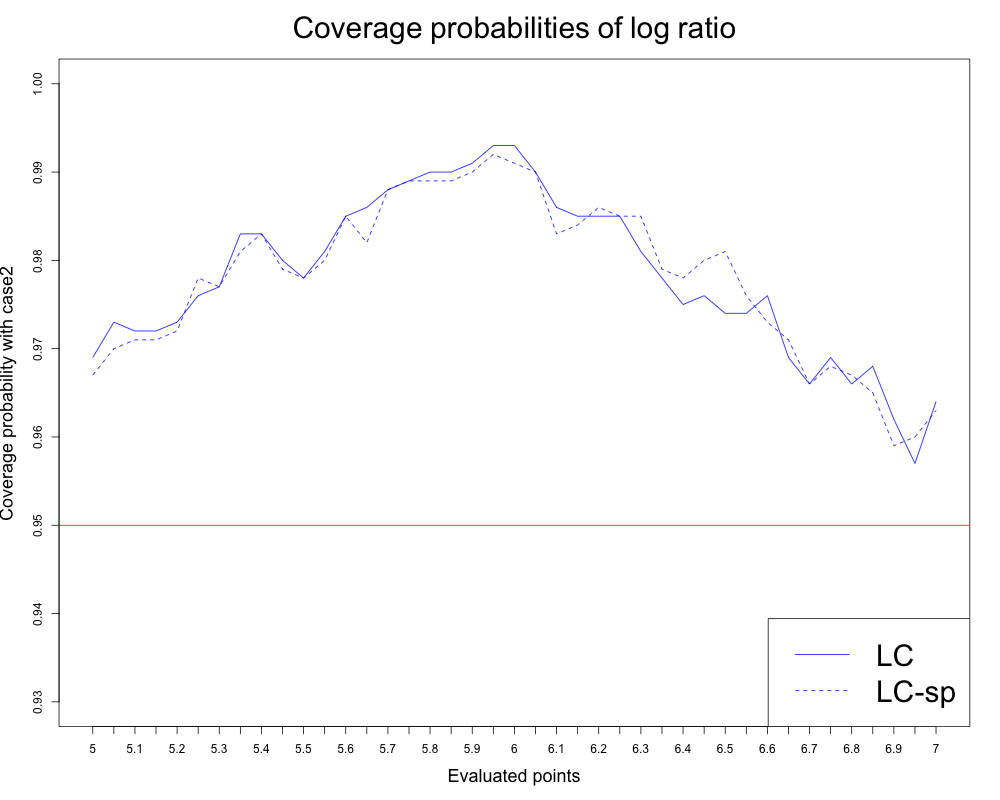}}
	\end{subfigure}
	\hspace{0.45cm}
	\begin{subfigure}[$\log(p_1/p_0)$; C2; $n=1000$]
		{\label{cont1000wm2}\includegraphics[width=70mm]{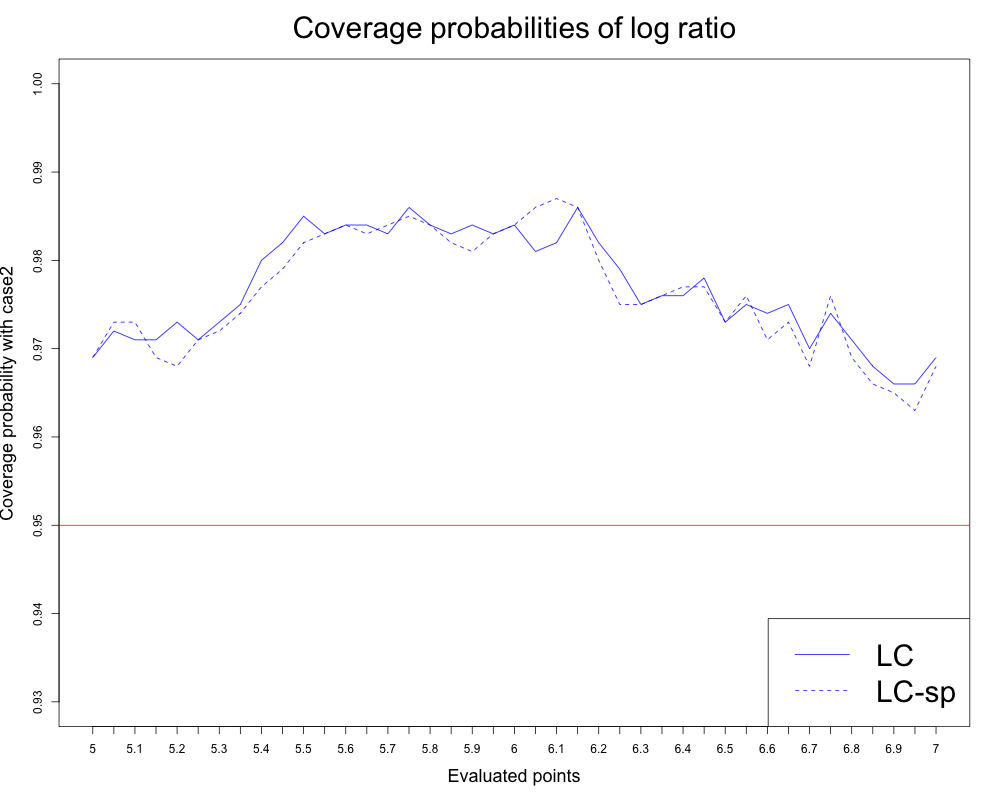}}
	\end{subfigure}\\
	\begin{subfigure}[$\log(p_1/p_0)$; C2; $n=2500$]
		{\label{cont2500wm2}\includegraphics[width=70mm]{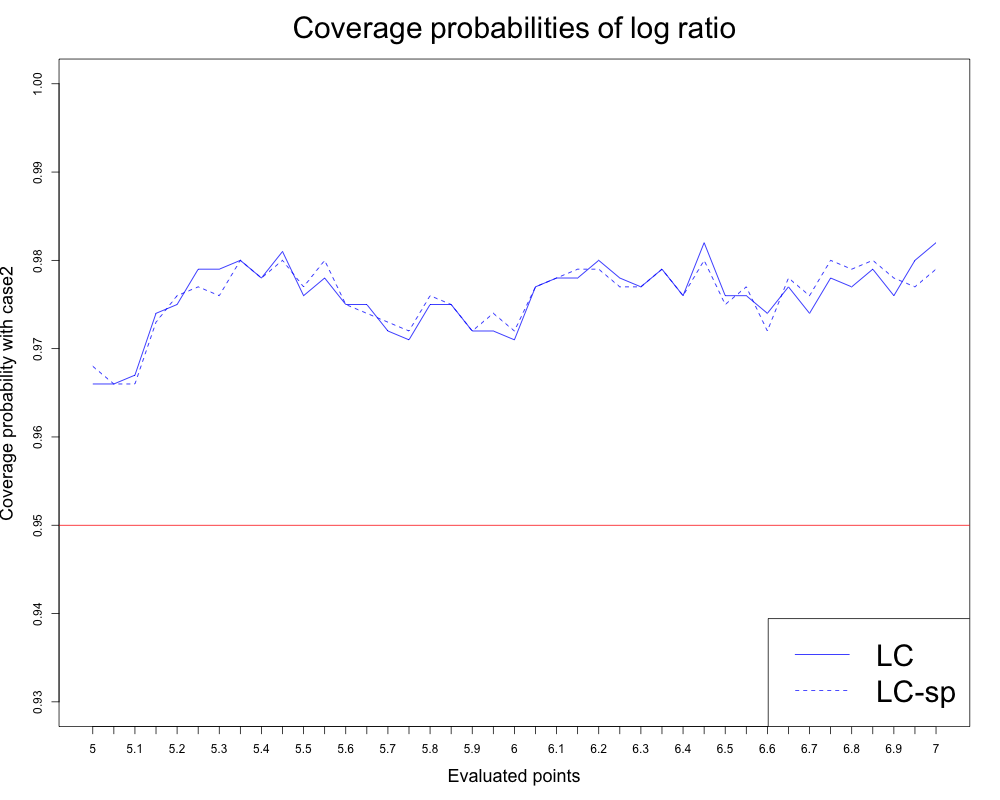}}
	\end{subfigure}
	\hspace{0.45cm}
	\begin{subfigure}[$\log(p_1/p_0)$; C2; $n=4000$]{\label{cont4000wm2}\includegraphics[width=70mm]{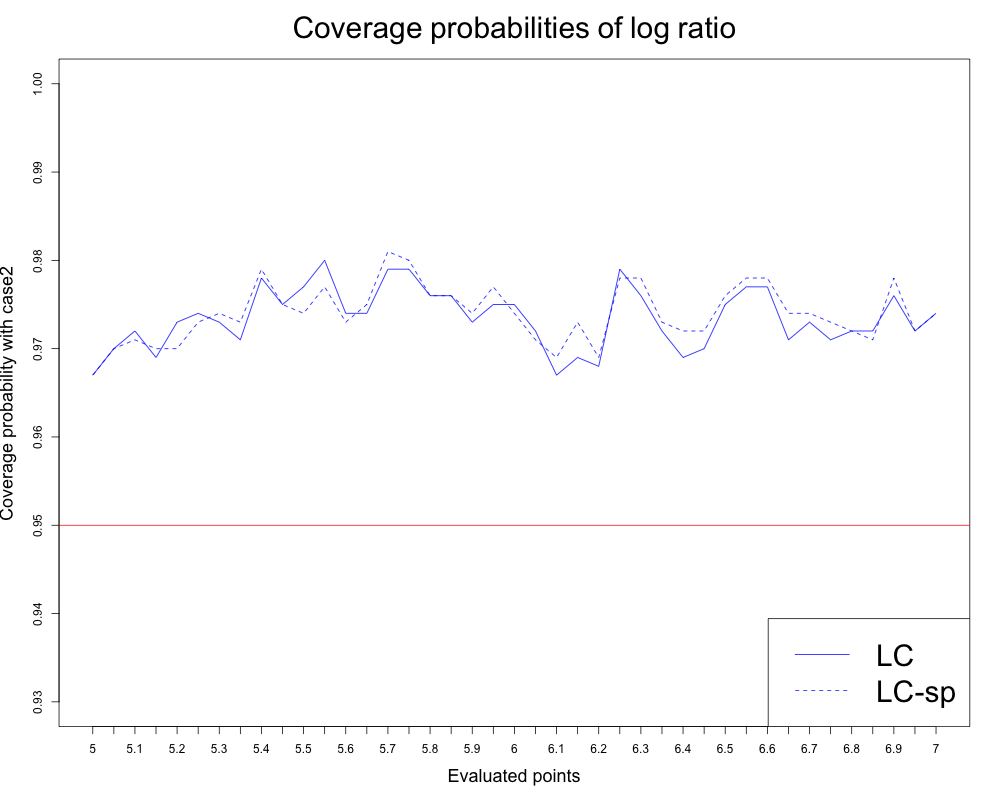}}
	\end{subfigure}\\
	\begin{subfigure}[$\log(p_1/p_0)$; C2; $n=6000$]
		{\label{cont6000wm2}\includegraphics[width=70mm]{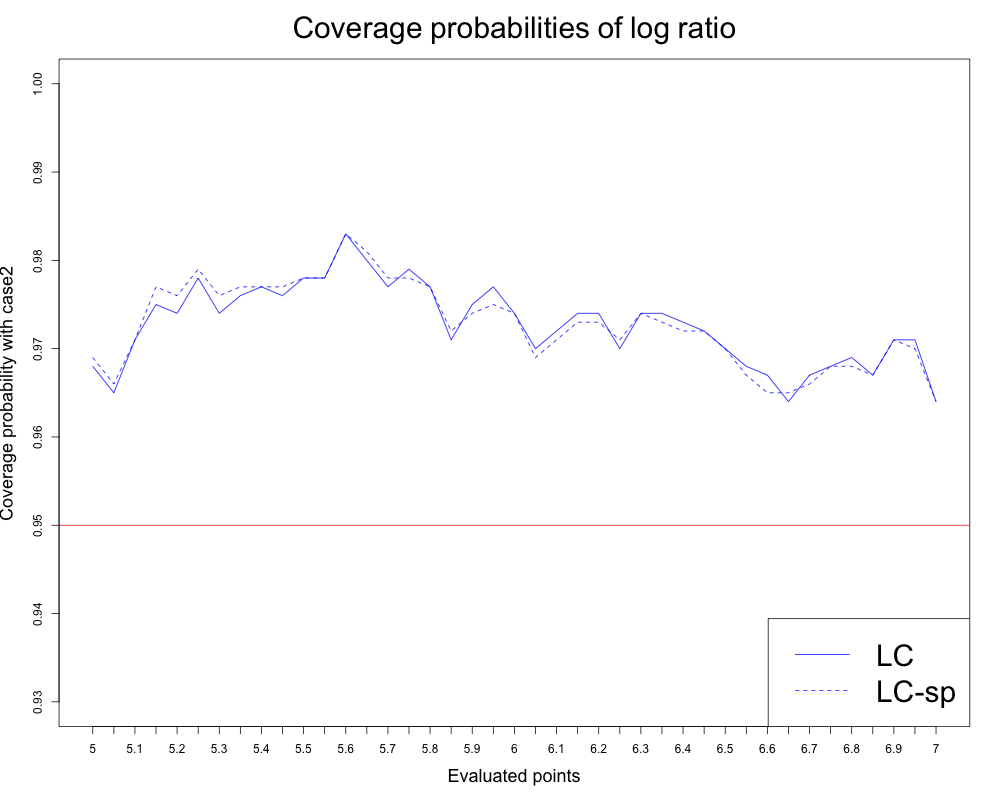}}
	\end{subfigure}
	\hspace{0.45cm}
	\begin{subfigure}[$\log(p_1/p_0)$; C2; $n=8000$]
		{\label{cont8000wm2}\includegraphics[width=70mm]{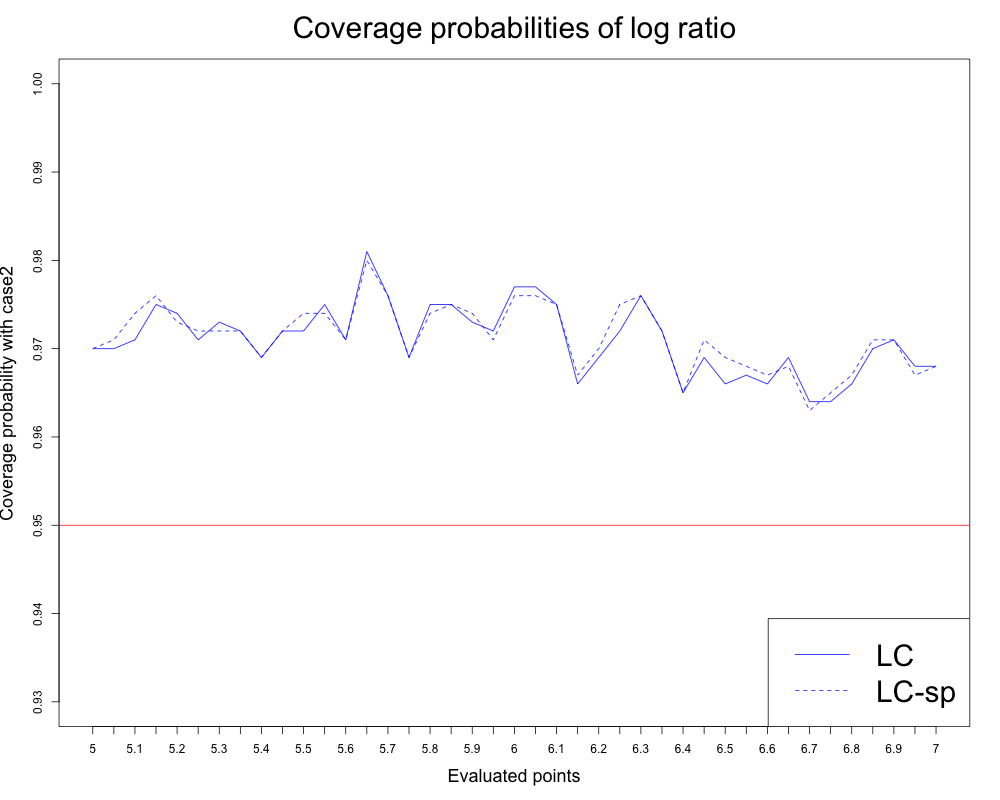}}
	\end{subfigure}
	\caption{\label{fig:contplot2s.wm} Notational details can be found in Figure \ref{fig:contplot2s.ww}.}
\end{figure}

\begin{figure}
	\centering
	\begin{subfigure}[$\log(p_1/p_0)$; C3; $n=500$]{\label{cont500mw2}\includegraphics[width=70mm]{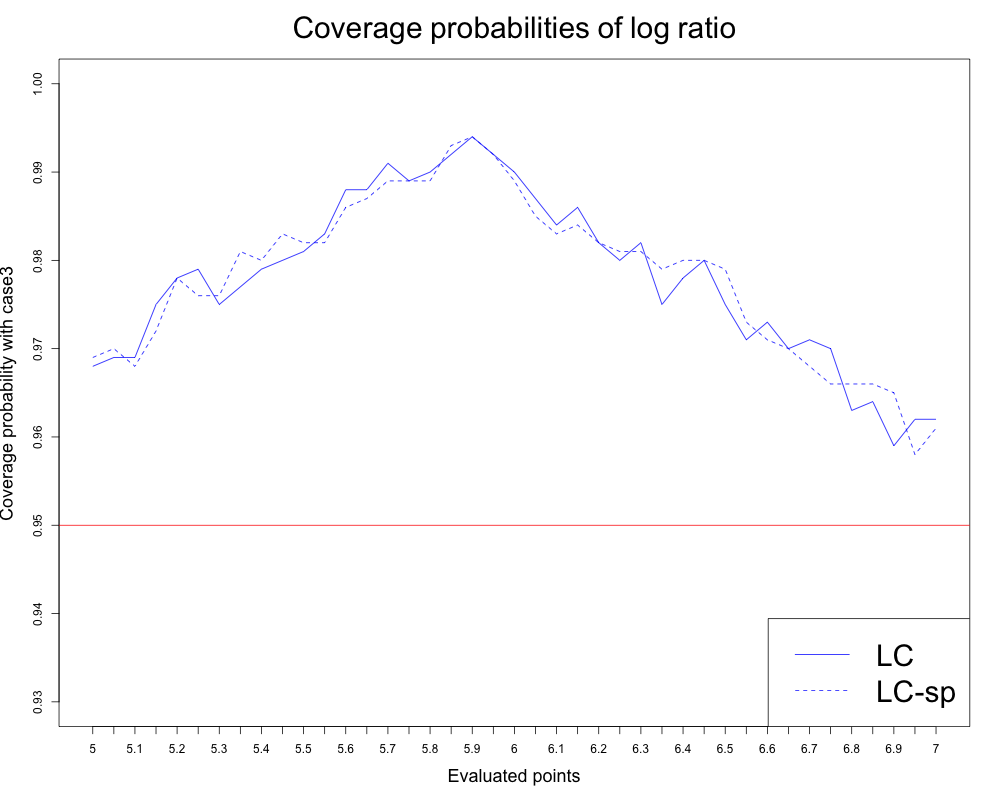}}
	\end{subfigure}
	\hspace{0.45cm}
	\begin{subfigure}[$\log(p_1/p_0)$; C3; $n=1000$]
		{\label{cont1000mw2}\includegraphics[width=70mm]{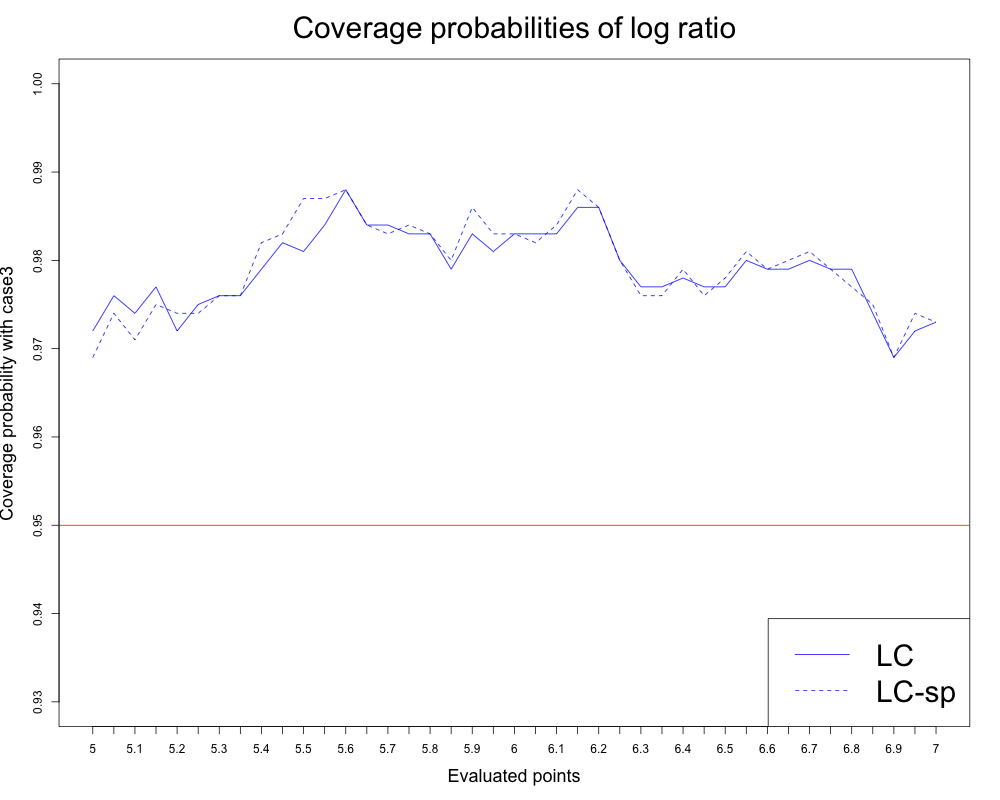}}
	\end{subfigure}\\
	\begin{subfigure}[$\log(p_1/p_0)$; C3; $n=2500$]
		{\label{cont2500mw2}\includegraphics[width=70mm]{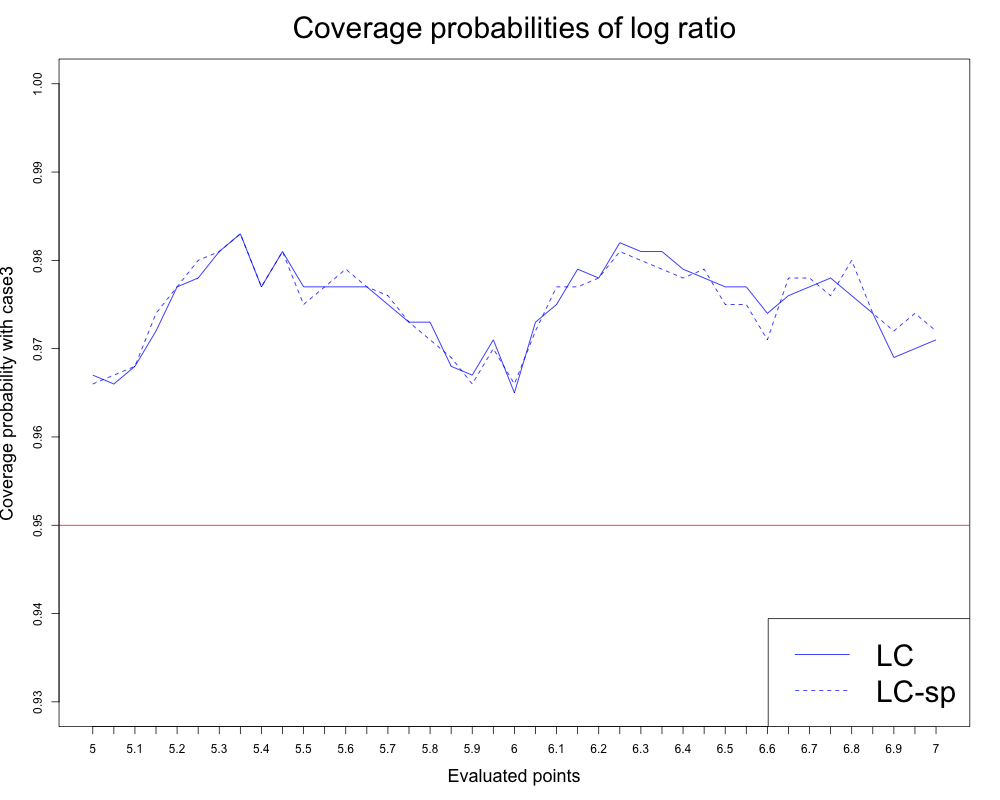}}
	\end{subfigure}
	\hspace{0.45cm}
	\begin{subfigure}[$\log(p_1/p_0)$; C3; $n=4000$]{\label{cont4000mw2}\includegraphics[width=70mm]{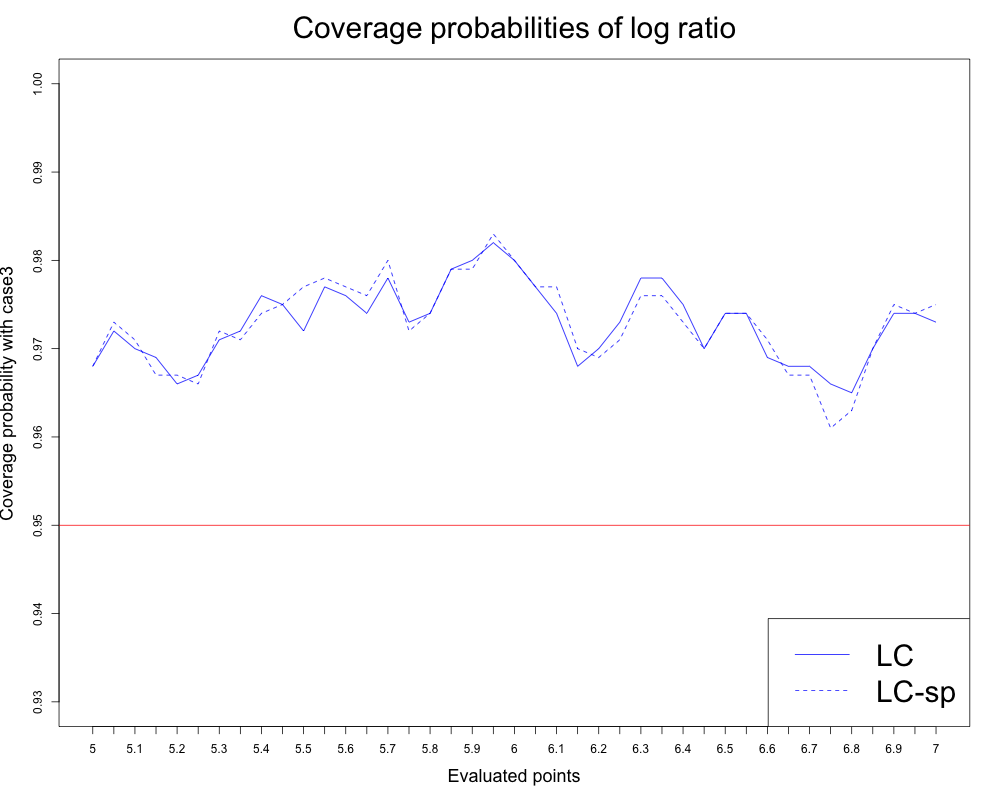}}
	\end{subfigure}\\
	\begin{subfigure}[$\log(p_1/p_0)$; C3; $n=6000$]
		{\label{cont6000mw2}\includegraphics[width=70mm]{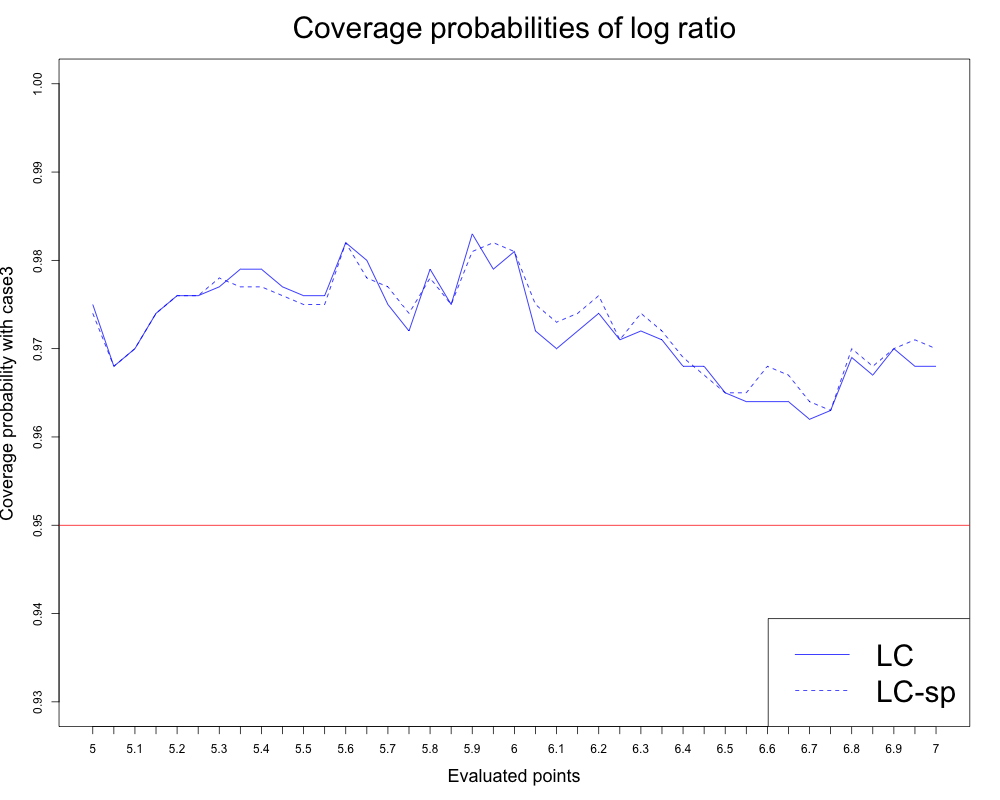}}
	\end{subfigure}
	\hspace{0.45cm}
	\begin{subfigure}[$\log(p_1/p_0)$; C3; $n=8000$]
		{\label{cont8000mw2}\includegraphics[width=70mm]{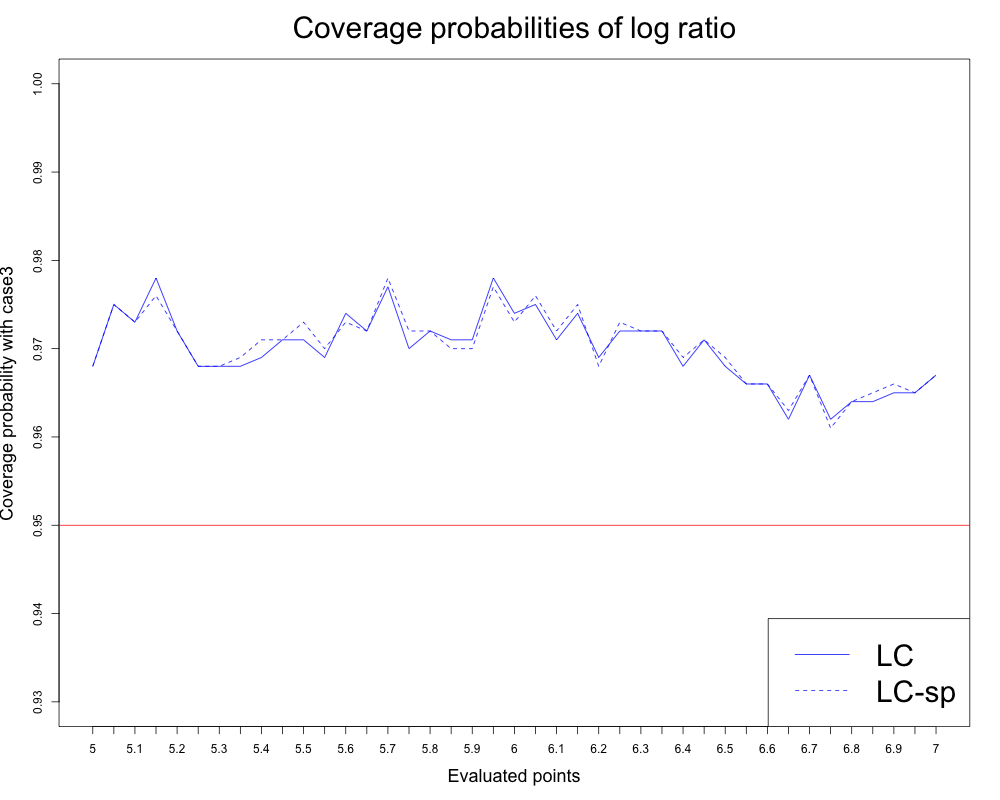}}
	\end{subfigure}
	\caption{\label{fig:contplot2s.mw}  Notational details can be found in Figure \ref{fig:contplot2s.ww}.}
\end{figure}

\clearpage
\bibliography{allref}


\end{document}